\documentclass[a4paper]{report}
\usepackage[left=3.00cm, right=2.50cm, top=2.30cm, bottom=3.50cm]{geometry}
\usepackage{amsmath}
\usepackage{amsfonts}
\usepackage{mathrsfs}  
\usepackage{amssymb}
\usepackage{amsthm}
\usepackage{verbatim}
\usepackage[square,sort&compress,comma,numbers]{natbib}
\usepackage[colorlinks=true,linkcolor=blue,citecolor=red]{hyperref}
\usepackage{graphicx}
\usepackage{epstopdf}
\usepackage[title]{appendix}
\usepackage{subcaption}

\graphicspath{ {./graph/} }

\newcommand{\bra}[1]{\langle #1\rvert}
\newcommand{\ket}[1]{\lvert #1\rangle}
\newcommand{\inner}[2]{\langle #1|#2\rangle}
\newcommand{\bracket}[3]{\left\langle #1\right\lvert#2\left\rvert#3\right\rangle}

\newcommand{\partiald}[2]{\frac{\partial #1}{\partial #2}}
\newcommand{\tr}{^{\mathrm{\tiny T}}}
\newcommand{\trace}{\mathrm{tr}}
\newcommand{\weyl}{\overset{\mathrm{Weyl}}{\longleftrightarrow}}
\newcommand{\mean}[1]{\langle #1 \rangle}

\newtheorem{lem}{Lemma}
\newtheorem{cor}{Corollary}
\newtheorem{pro}{Proposition}
\newtheorem{thm}{Theorem}
\newtheorem{defi}{Definition}
\newtheorem{axi}{Axiom}

\begin{document}
\title{Supremum of Entanglement Measure for Symmetric Gaussian States, and Entangling Capacity}
\author{PhD Thesis by Kao, Jhih-Yuan\\Supervisor: Chung-Hsien Chou\\National Cheng Kung University, Tainan, Taiwan
	}
\date{Thesis and Defense Completed in January, 2020\\ This edition uploaded to arXiv has been slightly edited}
\maketitle
\section*{Abstract}
In this thesis there are two topics: On the entangling capacity, in terms of negativity, of quantum operations, and on the supremum of negativity for symmetric Gaussian states.

Positive partial transposition (PPT) states are an important class of states in quantum information. We show a method to calculate bounds for entangling capacity, the amount of entanglement that can be produced by a quantum operation, in terms of negativity, a measure of entanglement. The bounds of entangling capacity are found to be associated with how non-PPT (PPT preserving) an operation is. A length that quantifies both entangling capacity/entanglement and PPT-ness of an operation or state can be defined, establishing a geometry characterized by PPT-ness. The distance derived from the length bounds the relative entangling capability, endowing the geometry with more physical significance.

For a system composed of permutationally symmetric Gaussian modes, by identifying the boundary of valid states and making necessary change of variables, the existence and exact value of the supremum of logarithmic negativity (and negativity likewise) between any two blocks can be shown analytically. Involving only the total number of interchangeable modes and the sizes of respective blocks, this result is general and easy to be applied for such a class of states.

Keywords: Entanglement, quantum operation, entangling capacity, Gaussian state, PPT

\chapter*{Acknowledgment}

First of all, I need to thank my family, who have been supportive of my decision. Next, I'd like to show my gratitude to my supervisor, Chung-Hsien Chou, whose advice and assistance have been very helpful, and who has given me lots of freedom. I'm also grateful to Ming-Yen Huang, and my not-so-many friends.

I'm thankful for the teaching or help from some teachers in our department, particularly Yong-Fan Chen, Yeong-Cheng Liang, Chopin Soo and Yueh-Nan Chen. Thanks to the all the committee members of my defense: (in in alphabetical order) Che-Ming Li, Chung-Hsien Chou, Feng-Li Lin, Hsi-Sheng Goan, Yueh-Nan Chen, Zheng-Yao Su, for willing to spend their precious time reading my thesis, attending my defense and giving me useful advice.

Special thanks to the people behind Collins dictionary, Wiktionary and other online dictionaries, so the reader doesn't have to suffer (as much) from my lacking pool of vocabulary.

Here I want to show my appreciation for the following musicians: (in alphabetical order) Aephanemer, Amorphis, Ayreon, Be'Lakor, Belzebubs, Borknagar, Disillusion, Enshine, Enslaved,\footnote{It's interesting to see how they have evolved since their founding around three decades ago---And in a good direction.} Finsterforst,\footnote{They are not very well-known, but I really love them. Their skill in songwriting is top-notch.} First Fragment, Fleshgod Apocalypse, Hail Spirit Noir, Hands of Despair, Hyperion, Insomnium,\footnote{The band the got me hooked to melancholic metal. Winter's Gate is one of my favorite albums.} Kanuis Kuolematon,\footnote{I think they find a perfect spot between doom and melodic death, and their vocals are great. Finnish melancholy at its best.} Moonsorrow,\footnote{V: H\"{a}vitetty is one of the best albums in my opinion.} Ne Obliviscaris,\footnote{They're one of the reasons why I became interested in prog metal. Portal of I is great, and Citadel is pure awesomeness.} Obscura, Opeth,\footnote{Before Watershed. Sorry, I just can't enjoy NewPeth. Still Life, Blackwater Park and Ghost Reveries always have an important place in my mind.} Periphery, Persefone,\footnote{Another reason why I started to dig prog metal. Spiritual Migration is mind-blowing.} Shade Empire, Shadow of Intent, Tribulation, Vorna,\footnote{Great Finnish doomy melodic death/black.} Wilderun, Wintersun, Wormwood, Xanthochroid, and some others I don't mention here. Their beautiful creation has been an integral part of my life. Seriously, if you're reading this right now, be sure to check their work, especially those I add footnotes to. It may take time to get used to, but their songs are masterfully crafted art.

BTW, I think this thesis will only be read by folks inside the physics community and perhaps some mathematician. If you are not one of them, why are you taking any interest in my humble work? Have I become a politician or celebrity?

\hfill{} January, 2020

\thispagestyle{empty}
\clearpage
\setcounter{page}{1}
\tableofcontents
\listoffigures
\chapter*{List of Symbols and Notations}
\begin{table}[hbtp!]
	\centering
	\begin{tabular}{cccc}
		$\mathbb{F}$ &  Section~\ref{sec:vec} & $\mathbb{Q}$ & Section~\ref{sec:set}\\
		CP, TP, HP    & Section~\ref{sec:liqo}  & $\mathcal{B}$, $\mathcal{L}$ & Section~\ref{sec:opnorm}\\   
		$\mathbb{R}$& Section~\ref{sec:orib} & $\mathbb{S}$ & Section~\ref{sec:set} \\
		$\rightarrow$, $\mapsto$ & Section~\ref{sec:map} & dom, ran & Section~\ref{sec:map}\\
		$\mathbb{C}$ & Section~\ref{sec:vec} & dim & Section~\ref{sec:span}\\
		$d(\cdot,\cdot)$ & Section~\ref{sec:normmet} & $||\cdot||$ & Section~\ref{sec:normmet}, \ref{sec:opnorm}, \ref{sec:inn}\\
		In calligraphy, e.g. $\mathcal{V}$ & Section~\ref{sec:vec} & ker & Section~\ref{sec:hom}\\
		$I$ & Section~\ref{sec:id} & $(\cdot|\cdot)$ & Section~\ref{sec:inn}, \ref{sec:HSinner}\\
		$\hat{\otimes}$ & Section~\ref{sec:HSTP} & $\perp$ & Section~\ref{sec:inn}\\
		$Q_W$, $Q_W^{-1}$ & Section~\ref{sec:weqde} & $W$ & Section~\ref{sec:wtr}\\
		$S$ & Section~\ref{sec:sym}, \ref{sec:qop} & $\mathrm{T}$, $\Gamma$ & Section~\ref{sec:TPT}\\
		$\weyl$ & Section~\ref{sec:weqde} & $L$ & Section~\ref{sec:liqo} \\
		$D_{1,\Gamma}$, $||\cdot||$ & Section~\ref{sec:d1g} & sup, inf, max, min & Section~\ref{sec:orib}\\
		$E$, $E_L$, $E_N$ & Section~\ref{sec:em}, \ref{sec:neg} & EC & Section~\ref{sec:EC}\\
		$\rho$ & Section~\ref{sec:deop} & $\mathbb{R}^n$, $\mathbb{C}^n$ & Section~\ref{sec:inn}\\
		$\mathbb{R}^{2n}$ & Section~\ref{sec:cm} & $\oplus$ & Section~\ref{sec:dip}\\
		$\otimes$ & Section~\ref{sec:tensor} & $\mathcal{H}^*$ & Section~\ref{sec:lfun}\\
		$\Pi$ & Section~\ref{sec:pro}, \ref{sec:id} & $\sigma$ & Section~\ref{sec:cov}\\
		$\Omega$ & Section~\ref{sec:syms} & $\dagger$ & Section~\ref{sec:ad}, \ref{sec:adl}\\
	    $||\cdot||_p$ & Section~\ref{sec:normop} & $U$ & Section~\ref{sec:isometry}\\
	    PPT & Section~\ref{sec:neg} & $\mathscr{T}$ & Section~\ref{sec:choi} \\
	    $\subset$, $\subseteq$ & Section~\ref{sec:set} & $L^p$, $L^2(\mathbb{R}^n)$ & Section~\ref{sec:l2s}\\
	    $H^\pm$, $\widetilde{H}^\pm$ & Section~\ref{sec:eig} & $\mathcal{U}$ & Section~\ref{sec:ur}
	\end{tabular}
\end{table}

\chapter{Introduction}
\label{ch:int}

\section{Overview}
First, let's have a quick overview of the two main topics of this thesis.
\subsection{Quantum Operations and Entangling Capability}

Entanglement has been found to be a useful resource for various tasks in quantum information \cite{Horodecki09}, so a problem arises: How to create entanglement? As an aspect of quantum states, this is the same as discerning what quantum processes, or quantum operations \cite{Cirac01,Gideke02,Bengtsson} can effectively produce this valuable resource, because operations govern how a state evolves or changes. There have been many studies on this problem, from various perspectives, such as how much entanglement an operation is able to produce/erase at most, on average, or per unit time \cite{Zanardi00,Dur01,Kraus01,Wolf03,Yukalov03,Leifer03,Bennett03,Chefles05,Linden09,Lari09,Campbell10} and what operations can produce the most entanglement (perfect entangler) \cite{Makhlin02,Zhang03,Rezakhani04,Cohen11}. Unitary operations are usually considered \cite{Zanardi00,Dur01,Kraus01,Leifer03,Bennett03,Chefles05,Linden09,Lari09,Rezakhani04,Cohen11}, while sometimes general quantum or Gaussian operations are investigated \cite{Wolf03,Campbell10}, with respect to various measures.

PPT states and operations have a profound importance in entanglement theory; for example, it was found no entanglement can be distilled from PPT states \cite{Horodecki98}. This class of states/operations is the subject of numerous studies, e.g. how to utilize them, and whether they are Bell non-local \cite{Peres96,Horodecki98,Peres99,Eggeling01,Vollbrecht02,Audenaert03,Pusey13,Vertesi14,Huber18}.

In this thesis, we quantify the capability of a quantum operation to produce entanglement by entangling capacity, defined as the maximal entanglement with respect to a given entanglement measure that can be created by a quantum operation \cite{Kraus01,Leifer03,Chefles05,Linden09,Campbell10}. It was found the existence of an ancilla, a system on which the operation isn't directly applied, may help boost entangling capacity \cite{Kraus01,Leifer03,Campbell10,Cohen11}.

We will obtain bounds (Proposition~\ref{pro}) for entangling capacities in terms of negativities. Since negativities bound teleportation capacity and distillable entanglement \cite{Bennett96,Horodecki98,Vidal02,Ishizaka04,Horodecki09}, our results give bounds for teleportation capacity and distillable entanglement that can be created by quantum operations.

Qualitatively, it is known that a PPT operation can't create negativity out of a PPT state \cite{Horodecki98,Cirac01}, and in this thesis, we would like to investigate the quantitative importance of PPT-ness of operations---A length, or norm associated with the bounds and PPT-ness, can be defined, by which, along with the distance or metric induced from it, we can provide entangling capacities of operations a geometric meaning. A strongly non-PPT operation, i.e. an operation that is ``longer'' in this norm, has the potential to create more negativity. In addition, the distance between operations can bound their relative entangling capability (Proposition~\ref{pro:2}). Therefore, this geometry of operations has physical importance. 

A method to find bounds of entangling capacity in terms of negativities was proposed in Ref.~\cite{Campbell10}. We will compare his approach with ours, and show that, albeit quite dissimilar in form, our bounds can lead to his. We are able to show the relation between separability and PPT-ness for unitary operations and pure states (Proposition~\ref{pro:sppt}).

Whenever there are bounds, it is natural to ask whether or when they can be saturated. Proposition~\ref{pro:ecs} will answer this question, and we will lay out a procedure to find the states with which to reach the bound.

\subsection{Symmetric Gaussian States and Suprema of Entanglement Measure}
Among various quantum systems of interest are continuous variable (CV) systems,\footnote{This is a term that is often used to refer to these systems. Later in this thesis we will give it a somewhat more rigorous definition, as the Hilbert space $L^2(\mathbb{R}^n)$.} of which Gaussian states have been a focus of researches, so has the entanglement and measures of information thereof \cite{Adesso04,Serafini05,Braunstein05,Adesso06,Serafini06,Fiurasek07,Adesso07,Giorda10,Adesso10,Serafini10,Weedbrook12,Adesso14}. Gaussian states are an important subset of CV states, not only because they have relatively less complex structures, but also because some important states in quantum optics, like coherent states and squeezed states, are Gaussian \cite{Adesso14}. 

The search for bounds of entanglement have been conducted for different kinds of systems in terms of various measures, e.g. entanglement of formation \cite{Bennett96} for two-mode Gaussian states \cite{Rigolin04,Nicacio14} and geometric measure of entanglement \cite{Wei03,Barnum01,Shimony95} for symmetric qubits \cite{Martin10,Aulbach10,Aulbach12}. 

The subject of interest in this thesis is symmetric Gaussian states: They're Gaussian states that are invariant under interchange of modes, which has garnered interests over the years \cite{Adesso04_2,Adesso05,Adesso08,Xu17,Kao16}. We will discuss the bipartite entanglement between two blocks of modes, both of which are part of the aforementioned symmetric modes. We're going to first talk about some basic properties of Gaussian states, and then of symmetric Gaussian states in Chapter~\ref{ch:gaf}, the primary subject of this work. Building upon this basis, we will show how to characterize a symmetric Gaussian state with proper variables so that a simple constraint can be established in Section~\ref{sec:par}. With this parameterization, we can then proceed to ascertain the suprema (least upper bounds) of entanglement (Proposition~\ref{pro:sup}), with respect to negativity and logarithmic negativity.

\section{Entanglement}
As this thesis focuses on quantum entanglement, here I will offer a brief outline of entanglement theory that is relevant to this thesis. For a comprehensive review the reader may refer to Ref.~\cite{Horodecki09}.
\subsection{Entanglement}
A bipartite state $\rho$ is said to be separable with respect to party A and B if it can be expressed as
\begin{equation}
\rho=\sum_i p_i \rho_i^A\otimes \rho_i^B,\;p_i\geq 0 \text{ and }\sum_i p_i=1,
\end{equation}
where $\rho_i^A$ and $\rho_i^B$ are both density operators. A bipartite state is entangled if it's not separable.

From the definition, we can see that if a local operation is performed on one party, then the state or density operator of the other won't be affected; on the contrary, if the state is entangled, we can expect otherwise. To experimentally verify this, since quantum measurements are probabilistic in nature, entanglement will be embodied in the correlation between measurement outcomes, e.g. Bell's inequality and CHSH inequality \cite{Bell64,Clauser69}. 

More generally, we anticipate entangled states to behave differently than separable states under local operations with classical communication, commonly shortened as LOCC \cite{Chitambar14}. The premise of LOCC is, they are operations easier and more practical (and cheaper) to perform in real life when the two parties are distant, so entangled states are expected to do certain things better under LOCC \cite{Plenio07,Horodecki09}.

\subsection{Entanglement Measure}\label{sec:em}
Since entanglement in itself is quite abstract, entanglement measures are one of the tools to help us quantify entanglement.
It is a non-negative function that ``measures'' the entanglement between two, or more parties, that is, its domain is the state, or density operator and the codomain is $[0,\infty)$. Some entanglement measures are more ``task-oriented,'' in that they quantify the amount of certain things that is needed for, or produced after a task. For example, the entanglement cost of a state $\rho$ reflects the best rate at which maximally entangled states can be converted to $\rho$ under LOCC. On the other hand, the entanglement of distillation of $\rho$ shows the best rate at which $\rho$ can be converted to the maximally entangled states \cite{Plenio07}.

There are some properties an entanglement measure $E$ is expected to satisfy:
\begin{itemize}
\item As mentioned above, it maps a density operator to a non-negative number.
\item $E(\rho)=0$ for all separable states $\rho$.
\item Suppose $S_i$ are sub-operations of an LOCC operation, and let $\overline{S_i(\rho)}$ denote the normalized state (i.e. having trace 1, so a valid density operator) after $S_i$, and $p_i$ the probability that $S_i$ is carried out. Then 
\begin{equation}
\sum_i p_i E(\overline{S_i(\rho)})\leq E(\rho),
\end{equation}
namely, an LOCC on average shouldn't create entanglement.
\end{itemize}
Some other properties may be satisfied by or required for certain measures, but we won't pursue them here \cite{Plenio07}. Those properties prompt an axiomatic approach to entanglement measures, in which a function that obeys them are deemed an entanglement measure. 

Such a concept can be generalized, for example, to Gaussian LOCC (GLOCC), which is an LOCC that maps Gaussian states (to be defined in Chapter~\ref{ch:gaf}) to Gaussian states: We may define a entanglement measure specifically for Gaussian states which don't increase under a GLOCC instead of any LOCC \cite{Gideke02,Wolf04}. 

The entanglement measures we're going to use are negativity and logarithmic negativity, denoted by $E_N$ and $E_L$ respectively. Even though they're (comparatively) easy to calculate, they still have operational meanings, as bounds for teleportation capacity and distillable entanglement \cite{Vidal02}. We will introduce them formally in Section~\ref{sec:neg}.

\section{How This Thesis Is Organized}
\begin{itemize}
\item Chapter~\ref{ch:math} will focus on some mathematics behind quantum mechanics. It will contain material that is missed out or not explained rigorously in standard physics lectures, and the reader may consider glancing through or skipping over some parts. If the reader already has a basic understanding of mathematics for quantum mechanics/information or some mathematical analysis, then Chapter~\ref{ch:math} isn't a must-read, but some symbols are introduced there, so he/she is encouraged to give it a quick look.
\item Chapter~\ref{ch:op} includes several topics on operators. They're essential for developing the tools required in later chapters. 
\item Chapter~\ref{ch:li} is centered around linear mappings on operators, in particular quantum operations. The concepts discussed in this chapter will be utilized heavily in Chapter~\ref{ch:EC}.
\item In Chapter~\ref{ch:EC} I will present one of the main topics of this thesis, concerning the entangling capacity of quantum operations.
\item Chapter~\ref{ch:gaf} is the preparation for Chapter~\ref{ch:SGS}, where theories for treating so-called CV states, Gaussian states and (multi)symmetric Gaussian states are explained.
\item In Chapter~\ref{ch:SGS} the other primary subject of this thesis is presented: How to find the suprema of block entanglement for symmetric Gaussian states, and their implications.
\item Chapter~\ref{ch:con} concludes this thesis.
\end{itemize}
Since this thesis has two quite distinct topics, not all material in Chapter~\ref{ch:op} and \ref{ch:li} are mandatory to understand either one of the main subjects; for example, Weyl quantization isn't needed for Chapter~\ref{ch:EC}, and the reader can make his/her own judgment to pass over certain pieces therein. 

A well-established and documented mathematical/physical important result will be presented as ``Theorem;'' those derived by me\footnote{Or those that have been already discovered by others but I'm not aware of. Undesirable, but I can't completely deny such a possibility.} will be presented as ``Proposition,'' ``Lemma'' or ``Corollary.''

Lots of materials in this thesis are based on Ref.~\cite{Kao18} and a paper that hasn't been published as yet, and Ref.~\cite{Kao16} to a much lesser degree, but I added a substantial amount of content (like Chapter~\ref{ch:math}) and made some changes and corrections, so it should be more rigorous and self-contained. 

\chapter{Mathematical Premier}\label{ch:math}
In this chapter, I'll give a brief introduction to some mathematical fundamentals of quantum mechanics. The reader may already be familiar with some of them, but here I'll adopt a somewhat more rigorous, or mathematical language. For some readers this may seem obscure and unconventional, but I think this can provide valuable perspectives on quantum mechanics, as it's built upon these interlocking cogs of mathematics. Of course it's impossible (and more importantly, far beyond my capacity) to give a comprehensive look at the math behind quantum mechanics, and a large part of it is nothing but a quick overview or outline. A great part of the material in this chapter is extracted and adapted from Ref.~\cite{Roman,Loomis,PapaRudin,GrandpaRudin,Einsiedler,HallQ,Woit,deGosson}.
\section{Sets and Mappings}
\subsection{Set}\label{sec:set}
Informally, a \textbf{set} is a collection of arbitrary things, and those things are called members or elements of the set. In this thesis, I use curly brackets to contain members of a set. For example, the set of my family members $=\{$my parents, my sibling, me$\}$. If $s$ is a member of a set $\mathbb{S}$, we write $s\in\mathbb{S}.$ Sometimes to avoid repetition we also call a set a ``collection'' or ``family,'' e.g. a collection (family) of sets instead of a set of sets. 

Usually some conditions are put on a set and its members, where a colon : (or a vertical bar $|$) is used to indicate the conditions. For example, the set of all rational numbers $\mathbb{Q}=\{\frac{i_1}{i_2}:i_1\in\mathbb{Z},i_2\in\mathbb{Z}\backslash \{0\} \}$\label{sym:Q}, where $\mathbb{Z}$ is the set of integers.  

A set $\mathbb{S}_2$ is called a \textbf{subset} of a $\mathbb{S}_1$, denoted by $S_2\subseteq S_1$ if all the members of $\mathbb{S}_2$ are in $\mathbb{S}_1$, i.e. $s_i\in \mathbb{S}_1$ $\forall s_i\in\mathbb{S}_2$; $\forall$ means ``for all.'' $\mathbb{S}_2$ is a proper subset of $\mathbb{S}_1$ if $\mathbb{S}_2\subseteq \mathbb{S}_1$ but $\mathbb{S}_2\neq \mathbb{S}_1$, denoted by $\mathbb{S}_2\subset\mathbb{S}_1.$ For instance, $\mathbb{Z}\subseteq \mathbb{Q}$ and $\mathbb{Z}\subset \mathbb{Q}$. Note if $s\in\mathbb{S}$, then $\{s\}\subseteq \mathbb{S}.$

Very often the statement ``A is B'' or ``A are B'' implies set inclusion: For example, ``sky is blue'' implies sky is in the set of blue things and ``integers are rational numbers'' means integers are a subset of (the set of) all ration numbers \cite{Loomis}.

\subsection{Mapping}\label{sec:map}
For a function/map/mapping $f$ from a set $\mathbb{S}_1$ to $\mathbb{S}_2$, denoted by $f:\mathbb{S}_1\rightarrow \mathbb{S}_2$, $\mathbb{S}_1$ is called its \textbf{domain} and $\mathbb{S}_2$ its \textbf{codomain}, where the domain of $f$ is denoted by $\mathrm{dom}f$. The set $f(\mathbb{S}_1):=\{f(x):x\in \mathbb{S}_2\}$, called the \textbf{range} or \textbf{image} of $f$, is a subset of the codomain, denoted by $\text{ran}f$. Note $f:\mathbb{S}_1\rightarrow \mathbb{S}_2$ is read as a mapping/function $f$ \emph{from} $\mathbb{S}_1$ \emph{to} $\mathbb{S}_2$ or $f$ \emph{on} $\mathbb{S}_1$ \emph{into} $\mathbb{S}_2$ \cite{Loomis}. 

For $f:\mathbb{S}_1\rightarrow \mathbb{S}_2$, if $\mathbb{S}_1'\subseteq \mathbb{S}_1$, $f\upharpoonright_{\mathbb{S}_1'}$, the \textbf{restriction} of $f$ to $\mathbb{S}_1'$ or $f$ restricted to $\mathbb{S}_1'$, is $f$ with its domain reduced to $\mathbb{S}_1'$, or simply written as $f_{\mathbb{S}_1'}$ \cite{Loomis}. The \textbf{composition} of two mappings $f$ and $g$ is the mapping $(f\circ g)(x):=f(g(x))$ if the image of $g$ is a subset of the domain of $f$.

A mapping $f:\mathbb{S}_1\rightarrow \mathbb{S}_2$ is \cite{Loomis}
\begin{itemize}
\item \textbf{injective} if it's one-to-one.
\item \textbf{surjective} if the image of $f$ is the same as the codomain;
\item \textbf{bijective} if it's injective and surjective. It's therefore a one-to-one correspondence between its domain and codomain.
\end{itemize}

We can clarify these concepts via the following theorem:
\begin{thm}\label{thm:inj}
$f:\mathbb{S}_1\rightarrow \mathbb{S}_2$ is injective if and only if there exists a mapping $f^{-1}:\mathbb{S}_2\rightarrow \mathbb{S}_1$ such that $f^{-1}\circ f$ is the identity mapping on $\mathbb{S}_1$.\footnote{A mapping $g:\mathbb{S}\rightarrow \mathbb{S}$ is said to be an identity mapping if $g(x)=x$ for all $x\in \mathbb{S}$. We will talk about the identity mapping on a vector space in Section~\ref{sec:id}.}
\end{thm}
\begin{proof}\hfill{}\\
``If'': If $f$ were not injective, then there would exist $x,y\in\mathbb{S}_1$ such that $x\neq y$ and $f(x)=f(y):=z$. Whether we define $f^{-1}(z)$ as $x$ or $y$, $f^{-1}\circ f$ can't be an identity mapping.   

\noindent ``Only if'': For every $y\in\text{ran}f\subseteq \mathbb{S}_2$, there's one and only one $x\in\mathbb{S}_1$ that satisfies $f(x)=y$, and we can define $f^{-1}(y)=x$ for such $y$. For other elements in $\mathbb{S}_2$, we can assign any values for $f^{-1}$, as long as the values are in $\mathbb{S}_1$. Hence such an $f^{-1}$ exists.
\end{proof}

Consequently, a mapping $f:\mathbb{S}_1\rightarrow \mathbb{S}_2$ is bijective if and only if there exists a mapping $f^{-1}:\mathbb{S}_2\rightarrow \mathbb{S}_1$ such that $f^{-1}\circ f$ and $f\circ f^{-1}$ are the identity mappings on $\mathbb{S}_1$ and $\mathbb{S}_2$ respectively.

We also use the notation $f:x\mapsto f(x)$ (read $f$ maps $x$ to $f(x)$)\label{sym:mapsto} to denote a mapping, especially when there's no need to give a specific name to the mapping. For example, $x\mapsto x^2$ is identical to the function $f(x)=x^2$. It's understood $x$ is any element in the domain of such a mapping. If the domain (and/or codomain) should be clarified, we write $\mathbb{S}_1\ni x\mapsto f(x)\in\mathbb{S}_2$.

\subsection{Index, Countability, Sequences and Series}\label{sec:index}
To indicate the members of a set $\mathbb{S}$ we often use an index set $\mathbb{K}$. An indexing function is a surjective mapping $\mathbb{K}\rightarrow\mathbb{S}$, so for any member $s\in\mathbb{S}$ there is always a corresponding $k\in\mathbb{K}$. 

We call two sets $\mathbb{S}_1$ and $\mathbb{S}_2$ equivalent if there exist a bijective mapping between them, denoted by $\mathbb{S}_1\sim \mathbb{S}_2$. Define
\begin{equation}
\mathbb{N}_n:=\{i:i\in\mathbb{N}\text{ and }i\leq n\} \text{ for } n\in \mathbb{N}\cup\{0\}.
\end{equation}
In other words, $\mathbb{N}_n$ is the subset of $\mathbb{N}$ whose elements are no larger than $n$, and $\mathbb{N}_0$ is the empty set. Then a set $\mathbb{S}$ is said to be\footnote{There are more than one ways of describing countability. Here I follow references like Ref.~\cite{Roman,Munkres}: A set is countable if it's finite or countably infinite.}
\begin{itemize}
\item \textbf{finite} if $\mathbb{S}\sim \mathbb{N}_n$ for some $n$, and $n$ is called the \textbf{cardinality} of $\mathbb{S}$;
\item \textbf{countably infinite} if $\mathbb{S}\sim \mathbb{N}$;
\item \textbf{countable} if it's finite or countably infinite;
\item \textbf{uncountable} if it's not countable.
\end{itemize}
\noindent For example, $\mathbb{Q}$ is countable, whereas the set of real numbers is uncountable, so the set of all irrational numbers is also uncountable \cite{BabyRudin,Munkres}.

A sequence $\{s_i\}$ can be thought of as a mapping $s:\mathbb{N}\rightarrow \mathbb{S}$, with $s_i=s(i)$ and $i\in\mathbb{N}$; to put it another way, $\mathbb{N}$ is an index set and $s$ is the indexing function, and a sequence is essentially $s$ together with $\mathbb{N}$ and $\mathbb{S}$. Note the set $\mathbb{S}$ doesn't need to be $\mathbb{R}$ or $\mathbb{C}$---we can have a sequence of, for example functions or vectors. In this thesis a sequence always refers to an infinite sequence, i.e. the index set is $\mathbb{N}$ rather than $\mathbb{N}_n$.

A series can be thought of as a sequence: Suppose $\mathcal{V}$ is a vector space (to be defined in Section~\ref{sec:vec}) and $t_i\in\mathcal{V}$ for all $i\in\mathbb{N}$. If $s=\sum_{i=1}^\infty t_i$, define the partial sum
\begin{equation}
s_n:=\sum_{i=1}^n t_i.
\end{equation}
Then the series is essentially the sequence $\{s_n\}$. It converges if the sequence $\{s_n\}$ converges (to be discussed in Section~\ref{sec:conban}) \cite{PapaRudin}.
\subsection{Ordered Sets, Intervals and Bounds}\label{sec:orib}
\begin{defi}[\textbf{Ordered set}]
A set $\mathbb{S}$ is said to be ordered if there exists a relation $<$, called \textbf{order}, such that the following two properties are satisfied \cite{BabyRudin}:
\begin{itemize}
\item For any $a,b\in\mathbb{S}$, one and only one of the three relations must be true: $a<b$, $a=b$, $b<a$.
\item Transitivity: For $a,b,c\in\mathbb{S}$, if $a<b$ and $b<c$ then $a<c$.
\end{itemize}
\end{defi}

$a>b$ is also used instead of $b<a$. The notation $a\leq b$ means $a<b$ or $a=b$, same for $a\geq b$ \cite{BabyRudin}. The real line $\mathbb{R}$ is a prime example of ordered sets. 

Before we proceed to the next important concept, let's take a look at intervals in $\mathbb{R}$\label{sym:R}. For any $b\geq a$ the following intervals are defined \cite{BabyRudin}:
\begin{enumerate}
	\item \textbf{(Open interval)} $(a,b):=\{x: a<x<b\}$.
	\item \textbf{(Closed interval)} $[a,b]:=\{x: a\leq x\leq b\}$.
	\item $[a,b):=\{x: a\leq x<b\}$.
	\item $(a,b]:=\{x:a< x\leq b\}$.
\end{enumerate}

\begin{defi}[\textbf{Upper and lower bounds}]
Given an ordered set $\mathbb{S}$, a subset $\mathbb{E}\subseteq \mathbb{S}$ is said to be \textbf{bounded (from) above} if there exists an $x\in\mathbb{S}$ such that $a\leq x$ for all $a\in\mathbb{E}$, and $a$ is called an \textbf{upper bound} of $\mathbb{E}$; similarly, it's said to be \textbf{bounded (from) below} if there exists an $y\in\mathbb{S}$ such that $y \leq a$ for all $a\in\mathbb{E}$, and $y$ is a \textbf{lower bound} of $\mathbb{E}$.
\end{defi}

As per the definition, the upper or lower bound of a subset (if it exists) isn't unique: For example both $1$ and $2$ are upper bounds of $(0,1)$. But what is to be defined is indeed unique:
\begin{defi}[\textbf{Supremum and infimum}]\label{def:sup}
A subset $\mathbb{E}$ of an ordered set $\mathbb{S}$ is said to have a \textbf{supremum} or \textbf{lowest upper bound} if there exists an $x\in\mathbb{S}$ such that $x$ is an upper bound of $\mathbb{E}$ and there's no upper bound of $\mathbb{E}$ smaller than $x$, i.e. any upper bound $y$ of $\mathbb{E}$ obeys $x\leq y$. If the supremum of a subset $\mathbb{E}$ exists, it's denoted by $\sup \mathbb{E}$.\footnote{Or ``lub'' in place of $\sup$, for lowest upper bound. Likewise the acronym of the greatest lower bound is "glb."} The \textbf{infimum} or \textbf{greatest lower bound} is defined in a similar way, denoted by $\inf$.
\end{defi}

The supremum or infimum of a subset isn't necessarily in the subset: For example, choosing $\mathbb{R}$ as the ordered set, $1$ is the supremum of both $(0,1)$ and $(0,1]$, but it's not in the former. The \textbf{maximum} of a subset $\mathbb{E}$ is simultaneously a supremum and an element of $\mathbb{E}$, denoted by $\max\mathbb{E}$, likewise for the \textbf{minimum}, denoted by $\min\mathbb{E}$ \cite{Zorich}. Therefore, $(0,1)$ doesn't have a maximum and minimum, and $(0,1]$ has the maximum $1$ but has no minimum.

An ordered set $\mathbb{S}$ has the \textbf{least-upper-bound property} if any non-empty subset $\mathbb{E}$ that is bounded from above always has a supremum. The real line $\mathbb{R}$ is equipped with such a property \cite{BabyRudin}. Note this doesn't imply any non-empty subset has a supremum---For example $(0,\infty):=\{x>0\}$ doesn't have a supremum.\footnote{Unless we're considering the extended real line \cite{BabyRudin,PapaRudin}.} 

\section{Vector Spaces and Linear Mappings}
\subsection{Vector Spaces and Norms}
\subsubsection{Vector space}
\label{sec:vec}
\begin{defi}[\textbf{Vector space}]
Let $\mathcal{V}$ be a set and $\mathbb{F}$ be a field, with a mapping $\mathcal{V}\times \mathcal{V}\ni (v_1,v_2) \mapsto v_1+v_2\in\mathcal{V}$ called \textbf{addition}, and another mapping
$\mathbb{F}\times\mathcal{V}\ni(a,v)\mapsto a v\in\mathcal{V}$, called \textbf{multiplication by a scalar}, where $a$ is often called a \textbf{scalar} in this context. Then $\mathcal{V}$ is said to be a vector space $\mathcal{V}$ over field $\mathbb{F}$ if \cite{Loomis,PapaRudin,Halmos,Roman}
\begin{description}
\item[A1.] Associativity of addition: $(v_1+v_2)+v_3=v_1+(v_2+v_3)$ for all $v_i\in\mathcal{V}$.
\item[A2.] Commutativity of addition: $v_1+v_2=v_2+v_1$ for all $v_i\in\mathcal{V}$.
\item[A3.] Existence of an identity, or zero under addition: There exists an element $0\in\mathcal{V}$ such that $0+v=v$ for all $v\in\mathcal{V}.$
\item[A4.] Existence of inverses under addition: For every $v_1\in\mathcal{V}$ there is exists a $v_2\in\mathcal{V}$ such that $v_1+v_2=0.$
\item[S1.] $(a_1 a_2) v=a_1 (a_2 v)$ for all $a_i\in\mathbb{F}$ and $v\in \mathcal{V}.$
\item[S2.] $(a_1+a_2)v=a_1v+a_2v$ for all $a_i\in\mathbb{F}$ and $v\in \mathcal{V}.$
\item[S3.] $a(v_1+v_2)=av_1+av_2$ for all $a\in\mathbb{F}$ and $v_i\in \mathcal{V}.$
\item[S4.] $1v=v$ for all $v\in \mathcal{V}.$
\end{description}
\end{defi}
Other properties can be deduced from these postulates \cite{Loomis}:
\begin{enumerate}
\item The zero element postulated in A3 is unique: If there's another $0_1$ that satisfies A3. Then by A1 and A4, if $v_1+v_2=0,$ then $(0_1+v_1)+v_2=0_1+(v_1+v_2)=0$, so $0_1=0.$
\item For each $v_1$ the $v_2$ of A4 is unique, and is called $-v_1$: If $v_1+v_2=v_3+v_2=0$, by A1 $v_1+v_2+v_3=v_1=v_3.$
\item $0v=0$ and $a0=0$: From S2 $av=(0+a)v=0v+av$, so $0v=0$; from S3 $av=a(0+v)=a0+av$ so $a0=0.$
\item $(-1)v=-v$: From S2, S4 and $0v=0$, $(1-1)v=v+(-1)v=0$, so $(-1)v=-v.$
\end{enumerate}
A1 to A4 basically say a vector space is an abelian group under addition, so point 1 and 2 above come naturally. We use the same symbol ``0'' for the zero elements of both $\mathcal{V}$ and $\mathbb{F}$, and the context dictates what the zero stands for. In this thesis, the field is usually $\mathbb{C},$ the complex field.\label{sym:C}

\subsubsection{Operations on mappings}
\label{sec:omap}
For functions $f$ and $f_i$ from a set $\mathbb{S}$ to a vector space $\mathcal{V}$, the following operations or mappings can be defined on them to create a new function:
\begin{itemize}
\item \textbf{Addition:}  $(f_1+f_2)(x):=f_1(x)+f_2(x)$ $\forall x\in\mathbb{S}$.
\item \textbf{Multiplication by a scalar:} $(cf)(x):=cf(x)$ for any scalar $c$ and $\forall x\in\mathbb{S}$.
\end{itemize}
Basically, they're defined \emph{pointwise}. We can also define a zero mapping $0: \mathbb{S}\rightarrow\mathcal{V}$ which satisfies
\begin{equation}
0(x)=0 \;\forall x\in\mathbb{S}.
\end{equation}
Note $0$ on the left side of the equation above is the zero mapping and $0$ on the right side is the zero vector in $\mathcal{V}.$ Since $f(x)$ is a vector, all the requirements for vectors (see Section~\ref{sec:vec}) are met by the two operations above, so mappings from a set to a vector space $\mathcal{V}$ can form a vector space, with the zero mapping being the zero vector.

\subsubsection{Span, Bases and Dimensions}\label{sec:span}
Let $\mathbb{S}$ be a nonempty subset of a vector space $\mathcal{V}$. A \textbf{linear combination} of  vectors in $\mathbb{S}$ is a \textbf{finite} sum $\sum_i c_i v_i$, where $c_i$ are arbitrary scalars and $v_i\in\mathbb{S}$. Note $\mathbb{S}$ doesn't need to be finite---it can even be uncountable.
\begin{defi}[\textbf{Span}]
The (linear) span of a nonempty subset $\mathbb{S}$ of a vector space is composed of all linear combinations of $\mathbb{S}$, denoted by $\text{span}(\mathbb{S})$. In other words, let $\mathbb{K}$ be an index set for $\mathbb{S}$, i.e. $\mathbb{S}=\{v_i:i\in\mathbb{K}\}$:
\begin{equation}
\mathrm{span}(\mathbb{S})=\left\{\sum_{i\in\mathbb{K}'}c_i v_i:c_i\in\mathbb{R}\text{ and } \mathbb{K}'\text{ are finite subsets of }\mathbb{K}\right\}.
\end{equation}
\end{defi}

$\text{span}(\mathbb{S})$ is a subspace of $\mathcal{V}$, and the smallest subspace of $\mathcal{V}$ that contains $\mathbb{S}.$ If $\text{span}(\mathbb{S})=\mathcal{V}$, we say $\mathbb{S}$ spans $V$. A vector space is \textbf{finite-dimensional} if it has a finite spanning set, else it's \textbf{infinite-dimensional}, and the \textbf{dimension} of a (finite-dimensional) vector space is the cardinality of the smallest spanning set; $\text{dim}\mathcal{V}$ refers to the dimension of a vector space $\mathcal{V}$. Some properties of finite-dimensional spaces we may take for granted don't necessarily exist for infinite-dimensional spaces, so they should be treated with extra care. 

\begin{defi}[\textbf{Linear independence}]\label{def:li}
Let $\mathcal{V}$ be a vector space. A non-empty subset $\mathbb{S}\subset \mathcal{V}$ is said to be \textbf{linearly independent} if $\mathbb{S}$ is the smallest spanning set of $\mathrm{span}(\mathbb{S})$; in other words, for any $v\in \mathbb{S}$, $\mathrm{span}(\mathbb{S})\neq \mathrm{span}(\mathbb{S}/\{v\})$. Equivalently, $\mathbb{S}\subset \mathcal{V}$ is linearly independent if and only if for any linear combination of vectors from $\mathbb{S}$: $\sum_i c_i v_i$ to be zero, the coefficients $c_i$ must all be zero.
\end{defi}
From the definition a linearly independent set must not contain $0$ \cite{Roman}. We can now define:
\begin{defi}[\textbf{Basis}]\label{def:basis}
Consider a vector space $\mathcal{V}$ and a set $\mathbb{B}=\{\alpha_i\}\subset \mathcal{V}$. If $\mathbb{B}$ spans $\mathcal{V}$ and is linearly independent, then $\mathbb{B}$ is called a basis of $\mathcal{V}$. As $\mathbb{B}$ spans $\mathcal{V},$ given $v\in\mathcal{V}$ we can always find scalars $c_i$ such that
\begin{equation}
v=\sum_i c_i \alpha_i,
\end{equation}
and $c_i$ are unique by linear independence. The scalars $c_i$ are called the \textbf{coordinates} of $v$ (with respect to $\mathbb{B}$).
\end{defi}

A basis $\mathbb{B}$ is the smallest set that spans $\mathcal{V}$, and the dimension of a vector space is the number of vectors in a basis. The basis of a vector space is not unique, but the dimension is a fixed number \cite{Loomis}. 
\subsubsection{Metrics and Norms}\label{sec:normmet}
\begin{defi}[\textbf{Metric}]
A metric or distance $d$ on a set $\mathbb{S}$ is a non-negative function such that for all $x,y,z\in\mathbb{S}$
\begin{enumerate}
	\item $d(x,y)>0$ for $x\neq y$ and $d(x,x)=0$.
	\item (Symmetry) $d(x,y)=d(y,x)$.
	\item \textbf{The triangle inequality}: $d(x,z)\leq d(x,y)+d(y,z)$.
\end{enumerate}
A set with a metric is called a \textbf{metric space} \cite{Loomis,PapaRudin}.
\end{defi}

The topology of a metric space can be constructed by open balls, defined as:
\begin{equation}
B_r(x)=\{y\in\mathbb{S}:d(y,x)< r\},\;r\geq 0,
\end{equation}
which is an open ball of radius $r$ about $x$ \cite{Loomis,PapaRudin}. A set is open if and only if it's a union of open balls. For example, in $\mathbb{R}$ an open interval is both an open set and an open ball, and a closed interval is a closed set and a closed ball; on the other hand $[a,b)$ is neither a closed nor open set, unless $a=b$, in which case it's an empty set and an open (and closed) set. Any subset $\mathbb{S}'$ of a metric space is a metric space with the metric given by $d\upharpoonright_{\mathbb{S}'\times \mathbb{S}'}.$ 

\begin{defi}[\textbf{Norm}]\label{def:norm}
If $\mathcal{V}$ is a complex (or real) vector space, a function $||\cdot||:\mathcal{V}\rightarrow \mathbb{R}$ that satisfies the following conditions
\begin{enumerate}
\item $||v||> 0$ for all nonzero $v\in\mathcal{V}$.
\item $||av||=|a|\; ||v||$ for all scalar $a$ and $v\in\mathcal{V}.$
\item \textbf{The triangle inequality}: $||v_1+v_2||\leq ||v_1||+||v_2||$.
\end{enumerate}
is called a \textbf{norm} or length \cite{Roman,Loomis,PapaRudin}. From 2. it's clear that $||0||=0$ and is thus the only vector with norm 0. A vector space with a norm is called a \textbf{normed vector space}. 
\end{defi}

A normed space $\mathcal{V}$ has a metric given by \cite{Loomis}
\begin{equation}
d(v_1,v_2):=||v_1-v_2||.
\end{equation}
Therefore a normed space is always a metric space, and the metric on a normed space is always assumed to that induced by norm. The associated topology on the space is called norm topology \cite{PapaRudin,Einsiedler,Kadison}.

\subsubsection{Convergence and Banach Spaces}
\label{sec:conban}
In a metric space,\footnote{More generally, convergence as a concept can be defined for topological spaces. The following definition corresponds to convergence in terms of metric topology. Please refer to Ref.~\cite{Munkres,Tu} for details.} a sequence $\{x_n\}$ is said to converge to a point $x$ if for any $\delta>0$ there exists an $N_\delta$ such that
\begin{equation}
d(x,x_n)<\delta \text{ for all }n>N. 
\end{equation}
For example, consider the metric space $[0,1]\subset \mathbb{R}$ with the metric $d(x_1,x_2)=|x_1-x_2|.$ The sequence $\{1/n:n\in\mathbb{N}\}$ converges to $0$.

In a metric space, a sequence $\{x_n\}$ is called a \textbf{Cauchy sequence} if for any $\delta>0$, there's always an $N_\delta$ such that
\begin{equation}
d(x_m,x_n)<\delta \text{ for all }m,n>N_\delta.
\end{equation}
Basically, the distance between each two elements in the sequence gets increasingly small as $N\rightarrow \infty$. 

By definition, a convergent sequence is a Cauchy sequence, but the reverse isn't necessarily true: For instance, now consider the metric space $(0,1]$. The sequence $\{1/n\}$ is clearly a Cauchy sequence, but it doesn't converge to any point in $(0,1]$, even though $1/n\in (0,1]$ for all $n\in\mathbb{N}$. This leads to the following definition:

\begin{defi}[\textbf{Completeness}]
A metric space where any Cauchy sequences converge is called a \textbf{complete space}. A complete normed space is called a \textbf{Banach space}.
\end{defi}

All finite-dimensional normed spaces are complete \cite{Loomis}; on the other hand, not all infinite-dimensional spaces are complete. Basically, completeness tells us what seems to be there is actually there---a sequence that seems to converge indeed converges to something in the space. On the contrary, in a incomplete space something that seems to converge may approach nothing, or to something outside the space (after completion \cite{Roman,Einsiedler}).

Let's use the previous example: $(0,1]$ isn't a compete space, as there are Cauchy sequences that don't converge, like $\{1/n\}$, but it can be completed by adding $\{0\}$ to it, becoming $[0,1]$. As limits are a frequent occurrence, completeness is a very desirable and useful property.

\subsection{Linear Mappings, Operators and the Dual Space}
\subsubsection{Linear Mappings and Bounded Linear Mappings}\label{sec:opnorm}

\begin{defi}[\textbf{Linear mappings}]
Suppose $\mathcal{V}$ and $\mathcal{W}$ are vector spaces. A mapping $T:\mathcal{V}\rightarrow \mathcal{W}$ is said to be a linear mapping (or map or transformation) if \cite{Loomis,BabyRudin}
\begin{equation*}
T(v_1+v_2)=T(v_1)+T(v_2) \text{ and } T(cv)=cT(v) \,\forall v,v_i\in \mathcal{V}_1 \text{ and scalar }c.
\end{equation*}
\end{defi}
With a linear mapping $T$, the parentheses around the argument are sometimes ignored, i.e. $Tv=T(v).$ Linear mappings from $\mathcal{V}$ to $\mathcal{W}$ form a vector space, because if $T_i$ are linear mappings, $T_1+T_2$ and $cT$ are as well \cite{Loomis,BabyRudin}. The zero mapping $0$ is the zero vector.

\begin{defi}[\textbf{Operator norm}]
If $\mathcal{V}$ and $\mathcal{W}$ are normed vector spaces, the operator norm of  a linear mapping $T:\mathcal{V}\rightarrow\mathcal{W}$ is defined as \cite{Bhatia,Loomis}
\begin{equation}
||T||:=\sup_{v\in\mathcal{V}}\frac{||Tv||}{||v||}.
\end{equation}
\end{defi}

The operator norm is indeed a norm---for example, by definition the operator norm of the zero mapping is 0, as a norm should. A linear mapping is \textbf{bounded} if $||T||<\infty.$ Bounded linear mappings from $\mathcal{V}$ to $\mathcal{W}$ form a normed space themselves, denoted by $\mathcal{B}(\mathcal{V},\mathcal{W})$; on the other hand, the space of \emph{all} linear mappings from $\mathcal{V}$ to $\mathcal{W}$ is denoted by $\mathcal{L}(\mathcal{V},\mathcal{W})$. It can be shown that if $T_1\in \mathcal{B}(\mathcal{V}_1,\mathcal{V}_2)$ and $T_2\in \mathcal{B}(\mathcal{V}_2,\mathcal{V}_3)$, $T_2\circ T_1\in \mathcal{B}(\mathcal{V}_1,\mathcal{V}_3).$ Other notations are also used for the space of bounded linear mappings in the literature, for example $\text{Hom}(\mathcal{V},\mathcal{W})$ (\text{Hom} for homomorphism) and $\mathcal{L}(\mathcal{V},\mathcal{W})$ \cite{Loomis,Roman}. If only finite-dimensional vector spaces are concerned, all linear mappings are bounded \cite{Loomis,Blackadar}. Also,
\begin{thm}\label{thm:dim}
Let $\mathcal{V}$ and $\mathcal{W}$ be finite-dimensional vector spaces. Then $\mathrm{dim}(\mathcal{B}(\mathcal{V},\mathcal{W}))=\mathrm{dim}(\mathcal{V})\mathrm{dim}(\mathcal{W})$.
\end{thm}

It can be shown continuity and boundedness are identical for linear mappings \cite{Loomis,PapaRudin}:
\begin{thm}\label{thm:boco}
For a linear mapping between normed spaces, the following statements are equivalent:
\begin{enumerate}
\item It's continuous at one point.
\item It's continuous (everywhere).\footnote{A mapping is said to be continuous if it's continuous everywhere.}
\item It's bounded.
\end{enumerate}
\end{thm}
\noindent Hence all linear mappings between finite-dimensional spaces are continuous.

\subsubsection{Null Space, Homomorphism and Isomorphism}\label{sec:hom}
\begin{defi}[\textbf{Null space}]
The null space or \textbf{kernel}\footnote{As the term ``kernel'' will be used in a different sense later, we'll avoid using this term to refer to the null space.} of a linear mapping $T:\mathcal{V}\rightarrow \mathcal{W}$ is the set \cite{Loomis,Roman,GrandpaRudin}
\begin{equation}
\ker T:=\{T(v)=0:v\in\mathcal{V}\}.
\end{equation} 
\end{defi}
The null space can be easily shown to be a (sub)space by linearity. A linear mapping is one-to-one (injective) if and only if its null space is $\{0\}$ \cite{Loomis}. 

A linear mapping from a vector space to another is also called a \textbf{homomorphism}, because the structure of, or operations on vector spaces are retained. If a linear mapping is bijective, it's called an \textbf{isomorphism}. Two vector spaces are said to be \textbf{isomorphic} if there's an isomorphism between them.

In finite-dimensional spaces, there's a very useful relation  \cite{Roman,Loomis}: 
\begin{thm}[\textbf{The rank plus nullity theorem}]\label{thm:rpn}
Suppose $\mathcal{V}$ and $\mathcal{W}$ are both finite dimensional spaces, and $T$ is a linear mapping from $\mathcal{V}$ to $\mathcal{W}$. Then
\begin{equation}
\mathrm{dim}(\ker T)+\mathrm{dim}(\mathrm{ran}T)=\mathrm{dim}(\mathcal{V}).
\end{equation}
Namely, the dimension of the kernel (\textbf{nullity}) plus that of the range (\textbf{rank}) is equal to the dimension of the domain.
\end{thm}  
Some immediate observations can be made of this theorem:
\begin{enumerate}
\item If $\mathrm{dim}(\mathcal{V})=\mathrm{dim}(\mathcal{W})$, a linear mapping is injective if and only if it's surjective (and therefore bijective).
\item If $\text{dim}(\mathcal{V})>\text{dim}(\mathcal{W})$, $T$ can't be one-to-one; if $\text{dim}(\mathcal{W})>\text{dim}(\mathcal{V})$, $T$ can't be surjective.
\item \emph{Two finite-dimensional vector spaces are isomorphic if and only if they have the same dimension.}
\end{enumerate}

When the dimension is infinite, even for linear mappings between the same space the remarks above break down. For example, a linear mapping can be injective without being surjective. Two infinite-dimensional spaces in general aren't isomorphic, so how a infinite-dimensional space is configured is crucial.

\subsubsection{Operators}
A linear mapping with the same domain and codomain is often called (linear) operators \cite{BabyRudin,Roman},\footnote{In this thesis, I'll only adhere to this convention strictly when dealing with bounded operators. In some literature generic linear mappings are also called operators \cite{Einsiedler,Halmos}. Actually, some linear mappings that are often called operators, like differential operators and integral operators, may have different domains and codomains \cite{Einsiedler}.} and the vector space of bounded linear operators on $\mathcal{V}$, $\mathcal{B}(\mathcal{V},\mathcal{V})$, will be denoted by $\mathcal{B}(\mathcal{V})$. 

The composition of bounded operators can be regarded as a ``multiplication'' between operators:
\begin{equation}
\mathcal{B}(\mathcal{V})\times \mathcal{B}(\mathcal{V})\ni (O_1,O_2):=O_1\circ O_2:=O_1O_2\in \mathcal{B}(\mathcal{V}),
\end{equation}
so operators form an algebra \cite{Blackadar,Loomis,Roman}:
\begin{subequations}
\begin{align}
(O_1O_2)O_3&=O_1(O_2O_3)\label{eq:a1}\\
O_1(O_2+O_3)&=O_1O_2+O_1O_3\label{eq:a2}\\
(O_1+O_2)O_3&=O_1O_3+O_2O_3\label{eq:a3}\\
c O_1 O_2:= c(O_1O_2)&=(cO_1)O_2=O_1(cO_2),\label{eq:a4}
\end{align}
\end{subequations}
where $c$ is any scalar. \eqref{eq:a1} comes from $(f\circ g)\circ h=f\circ (g\circ h)$ for all functions $f,g,h$ whose domains and codomains coincide. Linearity of operators leads to \eqref{eq:a2}, \eqref{eq:a3} and \eqref{eq:a4}.

\subsubsection{Linear Functionals and the Dual Space}
\label{sec:lf}
A \textbf{linear functional} on a vector space is a linear mapping from vectors to scalars. Note a field, like $\mathbb{R}$ or $\mathbb{C}$ is a one-dimensional vector space over itself. Like other linear mappings, linear functionals on a vector space $\mathcal{V}$ form a vector space, called the \textbf{dual space} of $\mathcal{V},$ denoted by $\mathcal{V}^*$. 

Now let's show for a finite-dimensional space $\mathcal{V}$, $\mathcal{V}^*$ is isomorphic to it \cite{Loomis}: Given a basis $\mathbb{B}=\{\alpha_i\}$ of a finite-dimensional space $\mathcal{V}$, a basis of $\mathcal{V}^*$ can be easily constructed by choosing linear functionals $l_i$ such that
\begin{equation}
l_i(\alpha_j)=\delta_{ij},
\end{equation}
and an \emph{isomorphism} $L$ from $\mathcal{V}$ to $\mathcal{V}^*$ can be established: For $v=\sum c_i\alpha_i \in\mathcal{V}$, we require
\begin{equation}
L(v)=\sum_i c_i l_i.
\end{equation}
$l_i$ are exactly the coordinate function(al)s with respect to $\mathbb{B}$, in that if $v=\sum c_i \alpha_i$ then $l_i(v)=c_i$, which are the coordinates.

\section{Hilbert Space}
\label{sec:inner}
\subsection{Inner Product, Inner Product Spaces and Hilbert Spaces}\label{sec:inn}
\begin{defi}[\textbf{Inner product}]\label{def:inn}
A complex vector space is called an \textbf{inner product space} $\mathcal{V}$ if for any pair of vectors $v_1$ and $v_2$ there's an associated complex number $(v_1|v_2)$, called an inner product, that satisfies the following properties \cite{Loomis,PapaRudin}:
\begin{enumerate}
\item $(v_1|v_2)^*=(v_2|v_1)$.
\item $(v_1|v_2+v_3)=(v_1|v_2)+(v_1|v_3)$ for any $v_i\in\mathcal{V}$.
\item $(v_1|cv_2)=c(v_1|v_2)$, where $c$ is any complex number and often called a scalar.\footnote{In mathematical texts, linearity is usually presumed on the first argument, rather than the second as most physicists do.}
\item $(v|v)\geq 0$ $\forall v\in\mathcal{V}.$
\item $(v|v)=0$ only if $v=0,$ the zero vector.
\end{enumerate}
\end{defi}

A real inner product space replaces property 1 with symmetry: $(v_1|v_2)=(v_2|v_1)$. An inner product space is naturally endowed with a norm:
\begin{equation}
||v||:=(v|v)^{1/2} \text{ for any }v\in\mathcal{V}.
\end{equation}
It can be shown this is indeed a norm \cite{Loomis,PapaRudin,Roman}. Unless otherwise stated, it's always assumed that the norm on an inner product space is that induced by inner product. A functional that is linear in the second argument and conjugate-linear in the first one is called a \textbf{sesquilinear form} \cite{HallQ}.\footnote{Conjugate-linearity is also called semi-linearity, $1/2$. Sesqui- refers to one and a half \cite{Einsiedler}.} From properties 1$\sim$3 an inner product is therefore a sesequilinear form. 
\begin{thm}\label{thm:in}
For all $v,w$ in an inner product space, they obey the following relations \cite{PapaRudin,Loomis,Roman,Blackadar}: 
\begin{enumerate}
\item \textbf{The polarization identity}: $4(w|v)=\left(||v+w||^2-||v-w||^2\right)+i\left(||v+iw||^2-||v-iw||^2\right)$.
\item \textbf{The parallelgoram law}: $||v+w||^2+||v-w||^2=2||v||^2+2||w||^2$.
\item \textbf{Cauchy-Schwarz inequality}: $|(v|w)|\leq \sqrt{(v|v)(w|w)}=||v||\,||w||$.
\end{enumerate}
\end{thm}

Two vectors $v$ and $w$ are said to be \textbf{orthogonal} if $(v|w)=0$, denoted by $v\perp w$\label{sym:ortho}; note the relation is symmetric: $v \perp w$ implies $w \perp v$ \cite{PapaRudin}. Given a subset $\mathbb{S}$ in an inner product space, $\mathbb{S}^\perp$ is the set
\begin{equation}
\mathbb{S}^\perp=\{v\perp w:w\in\mathbb{S}\}.
\end{equation}
It can be easily shown $\mathbb{S}^\perp$ is a subspace; furthermore, it's a closed subspace \cite{Loomis,PapaRudin}.

\begin{defi}[\textbf{Hilbert space}]
An inner product space is called a Hilbert space if it's complete with respect to the norm induced by inner product \cite{Loomis,PapaRudin,Einsiedler}, i.e. if it's also a Banach space \cite{Loomis,PapaRudin}. The symbol $\mathcal{H}$ denotes a Hilbert space.
\end{defi}

From now on an inner product space is always assumed to be a Hilbert space. As all finite-dimensional normed spaces are complete \cite{Loomis}, a finite-dimensional inner product space is always a Hilbert space. Finally let's give formal definitions for two inner product spaces that see frequent uses:
\begin{defi}[$\mathbb{R}^n$ and $\mathbb{C}^n$]\label{def:rncn}
$\mathbb{R}^n$ is a real inner product space of $n$-tuples of real numbers, $(a_1,\cdots,a_n)$, $a_i\in\mathbb{R}$. The basic vector operation is 
\begin{equation}
c_1(a_1,\cdots,a_n)+c_2(b_1,\cdots,b_n)=(c_1 a_1+c_2 a_2,\cdots,c_1 a_n+c_2 b_n),\;c_1,c_2\in\mathbb{R}.
\end{equation} 
Its inner product is:
\begin{equation}
(v_1|v_2):=\sum_{i} {a_1^i} a_2^i \text{ for } v_i=(a_1^i,\cdots,a_n^i).
\end{equation}	
Similarly, $\mathbb{C}^n$ is an inner product space of $n$-tuples of complex numbers, with 
\begin{equation}
c_1(a_1,\cdots,a_n)+c_2(b_1,\cdots,b_n)=(c_1 a_1+c_2 a_2,\cdots,c_1 a_n+c_2 b_n),\;c_1,c_2\in\mathbb{C}.
\end{equation} 
The inner product on $\mathbb{C}^n$ is defined as \cite{PapaRudin}:
\begin{equation}
(v_1|v_2):=\sum_{i} {a_1^i}^* a_2^i \text{ for } v_i=(a_1^i,\cdots,a_n^i)
\end{equation}
The inner product on $\mathbb{R}^n$ or $\mathbb{C}^n$ is often called a ``dot product'' $v_1\cdot v_2:=(v_1|v_2)$. A vector in $\mathbb{R}^n$ is often marked in boldface, e.g. $\mathbf{v}$.
\end{defi}

Clearly $\mathbb{R}^n$ and $\mathbb{C}^n$ are $n$-dimensional inner product spaces over $\mathbb{R}$ and $\mathbb{C}$ respectively. The inner products on $\mathbb{R}^n$ and $\mathbb{C}^n$ can be treated as matrix multiplications, to be discussed in Section~\ref{sec:ismat}. Later on in Section~\ref{sec:psv} we will consider the phase space, labeled as $\mathbb{R}^{2n}$, on which we don't assume an inner product like $\mathbb{R}^n$. 

\subsection{Fourier Expansion and Orthonormal Bases}
\label{sec:dual}
A set $\{\beta_i:i\in\mathbb{K}\}\subseteq \mathcal{H}$, where $\mathbb{K}$ is an index set equivalent to $\{\beta_i\}$ and isn't necessarily finite or even countable, is called an \textbf{orthonormal set} if $(\beta_i|\beta_j)=\delta_{ij}$ for all $i\neq j$. An orthonormal (more generally, orthogonal) set is linearly independent (Definition~\ref{def:li}) \cite{Loomis,Roman}. A maximal orthonormal set of a Hilbert space is called a \textbf{Hilbert basis/orthonormal basis/complete orthonormal set} \cite{Roman,Blackadar,PapaRudin}.\footnote{Bases (Definition~\ref{def:basis}) and orthonormal bases are distinct concepts: In a finite dimensional space, an orthonormal basis is also a basis, in that it spans the space, and any basis can be orthonormalized by Gram-Schmidt procedure---but this isn't necessarily true in infinite-dimensional spaces---An orthonormal basis may not be a basis \cite{Roman,Einsiedler}.}

With an orthonormal set $\{\beta_i\}$, the \textbf{Fourier expansion} of a vector $v\in\mathcal{H}$ with respect to it is\footnote{For the meaning of a summation over an uncountable set, please refer to Ref.~\cite{PapaRudin,Roman}.} \cite{Roman,Blackadar,PapaRudin}
\begin{equation}
\hat{v}=\sum_{i\in\mathbb{K}}(\beta_i|v)\beta_i.
\end{equation}
This series always converges, and only countable summands are nonzero \cite{Roman,PapaRudin}, and $(\beta_i|v)$ are known as \textbf{Fourier coefficients}. By \textbf{Bessel's inequality} \cite{Roman,Blackadar,PapaRudin},
\begin{equation}
||\hat{v}||\leq ||v||.
\end{equation}
Relative to an orthonormal set, the Fourier expansion is the \emph{unique best approximation}, that is to say
\begin{equation}
||v-\hat{v}||=\inf\left(||s-v||:s\in \text{cspan}(\{\beta_i\})\right),
\end{equation}
where cspan stands for the closure of the span; note $\hat{v}$ is in $\text{cspan}(\{\beta_i\})$ \cite{Roman,PapaRudin}.

The condition that Bessel's inequality becomes an equality is given by the following theorem \cite{Roman,PapaRudin}:
\begin{thm}[\textbf{Riesz-Fischer}]\label{thm:rf}
Given an orthonormal set $\mathcal{O}=\{\beta_i:i\in\mathbb{K}\}$ in a Hilbert space $\mathcal{H}$, the following are equivalent:
\begin{enumerate}
\item $\mathcal{O}$ is an orthonormal basis.
\item $\mathcal{O}^\perp =\{0\}$.
\item $\mathrm{cspan}(\mathcal{O})=\mathcal{H}$; in other words $\mathrm{span}(\mathcal{O})$ is dense in $\mathcal{H}.$
\item $v=\hat{v}$ for all $v\in\mathcal{H}.$ 
\item Equality holds in Bessel's inequality for all $v\in\mathcal{H}.$
\item \textbf{Parseval's equality} holds for all $v_1,v_2\in\mathcal{H}$, i.e.
\begin{equation}
(v_1|v_2)=(\hat{v}_1|\hat{v}_2)=\sum_i (v_1|\beta_i)(\beta_i|v_2).
\end{equation}
\end{enumerate}
\end{thm}
\noindent Hence, given an orthonormal basis, the Fourier expansion of any vector converges to the vector itself. Note whether a sequence converges depends on the topology, or norm in this case. Riesz-Fischer theorem states only that the vector and inner product in point 4 and 6 above converge with respect to inner-product induced norm, but whether they converge in terms of other topologies is another matter.

\subsection{Adjoint, Self-adjoint Operators and Unitary Mappings}
\subsubsection{Adjoint and self-adjoint operators}\label{sec:ad}
\begin{thm}[\textbf{Adjoint}]
For any bounded linear mapping $T\in\mathcal{B}(\mathcal{H}_1,\mathcal{H}_2)$, there always exists a unique bounded linear mapping $T^\dagger\in\mathcal{B}(\mathcal{H}_2,\mathcal{H}_1)$, called the \textbf{adjoint} of $T$,\footnote{Or Hermitian conjugate as is commonly known by physicists. Instead of $\dagger$, mathematicians use star (like $T^*$) for adjoint and a horizontal bar over a number (like $\bar{z}$) for complex conjugation \cite{HallQ}.} such that \cite{Roman,Blackadar,GrandpaRudin}
\begin{equation}
(v|T(w))=(L^\dagger(v)|w)\;\forall w\in\mathcal{H}_1\text{ and }\forall v\in\mathcal{H}_2.
\end{equation}
In addition, $||T||=||T^\dagger||.$
\end{thm}
For $O\in\mathcal{B}(\mathcal{H})$, if $O^\dagger=O$, it is called \textbf{self-adjoint}, or \textbf{Hermitian} when the field is $\mathbb{C}$ and \textbf{symmetric} when the field is $\mathbb{R}$ \cite{Roman}. As we consider the complex field in this thesis most of the time, ``self-adjoint'' and ``Hermitian'' are interchangeable as long as only bounded operators are considered, but the term ``Hermitian'' is usually used when dealing with finite-dimensional spaces. An operator $T\in\mathcal{B}(\mathcal{H})$ is said to be \textbf{normal} if $O^\dagger O=OO^\dagger$. Obviously, a self-adjoint operator is normal.

A mapping $f$ from a complex vector space $\mathcal{V}$ to another is said to be \textbf{conjugate linear} if \cite{Roman,GrandpaRudin}
\begin{equation}
f(v_1+v_2)=f(v_1)+f(v_2),\, f(av)=a^*f(v)\text{ for any } v,v_i\in\mathcal{V}\text{ and any scalar } a.
\end{equation}
The mapping $\dagger:T\mapsto T^\dagger$ can be shown to be conjugate linear and satisfy
\begin{equation}
(T_1 T_2)^\dagger=T_2^\dagger T_1^\dagger.
\end{equation}

The adjoint of an unbounded operator on a Hilbert space can also be defined, similarly for self-adjointness. However, as an unbounded operator isn't well-defined on the whole space, the domain of an unbounded operator needs to be specified (or chosen), and more technicality must be involved. Readers interested in such topics may refer to Ref.~\cite{GrandpaRudin,HallQ,Einsiedler}.

\subsubsection{Isometry and unitary mappings}
\label{sec:isometry}
A mapping $f$ from a metric space $\mathbb{S}_1$ to another metric space $\mathbb{S}_2$ is called an \textbf{isometry} if it preserves the metric, i.e. $d_{\mathbb{S}_1}(x,y)=d_{\mathbb{S}_2}(f(x),f(y))$ \cite{Tu,TuD}. If $\mathcal{V}$ and $\mathcal{W}$ are normed vector spaces, an isometry $f:\mathcal{V}\rightarrow \mathcal{W}$ satisfies $||v_2-v_1||_\mathcal{V}=||f(v_2)-f(v_1)||_\mathcal{W}$ \cite{Einsiedler}. We'll only deal with isometries on normed vector spaces, and from now on all isometries are assumed to be \emph{linear}. Since an isometry preserves length, it's one-to-one, for a vector has norm 0 if and only if it's the zero vector.

On an inner product space, by the polarization identity (Theorem~\ref{thm:in}) an isometry $T:\mathcal{H}_1\rightarrow \mathcal{H}_2$ should preserve inner product between any two vectors, i.e. $(Ty|Tx)_2=(y|x)_1$ for all $x,y\in\mathcal{H}_1$; conversely, any linear mapping which preserves the inner product is obviously an isometry \cite{Roman}. Since $(Ty|Tx)_2=(y|x)_1$ for all $x,y\in\mathcal{H}_1$, it implies $T^\dagger T=I_{\mathcal{H}_1}$. We can also use this to define a (linear) isometry between Hilbert spaces: We define an isometry as a linear mapping $T:\mathcal{H}_1\rightarrow \mathcal{H}_2$ that satisfies  $T^\dagger T=I_{\mathcal{H}_1}$.

An isometry $U:\mathcal{H}_1\rightarrow \mathcal{H}_2$ is said to be \textbf{unitary} if it's surjective. As an isometry is injective already, a (bounded) linear mapping $U$ on a Hilbert space is unitary if and only if it is a bijective isometry, and if and only if $U^\dagger U=I_{\mathcal{H}_1}$ and $UU^\dagger=I_{\mathcal{H}_2}$ \cite{Einsiedler,Blackadar,HallQ}. By definition, a unitary operator is normal. An isometry from an finite-dimensional inner product space to another with the same dimension is automatically unitary. 

When we say two Hilbert spaces are isomorphic, it always means there exists a unitary mapping between them: With the extra structure of inner product, we'd like the linear mapping to not only be bijective but also also preserve the inner product.

\subsection{Dual Space}
\subsubsection{Linear functionals and the dual Space}\label{sec:lfun}
With an inner product, for each vector $w$ a linear functional can be defined \cite{Loomis,GrandpaRudin}: 
\begin{equation}
l_w(v):=(w|v),\, v\in\mathcal{H}.\label{eq:lalpha} 
\end{equation}
Its linearity derives from the linearity of inner product. Note since the linear mapping $v\mapsto (w|v)$ is (uniformly) continuous \cite{PapaRudin}, $l_w(v)$ is a \emph{bounded} linear functional by Theorem~\ref{thm:boco}. The converse is true as well \cite{Roman,Blackadar,PapaRudin}: 
\begin{thm}[\textbf{Riesz-Fr\'{e}chet}]\label{thm:Riesz}
For any bounded linear functional $l$ on a Hilbert space $\mathcal{H}$, there's a unique $w$ such that $l(v)=(w|v)$ for all $v\in\mathcal{H},$ and $||w||=||l||.$
\end{thm}
\noindent Note the norm for $w$ is that induced by inner product, and the norm for $l$ is the operator norm. We will label the space formed by bounded linear functionals on $\mathcal{H}$ as $\mathcal{H}^*$.

Because $l_{Tw}(v)=(Tw|v)=(w|T^\dagger v),$ we have
\begin{equation}
l_{Tw}=l_w\circ T^\dagger.
\end{equation} 
In addition, the mapping from $v\mapsto l_v$ is conjugate linear, rather than linear, i.e.
\begin{equation}
l_{a_1 v_1+a_2 v_2}=a_1^* l_{v_1}+a_2^* l_{v_2}.
\end{equation}
By Theorem~\ref{thm:Riesz}, $v\mapsto l_v$ is bijective, so it's a \textbf{conjugate isomorphism} from $H$ to $\mathcal{H}^*$ \cite{Roman}. Due to conjugate linearity, it conforms to matrix multiplication which will be discussed later in Section~\ref{sec:ismat}.

If $\{\alpha_i\}$ is an orthonormal basis of $\mathcal{H}$ (Section~\ref{sec:dual}), then we can define such linear functionals in $\mathcal{H}^*$:
\begin{equation}
l_i(v):=(\alpha_i|v)\;\forall v\in\mathcal{H}.
\end{equation}
By Theorem~\ref{thm:rf} and \ref{thm:Riesz}, $\mathcal{H}^*=\mathrm{cspan}(\{l_i\})$.

\subsubsection{Define a linear mapping through linear functionals}

With a linear functional $l\in\mathcal{V}^*$ and a vector $w\in\mathcal{V}$, we can define such an operator on $\mathcal{V}$:
\begin{equation}
\mathcal{V}\ni v\mapsto l(v)w,
\end{equation}
Let's assume $\mathcal{V}$ and $\mathcal{W}$ are finite-dimensional spaces. Given a basis $\{\alpha_i\}$ of $\mathcal{V}$, as in Section~\ref{sec:lf} we can choose a basis $\{l_i\}$ of $\mathcal{V}^*$ by
\begin{equation}
l_i(\alpha_j)=\delta_{ij}.
\end{equation}
Let $T:\mathcal{V}\rightarrow \mathcal{W}$ be a linear mapping. If 
\begin{equation}
T(\alpha_i)=\beta_i,
\end{equation}
we can see
\begin{equation}
T(v)=\sum_i l_i(v) \beta_i\text{ for all }v\in\mathcal{V}.
\end{equation}
This construction also clearly demonstrates that if $\text{dim}(\mathcal{V})<\text{dim}(\mathcal{W})$, a linear mapping can't be surjective, and if $\text{dim}(\mathcal{V})>\text{dim}(\mathcal{W})$, a linear mapping can't be injective; see Section~\ref{sec:hom}.

Likewise, for any two vectors $w_1,w_2$ in a Hilbert space $\mathcal{H}$, we can define such an operator on $\mathcal{H}$:
\begin{equation}
h_{w_1,w_2}(v):=(w_2|v)w_1=w_1 l_{w_2}(v),\label{eq:haa}
\end{equation}
where $l_{w_2}$ is as defined in \eqref{eq:lalpha}, is an operator on $\mathcal{H}$. The linear combination of such operators is of course still an operator. Given an orthonormal basis $\{\alpha_i\}$, a vector $v\in\mathcal{H}$ is equal to
\begin{equation}
v=\sum_i (\alpha_i|v)\alpha_i,\label{eq:ab}
\end{equation}
which is the Fourier expansion (Section~\ref{sec:dual}). Given any bounded operator $O$ on $\mathcal{H}$, for all $v\in\mathcal{H}$
\begin{align}
Ov&=\sum_i (\alpha_i|Ov)\alpha_i\nonumber\\
&=\sum_{i,j}(\alpha_i|O\alpha_j)(\alpha_j|v)\alpha_i,\label{eq:oalpha1}
\end{align}
by linearity of inner product with respect to the second argument; namely, we obtain
\begin{equation}
O=\sum_{i,j} (\alpha_i|O\alpha_j) h_{i,j},\label{eq:Thij}
\end{equation}
where
\begin{equation}
h_{i,j}:=h_{\alpha_i,\alpha_j}.\label{eq:hii}
\end{equation}
Therefore, a bounded operator on a Hilbert space can be expanded in terms of $h_{i,j}$.

The approach above is heuristic in nature: We didn't deal with the convergence of the series, and \eqref{eq:oalpha1} and \eqref{eq:Thij} should be used with care---though they will pose no problems when the dimension is finite. The same problem will arise again in Section~\ref{sec:bas}. We will have a brief discussion about the convergence of such a series in Section~\ref{sec:id}.

\subsection{Orthogonal Direct Sum, Projections and the Identity Operator}\label{sec:dip}
Readers interested in content of this part may consult Ref.~\cite{Roman}.
\subsubsection{Direct sum}

Let $\mathcal{V}_1,\cdots ,\mathcal{V}_n$ be subspaces of a vector space $\mathcal{V}$. If for $\mathcal{V}_i \cap \mathcal{V}_j=\{0\}$ for any $i\neq j$ and for any $v\in\mathcal{V}$ there are $v_i\in\mathcal{V}_i$ such that $v=v_1+\cdots v_n$, then we say $\mathcal{V}$ is the \textbf{direct sum} of $\mathcal{V}_1,\cdots ,\mathcal{V}_n$, denoted by $\mathcal{V}=\mathcal{V}_1\oplus \cdots \oplus\mathcal{V}_n$. 

For subspaces $\mathcal{H}_1,\cdots ,\mathcal{H}_n$ of an inner product space $\mathcal{H}$, if $\mathcal{H}_i\perp \mathcal{H}_j$ for any $i\neq j$ and for any $v\in\mathcal{H}$ there are $v_i\in\mathcal{H}_i$ such that $v=v_1+\cdots v_n$, then $\mathcal{H}$ is said to be an \textbf{orthogonal direct sum} of $\mathcal{H}_1,\cdots ,\mathcal{H}_n$, denoted also by $\mathcal{H}_1\oplus\cdots \oplus\mathcal{H}_n$, because in this thesis the direct sum of inner product spaces is usually assumed to be orthogonal. 

Clearly an orthogonal direct sum is also a direct sum, but with a stronger requirement of orthogonality. For example, the space $\mathbb{R}^2$ is an orthogonal direct sum of $\text{span}((1,0))$ and $\text{span}((0,1))$, and it's also a direct sum of $\text{span}((1,0))$ and $\text{span}((1,1))$, but  $\text{span}((1,0))$ and $\text{span}((1,1))$ aren't orthogonal.

\subsubsection{Invariant subspaces and the direct sum of operators}
Let $\mathcal{V}=\mathcal{V}_1\oplus \mathcal{V}_2$ and $O\in\mathcal{L}(\mathcal{V})$. If $O(\mathcal{V}_1)\subseteq \mathcal{V}_1$, then $\mathcal{V}_1$ is said to be \textbf{invariant} under $O$; in other words, $O\upharpoonright_{\mathcal{V}_1}$ (Section~\ref{sec:map}) is an operator on $\mathcal{V}_1$. If both $\mathcal{V}_1$ and $\mathcal{V}_2$ are invariant subspaces of $O$, then we may express this as
\begin{equation}
O= O\upharpoonright_{\mathcal{V}_1}\oplus O\upharpoonright_{\mathcal{V}_2}.
\end{equation}
Hence, whenever we write $T=O_1\oplus O_2 \in\mathcal{L}(\mathcal{V})$, it implies there exist two invariant subspaces $\mathcal{V}_1$ and $\mathcal{V}_2$ of $O$, such that $\mathcal{V}_1\oplus \mathcal{V}_2=\mathcal{V}$ and $O\upharpoonright_{\mathcal{V}_i}=O_i$, and $O$ is called the \textbf{direct sum} of $O_1$ and $O_2$ \cite{Roman}.

From another point of view, suppose again $\mathcal{V}=\mathcal{V}_1\oplus \mathcal{V}_2$ and $O_1,O_2\in\mathcal{L}(\mathcal{V})$. If $\mathrm{ran}(O_i)\subseteq \mathcal{V}_i$ and if $\mathcal{V}_2\subseteq\ker(O_1)$ and $\mathcal{V}_1\subseteq\ker(O_2)$, i.e. $O_1(\mathcal{V}_2)=O_2(\mathcal{V}_1)=\{0\}$, then $O_i\upharpoonright_{\mathcal{V}_i}$ is an operator on $\mathcal{V}_i$, and $O_1+O_2=O_1\upharpoonright_{\mathcal{V}_1}\oplus O_2\upharpoonright_{\mathcal{V}_2}$---The operator $O_i\in\mathcal{L}(\mathcal{V})$ can be identified with $O_i\upharpoonright_{\mathcal{V}_i}\in\mathcal{L}(\mathcal{V}_i)$.

\subsubsection{Projection}\label{sec:pro}
For a vector space $\mathcal{V}$, if $\mathcal{V}=\mathcal{V}_1\oplus \mathcal{V}_2$, the linear operator $\Pi$ defined as
\begin{equation}
\Pi_{\mathcal{V}_1,\mathcal{V}_2}(v_1+v_2)=v_1 \;\forall v_1\in\mathcal{V}_1 \text{ and }\forall v_2\in\mathcal{V}_2
\end{equation}
is called a \textbf{projection} onto $\mathcal{V}_1$ along $\mathcal{V}_2$. By definition, $\ker \Pi_{\mathcal{V}_1,\mathcal{V}_2}=\mathcal{V}_2$, and $\mathrm{ran}\Pi_{\mathcal{V}_1,\mathcal{V}_2}=\mathcal{V}_1$. An operator $\Pi$ is a projection if and only if it's \emph{idempotent}, i.e. $\Pi^2=\Pi$. Two projections $\Pi_1$ and $\Pi_2$ are said to be \textbf{orthogonal} if $\Pi_1 \Pi_2 =\Pi_2 \Pi_1=0$ \cite{Roman}. 

On a Hilbert space $\mathcal{H}$, we have the following property \cite{Roman,Einsiedler}:
\begin{thm}
Suppose $\mathcal{H}$ is a Hilbert space and that $\mathcal{S}$ is a closed, and therefore complete subspace of $\mathcal{H}$. Then $\mathcal{S}^\perp$ is also a closed subspace, and $\mathcal{H}=\mathcal{S}\oplus \mathcal{S}^\perp$. 
\end{thm}
With this theorem in mind, let $\mathcal{S}$ be a closed subspace of a Hilbert space $\mathcal{H}$. A projection $\Pi_\mathcal{S}:= \Pi_{\mathcal{S},\mathcal{S}^\perp}$ is called an \textbf{orthogonal projection} onto $\mathcal{S}$. Furthermore, an operator $\Pi\in\mathcal{B}(\mathcal{H})$  is indempotent and self-adjoint if and only if it's an orthogonal projection onto a closed subspace \cite{Roman,HallQ}. Two orthogonal projections are orthogonal if and only if their images are orthogonal.\footnote{As confusing as they are, the context should be clear enough to distinguish these different concepts.}

\subsubsection{Identity operator and orthogonal projections}\label{sec:id}
For any vector space $\mathcal{V}$, an identity mapping $I:\mathcal{V}\rightarrow\mathcal{V}$ maps a vector to itself, i.e.
\begin{equation}
I(v)=v \text{ for any } v\in\mathcal{V}.\label{eq:identity}
\end{equation}
\eqref{eq:identity} alone ensures $I$ is linear, because $I(av)=av=aI(v)$ and $I(v_1+v_2)=v_1+v_2=I(v_1)+I(v_2)$. Because of its linearity it will be called the identity operator.

By Riesz-Fischer theorem (Theorem~\ref{thm:rf}), \eqref{eq:haa} and \eqref{eq:hii}, given an orthonormal basis $\{\alpha_i\}$ of a Hilbert space:
\begin{equation}
v=\hat{v}=\sum_i \beta_i (\beta_i|v)=\sum_i h_{i,i}(v)=I(v), \;\forall v\in\mathcal{H},\label{eq:vIv}
\end{equation}
so when the orthonormal basis is countable,
\begin{equation}
\sum_i h_{i,i}=I,\label{eq:I}
\end{equation}
i.e. $\sum_{i=1}^N h_{i,i}$ converges to $I$ as $N\rightarrow \infty$. It's worth mentioning that \eqref{eq:oalpha1} could be attained by applying \eqref{eq:I} twice. 

Here's a finer point about the convergence in \eqref{eq:I}: For a sequence of operators $\{T_n\}$, if $O_n v\rightarrow Ov$ for all $v$, it's called strong convergence \cite{Einsiedler}. It's easy to see \eqref{eq:I} doesn't converge in (operator) norm topology: Given any positive integers $m$, $n$ and $n<m$, $\sum_{i=1}^m h_{i,i}-\sum_{i=1}^n h_{i,i}$ is a projection to the space spanned by $\{\alpha_i: n<i\leq m\}$, and its operator norm is also $1$, so \eqref{eq:I} doesn't even form a Cauchy sequence (Section~\ref{sec:conban}) \cite{GrandpaRudin}. However, \eqref{eq:I} converges in the strong operator topology; see Ref.~\cite{Einsiedler}.

\eqref{eq:I} can be written in terms of projections: Again assume $\{\alpha_i\}$ is a countable orthonormal basis of a Hilbert space $\mathcal{H}$, and let $\Pi_{i}$ be the orthogonal projection onto $\text{span}(\alpha_i)$, which is a one-dimensional subspace. We can easily see $\Pi_{i}=h_{i,i},$ so
\begin{equation}
I=\sum_i \Pi_{i}.
\end{equation} 

More generally, if $\mathcal{S}$ is a closed subspace of a Hilbert space $\mathcal{H}$, and $\mathcal{S}$ has an orthonormal basis $\{\alpha_i\}$, then the orthogonal projection onto $\mathcal{S}$ is
\begin{equation}
\Pi_\mathcal{S}=\sum_i \Pi_{i},
\end{equation}
where $\Pi_{i}:=\Pi_{\text{span}(\alpha_i)}$, as above. The identity operator is essentially an (orthogonal) projection onto the whole space $\mathcal{H}$.

\subsection{Isomorphism with $\mathbb{C}^n$ and Matrices}\label{sec:ismat}
In this subsection only finite-dimensional spaces are considered. 
\subsubsection{$\mathbb{C}^n$ and $(\mathbb{C}^n)^*$}
\label{sec:iso}
For an $n$-dimensional Hilbert space $\mathcal{H}$, given an orthonormal basis $\{\alpha_i\}$, we can define a bijective linear mapping $g:\mathcal{H}\rightarrow\mathbb{C}^n$ (Definition~\ref{def:rncn}) by
\begin{equation}
g:\sum_{i}a_i\alpha_i\mapsto (a_1,a_2,\cdots,a_n).\label{eq:gai}
\end{equation}
Through this mapping $\mathcal{H}$ and $\mathbb{C}^n$ are isomorphic Hilbert spaces (Section~\ref{sec:isometry}):
\begin{equation}
(\sum_i b_i \alpha_i|\sum_i a_i \alpha_i)=(b_1,b_2\cdots,b_n)\cdot(a_1,a_2\cdots,a_n).
\end{equation}

Now let's talk about $(\mathbb{C}^n)^*$:
\begin{defi}
Let $(\mathbb{C}^n)^*$ be the space of $n$-tuples\footnote{It's more straightforward to define them as $n$-tuples for our purpose.} of complex numbers with a vector operation
\begin{equation}
c_1(a_1,\cdots,a_n)+c_2(b_1,\cdots,b_n)=(c_1 a_1+c_2 a_2,\cdots,c_1 a_n+c_2 b_n),\;c_1,c_2\in\mathbb{C},
\end{equation} 
and we require any $(b_1,b_2\cdots,b_n)\in (\mathbb{C}^n)^*$ satisfy the following property: For all $(a_1,\cdots,a_n)\in \mathbb{C}^n$, 
\begin{equation}
(b_1,b_2\cdots,b_n) \left( (a_1,\cdots,a_n)\right)=\sum_i b_i a_i. 
\end{equation}
\end{defi}

Any vector or tuple in $(\mathbb{C}^n)^*$ is clearly a linear functional on $\mathbb{C}^n$. There exists a conjugate linear mapping (or conjugate isomorphism) $*$ from $\mathbb{C}^n$ to $(\mathbb{C}^n)^*$:
\begin{equation}
*:\mathbb{C}^n\ni v=(a_1,\cdots,a_n)\mapsto v^*:=(a_1^*,\cdots,a_n^*)\in (\mathbb{C}^n)^*,\label{eq:*}
\end{equation}
so that
\begin{equation}
v^*(w)=(v|w) \text{ for any }v,w\in \mathbb{C}^n.
\end{equation}
Hence the conjugate isomorphism from $\mathcal{H}$ to $\mathcal{H}^*$, \eqref{eq:lalpha}, corresponds to this mapping $*$: For a vector $\sum_{i}a_i\alpha_i\in\mathcal{H}$ we have
\begin{equation}
\sum_{i}a_i\alpha_i\mapsto (a_1,a_2,\cdots,a_n)\in\mathbb{C}^n \mapsto (a_1^*,a_2^*,\cdots,a_n^*)\in(\mathbb{C}^n)^*,
\end{equation}
and 
\begin{equation}
l_w(v)=(w|v)=g(w)^*(g(v)) \text{ for all }w,v\in\mathcal{H}.\label{:inlw}
\end{equation}

We can define an isomorphism from $\mathbb{C}^n$ to ``column'' matrices $\overrightarrow{\mathbb{C}}^n$:
\begin{equation}
(a_1,a_2,\cdots,a_n)\mapsto \begin{pmatrix}
a_1\\a_2\\ \vdots\\a_n
\end{pmatrix}.\label{eq:gvec}
\end{equation} 
and a similar one between $(\mathbb{C}^n)^*$ and ``row'' matrices $\overrightarrow{\mathcal{C}^n}^*$. Therefore the inner product between $(b_1,\cdots,b_n),(a_1,\cdots,a_n)\in \mathbb{C}^n$ is equal to
\begin{align}
(b_1,\cdots,b_n)\cdot(a_1,\cdots,a_n)&=(b_1,\cdots,b_n)^* \left((a_1,\cdots,a_n)\right)\label{eq:ba}\\
&=(b_1^*,\cdots,b_n^*)\begin{pmatrix}
a_1\\a_2\\ \vdots\\a_n
\end{pmatrix},\label{eq:ba2}
\end{align} 
where in \eqref{eq:ba} $(b_1,\cdots,b_n)^*=(b_1^*,\cdots,b_n^*)\in (\mathbb{C}^n)^*$ and in \eqref{eq:ba2} $(b_1^*,\cdots,b_n^*)$ is the row matrix and \eqref{eq:ba2} is the usual matrix multiplication.

Composing \eqref{eq:gvec} and \eqref{eq:gai}, we obtain an isomorphism from $\mathcal{H}$ to column vectors $\overrightarrow{\mathbb{C}}^n$:
\begin{equation}
v=\sum_{i}a_i\alpha_i\mapsto \vec{v}:=\begin{pmatrix}
a_1\\a_2\\ \vdots\\a_n
\end{pmatrix},\label{eq:gi}
\end{equation}
and with the mapping $*$ from \eqref{eq:*} we have an isomorphism from $\mathcal{H}$ to row vectors $\overrightarrow{\mathcal{C}^n}^*$:
\begin{equation}
v=\sum_{i}a_i\alpha_i\mapsto \vec{v}^\dagger:=(a_1^*,\cdots,a_n^*)\in \overrightarrow{\mathcal{C}^n}^*.\label{eq:ggii}
\end{equation}
By \eqref{:inlw}, \eqref{eq:ba} and \eqref{eq:ba2}, they satisfy
\begin{equation}
(w|v)=\vec{w}^\dagger\cdot \vec{v}\;\forall w,v\in\mathcal{H}.
\end{equation}

\subsubsection{Operators}

Suppose the dimension of $\mathcal{H}$ is $n$. Given an orthonormal basis $\{\alpha_i\}$ of $\mathcal{H}$ define a bijective linear mapping from $\mathcal{B}(\mathcal{H})$ to the space of complex $n\times n$ matrices $\overrightarrow{\mathcal{C}}^{n\times n}$ as: 
\begin{equation}
\mathcal{B}(\mathcal{H})\ni O\mapsto \vec{O} \text{ with }\vec{O}_{ij}:=(\alpha_i|O\alpha_j).\label{eq:om}
\end{equation}
Its inverse is:
\begin{equation}
\vec{O}\mapsto \vec{O}_{ij}h_{i,j},\label{eq:mo}
\end{equation}
where $h_{i,j}$ was defined in \eqref{eq:haa} and \eqref{eq:hii}. \eqref{eq:Thij} suggests
\begin{equation}
O=\sum_{i,j} \vec{O}_{ij} h_{i,j}.\label{eq:OO}
\end{equation}
Furthermore, with two operators $N,O\in\mathcal{B}(\mathcal{H})$, for any $v \in\mathcal{H}$
\begin{align}
NOv&=\sum_{i,j,k,l}\vec{N}_{ij}\vec{O}_{kl}(\alpha_l|v)(\alpha_j|\alpha_k)\alpha_i\nonumber\\
&=\sum_{i,l}(\vec{N}\cdot\vec{O})_{il}h_{i,l}(v).
\end{align}
That is, through the mapping \eqref{eq:om} and \eqref{eq:mo}
\begin{equation}
\overrightarrow{O_1\circ O_2}\mapsto \vec{O}_1\cdot \vec{O}_2 \text{ and } \vec{O}_1\cdot\vec{O}_2\mapsto O_1\circ O_2.\label{eq:gg}
\end{equation} 
Operators on an $n$-dimensional Hilbert space and matrices $\overrightarrow{\mathcal{C}}^{n\times n}$ are algebraically isomorphic.

\subsubsection{Everything together}
After the isomorphism \eqref{eq:gi}, we can find
\begin{equation}
Ov=\vec{O}\cdot \vec{v}.\label{eq:Oij}
\end{equation}
Besides, via \eqref{eq:ggii} we obtain
\begin{equation}
l_w(Ov)=(w|Ov)=\vec{w}^\dagger\cdot \vec{O}\cdot \vec{v}.
\end{equation}
Along with \eqref{eq:gg}, it's been demonstrated that operators on $\mathcal{H}$, vectors in $\mathcal{H}$, and linear functionals defined by vectors in $\mathcal{H}$ can all be represented by matrices.

More generally, similar correspondences for other linear mappings and matrices can be established the same way, as long as the domains and codomains coincide. For example, it can be shown for any $T\in\mathcal{B}(\mathcal{H}_1,\mathcal{H}_2)$, where $\mathcal{H}_1$ and $\mathcal{H}_2$ are finite-dimensional,
\begin{equation}
T^\dagger \mapsto \vec{T}^\dagger;
\end{equation}
the adjoint of a linear mapping $T$ becomes the Hermitian conjugate of $T$'s corresponding matrix.

\subsection{Spectrum and Eigenvectors}\label{sec:spe}
As the reader is assumed to have a good grasp of finite-dimensional linear algebras, I will defer the proper definition of spectrum, and instead go straight for the spectral theorem for normal operators (Section~\ref{sec:ad}) on finite-dimensional spaces.
\subsubsection{Spectral theorem for finite-dimensional spaces}\label{sec:fispe}
\begin{defi}
Suppose $\mathcal{V}$ is a complex vector space. For $O\in\mathcal{B}(\mathcal{V})$, if there exists an $v\in \mathcal{V}$ such that $Ov=\lambda v$, $\lambda\in\mathbb{C}$, then $v$ is called an \textbf{eigenvector} of $O$ and $\lambda$ is the corresponding \textbf{eigenvalue}. The eigenvectors of a given eigenvalue form a subspace, called the \textbf{eigenspace}.
\end{defi}
\noindent In other words, any vector in an eigenspace is also an eigenvector. Now let's state the spectral theorem \cite{Roman}:
\begin{thm}[\textbf{Spectral theorem}]\label{thm:spe}
Suppose $\mathcal{H}$ is a finite-dimensional complex Hilbert space, and $O\in\mathcal{B}(\mathcal{H})$ is a normal operator. Then there exist subspaces $\mathcal{H}_1,\cdots,\mathcal{H}_n$ of $\mathcal{H}$ which are orthogonal to each other such that 
\begin{equation}
\mathcal{H}=\mathcal{H}_1\oplus\cdots\oplus \mathcal{H}_n,\label{eq:HV}
\end{equation}
and if $\Pi_i$ are the orthogonal projections onto $\mathcal{H}_i$,
\begin{equation}
O=\sum_{i=1}^n \lambda_i \Pi_i,\,\lambda_i\in\mathbb{C},\label{eq:Tn}
\end{equation}
$\lambda_i$ are exactly the eigenvalues of $O$. If $O$ is self-adjoint, then $\lambda_i$ are real, and if $O$ is unitary, $|\lambda_i|=1$. \eqref{eq:HV} and mutual orthogonality of $\mathcal{H}_i$ imply $\Pi_i \Pi_j=\delta_{ij} \Pi_{i}$ and $\sum_{i=1}^n \Pi_i=I.$

We can further require that all $\lambda_i$ be distinct; if so, $\mathcal{H}_i$ is the eigenspace of $\lambda_i$.   
\end{thm}

According to Theorem~\ref{thm:spe}, a vector $v$ is the eigenvector of $O$ if and only if $v$ belongs to one of $\mathcal{H}_i$, with eigenvalue $\lambda_i$. Two eigenvectors belonging to different subspaces $\mathcal{H}_i$ are orthogonal; for example, when they have different eigenvalues. 

By the spectral theorem, we can decompose $\mathcal{H}$ as the orthogonal direct sum (Section~\ref{sec:dip}) of the eigenspaces of a normal operator $O$. Namely, any vector $v$ in $\mathcal{H}$ can be decomposed as eigenvectors of a normal operator, by simply applying $I=\sum_i \Pi_i$ on $v$. It also implies for a normal operator $O$,\footnote{We also have similar relations concerning the range and kernel of an operator and its adjoint; see Ref.~\cite{GrandpaRudin,Einsiedler}.}
\begin{equation}
(\ker O)^\perp = \text{ran}O.\label{eq:keran}
\end{equation}

An eigenspace not being one-dimensional is commonly known as degeneracy by physicists.\footnote{In which case it may be convenient to introduce the concept of multisets for indicating the multiplicity  \cite{Roman}, but we won't do this in this thesis.} We can always find an orthonormal basis for an eigenspace---After finding a basis that spans the eigenspace, we can orthonormalize this basis by Gram-Schmidt method. As an eigenspace of a normal operator is always orthogonal to another eigenspace, whichever vectors are chosen as the orthonormal basis for the eigenspace, they're always orthogonal to eigenvectors belonging to other eigenspaces. A normal operator with ``degeneracy'' doesn't have a unique set of orthonormal eigenvectors. The orthonormal bases of all the eigenspaces $\mathcal{H}_i$ constitute an orthonormal basis of $\mathcal{H}$. 

An interesting example of ``degeneracy'' is the identity operator $I$, whose eigenvalue is 1 and the eigenspace is the entire $\mathcal{H}$. Any orthonormal basis comprises orthonormal eigenvectors of $I$.

\subsubsection{Spectrum}

\begin{defi}[\textbf{Spectrum}]
Let $\mathcal{V}$ be a Banach space. For $O\in\mathcal{B}(\mathcal{V})$, its spectrum is the set $\mathrm{Spec}(O)$ which comprises elements $\lambda$ such that $O-\lambda I$ doesn't have a bounded inverse.
\end{defi}

By the closed graph theorem or open mapping theorem, if a bounded operator has an inverse, its inverse is automatically bounded, so $O-\lambda I$ doesn't have a bounded inverse if and only if $O-\lambda I$ isn't bijective, in other words if and only if $O-\lambda I$ isn't one-to-one or isn't onto \cite{GrandpaRudin,Einsiedler,HallQ}. 

If there is a $v\in\mathcal{V}$ such that $Ov=\lambda v$, then $(O-\lambda I)v=0$, so $\lambda$ is in the spectrum of $\mathrm{Spec}(O)$. Hence eigenvalues of $O$ are in the spectrum. To phrase it a bit differently, $\lambda$ is an eigenvalue of $O$ if $O-\lambda I$ fails to be injective, that is if $\ker(O-\lambda I)$ contains not only the zero vector. The eigenspace associated with an eigenvalue $\lambda$ is then $\ker(O-\lambda I)$. 

If $\mathcal{V}$ is finite-dimensional, as discussed in Section~\ref{sec:hom} an operator $O\in\mathcal{B}(\mathcal{V})$ is injective if and only if it's surjective. Therefore, the spectrum is composed entirely of eigenvalues. Because an operator on a finite-dimensional space isn't one-to-one if and only if its determinant is zero, eigenvalues can be found by solving for $\det(O-\lambda I)=0$, as taught in elementary linear algebras.

On the other hand, if $\mathcal{V}$ is infinite-dimensional, in general the spectrum isn't discrete,\footnote{The discrete, or point spectrum of an operator is defined to be composed of all eigenvalues \cite{Einsiedler}.} and an operator may not even have eigenvectors (and eigenvalues) at all. The spectrum of a compact self-adjoint operator \cite{Loomis,Einsiedler,GrandpaRudin} is quite close to what we expect from operators on a finite-dimensional space, but not all bounded operators are compact. Also, even an unbounded operator may have a discrete spectrum, like the Hamiltonians of certain systems, such as harmonic oscillators.

It can be shown \cite{Einsiedler,HallQ}
\begin{thm}
For $O\in\mathcal{B}(\mathcal{\mathcal{V}})$, the spectrum is a closed, bounded, non-empty subset of $\mathbb{C}$. All elements $\lambda\in \mathrm{Spec}(O)$ obey $|\lambda|\leq ||O||$. If $O\in\mathcal{B}(\mathcal{\mathcal{H}})$ is self-adjoint, the spectrum is on the real line, and $\max (|\lambda|:\lambda \in\mathrm{Spec}(O))=||O||$.
\end{thm}

Even though we don't really need this theorem for finite-dimensional spaces, there are two points to be pointed out for those spaces: First, as there are only finite eigenvalues, the spectrum is a closed set, as it's a union of finite points. Second, from the spectral theorem for normal operators (Theorem~\ref{thm:spe}), the operator norm of a normal operator is exactly the largest absolute value of eigenvalues.

\subsubsection{Another form of the spectral theorem}
Having discussed the spectra of operators, now we can take a look at another form of the spectral theorem. We will skip some details; interested readers can refer to Ref.~\cite{HallQ,Blackadar}. 

Consider a finite-dimensional inner product space $\mathcal{H}$ and a normal operator $O\in\mathcal{B}(\mathcal{H})$. Define a vector-valued function $s$ on the spectrum $\mathrm{Spec}(O)$, called a \textbf{section}, and demand $s(\lambda_i)\in \mathcal{H}_i$ (Theorem~\ref{thm:spe}). We can further define an inner product for such functions, by
\begin{equation}
(s_1|s_2):=\sum_{\lambda_i\in\mathrm{Spec}(O)} (s_1(\lambda_i)|s_2(\lambda_i)).
\end{equation}
As $\mathrm{Spec}(O)$ is a finite set, sections form a finite-dimensional inner product space, called the \textbf{direct integral}, denoted by $\bigoplus_{\lambda_i\in\mathrm{Spec}(O)}\mathcal{H}_i$. The direct integral has the same dimension as $\mathcal{H}$, which can be verified. 

From Theorem~\ref{thm:spe}, requiring all $\lambda_i$ be distinct, then for any vector $v\in\mathcal{H}$
\begin{equation}
v=\sum_{\lambda_i\in\mathrm{Spec}(O)} \Pi_i v.
\end{equation}
This suggests we can define a mapping $U:v\mapsto s_v$ by 
\begin{equation}
(Uv)(\lambda_i)=s_v(\lambda_i):=\Pi_i v:=v_i\in\mathcal{H}_i \text{ for all }\lambda_i\in\text{Spec}(O).
\end{equation}
Such a section $s_v$ will satisfy
\begin{equation}
\sum_{\lambda_i\in\mathrm{Spec}(O)} s_v(\lambda_i)=v,
\end{equation}
and
\begin{equation}
(s_v|s_v)=\sum_i (\Pi_i v|\Pi_i v)=\sum_{i,j} (\Pi_i v|\Pi_j v)=(v|v)
\end{equation}
by $\Pi_i \Pi_j=\delta_{ij}\Pi_i$ and self-adjointness of orthogonal projections. Clearly $U$ is linear, so it is unitary. With this unitary mapping $U:\mathcal{H}\rightarrow \bigoplus_{\lambda_i\in\mathrm{Spec}(O)}\mathcal{H}_i$, we can find
\begin{equation}
(UOv)(\lambda_i)= \left(U \sum_j \lambda_i \Pi_i v\right)(\lambda_i)=\sum_j \lambda_j U(\Pi_j v)(\lambda_i)=\lambda_i \Pi_i v \text{ for all }v\in\mathcal{H},\label{eq:utv}
\end{equation}
because $\Pi_i\Pi_j=\delta_{ij}\Pi_i$. 

Given a section $s\in\bigoplus_{\lambda_i\in\mathrm{Spec}(O)}\mathcal{H}_i$, if $s(\lambda_i)=v_i\in \mathcal{H}_i$, then 
\begin{equation}
U^{-1}s=\sum_i v_i.
\end{equation}
Therefore by \eqref{eq:utv}
\begin{equation}
(UOU^{-1}s)(\lambda_i)=\lambda_i v_i=\lambda_i s(\lambda_i).
\end{equation}
Because $UOU^{-1}$ is an operator on $\bigoplus_{\lambda_i\in\mathrm{Spec}(O)}\mathcal{H}_i$, the equation above implies the unitary equivalence of $O$ on the space of sections is an operator that multiplies each section by $\lambda_i$, yielding a new section. We can now state the theorem:
\begin{thm}\label{thm:spe2}
Under the same assumption as Theorem~\ref{thm:spe}, we associate each point $\lambda_i$ in $\mathrm{Spec}(O)$ with an inner product space $\mathcal{H}_i$. There exists a unitary mapping $U$ from $\mathcal{H}$ to the direct integral $\bigoplus_{\lambda_i\in\mathrm{Spec}(O)}\mathcal{H}_i$ such that
\begin{equation}
(UOU^{-1})s(\lambda_i)=\lambda_i s(\lambda_i)
\end{equation}
for all $\lambda_i\in\mathrm{Spec}(O)$ and any $s\in\bigoplus_{\lambda_i\in\mathrm{Spec}(O)}\mathcal{H}_i$.
\end{thm}

The discussion before Theorem~\ref{thm:spe2} explains how to find the unitary mapping $U$, and shows properties thereof. If each spectral space $\mathcal{H}_i$ of $O$ is one-dimensional, then sections can thought of as complex (or real, if self-adjoint) functions on $\mathrm{Spec}(O)$, after choosing a basis for each $\mathcal{H}_i$. Theorem~\ref{thm:spe2} may seem redundant, but its infinite-dimensional counterpart proves fruitful, which we will touch upon informally in Section~\ref{sec:sper}.

\section{$L^2$-space}\label{sec:l2s}
In this section I'll give a quick review of $L^2$-space. Please read books such as Ref.~\cite{PapaRudin,Cohn,Einsiedler} for a more comprehensive picture.

\subsection{$L^p$-space}
\begin{defi}[$L^p$-\textbf{space}]
Let $(X,\mu)$ be a measure space \cite{PapaRudin,Cohn}, where $X$ is the underlying set and $\mu$ is a positive measure, and let $f$ be a complex measurable function on it. Define the $L^p$-norm for   $0<p<\infty$ as
\begin{equation}
||f||_p:=\left\{\int_X |f|^p \;d\mu\right\}^{1/p}.
\end{equation}
An $L^p$-space, denoted by $L^p(X,\mu)$ or $L^p(\mu)$, is the space formed by functions whose $L^p$-norms are finite. For $p=\infty$, define $L^\infty$-norm as the essential supremum \cite{PapaRudin,Einsiedler} of $f$. 
\end{defi}

A complex measurable function in $L^1(\mu)$ is called a \textbf{Lebesgue integrable} function, because its integral will have finite real and complex parts \cite{PapaRudin}. By Minkowski's inequality and $|f+g|\leq |f|+|g|$ \cite{PapaRudin}, an $L^p$-space is indeed a vector space (over $\mathbb{C}$).\footnote{As discussed in Section~\ref{sec:omap}, we assume pointwise addition and multiplication for vector operations on functions.} From the definition, it's apparent that if $f$ is in $L^p$-space, so is $f^*$, and they have the same $L^p$-norm.

The $L^p$-norm is actually a semi-norm, that is, there exist non-zero functions which have (semi-)norm 0 \cite{Einsiedler}. However, we can rid of this problem by establishing equivalence classes of functions by identifying functions which are the same \textbf{almost everywhere}, and the $L^p$-norm and the associated metric then become the norm and metric for equivalent classes of functions \cite{PapaRudin,Einsiedler}. In this context when we mention a function $f$, we actually refer to all functions that are identified as equivalent to $f$ i.e. an equivalent class, rather than a single function. 

There is a critical property of $L^p$-space \cite{PapaRudin}:
\begin{thm}
For $1\leq p\leq\infty$ and any positive measure $\mu$, $L^p(\mu)$ is a Banach space.
\end{thm}
Also, by H\"{o}lder's inequality \cite{PapaRudin}, if $f\in L^p(\mu)$, $g\in L^q(\mu)$ and $1/p+1/q=1$, then
\begin{equation}
||fg||_1\leq ||f||_p ||g||_q,\label{eq:Holderpq}
\end{equation}
so $fg$ is a Lebesgue integrable function.\footnote{$(fg)(x):=f(x)g(x)$, i.e. it's defined pointwise as well.} 

We're often interested in complex functions on $\mathbb{R}^n$. $L^p(\mathbb{R}^n)$ will denote $L^p$-space on $\mathbb{R}^n$ equipped with the Lebesgue measure, for whose definition please read the references listed at the start of this section.

\subsection{$L^2$-space and Fourier Transform}\label{sec:l2h}
\subsubsection{$L^2$-space as a Hilbert space}
An $L^2$-space is the space of functions whose $L^2$-norms are finite. Functions with finite $L^2$ norms are commonly referred to as \textbf{square-integrable} functions. By \eqref{eq:Holderpq}, if $f,g\in L^2(\mu)$
\begin{equation}
||f^*g||_1\leq ||f||_2 ||g||_2,\label{eq:HoCa}
\end{equation}
that is, $f^*g$ is integrable. Hence
\begin{equation}
\int_X f^*g \;d\mu
\end{equation}
has a definite value. As a result, we can define an inner product on an $L^2$-space by
\begin{equation}
(f|g):=\int_X f^* g\; d\mu \text{ for all }f,g\in L^2(\mu).
\end{equation}
$(f|f)$ is exactly $||f||_2^2$, so the $L^2$-norm is identical to the norm induced by this inner product, and this integral does satisfy all the requirements of an inner product, as laid out in Section~\ref{sec:inner}. Because $L^2$-space is complete, with this inner product it becomes a Hilbert space. Since $L^2$ space is a space of equivalence classes of functions, the zero vector is the class of functions identified with the zero function.

\subsubsection{A special case of $L^2$-spaces}
Now suppose $X=\{x_i\}$ is a finite set with cardinality $n$, and $\mu$ is the counting measure \cite{PapaRudin}. If $f$ and $g$ are complex bounded functions on $X$,
\begin{equation}
(f|g)=\sum_i f(x_i)^* g(x_i):=\sum_i f_i^* g_i.
\end{equation}
Hence we can associate this $L^2$-space with $\mathbb{C}^n$, or any $n$-dimensional space over $\mathbb{C}$.

An $L^2$-space with a counting measure is also denoted by $l^2$, and the space $X$ doesn't need to be finite. Fourier expansion can be thought of as an isometry between $L^2$ and $l^2$ spaces \cite{PapaRudin}.

\subsubsection{Fourier/Plancherel transform}\label{sec:plan}
The Fourier transform $\mathcal{F}$ takes a function on $\mathbb{R}^n$ to another also on $\mathbb{R}^n$. For any $f\in L^1(\mathbf{R}^n)$, it is defined as
\begin{equation}
\mathcal{F}(f)(\mathbf{k}):=\hat{f}(\mathbf{k}):=(2\pi)^{-n/2}\int_{\mathbb{R}^n} f(x)e^{-i\mathbf{k}\cdot \mathbf{x}}d\mathbf{x}.
\end{equation}
In addition to the mapping $\mathcal{F}$, the function $\hat{f}$ is also called the Fourier transform of $f$ as well. A function $f$ in $L^2(\mathbf{R}^n)$ doesn't always have a well-defined Fourier transform, because $f(x)e^{i\mathbf{k}\cdot \mathbf{x}}$ in general is not $L^1$. Fourier transform can be extended to $L^2$-functions by taking the limit \cite{PapaRudin,HallQ}:
\begin{equation}
\mathcal{F}(f)(\mathbf{k})=(2\pi)^{-n/2}\lim_{y\rightarrow\infty}\int_{|\mathbf{x}|\leq y} f(\mathbf{x})e^{-i\mathbf{k}\cdot \mathbf{x}}d\mathbf{x},
\end{equation}
and it becomes a unitary transform from $L^2(\mathbb{R}^n)$ to $L^2(\mathbb{R}^n)$, which is sometimes referred to as the Plancherel transform. The inverse of a Fourier transform, $\mathcal{F}^{-1}$, is
\begin{equation}
\mathcal{F}^{-1}(f)(\mathbf{x})=(2\pi)^{-n/2}\lim_{y\rightarrow\infty}\int_{|\mathbf{x}|\leq y} f(\mathbf{k})e^{i\mathbf{k}\cdot \mathbf{x}}d\mathbf{k}.
\end{equation}
Note that for an $L^2$-function $f$, $f$ and $\mathcal{F}^{-1}\mathcal{F}f$ are identical in terms of $L^2$-norm but not pointwise in general.

\subsection{Schwartz Space}\label{sec:Sch}
For an $n$-tuple of nonnegative integers $\boldsymbol{\alpha}=(\alpha_1,\cdots,\alpha_n)$, let $\mathbf{x}^{\boldsymbol{\alpha}}:=x_1^{\alpha_1}\cdots x_n^{\alpha_n}$ and $\partial^{\boldsymbol{\alpha}}:=\partial_1^{\alpha_n}\cdots \partial_n^{\alpha_n}$, where $\partial_i$ are the partial differential operators $\partial/\partial x_i$. Then we can define:
\begin{defi}[\textbf{Schwartz function}]
A Schwartz function $f$ is a smooth\footnote{Given the defining property, this statement of smoothness is actually redundant.} complex function on $\mathbb{R}^n$ such that
\begin{equation}
\lim_{|\mathbf{x}|\rightarrow\infty}|\mathbf{x}^{\boldsymbol{\alpha}} \partial^{\boldsymbol{\beta}}f(\mathbf{x})|=0,
\end{equation}
for all $\boldsymbol{\alpha},\boldsymbol{\beta}$ of nonnegative $n$-tuples. $\mathcal{S}(\mathbb{R}^n)$ will denote the set and space of all Schwartz functions on $\mathbb{R}^n$, named the \textbf{Schwartz space}.
\end{defi}

Schwartz functions are also refereed to as rapidly decreasing functions, for obvious reasons. Gaussian functions, and any functions that can be factorized into a Gaussian and a polynomial, are Schwartz functions. A Schwartz function has a well-defined Fourier transform; actually, the Fourier transform of a Schwartz function is also a Schwartz function. Because $\mathcal{S}(\mathbb{R}^n)$ is a dense subset of $L^2(\mathbb{R}^n)$, the Fourier transform can be extended to $L^2$-functions by continuity \cite{HallQ,PapaRudin}. Therefore instead of starting from $L^1$-functions, Plancherel transform can be defined via Schwartz functions.

\section{Tensor Product}\label{sec:tp}
Lots of material of this section is adapted from Ref.~\cite{Roman,HallQ}.

\subsection{Bilinear mapping}
\begin{defi}[\textbf{Bilinear mapping}]\label{def:bilinear}
Let $\mathcal{V}$, $\mathcal{W}$ and $\mathcal{X}$ be vector spaces over the same field $\mathbb{F}$. A function $f: \mathcal{V}\times\mathcal{W}\rightarrow\mathcal{X}$ is called a bilinear mapping if 
\begin{align*}
f(c_1 v_1+c_2 v_2,w)&=c_1 f(v_1,w)+c_2 f(v_2,w),\\
f(v,c_1 w_1+c_2 w_2)&=c_1 f(v,w_1)+c_2 f(v,w_2),
\end{align*}
for all $c_i\in\mathbb{F}$, $v,v_i\in\mathcal{V}$ and $w,w_i\in\mathcal{W}$. Namely, the mapping is linear in each of the two arguments. Bilinear mappings form a vector space, named $\text{hom}(\mathcal{V},\mathcal{W};\mathcal{X})$. If $\mathcal{X}$ is the base field $\mathbb{F}$, $f$ is called a \textbf{bilinear form}.
\end{defi}
Note $\mathcal{V}\times\mathcal{W}$ shouldn't be taken as a direct sum of vector spaces---it's simply a cartesian product of sets. Here are some instances of bilinear mappings: 
\begin{itemize}
\item Inner product of a space over $\mathbb{R}$ is bilinear, by linearity in the second argument and $(v_1|v_2)=(v_2|v_1)$ for all $v_1,v_2$.
\item The multiplication of two complex numbers is bilinear; same for any field.
\item More generally, the multiplication of an algebra is bilinear. The example the reader may be most familiar with is multiplication of matrices.
\item Suppose $\mathcal{V}$ and $\mathcal{W}$ are vector spaces. Then the mapping $\phi:\mathcal{L}(\mathcal{V},\mathcal{W})\times \mathcal{V}\rightarrow \mathcal{W}$ defined by
\begin{equation}
\phi(T,v)=T(v)
\end{equation} 
is bilinear.
\item The commutator of operators, $[X,Y]:=XY-YX$ is a bilinear mapping; in addition, it's antisymmetric, namely $[X,Y]=-[Y,X]$. 
\item For real functions on the phase space $\mathbb{R}^{2n}$, the Poisson bracket $\{f,g\}$ is an antisymmetric bilinear mapping; see Section~\ref{sec:psv}. 
\item Both the Poisson bracket and commutator are Lie brackets: The Lie bracket on a Lie algebra $\mathfrak{g}$, $[\cdot,\cdot]:\mathfrak{g}\times\mathfrak{g}\rightarrow\mathfrak{g}$ is a bilinear antisymmetric mapping. More precisely, a Lie algebra is an algebra equipped with a Lie bracket as the multiplication.
\end{itemize}

\subsection{Tensor Product}\label{sec:tensor}
As shown above bilinear mappings are indeed very common, and thereby comes the importance of tensor product:
\begin{defi}[\textbf{Universality of tensor product and tensors}]\label{def:tensor}
Let $\mathcal{V}$, $\mathcal{W}$ and $\mathcal{X}$ be vector spaces over the same field. A universal pair $(\mathcal{T},t:\mathcal{V}\times\mathcal{W}\rightarrow \mathcal{X})$ for bilinear mappings is a vector space $\mathcal{T}$ and a bilinear mapping $t$ such that for any bilinear mapping $f$ from $\mathcal{V}$ and $\mathcal{W}$ to any vector space $\mathcal{X}$ there exists a unique linear mapping $\tau_f:\mathcal{T}\rightarrow\mathcal{X}$ that has the following property:
\begin{equation}
f=\tau_f\circ t.
\end{equation}
$\tau_f$ is called the \textbf{mediating morphism} for $f$, $\mathcal{T}$ the \textbf{tensor product} of $\mathcal{V}$ and $\mathcal{W}$, denoted by $\mathcal{V}\otimes\mathcal{W}$, and the mapping $t$ a \textbf{tensor map} or \textbf{tensor product}. 

Vectors in $\mathcal{V}\otimes\mathcal{W}$ are called \textbf{tensors}. $t(v,w)$ is denoted by $v\otimes w$,and a tensor of the form $v\otimes w$ is called \textbf{decomposable}.
\end{defi}
Simply put, tensor product turns any bilinear mapping on $\mathcal{V}$ and $\mathcal{W}$ into a linear mapping on $\mathcal{V}\otimes\mathcal{W}$. Universal pairs are unique up to isomorphisms, so even though there exist multiple universal pairs, they can be considered equivalent, and any such a space can be referred to as ``the'' tensor product \cite{Roman}. 

The mapping from a bilinear mapping to the corresponding mediating morphism, $f\mapsto \tau_f$ is actually an isomorphism:
\begin{thm}
Let $\mathcal{V}$, $\mathcal{W}$ and $\mathcal{X}$ be vector spaces over the same field. Define the mapping $\phi:\text{hom}(\mathcal{V},\mathcal{W};\mathcal{X})\ni f\mapsto \tau_f\in \mathcal{L}(\mathcal{V}\otimes\mathcal{W},\mathcal{X})$ as the mediating morphism mapping. Then $\phi$ is linear, and it's an isomorphism, that is, $\text{hom}(\mathcal{V},\mathcal{W};\mathcal{X})$ and $\mathcal{L}(\mathcal{V}\otimes\mathcal{W},\mathcal{X})$ are isomorphic.
\end{thm}

Constructing the tensor product can be carried out in a coordinate (or basis) dependent or a coordinate free way.
\subsubsection{Coordinate dependent}
Let $\{e_i\}$ and $\{f_j\}$ be bases of $\mathcal{V}$ and $\mathcal{W}$. Then the mapping $t$ can be defined by assigning arbitrary values to $t(e_i,f_j)$, and extending it by bilinearity:
\begin{equation}
t(\sum_i a_i e_i,\sum_j b_j f_j)=\sum_{i,j}a_i b_j t(e_i,f_j).
\end{equation}
Since $t$ is the most general bilinear mapping, we demand $t(e_i,f_j)$ be linearly independent, and use the symbol $e_i\otimes f_j$ for $t(e_i,f_j)$. That is, we let 
\begin{equation}
\{e_i\otimes f_j\}
\end{equation}
be the basis of $t(\mathcal{V},\mathcal{W})$.

Such a construction does yield a tensor product. If $g$ is a bilinear mapping on $\mathcal{V}$ and $\mathcal{W}$, then the mediating morphism can be obtained by $\tau_g(e_i\otimes f_j):=g(e_i,f_j)$ and extending it by linearity, which is indeed a unique linear mapping on $t(\mathcal{V},\mathcal{W})$, so $t(\mathcal{V},\mathcal{W})$ can be identified with $\mathcal{V}\otimes\mathcal{W}$. From this construction we can also find
\begin{thm}
Suppose $\mathcal{V}$ and $\mathcal{W}$ are finite-dimensional vector spaces. Then
\begin{equation}
\mathrm{dim}(\mathcal{V}\otimes \mathcal{W})=\mathrm{dim}(\mathcal{V})\mathrm{dim}(\mathcal{W}).
\end{equation}
\end{thm}
\subsubsection{Coordinate independent}
In this approach, unlike the previous one, no bases of $\mathcal{V}$ and $\mathcal{W}$ are presumed, but additional mathematical concepts are required \cite{TuD,Roman}. The reader can skip this part if she/he wants. 

Again suppose $\mathcal{V}$ and $\mathcal{W}$ are vector spaces over the same field $\mathbb{T}$. Let $\mathcal{T}_{\mathcal{V}\times \mathcal{W}}$ be a vector space over $\mathbb{F}$ with basis $\mathcal{V}\times \mathcal{W},$ i.e. its basis is
\begin{equation}
\{(v,w):v\in\mathcal{V}, w\in\mathcal{W} \}.
\end{equation} 
Hence the space is infinite-dimensional, and in this space two vectors of the form $(v,w)$ are identical, i.e. $(v_1,w_1)=(v_2,w_2)$ if and only if $v_1=v_2$ and $w_1=w_2$; additionally, $c(v_1,v_2)\neq (cv_1,cv_2)$.

Let $\mathcal{S}$ be a subspace generated by such vectors:
\begin{align}
&c_1(v_1,w)+c_2(v_2,w)-(c_1v_1+c_2v_2,w),\nonumber\\
&c_1(v,w_1)+c_2(v,w_2)-(v,c_1w_1+c_2w_2).\label{eq:cvw}
\end{align}
Since $\mathcal{V}\otimes\mathcal{W}$ is universal, vectors like \eqref{eq:cvw}, by Definition~\ref{def:bilinear}, are expected to be zero in the tensor product space---a linear mapping always maps zero to zero. Hence, we identify the tensor product with the quotient space:
\begin{equation}
\frac{\mathcal{T}_{\mathcal{V}\times \mathcal{W}}}{\mathcal{S}}.
\end{equation} 
For example, $c(v,w)=(cv,w)=(v,cw)$, in particular $(0,w)=(v,0)=0$ by such identification. Any element of this space is the linear combination of cosets \cite{Roman}
\begin{equation}
\sum_i [(v_i,w_i)+S].
\end{equation}
Let $v\otimes w$ denote the coset $(v,w)+\mathcal{S}$. Then the space is composed of vectors of this form
\begin{equation}
\sum_i u_i\otimes v_i.
\end{equation}
The tensor map to this quotient space then is 
\begin{equation}
t(u,v)=u\otimes v=(u,v)+\mathcal{S}.
\end{equation}
It can be shown this does lead to a tensor product, i.e. $(\frac{\mathcal{F}_{\mathcal{V}\times\mathcal{W}}}{\mathcal{S}},t)$ is a universal pair \cite{Roman}.

\subsubsection{Rank of a Tensor}\label{sec:rank}
Given bases $\{e_i\}$ and $\{f_j\}$ of vector spaces $\mathcal{V}$ and $\mathcal{W}$, a tensor can be written as
\begin{equation}
x=\sum_{i,j} c_{ij} e_i\otimes f_j.
\end{equation}
$c_{ij}$ is called the \textbf{coordinate matrix} of the tensor.\footnote{The coordinate matrix is quite often treated as ``the'' tensor by physicists, which, even though is not entirely wrong per se, it offers a fragmentary perspective on what tensors truly are.} The rank of the coordinate matrix is independent of the choice of bases, and is defined as the \textbf{rank} of the tensor \cite{Roman}. A decomposable tensor therefore has rank 1.

\subsection{Linear Mappings and Tensor Product}\label{sec:lmtensor}

\begin{thm}\label{thm:emiso}
Suppose $\mathcal{V}$, $\mathcal{W}$, $\mathcal{V}'$ and $\mathcal{W}'$ are vector spaces. There exists a unique linear mapping $\theta: \mathcal{L}(\mathcal{V},\mathcal{V}')\otimes \mathcal{L}(\mathcal{W},\mathcal{W}')\rightarrow \mathcal{L}(\mathcal{V}\otimes\mathcal{W},\mathcal{V}'\otimes\mathcal{W}')$, defined by
\begin{equation}
\left(\theta(f\otimes g)\right)(v\otimes w):=f(v)\otimes g(w),\; v\in\mathcal{V},\;w\in\mathcal{W}\;f\in\mathcal{L}(\mathcal{V},\mathcal{V}')\;g\in\mathcal{L}(\mathcal{W},\mathcal{W}').
\end{equation} 
$\theta$ is injective, and an isomorphism when $\mathcal{V}$ and $\mathcal{W}$ are finite-dimensional.
\end{thm}

The tensor product of two linear mappings has often been treated as a linear mapping on the tensor product, and Theorem~\ref{thm:emiso} gives it a sound justification. By applying Theorem~\ref{thm:emiso} and a field being a vector space itself, the correspondence between several kinds of mappings and spaces can be established; interested readers can refer to Ref.~\cite{Roman}. 

\subsection{Multilinear Mappings and Tensor Product}
\begin{defi}[\textbf{Multilinear mapping}]
Suppose $\mathcal{V}_i$, $i=1,\cdots N$ and $\mathcal{W}$ are vector spaces over the same field $\mathbb{F}$. A mapping $f:\mathcal{V}_1\times \mathcal{V}_2\cdots\times\mathcal{V}_N\rightarrow \mathcal{W}$ is said to be a multilinear mapping if it's linear in each argument. If $\mathcal{W}$ is the base field $\mathbb{F}$, $f$ is called a multilinear form. 
\end{defi}
\noindent Clearly a bilinear mapping is also a multilinear mapping. Treating an $n\times n$ matrix as $n$ vectors in $\mathbb{C}^n$, a determinant is a multilinear form.

By the same token, the tensor product of multiple vector spaces can be also defined by universality: The tensor product space is universal for multilinearity, that is, any multilinear mapping is a composition of the tensor product and a linear mapping (mediating morphism), as discussed in Section~\ref{sec:tensor}. Please read Ref.~\cite{Roman} for details.

Tensor product of multiple vector spaces has two important properties:
\begin{thm}
Suppose $\mathcal{V}_i$, $i=1,\cdots n$ and $\mathcal{W}_j$, $j=1,\cdots m$ are vector spaces over the same field. Then,
\begin{description}
\item[Associativity:] $(\mathcal{V}_1\otimes\cdots\otimes \mathcal{V}_n)\otimes (\mathcal{W}_1\otimes \cdots \otimes\mathcal{W}_m)$ and $\mathcal{V}_1\otimes\cdots\otimes \mathcal{V}_n\otimes \mathcal{W}_1\otimes \cdots \otimes\mathcal{W}_m$ are isomorphic.
\item[Commutativity:] For any permutation $\sigma$ of $\{1,\cdots,n\}$, $\mathcal{V}_1\otimes\cdots\otimes \mathcal{V}_n$ and $\mathcal{V}_{\sigma(1)}\otimes\cdots\otimes \mathcal{V}_{\sigma(n)}$ are isomorphic.
\end{description}
\end{thm}
\noindent Hence, the orders of tensor product don't matter---They're all the same, up to isomorphism. There also exist embeddings (and isomorphisms) like Theorem~\ref{thm:emiso} for tensor product of multiple spaces, but I won't list them here; again please refer to Ref.~\cite{Roman}.

\subsection{Hilbert Space Tensor Product}\label{sec:HSTP}
So far the construction of tensor product has been among ``vanilla'' vector spaces, that is those without extra structures. Now let's shift our focus to inner product spaces. The following theorem allows us to have a ``nice'' inner product on the tensor product of two inner product spaces \cite{HallQ}:
\begin{thm}\label{thm:HTP}
Let $\mathcal{V}$ and $\mathcal{W}$ be inner product spaces, with inner products $(\cdot|\cdot)_\mathcal{V}$ and $(\cdot|\cdot)_\mathcal{W}$. Then there exists a unique inner product on $\mathcal{V} \otimes \mathcal{W}$ such that
\begin{equation}
(v_1\otimes w_1|v_2\otimes w_2)=(v_1| v_2)_\mathcal{V} (w_1|w_2)_\mathcal{W}\label{eq:OQ}
\end{equation}
for all $v_1,v_2\in\mathcal{V}$ and $w_1,w_2\in\mathcal{W}$.
\end{thm}
\noindent Given any two inner product spaces $\mathcal{V}$ and $\mathcal{W}$, $\mathcal{V}\otimes\mathcal{W}$ is always assumed to be an inner product space with such an inner product.

For two Hilbert spaces $\mathcal{H}_1$ and $\mathcal{H}_2$, if they're infinite-dimensional, the tensor product $\mathcal{H}_1\otimes \mathcal{H}_2$ in general is incomplete. To remedy this, define the \textbf{Hilbert space inner product} of $\mathcal{H}_1$ and $\mathcal{H}_2$ as the completion of $\mathcal{H}_1\otimes \mathcal{H}_2$ with respect to the inner product in Theorem~\ref{thm:HTP}, denoted by $\mathcal{H}_1\hat{\otimes} \mathcal{H}_2$. This space has an orthonormal basis obtained from those of the constituent spaces \cite{HallQ}:
\begin{thm}\label{thm:bah}
Let $\{a_i\}$ and $\{b_j\}$ be orthonormal bases of Hilbert spaces $\mathcal{H}_1$ and $\mathcal{H}_2$, respectively. Then $\{a_i\otimes b_j\}$ is an orthonormal basis of $\mathcal{H}_1\hat{\otimes} \mathcal{H}_2$.
\end{thm}

\subsection{Tensor Product of Functions}\label{sec:tf}
Suppose $X$ and $Y$ are two sets, and let $\mathcal{V}_X$ and $\mathcal{V}_Y$ be the spaces of functions from $X$ to $\mathbb{C}$ and from $Y$ to $\mathbb{C}$, respectively. We can define a mapping $\pi$ on $\mathcal{V}_X \times \mathcal{V}_Y$ such that
\begin{equation}
\pi(f,g)(x,y):=f(x)g(y) \text{ for all }x\in X,\,y\in Y,\, f\in\mathcal{V}_X,\, g\in\mathcal{V}_Y.\label{eq:pifg}
\end{equation}
$\pi$ is bilinear over $\mathbb{C}$ because multiplication of two complex numbers is bilinear; explicitly: 
\begin{align*}
\pi(c_1 f_1+c_2 f_2,g)(x,y)&=(c_1 f_1+c_2f_2)(x)g(y)\\
&=c_1 f_1(x)g(y)+c_2 f_2(x)g(y)\\
&=[c_1 \pi(f_1,g)+c_2\pi(f_2,g)](x,y), \text{ for all }x\in X,\,y\in Y,\,c_1,c_2\in\mathbb{C},\,f_1,f_2\in\mathcal{V}_X,
\end{align*}
implying $\pi(c_1 f_1+c_2 f_2,g)$ and $c_1 \pi(f_1,g)+c_2\pi(f_2,g)$ are identical; same for the second argument. We can therefore construct a mediating morphism $\tau_\pi (f\otimes g):=\pi(f,g)$, as the tensor product is the most fundamental bilinear mapping.

Having see how the tensor product can define new functions, here comes a theorem relevant to quantum mechanics \cite{HallQ}:
\begin{thm}\label{thm:L2}
Suppose $(X_1,\mu_1)$ and $(Y,\mu_2)$ be $\sigma$-finite measure spaces. Then $L^2(X\times Y,\mu_1\times \mu_2)$ and $L^2(X,\mu_2) \hat{\otimes} L^2(Y,\mu_2)$ are isomorphic, where $\pi$ as defined in \eqref{eq:pifg} is a unitary mapping.
\end{thm}
Beware of the Hilbert space tensor product. As the spaces are isomorphic, the mediating morphism is usually ignored, and the tensor product of two $L^2$-functions $f\otimes g$ is regarded as a function on $X\times Y$; nevertheless it's worthwhile to know the reason behind such an identification.

The mapping $\pi$ gives us a function on $X\times Y$. If $X$ and $Y$ are vector spaces, then $X\times Y$ can be identified with $X\oplus Y$. For example, if $f$ and $g$ are complex functions on $\mathbb{R}^{n_1}$ and $\mathbb{R}^{n_2}$, then $f\otimes g$ can be considered a complex function on $R^{n_1+n_2}$. By the theorem above (and the fact that Lebesgue measure $m^{n_1+n_2}$ is the completion of $m^{n_1}\times m^{n_2}$ \cite{PapaRudin}), $L^2(\mathbb{R}^{n_1+n_2})$ and $L^2(\mathbb{R}^{n_1})\hat{\otimes}L^2(\mathbb{R}^{n_2})$ are isomorphic Hilbert spaces.

\section{Quantum States and Dirac Notation}
\label{sec:qstate}
\subsection{Quantum Mechanics}
Here I will list some axioms of quantum mechanics. The axioms on probabilities and measurements won't be included.
\begin{axi}
The state of a quantum system is described by a nonzero vector in a complex Hilbert space $\mathcal{H}$. Two vectors $\psi_2$ and $\psi_1$ are considered the same state if there exists a nonzero scalar $c$ such that $\psi_2=c\psi_1$. 
\end{axi}
Therefore, a quantum state is actually an equivalent class of vectors. We usually (if not always) normalize the vectors such that they have norm 1. If so, then two vectors $\psi_2$ and $\psi_1$ belong to the same state if there's a scalar $c$ with $|c|=1$ such that $\psi_2=c\psi_1.$
\begin{axi}
The evolution of a quantum state in a (closed) system is governed by a one-parameter strongly continuous \cite{HallQ} unitary transformation, with time $t$ being the parameter. The associated infinitesimal generator is called the Hamiltonian.
\end{axi}
By Stone's theorem, a one-parameter strongly continuous unitary transformation $U(t)$ can always be expressed in terms of an exponential map, $e^{itA}$, where $A$ is called the infinitesimal generator, and it's a self-adjoint operator \cite{HallQ,Einsiedler}. Since quantum state vectors are normalized, by unitarity a quantum state always stays on the $1$-sphere as it evolves.

\begin{axi}\label{ax:ob}
Quantum observables are self-adjoint operators on $\mathcal{H}$. To each real function on a classical phase space there corresponds a self-adjoint operator. 
\end{axi}

Self-adjointness of quantum observables can be considered both an assumption and a necessity: Spectral theorem for normal operators ensures orthogonality of disjoint spectral spaces, so that states/vectors belong to distinct eigenspaces\footnote{If the Hilbert space is infinite-dimensional, then in general we should consider the spectral space, unless the spectrum is discrete; see Ref.~\cite{Einsiedler,HallQ,GrandpaRudin}.} can be distinguished from one another by inner product, and the spectrum of a self-adjoint operator is in the real line. Additionally, if an operator (on a finite-dimensional Hilbert space) has orthogonal eigenspaces and real eigenvalues, then it must be self-adjoint.

The process of converting a real function on a phase space to a self-adjoint operator is called quantization \cite{HallQ,deGosson}. For example, the $i$-th position coordinate function on the phase space $\mathbb{R}^{2n}$ becomes the position operator $x_i$; the $i$-th component of the angular momentum becomes the $L_i$ angular momentum operator. 

Quantization of $x_i$ and $p_i$ can also be thought of as finding an irreducible unitary representation of the Heisenberg group (and algebra); the uniqueness (up to unitary transformation) of such a representation is ensured by Stone-von Neumann theorem. The irreducible representation is exactly $L^2(\mathbb{R}^n)$---hence we describe the motion of a particle moving in $\mathbb{R}^n$ as a state in $L^2(\mathbb{R}^n)$ \cite{HallQ,Woit,deGosson}. We'll have a short discussion about these operators in Section~\ref{sec:pomo}, and will talk about a quantization scheme in Section~\ref{sec:Weyl}.

\begin{axi}
Let $\mathcal{H}_A$ and $\mathcal{H}_B$ be the spaces describing two quantum systems. The state of the composite system is in the Hilbert space $\mathcal{H}_A\hat{\otimes}\mathcal{H}_B$.
\end{axi}

The inner product on $\mathcal{H}_A\hat{\otimes}\mathcal{H}_B$ is the one presented in Theorem~\ref{thm:HTP}. It should be emphasized (again) that it's a Hilbert space tensor product, rather than a mere tensor product. This axiom should be compared with Theorem~\ref{thm:L2}---As the state of a particle moving in $\mathbb{R}^n$ is a nonzero vector in $L^2(\mathbb{R}^n)$, the state of two moving particle is described by a nonzero vector in $L^2(\mathbb{R}^{n+n})$.

\subsection{Density Operator}\label{sec:deop}
The state of a system can't be always described by a ``wave function.'' For example, to portray the state of system A for a state vector in $\mathcal{H}_A\hat{\otimes}\mathcal{H}_B$ in general can't be achieved by a vector in $\mathcal{H}_A$. Hence comes the following concept \cite{HallQ}:
\begin{defi}
A linear functional $\Phi: \mathcal{B}(\mathcal{H})\rightarrow \mathbb{C}$ is called \textbf{a family of expectation values} if
\begin{enumerate}
\item $\Phi(I)=1$.
\item $\Phi(O)$ is real if $O\in\mathcal{B}(\mathcal{H})$ is self-adjoint.
\item $\Phi(O)\geq 0$ if $O$ is a positive operator.
\item Continuity with respect to strong convergence: For any sequence $\{O_n\}$ of operators in $\mathcal{B}(\mathcal{H})$, if $||O_n v-Ov||\rightarrow 0$ for any $v\in\mathcal{H}$, then $\Phi(O_n)\rightarrow \Phi(O)$.
\end{enumerate}
\end{defi}  
The motivation is that given an observable the state has a corresponding expectation value, and the four properties above are reasonable requirements of it.
\begin{defi}\label{def:den}
An operator $\rho\in\mathcal{B}(\mathcal{H})$, where $\mathcal{H}$ is a complex Hilbert space, is called a density operator if it it positive and has trace 1. It's also called a density matrix, state matrix, or state operator.
\end{defi}
A density matrix may also refer to the matrix isomorphism of a density operator with respect to an orthonormal basis. By definition a density operator is a trace class operator; see Section~\ref{sec:HSinner} for more properties of such class of operators. Families of expectation values and density operators are in fact two sides of the same coin:
\begin{thm}\label{thm:fev}
For any family of expectation values $\Phi$, there's always a unique density operator $\rho$ such that $\Phi(O)=\mathrm{tr}(\rho O)=\mathrm{tr}(O\rho)$ for all $O\in\mathcal{B}(\mathcal{H})$.
\end{thm}

Axiom 1 can be adjusted accordingly:
\begin{axi}
The state of a quantum system is described by a density operator $\rho\in\mathcal{B}(\mathcal{H})$. The expectation value of an observable $O$ on $\mathcal{H}$, if it's bounded, is equal to $\mathrm{tr}(\rho O)=\mathrm{tr}(O\rho)$.
\end{axi}
A density operator is often just called ``the state.'' Theorem~\ref{thm:fev} suggests that a quantum state can also be described by a (unique) family of expectation values; for example in Ref.~\cite{Stormer} it's the family of expectation values instead of the density operator that is referred to as ``the state.'' A quantum state $\rho$ is said to be pure if $\rho$ is an orthogonal projection onto a one-dimensional subspace, and is said to be mixed otherwise. Whereas there's ambiguity due to phase (or a nonzero constant if not normalized) when dealing with state vectors, this problem isn't present in density operators---A quantum state corresponds to exactly one density operator.

We will discuss more properties of density operators in Section~\ref{sec:HSinner}, after formally introducing trace-class operators.

\subsection{Dirac Notation}
In quantum mechanics, there are some conventions of notation, which are not always strictly followed \cite{HallQ}: 
\begin{itemize}
\item A vector in $\mathcal{H}$ is encased in a \textbf{ket}, for example $\ket{\psi},$ so
\begin{equation}
\ket{a\psi}=a\ket{\psi},\,a\in\mathbb{C}.
\end{equation}
\item The inner product between any two vectors $\ket{\psi_1}$ and $\ket{\psi_2}$ in $\mathcal{H}$ is denoted by $\inner{\psi_1}{\psi_2}$.
\item The continuous functional $l_\phi:\ket{\psi}\mapsto \inner{\phi}{\psi}$, where $\ket{\phi}$ can be any vector in $\mathcal{H}$, is enclosed in a \textbf{bra}:
\begin{equation}
\bra{\phi}:=l_{\phi}.
\end{equation}
Because the mapping $\phi\mapsto l_\phi$ is conjugate linear (Section~\ref{sec:lfun}), $\bra{a\phi}=a^*\bra{\phi}$. As Riesz-Fr\'{e}chet theorem (Theorem~\ref{thm:Riesz}) entails, the dual space of $\mathcal{H}$, $\mathcal{H}^*$, is composed of ``bras.''

\item For a linear mapping $T$ on $\mathcal{H}$, the vector $T\psi$ for $\psi\in\mathcal{H}$ is very often denoted by $T\ket{\psi}$. Likewise the inner product $\langle \psi_1|T\psi_2\rangle$ is often denoted by $\bracket{\psi_1}{T}{\psi_2}$. 
\item The mapping $\ket{\phi}\mapsto \ket{\psi_1} \langle \psi_2|\phi\rangle$, \eqref{eq:haa}, is denoted by
\begin{equation}
\ket{\psi_1}\bra{\psi_2}
\end{equation}

\item The tensor map $\otimes$ and identity operators are sometimes ignored: For $\psi_i\in\mathcal{H}_i$, $\ket{\psi_1}\ket{\psi_2}$ means $\ket{\psi_1}\otimes \ket{\psi_2}$, similar for ``bras.'' For a linear mapping $T_1$ on $\mathcal{H}_1$, $T_1\otimes I_{\mathcal{H}_2}$ is often written simply as $T_1$.

\end{itemize}

There are some notable mappings:
\begin{enumerate}
\item An orthogonal projection $\Pi$ onto a closed subspace whose orthonormal basis is $\{\ket{\alpha_i}\}$ is equal to
\begin{equation}
\Pi=\sum_i\ket{\alpha_i}\bra{\alpha_i}.
\end{equation}
In particular, if the subspace is one-dimensional then $\Pi=\ket{\alpha}\bra{\alpha}$, and if $\mathcal{H}$ has an orthonormal basis $\{\ket{\beta_i}\}$, the identity operator is equal to
\begin{equation}
I=\sum_i \ket{\beta_i}\bra{\beta_i};
\end{equation}
see Section~\ref{sec:id}.

\item $(\ket{\psi_1}\bra{\psi_2})(\ket{\psi_3}\bra{\psi_4})=\inner{\psi_2}{\psi_3}\ket{\psi_1}\bra{\psi_4}$, as $
h_{\psi_1,\psi_2}\circ h_{\psi_3,\psi_4}=(\psi_2|\psi_3)h_{\psi_1,\psi_4}.$

\item The adjoint of $\ket{\psi_1}\bra{\psi_2}$ is $\ket{\psi_2}\bra{\psi_1}$ because\footnote{We can also choose $\phi_1=\phi_2=\phi$, because if $(\phi|O\phi)=0$ for all $v$ then $O=0$ \cite{GrandpaRudin}.}
\begin{equation}
\Big(\phi_1\Big| \psi_1(\psi_2|\phi_2)\Big)=(\psi_2|\phi_2)(\phi_1|\psi_1)=\Big(\psi_2(\psi_1|\phi_1)\Big|\phi_2\Big)\text{ for any }\phi_1,\,\phi_2\in\mathcal{H}.
\end{equation} 

\item For $O\in\mathcal{B}(\mathcal{H})$ and $\ket{\psi_1},\ket{\psi_2}\in\mathcal{H}$, the mapping $\ket{\psi_1}\bra{\psi_2} O$, as a composition $(\ket{\psi_1}\bra{\psi_2})\circ O$, equals $ \ket{\psi_1}\bra{O^\dagger\psi_2}$ because for any $\ket{\phi}\in\mathcal{H}$
\begin{equation}
\ket{\psi_1}\bra{\psi_2} O_2\ket{\phi}=\ket{\psi_1}\inner{O^\dagger \psi_2}{\phi}=(\ket{\psi_1}\bra{O^\dagger\psi_2})\ket{\phi}.
\end{equation}
\end{enumerate}

\section{Group Actions and Representations}
We won't go through all the details and definitions for the material in this and the next section. Readers interested in these topics can refer to references such as Ref.~\cite{HallQ,Hall,Woit,Hilgert,Einsiedler}.
\subsection{Group Action}
\begin{defi}[\textbf{Group action}]
The group action of a group $G$ on a set $X$ is a mapping $\cdot: G\times X\rightarrow X$ such that 
\begin{equation}
g_1\cdot(g_2\cdot x)=(g_1 g_2)\cdot x \text{ for all }g_1,g_2\in G \text{ and }x\in X.
\end{equation}
Here $g_1 g_2$ denotes the group multiplication $g_1\cdot g_2$.
\end{defi}

If $G$ is a topological group (a group with a topology, for which multiplication and inversion are continuous), it's often required that the group action be continuous. An example of group action is $\mathbb{C}/\{0\}$, with multiplication of complex numbers as group multiplication, on a complex vector space $\mathcal{V}$, defined by\footnote{This action, being a linear mapping, also constitutes a representation as will be discussed in the next subsection---as trivial a representation as it is.}
\begin{equation}
c\cdot v:=cv,\,c\in\mathbb{C}/\{0\}\text{ and } v\in\mathcal{V}.
\end{equation}

Properties of groups put some constraints on the group action---Let $e$ be the identity of group $G$. It should be obvious that 
\begin{equation}
e\cdot x=x\text{ for all } x\in X.\label{eq:ex}
\end{equation}
In addition,
\begin{equation}
g\cdot (g^{-1}\cdot x)=g^{-1}\cdot(g\cdot x)=e\cdot x=x,\label{eq:inv}
\end{equation}
so any group action is invertible/bijective, in that given any $g\in G$ the mapping $x\mapsto g\cdot x$ always has an inverse mapping $x\mapsto g^{-1} x$. 

Through a group action we obtain a group homomorphism: If we define
\begin{equation}
f_g(x):=g\cdot x \text{ for all }g\in G\text{ and }x\in X,
\end{equation} 
then 
\begin{equation}
f_{g_1}\circ f_{g_2}=f_{g_1 g_2}.\label{eq:com}
\end{equation} 
From \eqref{eq:ex} we see that $f_e$ is the identity function, and \eqref{eq:inv} shows that $f_{g^{-1}}={f_g}^{-1}$. Hence, the image of $G\ni g\mapsto f_g$ is also a group, and $g\mapsto f_g$ is a group homomorphism. 

To put it another way, stemming from a group action demands that $\{f_g\}$ have an identity function, and that every function in $\{f_g\}$ have an inverse. These two conditions make $\{f_g\}$ a group (of invertible functions), with function composition as the multiplication. Associativity of a group is already built-in for functions: $f_1\circ (f_2\circ f_3)=(f_1\circ f_2)\circ f_3$ for any functions $f_i$ if their domains and codomains coincide. \eqref{eq:com} then renders $g\mapsto f_g$ a homomorphism. 

The discussion above inspires another definition of group actions:
\begin{defi}[\textbf{Another definition of group actions}]
Suppose $X$ is a set and $G$ is a group. Let $R(X)$ be the set of bijective functions from $X$ to $X$\footnote{In other words, automorphisms on $X$.}, and $h:G\rightarrow R(X)$ be any group homomorphism. A group action is the mapping $G\times X \ni(g,x)\mapsto h(g)(x)$. 
\end{defi}

To end this part, let me give another example of group actions: Take $\mathbb{R}$ as an addition group. For $a\in\mathbb{R}$ we can define an action on $\mathbb{R}^3$ by $a\cdot (x,y,z)=(x-a,y,z)$. This corresponds to mappings $f_a$ that shift the $x$-coordinate of an $\mathbb{R}^3$ vector. Such mappings form a group, and $a\mapsto f_a$ is a group homomorphism (and isomorphism). They are a subgroup of the group of all translations on $\mathbb{R}^3$, which group is isomorphic to $\mathbb{R}^3$.

\subsection{Representation}
\begin{defi}[\textbf{General linear group}]
A general linear group on a vector space $\mathcal{V}$ is the group of invertible operators, denoted by $GL(\mathcal{V})$.
\end{defi}

If $\mathcal{V}$ is the space $\mathbb{C}^n$, then the corresponding general linear group is commonly denoted by $GL(n,\mathbb{C})$, and $GL(n,\mathbb{C})$ is identified as a matrix group. The space $\mathcal{V}$ can be any vector space---it doesn't need to be finite-dimensional. Note a general linear group isn't a vector space.

\begin{defi}[\textbf{Representation}]
Let $G$ be a group and $\mathcal{V}$ be a vector space. A group representation is a homomorphism $\pi: G\rightarrow GL(\mathcal{V})$.
\end{defi}

From the discussion in the previous subsection, a representation can be identified with a group action that is linear in $X=\mathcal{V}$, and we can define a representation that way. For a group that is expressible in terms of matrices, in other words, isomorphic to a matrix group, it has a natural representation: For example, a natural representation of $SO(3)$ on $\mathbb{R}^3$ is the group of $3\times 3$ orthogonal matrices with unit determinant, and the homomorphism $\pi$ can be thought of as $I$, the identity mapping. 

Another well-known representation is the adjoint representation for a matrix Lie group $G$, $Ad: G\rightarrow GL(\mathfrak{g})$: Here $\mathfrak{g}$ is the Lie algebra of group $G$, and $\mathfrak{g}$ itself is a (real) vector space of matrices.\footnote{To be more precise, the tangent space at the identity of $G$.} The mapping $Ad(g)$ is defined as $Ad(g): \mathfrak{g}\ni X\mapsto g X g^{-1}$, which is clearly linear. It can be shown that $g X g^{-1}$ is in $\mathfrak{g}$ for any $g\in G$ and $X\in\mathfrak{g}$, and that the adjoint representation $G\ni g\mapsto Ad(g)$ is indeed a homomorphism.

Here's a type of representations we're interested in: Let $G$ be a group, $X$ a set and $\mathcal{V}$ a vector space. The functions $X\rightarrow \mathcal{V}$ form a vector space, denoted by $\mathcal{K}(X,\mathcal{V})$, and suppose $f$ is one such function. If there's a group action of $G$ on $X$, then 
\begin{equation}
f_g(x):=f(g^{-1}\cdot x) \;\forall x\in X\label{eq:f_g}
\end{equation}
is another function in $\mathcal{K}(X,\mathcal{V})$. This gives rise to a representation \cite{Woit,HallQ}:
\begin{thm}
If there's a group action of $G$ on $X$, the mapping $\pi$ defined by
\begin{equation}
\pi(g)f:=f_g \text{ for all } f\in\mathcal{K}(X,\mathcal{V})
\end{equation}
is a representation, i.e. $\pi: G\rightarrow GL(\mathcal{K}(X,\mathcal{V}))$ is a homomorphism.
\end{thm}

\begin{proof}
First, it's clear that given $g\in G$, the mapping $\pi(g):f\mapsto f_g$ is a linear mapping:
\begin{equation*}
(cf)_g= c (f_g)\text{ and } (f_1+f_2)_g={f_1}_g+{f_2}_g.
\end{equation*}
We can also easily see that $(f_{g})_{g^{-1}}=(f_{g^{-1}})_{g}=f$ and $f_e=f$, so 
\begin{equation*}
\pi(e)=I \text{ and } \pi(g^{-1})\pi(g)=\pi(g)\pi(g^{-1})=I \;\forall g\in G.
\end{equation*}
To see $\pi(g_1 g_2)=\pi(g_1)\pi(g_2)$, we can show $(f_{g_2})_{g_1}=f_{g_1 g_2}$ for all $f\in\mathcal{K}(X,\mathcal{V})$:
\begin{equation*}
f_{g_1 g_2}(x)=f\left((g_1g_2)^{-1}\cdot x\right)=f\left((g_2^{-1} g_1^{-1})\cdot x\right)=f\left(g_2^{-1}\cdot (g_1^{-1}\cdot x)\right)=f_{g_2}\left(g_1^{-1}\cdot x\right)=(f_{g_2})_{g_1}(x),
\end{equation*}
true for all $x\in X$. Therefore $\pi$ is a representation.
\end{proof}

The proof should make it obvious why \eqref{eq:f_g} was defined that way: If it were instead defined as $f_g(x):=f(g\cdot x)$, then $\pi$ wouldn't be a homomorphism unless, for example, $G$ is abelian. $\mathbb{R}^n$, $\mathbb{C}^n$, $SO(2)$ and $U(1)$ are instances of abelian groups.

\subsection{Unitary Representation}\label{sec:ur}
Let $(X,\mu)$ be a measure space. We say a mapping $f:X\rightarrow X$ is \textbf{measure-preserving} if
\begin{equation}
\mu(f(E))=\mu(E)\label{eq:muE}
\end{equation}
for all measurable sets $E$. Similarly, a group action on a measure space is said to be measure-preserving if \cite{Einsiedler}
\begin{equation}
\mu(g\cdot E)=\mu(E)
\end{equation}
for all $g\in G$ and all measurable sets $E$. Note it's implied the mapping and the group action always bring a measurable set to another measurable set.\footnote{Another way to say this, though maybe quirky, is that the mapping is an endomorphism and the group action is an automorphism on the $\sigma$-algebra. \eqref{eq:muE} implies it's an endomorphism, and being a group action implies bijectivity.} If $X=\mathbb{R}^n$ and $\mu$ is the Lebesgue measure, translations and any operators $O$ on $\mathbb{R}^n$ with $|\det (O)|=1$, e.g. rotations, are measure-preserving.

Suppose $\mathcal{H}$ is a Hilbert space. A representation $\pi$ is called a \textbf{unitary representation} if $\pi(g)$ is unitary for all $g\in G$, i.e. $\pi:G\rightarrow \mathcal{U}(\mathcal{H})$ is a homomorphism, where $\mathcal{U}(\mathcal{H})$ is the group of unitary operators on $\mathcal{H}$. Suppose $G$ is a group, $(X,\mu)$ is a measure space and there's a measuring-preserving group action of $G$ on $X$. As discussed in Section~\ref{sec:l2h}, $L^2(X,\mu)$ is a Hilbert space. For all $f\in L^2(X,\mu)$ let's define a mapping $f_g$ like \eqref{eq:f_g}. Then $f_g$ is also in $L^2(X,\mu)$ and $||f_g||_2=||f||_2$, because
\begin{equation}
||f_g||_2^2=\int_X |f_g|^2 d\mu=\int_X |f(g^{-1}\cdot x)|^2 d\mu=\int_X |f(x)|^2 d\mu=||f||_2^2,
\end{equation}
by the group action being measure-preserving and invertible. Therefore, 
\begin{thm}\label{thm:gxl2}
For a group $G$ with a meaure-preserving action on $(X,\mu)$, the mapping $\pi$, defined as $\pi(g)f=f_g$ for all $g\in G$ and $f\in L^2(X,\mu)$, is a unitary representation on $L^2(X,\mu)$.
\end{thm}

\section{Quantum Operators on $L^2(\mathbb{R}^n)$ and Representations}
\label{sec:pomo}
In this section the position and momentum operators will be labeled in uppercase $X_i$ and $P_i$, to avoid confusion with the position and momentum coordinate functions $x_i$ and $p_i$.
\subsection{Position and Momentum Operators}
To describe the (pure) state of a particle in $L^2(\mathbb{R}^n)$, we use a function in $L^2(\mathbb{R}^n)$. If there are $n$ particles moving in $\mathbb{R}^m$, then their state is a function in $L^2(\mathbb{R}^{mn})$. The position operators are the operators $X_i$ which multiply a (nice enough) function by $x_i$: $(X_i f)(\mathbf{x})=x_i f(\mathbf{x})$. The momentum operators $P_j$ are differential operators\footnote{On Schwartz functions both position and momentum operators are well-defined (Section~\ref{sec:Sch}).} $(h/i)\partial/\partial x_j$: $P_j f=(h/i)\partial f/\partial x_j$.

The way that momentum operators are defined can be thought of as a generalization of de Broglie's concept: A particle moving in $\mathbb{R}$ with a definite momentum $p$ corresponds to a wave with wave vector $k=p/\hbar$, that is, a wave of the form $e^{ikx}=e^{ipx/\hbar}$. This is exactly the eigenvector of the operator $(\hbar/i)d/dx$ with eigenvalue $p$ \cite{HallQ}. 

Nevertheless, $e^{ipx/\hbar}$ is not an $L^2$-function on $\mathbb{R}$, so there's no states with definite momenta in $L^2(\mathbb{R})$. To remedy this, we observe that the Fourier transform and its inverse transform looks like a ``superposition'' of $e^{ipx/\hbar}$; in addition, the Fourier transform brings the momentum operator $(h/i) d/dx$ to an operator that multiplies the Fouier transform $\hat{f}$ by $p$, just like the position operator on a function $f$. This suggests the Fourier transform is the gateway to turning a position wave function to a momentum wave function, and that the momentum operator $(h/i)d/dx$ is indeed the right choice. 

\subsection{Unitary Representations on $L^2(\mathbb{R}^n)$}

\subsubsection{Momentum}
If $(X,\mu)$ is $\mathbb{R}^n$ with the Lebesgue measure, any group action preserving the Lebesgue measure has a corresponding unitary representation on $L^2(\mathbb{R}^n)$ (Theorem~\ref{thm:gxl2}). Consider a unitary representation of $\mathbb{R}^n$, as an addition group, on $L^2(\mathbb{R}^n)$: Let the representation be $\mathbb{R}^n\ni \mathbf{a}\mapsto U_\mathbf{a}$ such that 
\begin{equation}
(U_\mathbf{a} f)(\mathbf{x})=f(\mathbf{x}+\mathbf{a})
\end{equation}
$\{U_{\mathbf{a}}\}$ is a Lie group, and the Lie algebra of this group has a basis $\{A_j\}$ (for nice enough functions): 
\begin{equation}
A_j=-i\frac{\partial}{\partial x_j}, \;j=1,\cdots,n,
\end{equation}
which, except for difference in coefficients, are the momentum operators; or we can say that there's a Lie algebra representation $\mathbf{e}_i\mapsto A_i$ corresponding to this unitary representation.

\subsubsection{Position}
The previous unitary representation is a result of Theorem~\ref{thm:gxl2}. Now let's take a look at another unitary representation that doesn't derive from this theorem: Define unitary operators $V_\mathbf{a}$ for all $\mathbf{a}\in\mathbb{R}^n$ on $L^2(\mathbb{R}^n)$ by
\begin{equation}
(V_\mathbf{a} f)(\mathbf{x})=e^{i \mathbf{a}\cdot \mathbf{x}}f(\mathbf{x}) \text{ for all }f\in L^2(\mathbb{R}^n),
\end{equation}
where $\mathbf{a}\cdot \mathbf{x}$ is the inner/dot product on $\mathbb{R}^n$. The mapping $\mathbf{a}\mapsto V_\mathbf{a}$ again is a unitary representation.\footnote{It can be easily shown that $\{V_\mathbf{a}\}$ is a group of unitary operators, and that $\mathbf{a}\mapsto V_\mathbf{a}$ is a homomorphism.}  The Lie algebras of the group $\{V_\mathbf{a}\}$ has a basis $\{B_j\}$, which multiply functions by $x_j$:
\begin{equation}
(B_j f)(\mathbf{x})=x_j f(\mathbf{x}).
\end{equation}
Again, there's a corresponding Lie algebra homomorphism $\mathbf{e}_j\mapsto B_j$. $B_j$ are exactly the position operators.

\subsubsection{$SO(n)$ and angular momentum operators}

Let's consider the matrix Lie group $SO(n)$, defined by $n\times n$ orthogonal matrices with determinant 1, which has a natural representation, or group action $\pi=I$ on $\mathbb{R}^n$ or $\mathbb{C}^n$. With this and Theorem~\ref{thm:gxl2} we can define a unitary representation $\pi$ of $SO(n)$ on $L^2(\mathbb{R}^n)$ by
\begin{equation}
\pi(g)f(\mathbf{x})=f(g^{-1}\cdot \mathbf{x}) \text{ for all }\mathbf{x}\in\mathbb{R}^n\text{ and }f\in L^2(\mathbb{R}^n).\label{eq:rot}
\end{equation}
This is a representation of the rotation group on functions on the Euclidean space.

Now consider such a unitary representation $\pi$ for $SO(3)$. The Lie algebra of $SO(3)$, $so(3)$, has a basis
\begin{equation}
F_1=\begin{pmatrix}
0&0&0\\
0&0&-1\\
0&1&0
\end{pmatrix},\;
\begin{pmatrix}
0&0&1\\
0&0&0\\
-1&0&0
\end{pmatrix},\;
\begin{pmatrix}
0&-1&0\\
1&0&0\\
0&0&0
\end{pmatrix}.
\end{equation}
They satisfy the commutation relations $[F_1,F_2]=F_3$, $[F_2,F_3]=F_1$, $[F_3,F_1]=F_2$. The Lie algebra representation corresponding to $\pi$ has $F_j\mapsto L_j/(i\hbar)$, where $L_j$ are the angular momentum operators:
\begin{equation}
L_1=X_2 P_3-X_3 P_2,\, L_2= X_3 P_1-X_1 P_3,\, L_3= X_1 P_2-X_2 P_1.
\end{equation}
As a Lie algebra representation, the angular momentum operators should obey the commutation relations as $F_j$ do, for example
\begin{equation}
[F_1,F_2]=F_3\mapsto \left[\frac{L_1}{i\hbar},\frac{L_2}{i\hbar}\right]=\frac{L_3}{i\hbar}\Rightarrow [L_1,L_2]=i\hbar L_3.
\end{equation}
That is to say, as the basis of the Lie algebra unitary representation of $so(3),$ the commutation relations of $L_i$ stem from the structure of $so(3)$. In addition, similar to classical mechanics where angular momenta ``generate'' rotation of $\textbf{x}$ and $\textbf{p}$, angular momentum operators in quantum mechanics are the infinitesimal generators of the rotation group \eqref{eq:rot}.

\subsection{Representations and Wave Functions}

It can be found
\begin{equation}
\widehat{U_\mathbf{a}f}(\mathbf{k})=e^{-i\mathbf{a}\cdot \mathbf{k}}\hat{f}(\mathbf{k})
\end{equation}
and
\begin{equation}
\widehat{V_\mathbf{a}f}(\mathbf{k})=\hat{f}(\mathbf{k}+\mathbf{a}).
\end{equation}
Therefore, a translation corresponds to a phase shift in the Fourier transform, and a phase shift corresponds to a translation in the Fourier transform. Ignoring the factor $\hbar$, from a physical point of view this means 
\begin{itemize}
\item A phase shift in the position wave function corresponds to a shift of coordinate in the momentum wave function; similarly, shifting the coordinate in the position wave function corresponds to a phase shift in the momentum wave function. 
\item Position operators are the infinitesimal generators of translating the momentum, and momentum operators are the infinitesimal generators of translating the position.
\item Translations in both momentum and position are unitary representations of $\mathbb{R}^n$, and position and momentum operators are bases of the associated Lie algebra representations.  $[X_i,X_j]=[P_i,P_j]=0$ reflects $\mathbb{R}^n$ being abelian.  
\end{itemize}

Th operators $U_\mathbf{a}$ and $V_\mathbf{a}$ are special cases of Weyl-Heisenberg operators. They're related to the unitary representation of the Heisenberg group (sometimes known as the Weyl group) on $L^2(\mathbb{R}^n)$ \cite{Woit,deGosson}. 
\subsection{Spectral Theorem and Realization of a State}\label{sec:sper}
For details in this part please refer to Ref.~\cite{Einsiedler,HallQ}.

By the spectral theorem in the direct integral formulation, of which the finite-dimensional version is Theorem~\ref{thm:spe2}, for any self-adjoint operator $O$ on $\mathcal{H}$ there is a unitary mapping $U$ from $\mathcal{H}$ to a Hilbert space (the direct integral) formed by sections $s$, such that $UOU^{-1}$ is a multiplication operator: $(UTU^{-1})(s)(\lambda)=\lambda s(\lambda)$, $\lambda \in \mathrm{Spec}(O)$. For the momentum operators, the Fourier transform $\mathcal{F}$ is exactly the unitary mapping $U$: For $f\in L^2(\mathbb{R}^n)$ and $j$-th component momentum operator $P_j$, $U(P_jf)= UP_j{U}^{-1} Uf=p_j Uf$,\footnote{Here $p_j$ is considered a ``coordinate'' function: $p_j(\mathbf{p})=p_j$.} i.e. a momentum operator $P_j$ is an operator that multiplies a function (section) by $p_j$ on the momentum space. 

A benefit of this form of spectral theorem\footnote{Another commonly seen formulation of the spectral theorem is one that utilizes the projection-valued measure and functional calculus; see Ref.~\cite{HallQ,Einsiedler,GrandpaRudin}.} is that a ``quantum state'' can be regarded as an independent entity---The realization of a quantum state is a vector or section in the direct integral associated with a self-adjoint operator, in which direct integral the self-adjoint operator and its functional calculus act as multiplication operators. The spaces in which the state is realized are unitarily related, so they can be considered equivalent. The position and momentum wave functions are unitarily the same,\footnote{As stated in the Plancherel theorem \cite{PapaRudin,HallQ}; also see Section~\ref{sec:plan}.} just realized in different spaces, or direct integrals: one corresponding to the position operators, the other to the momentum operators. If a self-adjoint operator has a discrete spectrum, e.g. the Hamiltonian of a harmonic oscillator, then it's akin to Theorem~\ref{thm:spe2}. 

In standard quantum mechanics, we often express a particular realization like an inner product, e.g. $\inner{p}{\psi}$. However such an inner product may not exist, as the vector to ``realize'' the state, in this case $\ket{p}$ may not be in the Hilbert space. However, this ``inner product'' may be seen as a distribution or linear functional that may not be continuous; see Ref.~\cite{HallQ,deGosson,Einsiedler,Grubb}.

It should be reemphasized that unlike finite-dimensional spaces for which spaces of the same dimension are isomorphic, how infinite-dimensional spaces are configured decide whether they're isomorphic. For example, assuming $\mathbb{S}$ is a compact subset of $\mathbb{R}$, $L^2(\mathbb{S})$, unsurprisingly, isn't isomorphic to $L^2(\mathbb{R})$. Describing the space of states clearly may seem superfluous, but it's paramount---As an instance, the ``momentum'' and position operator on $L^2([0,1])$ don't obey the uncertainty relation \cite{HallQ}.

\section{Phase Spaces and Symplectic Vector Spaces}\label{sec:psv}
In this short section we will develop some concepts required for later chapters. Please consult Ref.~\cite{Goldstein,Woit,deGosson} for more details.
\subsection{Classical Mechanics and the Phase Space}\label{sec:cm}
As a starter, let's lay out the axioms of classical (Hamiltonian) mechanics:
\begin{axi}
The state of a classical mechanical system is given by a point in $\mathbb{R}^{2n}=\mathbb{R}^n\oplus \mathbb{R}^n,$ where the former $\mathbb{R}^n$ describes the position $\mathbf{x}$ and the latter the (conjugate) momentum $\mathbf{p}$ of the system. In this context, the space $\mathbb{R}^{2n}$ is usually called a \textbf{phase space}.
\end{axi}
\begin{axi}
A (classical) observable is a real function on the phase space $\mathbb{R}^{2n}$. For example the coordinate functions or canonical coordinates $x_i$ and $p_i$ are observables, for which 
\begin{align}
x_i((\mathbf{x},\mathbf{p}))&=\mathbf{x}_i,\\
p_i((\mathbf{x},\mathbf{p}))&=\mathbf{p}_i.
\end{align}
\end{axi}
\begin{axi}
There's a distinguished observable, called the Hamiltonian of the system, denoted by $h$. The evolution of a system is determined by Hamilton's equations:
\begin{align}
\frac{dx_i}{dt}&=\frac{\partial h}{\partial p_i},\\
\frac{dp_i}{dt}&=-\frac{\partial h}{\partial x_i}.
\end{align}
\end{axi}
Here to ease matters we don't assume explicit time-dependency of observables, including $h$; namely they are functions on the phase space $\mathbb{R}^{2n}$, but not on $\mathbb{R}^{2n}\times \mathbb{R}$ or $\mathbb{R}^{2n}\oplus \mathbb{R}$. Now let's introduce
\begin{defi}[\textbf{Poisson bracket}]
The Poisson bracket $\{\cdot,\cdot\}$ is a bilinear mapping on functions on the phase space $\mathbb{R}^{2n}$, defined as
\begin{equation}
\{f_1,f_2\}:=\sum_{i=1}^n \left(\partiald{f_1}{x_i}\partiald{f_2}{p_i}- \partiald{f_1}{p_i}\partiald{f_2}{x_i}\right).
\end{equation}
\end{defi}

The time-dependency of an observable $f$ (that doesn't depend on time explicitly) along a trajectory of motion $(\mathbf{x}(t),\mathbf{p}(t))$ is 
\begin{equation}
\frac{df}{dt}(\mathbf{x}(t),\mathbf{p}(t))=\{f,h\},
\end{equation}
by Hamilton's equations. Specifically, Hamilton's equations can be rewritten as
\begin{equation}
\frac{dx_i}{dt}=\{p_i,h\},\;\frac{dp_i}{dt}=\{x_i,h\}.
\end{equation}
An observable $f$ is conserved under a Hamiltonian $h$ if $\{f,h\}=0$. 

Furthermore, the Poisson bracket is anti-symmetric and obeys the Jacobi identity, so it's a Lie bracket of functions on $\mathbb{R}^{2n}$. In other words, functions on $\mathbb{R}^{2n}$ form an infinite-dimensional Lie algebra under the Poisson bracket. The coordinate functions $x_i$ and $p_i$ are a basis of $(\mathbb{R}^{2n})^*$, and along with the function $\hat{1}$ that yields $1$ on every point of the phase space, they form a $2n+1$-dimensional subalgebra $(\mathbb{R}^{2n})^*\oplus \mathbb{R}$:
\begin{equation}
\{x_i,x_j\}=\{p_i,p_j\}=\{x_i,\hat{1}\}=\{p_i,\hat{1}\}=0,\, \{x_i,p_j\}=\delta_{ij}\hat{1}.\label{eq:Ms}
\end{equation} 
This subalgebra is isomorphic to the Heisenberg algebra \cite{Woit}: As the position and momentum operators on $L^2(\mathbb{R}^n)$ come from a unitary representation of the Heisenberg algebra, they are also the representation of the subalgebra formed by $x_i,p_i,\hat{1}$.

\subsection{Symplectic Vector Space}\label{sec:sym}

\begin{defi}[\textbf{Symplectic form}]
A symplectic form $\omega$ on a real vector space $\mathcal{V}$ is an anti-symmetric bilinear form (Definition~\ref{def:bilinear}) that is also non-degenerate, i.e. $\omega(v,w)=0$ for all $v\in\mathcal{V}$ if and only if $w=0$. A vector space equipped with a symplectic form is called a \textbf{symplectic vector space}.
\end{defi}

There's a standard symplectic form on $\mathbb{R}^{2n}$: Let $\hat{e}_i$ and $\hat{f}_i$ denote the standard basis of $\mathbb{R}^{2n}$,\footnote{$\hat{e}_i$ for the ``position'' components and $\hat{f}_i$ for the ``momentum'' ones.}
\begin{equation}
\omega(\hat{e}_i,\hat{e}_j)=\omega(\hat{f}_i,\hat{f}_j)=0,\,\omega(\hat{e}_i,\hat{f}_j)=-\delta_{ij},
\end{equation}
and extend this by bilinearity and anti-symmetry; explicitly
\begin{equation}
\omega\left((x_1,\cdots,p_n),(x_1',\cdots,p_n')\right)=\sum_i\left( p_i x_i'-p_i' x_i\right).
\end{equation}
We will always use this standard symplectic form on $\mathbb{R}^{2n}$.

\begin{defi}[\textbf{Symplectic group}]
	The standard symplectic group $Sp(n)$ is the group of all bijective linear mappings $\mathbb{R}^{2n}\rightarrow \mathbb{R}^{2n}$ that preserve the symplectic form $\omega$ on $\mathbb{R}^{2n}$, i.e. $\omega(z,z)=\omega(Sz,Sz')$ for all $z,z'\in\mathbb{R}^{2n}$ and $S\in Sp(n)$. 
\end{defi}

A symplectic mapping to a symplectic vector space is like a unitary mapping to a Hilbert space. A transformation (mapping) on the phase space is called a canonical transformation if the coordinates remain canonical. It can be shown a transformation is a (restricted) canonical transformation if and only if its differential (or Jacobian) is symplectic, and if and only if the Poisson bracket is invariant under such a transformation \cite{Goldstein}. In mathematics, a restricted canonical transformation is known as a symplectomorphism \cite{deGosson,Woit}.

The differential of a linear mapping is the mapping itself \cite{Loomis}, so a linear symplectomorphism is a symplectic mapping. In other words, as far as linear mappings are concerned, a phase space remains a phase space under a transformation if and only if the transformation is symplectic. This is the kind of symplectomorphism we will be interested in.

\chapter{Several Topics of Operators}\label{ch:op}
\section{Positive Operators}
\begin{defi}
An operator $P\in\mathcal{B}(\mathcal{H})$ is said to be positive (or non-negative or positive semi-definite) if
\begin{equation}
(v|Pv)\geq 0\;\forall v\in\mathcal{H}.\label{eq:po}
\end{equation}
\end{defi}

It can be easily shown that a positive operator is self-adjoint \cite{Einsiedler,PapaRudin}: Let $P$ be a positive operator, and $P=H_1+i H_2$, where $H_1$ and $H_2$ are self-adjoint. Because
\begin{equation}
(v|Pv)=(v|(H_1+iH_2)v)=(v|H_1v)+i(v|H_2v)\geq 0,
\end{equation}
Because $H_1$ and $H_2$ are self-adjoint, $(v|H_1v)$ and $(v|H_2v)$ are always real. For $(v|Pv)$ to be real, $(v|H_2v)$ must be 0 for all $v\in\mathcal{H}$. Because an operator $O=0$ if (and only if) \cite{PapaRudin}
\begin{equation}
(v|Ov)=0 \text{ for all }v\in\mathcal{H},
\end{equation}  
we find $H_2=0$, and thus $P$ is self-adjoint.

The condition \eqref{eq:po} implies $P$ is allowed to have a null space larger than $\{0\}$; on the other hand, a positive-definite operator $P'>0$ obeys $(v|Pv)> 0$ for all nonzero $v\in\mathcal{H}$, so $\ker P'=\{0\}$.

\section{Hilbert-Schmidt and Trace-class Operators}
\label{sec:HSinner}
In this section we will have a glance at two important classes of operators. For more information about such operators the reader may refer to Ref.~\cite{HallQ,Einsiedler,Blackadar,Kadison,Davies,deGosson}. 
\subsection{Trace and Hilbert-Schmidt Inner Product}
\begin{defi}[\textbf{Trace}]
Trace of an operator $O\in\mathcal{B}(\mathcal{H})$ is defined as \cite{Blackadar}
\begin{equation}
\trace O:=\sum_{i}(\beta_i|O (\beta_i)).\label{eq:tr}
\end{equation}
\end{defi}
\begin{defi}[\textbf{Hilbert-Schmidt inner product}]
For $T_1,T_2\in \mathcal{B}(\mathcal{H}_1,\mathcal{H}_2)$, their Hilbert-Schmidt inner product is \cite{Bhatia}:
\begin{equation}
(T_2|T_1):=\trace (T_2^\dagger\circ T_1),\label{eq:in1}
\end{equation}
where with an orthonormal basis $\{\beta_i\}$ in $\mathcal{H}$. Note $T_2^\dagger\circ T_1\in\mathcal{B}(\mathcal{H}_1).$ 
\end{defi}
Another way of expressing \eqref{eq:in1} is
\begin{equation}
(T_2|T_1)=\sum_{i}(T_2 (\beta_i)|T_1 (\beta_i)).\label{eq:in2}
\end{equation}
It's an inner product, because it satisfies all the properties for an inner product; see Section~\ref{sec:inner}.

\subsection{Hilbert-Schmidt and Trace-class Operators}\label{sec:htop}
When the space is infinite-dimensional, the trace of an operator and Hilbert-Schmidt inner product between operators don't necessarily exist, which inspires the following definition:
\begin{defi}
A bounded operator on a Hilbert space $O\in\mathcal{B}(\mathcal{H})$ is said to be a \textbf{trace-class operator} if $\trace |O|<\infty$, and a \textbf{Hilbert-Schmidt operator} if $(O|O)=\trace O^\dagger O<\infty$. 
\end{defi}
Accordingly, the trace of a trace-class operator always exists, so does the inner product between two Hilbert-Schmidt operators. When the dimension is infinite, not all bounded operators are trace-class or Hilbert-Schmidt, e.g. the identity operator. With Hilbert-Schmidt inner product, the space of Hilbert-Schmidt operators (on a Hilbert space) is a Hilbert space \cite{HallQ}. Both trace class and Hilbert-Schmidt operators are compact operators, so they have discrete spectra \cite{Blackadar,Einsiedler}, and trace-class operators are Hilbert-Schmidt operators \cite{Davies,Blackadar,HallQ,deGosson}. If $T$ is trace-class and $O$ is bounded, then $TO$ and $OT$ are both trace-class \cite{Davies,Blackadar,HallQ}, so given a quantum state the expectation value of a bounded operator exists and is finite.\footnote{By this it's implied that the codomain of the trace is the extended real line $[-\infty,\infty]$ instead of the real line $(-\infty,\infty)$.}

There are two important theorems concerning trace class operators, but first we need to define \cite{Einsiedler,Davies,HallQ}
\begin{defi}[\textbf{Integral operator}]
Let $(X,\mu)$ be a measure space, and $f$ any measurable complex function on $(X,\mu)$. The mapping $T$
\begin{equation}
(Tf)(x):=\int_X k(x,y) f(y) \,d\mu(y)
\end{equation}
with $k: X\times X\rightarrow \mathbb{C}$, is an operator. This kind of operators are called integral operators and the function $k$ is called the (integral) \textbf{kernel} of $T$.
\end{defi}
The following theorem not only shows Hilbert-Schmidt operators on $L^2(\mathbb{R}^n)$ are integral operators but also conveys some of the conditions under which an integral operator is well-defined \cite{HallQ,Davies}:
\begin{thm}\label{thm:hsi}
Let $T$ be an integral operator with integral kernel $k\in L^2(\mathbb{R}^n\times \mathbb{R}^n)$. Then for any $f\in L^2(\mathbb{R}^n)$, $Tf\in L^2(\mathbb{R}^n)$, and $T$ is a Hilbert-Schmidt operator on $L^2(\mathbb{R}^n)$. Conversely, any Hilbert-Schmidt operator on $L^2(\mathbb{R}^n)$ has a corresponding integral operator with a unique kernel in $L^2(\mathbb{R}^n\times \mathbb{R}^n)$.
\end{thm}
A positive trace-class operator on $L^2(\mathbb{R}^n)$, as a Hilbert-Schmidt operator, is an integral operator as well, and its trace may be determined by: 
\begin{thm}[\textbf{Mercer's theorem}]
If a positive and bounded integral operator $T$ has a continuous kernel $k$, then
\begin{equation}
\mathrm{tr} O=\int_X k(x,x)\,d\mu(x)
\end{equation}
\end{thm}
Note Mercer's theorem doesn't imply such an integral operator is trace-class: Its trace may be $+\infty.$  

\subsection{Another Look at Density Operators}\label{sec:ado}
A density operator $\rho$ (Section~\ref{sec:deop}) is said to be pure if it's an orthogonal projection onto a one-dimensional space; in other words, it's a mapping of this form $\ket{\psi}\bra{\psi}$, where $\ket{\psi}$ is normalized. If the Hilbert space is $L^2(\mathbb{R}^n)$, what this mapping does on an $L^2$-function $f$, \emph{by definition}, is 
\begin{equation}
(\ket{\psi}\bra{\psi} f)(\mathbf{x})=\psi(\mathbf{x})\int_{\mathbf{R}^n}\psi^*(\mathbf{y})f(\mathbf{y})  \,d\mathbf{y}=\int_{\mathbf{R}^n}\psi(\mathbf{x})\psi^*(\mathbf{y})f(\mathbf{y})  \,d\mathbf{y}.
\end{equation}
Therefore, $\ket{\psi}\bra{\psi}$ is an integral operator with kernel $k(\mathbf{x},\mathbf{y})=\psi(\mathbf{x})\psi^*(\mathbf{y})$. As an orthogonal projection onto a one-dimensional subspace, $\trace(\ket{\psi}\bra{\psi})=1$, agreed with Mercer's theorem if $\psi$ is continuous: $\int k(\mathbf{x},\mathbf{x})\,d\mathbf{x}=\int |\psi(\mathbf{x})|^2\,d\mathbf{x}=1$. 

In general, by definition a density operator is trace-class. According to Theorem~\ref{thm:hsi}, any density operator on $L^2(\mathbb{R}^n)$, being trace-class and therefore Hilbert-Schmidt, has a corresponding integral operator with a kernel in $L^2(\mathbb{R}^n\times \mathbb{R}^n)$. This is often denoted by 
\begin{equation}
\rho f(\mathbf{x})= \int_{\mathbb{R}^n} \rho(\mathbf{x},\mathbf{y}) f(\mathbf{y}) \,d\mathbf{y} \;\text{for any }f\in L^2(\mathbb{R}^n),
\end{equation}
where $\rho$ on the right side of the equality is an operator, whereas $\rho(\mathbf{x},\mathbf{y})$ on the left side is the kernel of $\rho$. Having trace 1, by Mercer's theorem the kernel $\rho(\mathbf{x},\mathbf{y})$, if \emph{continuous}, satisfies $\int_{\mathbb{R}^n} \rho(\mathbf{x},\mathbf{x})\,d \mathbf{x}=1$. The kernel is sometimes referred to as a ``density matrix.'' Beware of the application of Mercer's theorem: In general $\int_{\mathbb{R}^n} \rho(\mathbf{x},\mathbf{x})\,d \mathbf{x}\neq 1$ without extra conditions on the kernel or $\rho$ \cite{deGosson}.

Now let's consider a density operator $\rho$ on any separable Hilbert space $\mathcal{H}$. Being trace-class and thus compact, it has a discrete spectrum, and there exist orthogonal subspaces $\mathcal{H}_i$ of $\mathcal{H}$ such that \cite{Einsiedler,Loomis,deGosson}
\begin{equation}
\rho=\sum_i \lambda_i \Pi_{\mathcal{H}_i},
\end{equation} 
where $\lambda_i$ are eigenvalues of $\rho$ and $\Pi_{\mathcal{H}_i}$ are orthogonal projections onto $\mathcal{H}_i$. As a positive operator with unit trace, $\lambda_i\geq 0$ and $\sum_i \lambda_i=1$. If each $\mathcal{H}_i$ has an orthonormal basis $\{\psi_{i,j}\}$, then in terms of bra-ket we obtain
\begin{equation}
\rho=\sum_{i,j(i)}\lambda_i \ket{\psi_{i,j}}\bra{\psi_{i,j}}.\label{eq:rhoij}
\end{equation}

\subsection{Basis}\label{sec:bas}
With Hilbert-Schmidt inner product, we can choose an orthonormal basis for the Hilbert space formed by Hilbert-Schmidt operators. If the Hilbert space $\mathcal{H}$ they operate on has an orthonormal basis $\{\ket{a_i}\}$, we can define an orthonormal basis by \cite{HallQ,dePillis67,Jamiolkowski72}
\begin{equation}
	E_{ij}:=\ket{a_i}\bra{a_j}.\label{eq:eij}
\end{equation}
Clearly $E_{ij}^\dagger=E_{ji}$, so $\{E_{ij}^\dagger\}=\{E_{ij}\}$. The Fourier expansion of a Hilbert-Schmidt operator $O$ in terms of this basis is
\begin{equation}
O=\sum_{i,j} (E_{ij}|O)E_{ij}=\sum_{i,j}\bracket{a_i}{O}{a_j}\ket{a_i}\bra{a_j}.
\end{equation}
We can also write this as (c.f. Section~\ref{sec:id})
\begin{equation}
\mathcal{I}=\sum_{i,j} \Pi_{E_{ij}},
\end{equation}
where $\mathcal{I}$ is the identity mapping/operator on the space of Hilbert-Schmidt operators, and $\Pi_{E_{ij}}$ is the orthogonal projection onto the one-dimensional subspace spanned by $E_{ij}$. In a more quantum-mechanical language,
\begin{equation}
\mathcal{I}=\sum_{i,j}|E_{ij})(E_{ij}|.
\end{equation}

In terms of matrices, $E_{ij}$ is a matrix whose entries are $1$ at $(i,j)$ and 0 elsewhere, and the Fourier expansion is expressing a matrix as a linear combination of matrices $E_{ij}$ \cite{Choi75}. In general, the series will converge in the norm induced by Hilbert-Schmidt inner product, but whether it converges ``strongly'' needs further justification. Since we will use such expansions only in the finite-dimensional case, this won't be a problem.

\section{Weyl Quantization}\label{sec:Weyl}
Here we will examine how quantum operators can be linked to functions on the phase space. This will be useful in analyzing certain infinite-dimensional quantum systems.
\subsection{Weyl Quantization and Dequantization}\label{sec:weqde}
Weyl quantization is one of the quantization schemes, which map functions on the phase space to operators on $L^2$-functions on the configuration space, c.f. Axiom~\ref{ax:ob}. The process of finding the corresponding classical function on the phase space of a quantum operator/observable is called dequantization, which, like quantization, is not unique \cite{HallQ,deGosson}.
\begin{defi}[\textbf{Weyl quantization}]
Suppose $f$ is any function in $L^2(\mathbb{R}^{2n})$. Define $k: \mathbb{R}^{n}\times \mathbb{R}^{n}\rightarrow \mathbb{C}$ as
\begin{equation}
	k(\mathbf{x},\mathbf{y}):=\frac{1}{(2\pi\hbar)^n}\int_{\mathbb{R}^n}f(\frac{\mathbf{x}+\mathbf{y}}{2},\mathbf{p})e^{-i(\mathbf{y-}\mathbf{x})\cdot\mathbf{p}/\hbar}\,d\mathbf{p}.
\end{equation}
	The integral operator $Q_W(f)$ with kernel $k$, which is an operator on $L^2(\mathbb{R}^n)$, is called the Weyl quantization of $f$.
\end{defi}
The integral above may not always converge: Like the Fourier (Plancherel) transform on $L^2$-functions, it can be treated as $\lim_{R\rightarrow\infty}\int_{|\mathbf{p}|\leq R}$, rather than $\int_{\mathbb{R}^n}$. Even though the function is assumed to be $L^2$, there are non-$L^2$-functions with well-defined Weyl quantization. For example, the function $x_i p_i$ on $\mathbb{R}^{2n}$ has a Weyl quantization $(X_i P_i+P_i X_i)/2$, and the Weyl quantization of $f=1$ is $I$ \cite{deGosson}. More generally, the Weyl quantization of any polynomial $(\mathbf{a}\cdot \mathbf{x}+\mathbf{b}\cdot \mathbf{p})^j$ is $(\mathbf{a}\cdot \mathbf{X}+\mathbf{b}\cdot \mathbf{P})^j$ for all $\mathbf{a},\mathbf{b}\in\mathbb{R}^n$ and non-negative $j$, which is commonly known as \emph{symmetric ordering} by physicists \cite{HallQ,Adesso14,Barnett}.

The inverse of Weyl quantization, or dequantization, is \cite{HallQ,deGosson}:
\begin{thm}\label{thm:we}
Let $O$ be a Hilbert-Schmidt operator on $L^2(\mathbb{R}^n)$ with kernel $k(\mathbf{x},\mathbf{y})$. The inverse mapping of Weyl quantization is
\begin{align}
Q_W^{-1}(T)(\mathbf{x},\mathbf{p})&=\int_{\mathbb{R}^n}k(\mathbf{x}-\frac{ \mathbf{y}}{2},\mathbf{x}+\frac{\mathbf{y}}{2})e^{i\mathbf{y}\cdot\mathbf{p}/\hbar}\,d\mathbf{y}\\
&=2^n\int_{\mathbb{R}^n}k(\mathbf{x}-\mathbf{y},\mathbf{x}+\mathbf{y})e^{2i\mathbf{y}\cdot\mathbf{p}/\hbar}\,d\mathbf{y},
\end{align}
and $Q_W^{-1}(O)$ is in $L^2(\mathbb{R}^{2n})$. In addition, for $f\in L^2(\mathbb{R}^{2n})$, $Q_W(f^*)=Q_W(f)^\dagger$.
\end{thm}
This theorem shows that Hilbert-Schmidt operators on $L^2(\mathbb{R}^n)$ can be dequantized, and that if $f$ is a real function, then its Weyl quantization will be self-adjoint. The ``dequantized'' function will be called the \textbf{Weyl symbol} of an operator. The pair of an operator $T$ and its Weyl symbol $t$, or a function $t$ on the phase space and its Weyl quantization $T$, following Ref.~\cite{deGosson}, will be denoted by
\begin{equation}
T\weyl t.\label{eq:we}
\end{equation}

\subsection{Wigner Transform}\label{sec:wtr}

\begin{defi}[\textbf{Wigner transform}]\label{def:Wig}
Let $\psi\in L^2(\mathbb{R}^n)$. The Wigner transform $W:L^2(\mathbb{R}^n)\rightarrow L^2(\mathbb{R}^{2n})$ is defined as $\psi\mapsto (2\pi\hbar)^{-n} Q_W^{-1}(\ket{\psi}\bra{\psi})$, explicitly, 
\begin{equation}
W(\psi)(\mathbf{x},\mathbf{p})=\frac{1}{(2\pi\hbar)^n}\int_{\mathbb{R}^n}\psi(\mathbf{x}-\frac{\mathbf{y}}{2})\psi^*(\mathbf{x}+\frac{\mathbf{y}}{2}) e^{i \mathbf{y}\cdot\mathbf{p}/\hbar}\,d\mathbf{y}. 
\end{equation}
\end{defi}
As discussed in Section~\ref{sec:ado}, $\ket{\psi}\bra{\psi}$ is an integral operator with kernel $k(\mathbf{x},\mathbf{y})=\psi(\mathbf{x})\psi^*(\mathbf{y})$, and it's necessarily Hilbert-Schmidt, so the Wigner transform of an $L^2$-function is also $L^2$, by Theorem~\ref{thm:we}. The Wigner transform essentially brings a function on the configuration space to a function on the phase space. There's a similar transform, called Wigner-Moyal transform, which is a bilinear mapping on two functions $\psi,\phi$, replacing $\ket{\psi}\bra{\psi}$ with $\ket{\psi}\bra{\phi}$; interested readers can refer to Ref.~\cite{deGosson}.

There's an important property of Wigner transform due to Moyal \cite{deGosson}:
\begin{thm}\label{thm:Wexp}
Suppose $f$ is a function on $\mathbb{R}^{2n}$ and has a well-defined Weyl quantization. Then for any $\psi\in L^2(\mathbb{R}^n)$
\begin{equation}
\langle\psi|Q_W(f)\psi\rangle=\int_{\mathbb{R}^{2n}} W(\psi)f \,d\mathbf{x}\,d\mathbf{p}.
\end{equation}
\end{thm}
This theorem implies that the expectation value of an observable over a state $\ket{\psi}$, if the observable has a classical counterpart, can be computed in the phase space. In particular,
\begin{equation}
||\psi||^2=\inner{\psi}{\psi}=\int_{\mathbb{R}^{2n}} W(\psi)\,d\mathbf{x}\,d\mathbf{p},
\end{equation}
by choosing $f=1$ so that $Q_W(f)=I$. Therefore, if $\psi$ is normalized, this integral is equal to 1. 

\subsection{Wigner Quasi-probability Distribution}
\begin{defi}[\textbf{Wigner quasi-probability distribution}]
The Weyl symbol of a density operator $\rho$ on $L^2(\mathbb{R}^n)$ divided by $(2\pi\hbar)^n$ is called the Wigner quasi-probability distribution of $\rho$, denoted by $W_\rho$ \cite{Adesso14}. Namely (c.f. Theorem~\ref{thm:we})
\begin{equation}
W_\rho:=\frac{1}{(2\pi\hbar)^n}Q_W^{-1}(\rho)=\frac{1}{(2\pi\hbar)^n}\int_{\mathbb{R}^n}\rho(\mathbf{x}-\frac{ \mathbf{y}}{2},\mathbf{x}+\frac{\mathbf{y}}{2})e^{i\mathbf{y}\cdot\mathbf{p}/\hbar}\,d\mathbf{y},
\end{equation}
where $\rho$ on the right side of the equation is the kernel of the operator $\rho$.
\end{defi}

Comparing this with Definition~\ref{def:Wig}, the Wigner quasi-probability distribution of a pure state $\ket{\psi}\bra{\psi}$ is exactly the Wigner transform of $\psi$, i.e. $W_\rho=W(\psi)$ \cite{deGosson}. More generally, 
\begin{thm}\label{thm:Wig}
Decomposing a density operator $\rho$ on $L^2(\mathbb{R}^n)$ as orthogonal projections onto one-dimensional spaces as in \eqref{eq:rhoij},
\begin{equation*}
\rho=\sum_{i,j(i)}\lambda_i \ket{\psi_{i,j}}\bra{\psi_{i,j}},
\end{equation*}
the (integral) kernel of $\rho$ then is
\begin{equation}
\rho(\mathbf{x},\mathbf{y})=\sum_{i,j(i)} \lambda_i\psi_{i,j}(x)\psi^*_{i,j}(y),
\end{equation} 
and the Wigner quasi-probability distribution is
\begin{equation}
W_\rho=\sum_{i,j(i)} \lambda_i W(\psi_{i,j}),
\end{equation}
so
\begin{equation}
\int_{\mathbb{R}^{2n}} W_\rho\,d\mathbf{x}\,d\mathbf{p}=1,
\end{equation}
\end{thm}

Given an operator $T$ on $L^2(\mathbb{R}^n)$, Theorem~\ref{thm:Wig} along with Theorem~\ref{thm:Wexp} seems to imply \cite{Adesso14}
\begin{align}
\trace (T\rho)&=\sum_{i,j(i)}\lambda_i\int_{\mathbb{R}^{2n}} t  W(\psi_{i,j})\,d\mathbf{x}\,d\mathbf{p}\label{eq:orho0}\\
&\overset{?}{=} \int_{\mathbb{R}^{2n}} t \sum_{i,j(i)}\lambda_i W(\psi_{i,j})\,d\mathbf{x}\,d\mathbf{p} =\int_{\mathbb{R}^{2n}} t W_\rho \,d\mathbf{x}\,d\mathbf{p}.\label{eq:orho}
\end{align}
For \eqref{eq:orho} to be true some extra conditions are required. 
\subsubsection{Trace-class observable}
It can be shown \cite{deGosson}
\begin{thm}
If both $A$ and $B$ are trace-class and self-adjoint operators on $L^2(\mathbb{R}^n)$ with $A\weyl a$ and $B\weyl b$, then
\begin{equation}
\trace A= \frac{1}{(2\pi\hbar)^n}\int a \,d\mathbf{x}\,d\mathbf{p},
\end{equation}
and\footnote{As trace-class operators are Hilbert-Schmidt, $\trace(AB)$ exists.}
\begin{equation}
\trace (AB)=\frac{1}{(2\pi\hbar)^n}\int ab \,d\mathbf{x}\,d\mathbf{p}.
\end{equation}
\end{thm}
Therefore \eqref{eq:orho} holds true for a trace-class observable $T$. 
\subsubsection{Utilizing dominated convergence}
The argument below is incomplete as will be explained toward the end, so please take it with a grain of salt. Lebesgue's dominated convergence theorem \cite{PapaRudin,Cohn} provides a sufficient condition for whether the limit outside the integral sign can be taken inside. For example, if $T$ is a positive operator with $T \weyl t$, then $\bracket{\psi_{i,j}}{T}{\psi_{i,j}}\geq 0,$ so for $N_1\geq N_2$
\begin{equation}
\sum_{i}^{N_1}\sum_{j(i)}\bracket{\psi_{i,j}}{T}{\psi_{i,j}}\geq \sum_{i}^{N_2}\sum_{j(i)}\bracket{\psi_{i,j}}{T}{\psi_{i,j}}.
\end{equation}
Hence, if $t W_\rho\in L^1(\mathbb{R}^{2n})$, then by the dominated convergence theorem \eqref{eq:orho} will be true. In the same vein, if there's a positive operator $T'\geq T$ which has a Weyl symbol, and if the product of the Weyl symbols is $L^1$, then \eqref{eq:orho} holds.

Now suppose $W_\rho$ is a Schwartz function (Section~\ref{sec:Sch}) and $t$ is a polynomial in $\mathbf{x}$ and $\mathbf{p}$. Because there exists a positive polynomial larger than $t$ everywhere, and because the product of such a polynomial and a Schwartz function is Schwartz, which is also an $L^1$-function \cite{Einsiedler}, the integral converges to \eqref{eq:orho}. Simply put, \eqref{eq:orho} is correct for a density operator $\rho$ with a Schwartz $W_\rho$ and an observable whose Weyl symbol is a polynomial in $\mathbf{x}$ and $\mathbf{p}$.

There's actually one issue that hasn't been resolved: is the equality in \eqref{eq:orho0} correct? To be more precise, since $T$ may be an unbounded operator, we don't know whether $T\sum_{i,j}\ket{\psi_{i,j}}\bra{\psi_{i,j}}=\sum_{i,j}T\ket{\psi_{i,j}}\bra{\psi_{i,j}}$, and whether $\trace\sum_{i,j}T\ket{\psi_{i,j}}\bra{\psi_{i,j}}=\sum_{i,j}\trace T\ket{\psi_{i,j}}\bra{\psi_{i,j}}$. Given the classes of operators $\rho$ and $T$ are, it may be easier to verify $\trace T\rho=\int t W_\rho \,d\mathbf{x}\,d\mathbf{p}$ directly, but we'll stop pursuing this problem here, and take it for granted.

\subsection{Metaplectic Group}
Here we will glimpse the metaplectic group. As to introduce it in a even remotely rigorous way requires many tools beyond the scope of this thesis, we will only touch on what is essential for this work. Please refer to Ref.~\cite{deGosson,Woit,Simon94,Arvind95} if interested.

A metaplectic group $Mp(n)$ is a subgroup of $\mathcal{U}(L^2(\mathbb{R}^n))$, the group of unitary operators on $L^2(\mathbb{R}^n)$. There exists a group homomorphism $\pi_{Mp}$ from $Mp(n)$ \emph{onto} $Sp(n)$ (Section~\ref{sec:sym}), with kernel $\{I,-I\}$. There's a vital property of $Mp(n)$ \cite{deGosson}:
\begin{thm}\label{thm:meta}
For every Weyl operator with the corresponding symbol $A\overset{\mathrm{Weyl}}{\longleftrightarrow} a$, if $\hat{S}\in Mp(n)$ has $\pi_{Mp}(\hat{S})=S\in Sp(n)$,
\begin{equation}
a\circ S\overset{\mathrm{Weyl}}{\longleftrightarrow} \hat{S}^{-1}A \hat{S}.
\end{equation}
\end{thm}

This theorem implies that given a density operator $\rho$ and its Wigner quasi-probability distribution $W_\rho$, we can perform a symplectic mapping $S$ on the phase space, resulting in another phase-space function $W_\rho\circ S$, which amounts to another density operator $\rho'$ that is unitarily equivalent to $\rho$. Beware of surjectivity of the homomorphism $\pi:Mp(n)\rightarrow Sp(n)$: This means there always exists one (actually two) $\hat{S}$ for every $S$, so every symplectic transform on the phase space is possible.

\section{Hermitian and Positive Operators on a Finite-Dimensional Hilbert Space}\label{sec:HPF}
In this section the Hilbert space $\mathcal{H}$ is always assumed to be finite-dimensional. The symbol $H$ refers to a Hermitian operator, and $P$ to a positive operator. 
\subsection{Norms on Operators}\label{sec:normop}
In addition to the operator norm introduced in Section~\ref{sec:opnorm}, in this thesis several norms on operators will be used. In this section I'll list them and show several properties thereof relevant to our topic.
\subsubsection{Operator Norm}
When the space in question is a Hilbert space, the norm is often that induced by inner product, and in terms of bra-ket: 
\begin{equation}
||O||:=\sup_{v\in\mathcal{V}}\frac{\bracket{v}{O^\dagger O}{v}^{1/2}}{\inner{v}{v}^{1/2}}.
\end{equation}
For a Hermitian operator $H$,
\begin{equation}
||H||=\max(|h_i|:h_i \text{ are the eigenvalues of }H).
\end{equation}

\subsubsection{Trace Norm}
First, for $O\in\mathcal{B}(\mathcal{H})$ define \cite{Bhatia}
\begin{equation}
|O|:=\sqrt{O^\dagger O},
\end{equation}
as $O^\dagger O\geq0.$ Trace norm then is \cite{Bhatia,Vidal02}
\begin{equation}
||O||_1:=\trace |O|.
\end{equation}
For a Hermitian operator $H$,
\begin{equation}
||H||_1=\sum_i |h_i|,\,h_i:\text{ eigenvalues of }H.\label{eq:Htr}
\end{equation}
Obviously, for any positive operator $P$, its trace is equal to its trace norm, $\trace P=||P||_1$, and I will use them interchangeably.

\subsubsection{Hilbert-Schmidt Norm}
Hilbert-Schmimdt norm, as its name suggests, is a norm induced by Hilbert-Schmidt inner product \cite{Bhatia}:
\begin{equation}
||O||_2=(O|O)^{1/2}=\sqrt{\trace O^\dagger O}=\sqrt{\trace |O|^2}.
\end{equation}
For a Hermitian operator $H$,
\begin{equation}
||H||_1=\sqrt{\sum_i |h_i|^2},\,h_i:\text{ eigenvalues of }H.
\end{equation}

In Appendix~\ref{app:isHS}, I'll discuss isometry of specific mappings with respect to Hilbert-Schmidt norm, and why it won't be used in Chapter~\ref{ch:EC}. 

\subsubsection{Schatten $p$-norm}
For an operator $O$ its Schatten $p$-norm is \cite{Bhatia,Rastegin12,Horn}
\begin{equation}
||O||_p=\left(\trace |O|^p\right)^{1/p},\,p\geq 1.
\end{equation}
For $p=\infty$ it's (defined as) the largest singular value of $O$, hence the same as the operator norm; for $p=1$ it's exactly the trace norm; and Hilbert-Schmidt norm is Schatten $2$-norm  \cite{Bhatia,Rastegin12}. For a Hermitian operator $H$ whose eigenvalues are $h_i$,
\begin{equation}
||H||_p=\left(\sum_i |h_i|^p\right)^{1/p}.
\end{equation}

\subsection{Inner Product between Hermitian and Positive Operators}
\label{sec:innerHP}
For two positive operators $P_1,P_2\in\mathcal{B}(\mathcal{H})$, by the spectral decomposition $P_1=\sum_i p_i \ket{\psi_i}\bra{\psi_i}$ we can  find \cite{Campbell10}:\footnote{The inner product here is Hilbert-Schmidt inner product, introduced in Section~\ref{sec:HSinner}.}
\begin{equation}
\min_i p_i\trace P_2 \leq (P_1| P_2)\leq \max_i p_i \trace P_2=||P_1||\,||P_2||_1,\label{eq:P1P2}
\end{equation}
by $\trace (P_1P_2)=\sum_i p_i \bracket{\psi_i}{P_2}{\psi_i}$ and $\trace P_2=\sum_i \bracket{\psi_i}{P_2}{\psi_i}.$

The inequality on the right of \eqref{eq:P1P2} can be regarded as an application of H\"{o}lder's inequality \cite{PapaRudin}. Either side of \eqref{eq:P1P2} becomes an equality if and only if $(\ker P_2)^\perp=\mathrm{ran}P_2$ is in the eigensapce of $P_1$ with the maximum or minimum eigenvalue, e.g. if $P_1=p I,$ $\trace P_1 P_2=p||P_2||_1$; $\ker$ refers to the kernel or null space (Section~\ref{sec:hom}).

Matrix H\"{o}lder inequality \cite{Baumgartner11} states that for $1\leq p,q\leq \infty$ satisfying $1/p+1/q=1$:
\begin{equation}
|\trace O_1^\dagger O_2|=|(O_1,O_2)|\leq ||O_1||_p||O_2||_q,\label{eq:Holder}
\end{equation}
where $||\cdot||_p$ and $||\cdot||_q$ are the Schatten norm. It corresponds to \eqref{eq:P1P2} at $p=1$ and $q=\infty$, magnitude-wise.

By applying the same procedure as \eqref{eq:P1P2}, an inequality similar to \eqref{eq:P1P2} for Hermitian operators $H_1$ and $H_2$ can be found:
\begin{equation}
-||H_1||\,||H_2||_1\leq (H_1| H_2)\leq ||H_1||\,||H_2||_1,\label{eq:H1H2}
\end{equation}
which is identical to \eqref{eq:Holder}. 

Given any Hermitian $H_1$, we can always find a Hermitian operator $H_2$ that satisfies either side \eqref{eq:H1H2} in pretty much the same way as the process laid out below \eqref{eq:P1P2}.\footnote{It's worth mentioning that if $(H_1|H_2)=||H_1||\,||H_2||_1,$ then $(H_1|-H_2)=-||H_1||\,||-H_2||_1,$ and vice versa.}  Hence, we can make two quick observations:
\begin{itemize}
\item The inner product between two positive operators is non-negative.
\item For a Hermitian operator $H$ if $(H|P)\geq 0$ for any $P\geq 0$, then $H\geq 0$.
\end{itemize}

\subsection{Decomposition of a Hermitian Operator}\label{sec:eig}
Any Hermitian operator $H\in \mathcal{B}(\mathcal{H})$ can be decomposed as \cite{Vidal02}
\begin{equation}
H=\widetilde{H}^+-\widetilde{H}^-,\,\widetilde{H}^\pm\geq 0.\label{eq:Hpma}
\end{equation}
Such decompositions are not unique. For example, with the spectral decomposition of $H=\sum_i h_i \ket{\psi_i}\bra{\psi_i}$, $h_i$ being its eigenvalues and $\ket{\psi_i}$ the eigenvectors, by defining 
\begin{equation}
H^\pm=\sum_{i,\pm} h_i^\pm\ket{\psi_i^\pm}\bra{\psi_i^\pm},\label{eq:Hpm}
\end{equation}
where $h_i^\pm$ are the positive/negative eigenvalues and $\ket{\psi_i^\pm}$ the corresponding eigenvectors, we obtain 
\begin{equation}
H=H^+-H^-,\,H^\pm\geq 0.
\end{equation}
By \eqref{eq:Htr} \cite{Vidal02},
\begin{equation}
||H||_1=\trace H^++\trace H^-=\trace H+2\trace H^-,\label{eq:H+H-}
\end{equation}
where the second equality comes from linearity of trace: 
\begin{equation}
\trace H= \trace (H^+-H^-)=\trace H^+- \trace H^-.
\end{equation}

In this thesis, we will mark \emph{any decomposition} of a Hermitian operator $H$ like \eqref{eq:Hpma} as $\widetilde{H}^\pm$, and the decomposition \emph{according to spectral decomposition} \eqref{eq:Hpm} as $H^\pm$. With this in mind, let's define
\begin{equation}
\mathcal{H}^\pm:=(\ker H^\pm)^\perp=\mathrm{ran}H^\pm\text{ and }\widetilde{\mathcal{H}}^\pm:=(\ker \widetilde{H}^\pm)^\perp=\mathrm{ran}\widetilde{H}^\pm,
\end{equation}
by \eqref{eq:keran}. Then $\mathcal{H}$ is an orthogonal direct sum (Section~\ref{sec:dip}) of 
\begin{equation}
\mathcal{H}=\mathcal{H}^+\oplus \mathcal{H}^-\oplus\ker H,
\end{equation}
by spectral theorem (Theorem~\ref{thm:spe}).

Here's a lemma concerning such a decomposition, which is a generalization of Lemma 2 from Ref.~\cite{Vidal02}:
\begin{lem}\label{lem:or}
For any Hermitian operator $H$ on a finite-dimensional Hilbert space $\mathcal{H}$, among all possible such decompositions: $H=\widetilde{H}^+-\widetilde{H}^-,$ $\widetilde{H}^\pm\geq 0$, the spectral decomposition $H=H^+-H^-$ is the \textbf{unique} one that minimizes $\trace\widetilde{H}^+$,  $\trace\widetilde{H}^-$, and $\trace(\widetilde{H}^++\widetilde{H}^-)$; minimizing any one of them is the same as minimizing all of them. A decomposition in which $\widetilde{\mathcal{H}}^\pm$ are orthogonal is equivalent to the spectral decomposition.
\end{lem}

\begin{proof}
\hfill

\noindent\emph{1. Spectral decomposition minimizes traces}

We can find $\mathcal{H}^\pm\subseteq \widetilde{\mathcal{H}}^\pm$; otherwise there would exist nonzero $\ket{\psi^\pm}\in \mathcal{H}^\pm$ such that $\bracket{\psi^\pm}{H}{\psi^\pm}=\bracket{\psi^\pm}{\widetilde{H}^+-\widetilde{H}^-}{\psi^\pm}$ doesn't come out positive/negative, because $\widetilde{H}^\pm\geq 0$. Furthermore, since $\widetilde{\mathcal{H}}^\pm$ may intersect $\mathcal{H}^\mp$ (and $\ker H$),
\begin{equation}
\big|\bracket{\psi^\pm}{H}{\psi^\pm}\big|=\bracket{\psi^\pm}{H^\pm}{\psi^\pm}\leq\bracket{\psi^\pm}{\widetilde{H}^\pm}{\psi^\pm},\label{eq:ineq1}
\end{equation}
because $\bracket{\psi^\pm}{H}{\psi^\pm}=\bracket{\psi^\pm}{\widetilde{H}^+}{\psi^\pm}-\bracket{\psi^\pm}{\widetilde{H}^-}{\psi^\pm}$ and $\widetilde{H}^\pm\geq 0$. Hence
\begin{equation}
\trace\widetilde{H}^\pm\geq \trace H^\pm,\label{eq:ineq2}
\end{equation}
i.e. spectral decomposition minimizes $\trace \widetilde{H}^\pm$ and $\trace \widetilde{H}^++\trace \widetilde{H}^-.$

\noindent\emph{2. Only spectral decomposition minimizes traces}

Suppose we want to minimize both $\trace\widetilde{H}^\pm$ at once. Because $\widetilde{H}^\pm\geq 0$ and because of \eqref{eq:ineq1}, for both sides of \eqref{eq:ineq2} to be equal, i.e. to minimize both $\trace\widetilde{H}^\pm$, $\widetilde{\mathcal{H}}^\pm$ must be the same as $\mathcal{H}^\pm$ (as $\mathcal{H}^\pm\subseteq \widetilde{\mathcal{H}}^\pm$), implying
\begin{equation}
\pm\bracket{\psi^\pm}{H}{\phi^\pm}=\bracket{\psi^\pm}{\widetilde{H}^\pm}{\phi^\pm}=\bracket{\psi^\pm}{H^\pm}{\phi^\pm} \label{eq:eq3}
\end{equation}
and
\begin{equation}
\bracket{\psi^\pm}{\widetilde{H}}{\phi^\mp}=0,
\end{equation}
where $\ket{\psi^\pm}$ and $\ket{\phi^\pm}$ are any vectors in $\mathcal{H}^\pm$; the same is true for any vectors in $\ker H.$ Therefore $\bracket{\psi}{\widetilde{H}^\pm}{\phi}=\bracket{\psi}{H^\pm}{\phi}$ for any $\ket{\phi},\ket{\psi}\in \mathcal{H}$ as $\mathcal{H}=\mathcal{H}^+\oplus \mathcal{H}^-\oplus\ker H,$ so $\widetilde{H}^\pm=H^\pm$. Because both $\trace\widetilde{H}^+$ and $\trace\widetilde{H}^-$ are minimized, so is $\trace(\widetilde{H}^++\widetilde{H}^-)$. 

To minimize $\trace(\widetilde{H}^++\widetilde{H}^-)$ both $\trace\widetilde{H}^+$ and $\trace\widetilde{H}^-$ should be minimized, leading to the same conclusion. If we want to minimize only one of $\trace\widetilde{H}^+$ and $\trace\widetilde{H}^-$, say $\trace\widetilde{H}^-$, then by \eqref{eq:ineq2} $\trace\widetilde{H}^-=\trace H^-$. Because $\trace H=\trace H^+-\trace H^-=\trace \widetilde{H}^+-\trace\widetilde{H}^-,$ $\trace\widetilde{H}^+=\trace H^+,$ i.e. both $\trace\widetilde{H}^+$ and $\trace\widetilde{H}^-$ are minimized. Hence minimizing any of them is identical to minimizing all of them, and minimization of traces requires spectral decomposition.

\noindent\emph{3. Orthogonality is equivalent to spectral decomposition}

It's already known that if $\widetilde{H}^\pm=H^\pm$ then $\widetilde{\mathcal{H}}^+\perp\widetilde{\mathcal{H}}^-$, so we only have to show the reverse: Does $\widetilde{\mathcal{H}}^+\perp\widetilde{\mathcal{H}}^-$ lead to $\widetilde{H}^\pm=H^\pm$?

Because $\widetilde{\mathcal{H}}^+\perp\widetilde{\mathcal{H}}^-$, $\mathcal{H}^\pm $ as subsets of $\widetilde{\mathcal{H}}^\pm$\footnote{See the first part of the proof.} are also perpendicular to $\widetilde{\mathcal{H}}^\mp$, i.e. $\widetilde{\mathcal{H}}^\pm \perp\mathcal{H}^\mp$, which implies $\widetilde{\mathcal{H}}^\pm \subseteq \mathcal{H}^\pm\oplus\ker H$. For any vector $\ket{\psi}\in\ker H$, $\widetilde{H}^+\ket{\psi}=\widetilde{H}^-\ket{\psi}$; however, because $\widetilde{\mathcal{H}}^+\perp\widetilde{\mathcal{H}}^-$, from \eqref{eq:keran} $\widetilde{H}^+\ket{\psi}=\widetilde{H}^-\ket{\psi}=0,$ i.e. $\ker H\subseteq\ker\widetilde{H}^\pm$ and hence $\ker H\perp \widetilde{H}^\pm$, which, by $\widetilde{\mathcal{H}}^\pm \subseteq \mathcal{H}^\pm\oplus\ker H$, implies $\widetilde{\mathcal{H}}^\pm \subseteq \mathcal{H}^\pm$. Since $\mathcal{H}^\pm\subseteq \widetilde{\mathcal{H}}^\pm$, we conclude $\mathcal{H}^\pm=\widetilde{\mathcal{H}}^\pm$, so according to the arguments from the second part of the proof $\widetilde{H}^\pm=H^\pm$. 

Or rather, we can directly apply the spectral theorem (Theorem~\ref{thm:spe}): Since $\widetilde{\mathcal{H}}^\pm$ are orthogonal, the eigenspaces of $\widetilde{\mathcal{H}}^\pm$ are orthogonal. After spectral-decomposing $\widetilde{H}^\pm$, $\widetilde{H}^+-\widetilde{H}^-$ is already a spectral decomposition of $H$.
\end{proof}

From Lemma~\ref{lem:or}, $\widetilde{H}^\pm$ and  $H^\pm$ are related by
\begin{equation}
\widetilde{H}^\pm=H^\pm+H'\geq 0 \text{ and } \trace H'> 0, \label{eq:de}
\end{equation}
where $H'$ is also a Hermitian operator. From \eqref{eq:ineq1} and $\mathcal{H}^\pm\subseteq \widetilde{\mathcal{H}}^\pm$ it may seem $\widetilde{H}^\pm\geq H^\pm$, but this isn't true in general---what \eqref{eq:ineq1} implies is actually $\widetilde{H}^\pm\geq H^\pm$ \emph{to the restriction of} $\mathcal{H}^\pm.$ To be specific, for a vector $\ket{\psi^+}+\ket{\psi^-}$, $\ket{\psi^\pm}\in \mathcal{H}^\pm$, the value of
\begin{align}
\left(\bra{\psi^+}+\bra{\psi^-}\right)H'\left(\ket{\psi^+}+\ket{\psi^-}\right),
\end{align}
even though $\bracket{\psi^\pm}{H'}{\psi^\pm}\geq 0$, isn't necessarily non-negative, as it depends on $\bracket{\psi^\pm}{H'}{\psi^\mp}$ as well. In terms of matrices, it just reflects the fact that \emph{a matrix that has non-negative diagonal entries aren't necessarily positive (semi-definite)}.

\subsection{Ensembles of Positive Operators}
\label{sec:en}
For any positive operators $P$, we can always express it as a mixture of projections \cite{Hughston93,Horodecki97}:
\begin{equation}
P=\sum_i p_i \ket{\psi_i}\bra{\psi_i},\;p_i>0,
\end{equation}
with $||\psi_i||=1$. An example is of course the spectral decomposition of $P$. We can absorb $p_i$ into $\psi_i$, and it becomes
\begin{equation}
P=\sum_i \ket{\tilde{\psi}_i}\bra{\tilde{\psi}_i},
\end{equation}
where $\ket{\tilde{\psi}_i}=\sqrt{p_i }\ket{\psi_i}$ aren't normalized in general. The nonzero vectors $\{\ket{\tilde{\psi}_i}\}$ form a so-called $P$-ensemble. The ensemble obtained through spectral decomposition is called an eigenensemble \cite{Hughston93,Horodecki97}. In the lemma below for convenience the tildes are ignored:

\begin{lem}\label{lem:oen}
Suppose $P_1$ and $P_2$ are two positive operators, and they have such ensembles:
\begin{equation*}
P_i=\sum_j |\psi_j^i\rangle \langle\psi_j^i |,
\end{equation*}
where each $\{|\psi_j^i\rangle\}$ is a set of nonzero vectors that aren't necessarily normalized or mutually orthogonal. Then each eigenvector/eigenspace of $P_2$ corresponding to a nonzero eigenvalue is orthogonal to each of $P_1$ if and only if
\begin{equation*}
\langle \psi_i^1|\psi_j^2\rangle=0 \;\forall i,j,\label{eq:lm4}
\end{equation*}
in other words, if and only if
\begin{equation}
\mathrm{ran}P_1\perp \mathrm{ran}P_2.
\end{equation}
Note $\mathrm{ran}O=(\ker O)^\perp$ for a normal operator $O$, \eqref{eq:keran}.
\end{lem}

\begin{proof}
From Ref.~\cite{Hughston93}, if the eigenensembles of $P_i$ are
\begin{equation}
P_i=\sum_{j(i)} \ket{e_j^i}\bra{e_j^i},\label{eq:pie}
\end{equation}
where $\ket{e_j^i}$ aren't normalized, then there exist unitary matrices $U^i$ such that\footnote{See Section~\ref{sec:hom} for the definition of a rank.}
\begin{equation}
\begin{cases}
\ket{e_j^i}=\sum_k U_{jk}^i\ket{\psi_k^i}& j\leq \text{Rank}(P_i)\\
\ket{0}=\sum_k U_{jk}^i\ket{\psi_k^i}& j> \text{Rank}(P_i)
\end{cases},\label{eq:pe1}
\end{equation}
and
\begin{equation}
\ket{\psi_j^i}=\sum_k (U^i)^{-1}_{jk}\ket{e_k^i}.\label{eq:pe2}
\end{equation}
If $\inner{e^1_i}{e^2_j}=0$ $\forall i,j$, from \eqref{eq:pe2} we can see $\inner{\psi^1_i}{\psi^2_j}=0$ $\forall i,j$. Similarly, if $\inner{\psi^1_i}{\psi^2_j}=0$ then from \eqref{eq:pe1} $\inner{e^1_i}{e^2_j}=0$ $\forall i,j$. In other words, because each $|\psi_j^i \rangle$ is a linear combination of $\{\ket{e_j^i}\}$, and each $\ket{e_j^i}$ is also linear combinations of $\{|\psi_j^i \rangle\}$, orthogonality of two ensembles guarantees that of the other two. 

Here's a more elegant proof: The orthogonality of the vectors as stated in the lemma is satisfied if and only if their spans are orthogonal. From \eqref{eq:pie}, the spans of $\{\psi^i_j\}$ are identical to $(\ker P_i)^\perp=\mathrm{ran}P_i$, which of course are also identical to the spans of $\{e_j^i\}$. Therefore the orthogonality of one pair implies that of the other. 
\end{proof}

\chapter{Linear Mappings from Operators to Operators}\label{ch:li}
In this chapter and the next, the spaces $\mathcal{H}$ are in general assumed to be finite-dimensional. 
\section{Linear Mappings from Operators to Operators and Quantum Operations}\label{sec:liqo}
Many quantum processes are linear by nature, such as
\begin{itemize}
\item The evolution of a closed system, which is determined by the Hamiltonian, or the corresponding unitary operator.
\item Measurements, whether they are the standard collapsing model or more general POVM (positive operator valued measure) measurements.
\item The evolution of an open system with a fixed environment initial state: $\trace_E \left(U(\rho_S\otimes \rho_E)U^\dagger\right)$, where $E$ and $S$ stand for the environment and system. This is linear in $\rho_S$, the input state.   
\end{itemize} 
As such, the properties of this kind of linear mappings are of great interest. 

In this chapter, we will discuss linear mappings from operators to operators, explicitly, linear mappings in such a space
\begin{equation}
\mathcal{B}(\mathcal{B}(\mathcal{H}_1),\mathcal{B}(\mathcal{H}_2)),
\end{equation}
and a linear mapping $L$ in this space is said to be \cite{dePillis67,Jamiolkowski72,Bengtsson,Choi75}
\begin{itemize}
\item \textbf{Hermicity-preserving (HP)} if $L(H)$ remains Hermitian, namely $L(H)^\dagger=L(H)$ for any Hermitian operator $H$ in $\mathcal{B}(\mathcal{H}_1)$, which is equivalent to\footnote{Because any operator is the sum of a Hermitian operator and an anti-Hermitian operator.} $L(O^\dagger)=L(O)^\dagger$ for any $O\in\mathcal{B}(\mathcal{H}_1)$; in other words, $L$ and the adjoint mapping (Section~\ref{sec:ad}) commute;

\item \textbf{positive} if $L(P)\geq 0$ for any positive operator $P\in\mathcal{B}(\mathcal{H}_1)$;

\item \textbf{completely positive (CP)} if $\mathcal{I}\otimes L$ is a positive mapping for any identity mapping $\mathcal{I}\in \mathcal{B}(\mathcal{B}(\mathcal{H}_i))$, where $\mathcal{H}_i$ is any finite-dimensional Hilbert space;

\item \textbf{trace-preserving (TP)} if $\trace L(O)=\trace O$ for any $O\in\mathcal{B}(\mathcal{H}_1)$.
\end{itemize}
Positive and CP mappings are HP.

It's worth mentioning (however obvious it may seem) that the four properties above are preserved under composition. I'll demonstrate preservation of complete positivity: Assume $L_1: \mathcal{B}(\mathcal{H}_1)\rightarrow \mathcal{B}(\mathcal{H}_2)$ and $L_2: \mathcal{B}(\mathcal{H}_2)\rightarrow \mathcal{B}(\mathcal{H}_3)$ are CP. Then
\begin{equation}
(\mathcal{I}\otimes L_2)\circ (\mathcal{I}\otimes L_1)=\mathcal{I}\otimes (L_2\circ L_1), \,\mathcal{I}\in \mathcal{B}(\mathcal{B}(\mathcal{H}_i)).
\end{equation} 
Because $[(\mathcal{I}\otimes L_2)\circ (\mathcal{I}\otimes L_1)](P)\geq 0$ for any positive operator $P\in \mathcal{B}(H_i\otimes H_1)$, $L_2\circ L_1$ is CP.

I'll usually call linear mappings from operators to operators simply ``linear mappings,'' but the context dictates what they actually mean. A mapping in $\mathcal{B}(\mathcal{B}(\mathcal{H}))$ is actually an operator, namely an operator on operators, but to prevent confusion I avoid such a wording. 
  
\subsection{Quantum Operation}\label{sec:qop}
Let $\mathcal{H}_1$ and $\mathcal{H}_2$ be two (finite-dimensional) Hilbert spaces. There are two types of quantum operations:
\begin{itemize}
\item A \textbf{deterministic quantum operation} is a CPTP mapping $S: \mathcal{B}(\mathcal{H}_1)\rightarrow \mathcal{B}(\mathcal{H}_2)$.
\item A \textbf{probabilistic}/stochastic operation $S$ is composed of sub-operations $S_i: \mathcal{B}(\mathcal{H}_1)\rightarrow \mathcal{B}(\mathcal{H}_2)$, each of which is CP, such that $\sum_i S_i$ is TP (and CP). The probability that a sub-operation $S_i$ occurs is $\trace S_i(\rho),$ and the resultant state is $S_i(\rho)/(\trace S_i(\rho)).$ 
\end{itemize}

Because a density operator is positive and has unit trace, quantum operations are CP so that a density operator remains positive, and they are TP as a whole so that the probabilities add up to 1 and the trace stays at 1. A deterministic operation can be considered a probabilistic one that has a single sub-operation.

I'll use the symbol $L$ for a linear mapping (from operators to operators), and $S$ and $S_i$ for quantum operations and sub-operations thereof.\footnote{The symbol $S$ refers to ``superoperators,'' as they're operators or linear mappings on operators \cite{Rains99,Breuer}. However I won't use this term.} In general we use $S$ without a subscript for a TP or deterministic operation, but sometimes I'll be a bit loose on this, as in Section~\ref{sec:up}. 

The most well-known probabilistic operation may be the collapsing model of measurement: Suppose we're measuring an observable $H$ whose spectral decomposition (see Theorem~\ref{thm:spe}) is $H=\sum_i \lambda_i \Pi_i$. The probability for an outcome $i$ to occur is 
\begin{equation}
p_i=\trace S_i(\rho)=\trace(\Pi_i \rho \Pi_i),
\end{equation}
and the state will become $S_i(\rho)/p_i=\Pi_i \rho \Pi_i/p_i$. The probabilities sum up to 1: 
\begin{equation}
\sum_i \trace (\Pi_i \rho \Pi_i)=\sum_i \trace (\Pi_i\rho)=\trace [(\sum_i \Pi_i) \rho]=\trace \rho=1,
\end{equation}
which is also equal to $\trace S(\rho)=\trace \sum S_i(\rho)$, so $S=\sum S_i$ is TP.

In this thesis a deterministic operation $S$ for which $S(\rho)=U\rho U^\dagger$ with unitary $U$ will be called a \textbf{unitary operation}. This name suggests the unitarity of $U$, though $S$ is also a unitary mapping: Because it's an isometry with respect to Hilbert-Schmidt inner product and the domain and codomain have the same dimension, it's unitary (Section~\ref{sec:isometry} and \ref{sec:HSinner}), and its inverse is $S^{-1}(A)=U^\dagger A U$.

We will also consider a sub-operation $S_i$ of a probabilistic operation alone. A sub-operation $S_i$ is a CP mapping with $S_i^\dagger(I)\leq I$, for which the significance of $S_i^\dagger(I)$ will be addressed soon.

\section{Adjoint of a Linear Mapping}\label{sec:adl}
For a linear mapping $L:\mathcal{B}(\mathcal{H}_1)\rightarrow \mathcal{B}(\mathcal{H}_2)$ \cite{Jamiolkowski72,Stormer}, 
\begin{equation}
(O_2|L(O_1))=(L^\dagger(O_2)|O_1)\;\forall O_i\in\mathcal{B}(\mathcal{H}_i).
\end{equation}
$L^\dagger: \mathcal{B}(\mathcal{H}_2)\rightarrow\mathcal{B}(\mathcal{H}_1)$ is the adjoint of $L$ with respect to Hilbert-Schmidt inner product (Section~\ref{sec:ad} and \ref{sec:HSinner}). In particular, with $I_i$ being the identity operator on $\mathcal{H}_i$, we're interested in $L^\dagger(I_{2})\in\mathcal{B}(\mathcal{H}_1)$, because
\begin{equation}
\trace L(O_1)=(I_{2}|L(O_1))=(L^\dagger(I_{2})|O_1).
\end{equation}
The trace of an operator after a linear mapping $L$ can be expressed in terms of an inner product between the operator and $L^\dagger(I_2).$ From now on if there's no risk of confusion, we will ignore the subscript of the identity operator, e.g. $L^\dagger(I)$ stands for $L^\dagger(I_2)$.

\begin{thm}
A linear mapping $L:\mathcal{B}(\mathcal{H}_1)\rightarrow \mathcal{B}(\mathcal{H}_2)$ is TP if and only if $L^\dagger(I_2)=I_1$ \cite{Bengtsson}.
\end{thm}
\begin{proof}
$L$ is TP if and only if
\begin{equation}
\trace L(O)=(L^\dagger(I_2)|O)=(I_1|O) \;\forall O\in\mathcal{B}(\mathcal{H}_1),
\end{equation} 
so 
\begin{equation}
(L^\dagger(I_2)-I_1|O)=0 \;\forall O\in\mathcal{B}(\mathcal{H}_1).
\end{equation}
If $L^\dagger(I)-I\neq 0$,\footnote{$0$ here is the zero operator on $\mathcal{H}_1$.} then it would violate positive-definiteness of inner product (Section~\ref{sec:inn}) because $O$ can be $L^\dagger(I)-I.$ Therefore $L^\dagger(I_2)=I_1$.
\end{proof}

If $L:\mathcal{B}(\mathcal{H}_1)\rightarrow \mathcal{B}(\mathcal{H}_2)$ is HP, it always admits an operator-sum representation \cite{Kraus71,Choi75,Bengtsson}: For all $O\in\mathcal{B}(\mathcal{H}_1)$ the mapping $L$ can be expresses as
\begin{equation}
L(O)=\sum_i c_i V_i O V_i^\dagger,\,V_i\in\mathcal{B}(\mathcal{H}_1,\mathcal{H}_2) \text{ and }c_i\in\mathbb{R}.\label{eq:opsum}
\end{equation}
Therefore
\begin{equation}
L^\dagger(I)=\sum_i c_i V_i^\dagger V_i.\label{eq:Kr}
\end{equation}
From this it's apparent $L^\dagger(I)$ is Hermitian if $L$ is HP. If $L$ is CP, then we can always find an operator-sum representation of it with $c_i>0$ in \eqref{eq:opsum}; we will demonstrate how to find it in Section~\ref{sec:opsum}.

From \eqref{eq:Kr} it's obvious that $L^\dagger(I)\geq 0$ for a CP mapping $L$, and since a probabilistic operation is CPTP as a whole and each sub-operation is CP, $S_i^\dagger(I)\leq I$. \eqref{eq:P1P2} also provides another perspective on why $S_i^\dagger(I)\leq I$, else the trace (and probability) could surpass 1.

\section{Choi Isomorphism}\label{sec:choi}
\subsection{Choi Isomorphism}
Let $\mathcal{H}_1$ and $\mathcal{H}_2$ be finite-dimensional Hilbert spaces, $\{a_i\}$ be an orthonormal basis of $\mathcal{H}_1$, and
\begin{equation}
E_{ij}:=\ket{a_i}\bra{a_j}
\end{equation}
be a basis of $\mathcal{B}(\mathcal{H}_1)$ as shown in Section~\ref{sec:bas}. \textbf{Choi isomorphism}\footnote{It's often called Choi-Jamio\l{}kowski isomorphism. However, the mapping studied by de Pillis and Jamio\l{}kowski is different \cite{dePillis67,Jamiolkowski72,Jiang13}.} $\mathscr{T}$ for linear mappings $L:\mathcal{B}(\mathcal{H}_1)\rightarrow \mathcal{B}(\mathcal{H}_2)$ is defined as \cite{Choi75}
\begin{equation}
\mathscr{T}(L):=\mathcal{I}\otimes L(\sum_{i,j} E_{ij}\otimes E_{ij})=\sum_{i,j} E_{ij}\otimes L(E_{ij}),
\end{equation}
where $\mathcal{I}$ is the identity mapping on $\mathcal{B}(\mathcal{H}_1)$. Therefore, Choi isomorphism brings $L\in\mathcal{B}(\mathcal{B}(\mathcal{H}_1),\mathcal{B}(\mathcal{H}_2))$ to an operator $\mathscr{T}(L)$ in $\mathcal{B}(\mathcal{H}_1)\otimes \mathcal{B}(\mathcal{H}_2)$ or $\mathcal{B}(\mathcal{H}_1\otimes \mathcal{H}_2)$, as these two spaces are isomorphic, c.f. Theorem~\ref{thm:emiso}. In this thesis Choi isomorphism refers to both the mapping $\mathscr{T}$ itself and the image thereof, namely $\mathscr{T}(L)$, and the context should be clear enough to distinguish which is which.

Choi isomorphism is an important tool in studying CP and HP mappings because
\begin{thm}\label{thm:Choi}
A linear mapping $L$ is CP if and only if $\mathscr{T}(L)\geq 0$, and that $L$ is HP if and only if $\mathscr{T}(L)$ is Hermitian \cite{Choi75}.
\end{thm}
\noindent With Theorem~\ref{thm:Choi}, we can easily distinguish CP mappings.

If $L:\mathcal{B}(\mathcal{H}_1)\rightarrow \mathcal{B}(\mathcal{H}_2)$ is TP,
\begin{equation}
\trace \mathscr{T}(L)=\trace\sum_{i,j} E_{ij}\otimes L(E_{ij})=\sum_{ij} \delta_{ij}^2=\mathrm{dim}\mathcal{H}_1.\label{eq:d}
\end{equation}
Therefore for a deterministic operation $S: \mathcal{B}(\mathcal{H}_1)\rightarrow \mathcal{B}(\mathcal{H}_2)$, $\trace\mathscr{T}(S)=\mathrm{dim}\mathcal{H}_1$.

Choi isomorphism is an isomorphism because it is linear and invertible, which will be shown very soon. Hence, we can treat linear mappings like operators, and many tools for operators can be utilized to handle problems about linear mappings from operators to operators. The dimension of its domain and codomain is
\begin{equation}
\mathrm{dim}\left[\mathcal{B}(\mathcal{B}(\mathcal{H}_1),\mathcal{B}(\mathcal{H}_2))\right]=\mathrm{dim}\left[\mathcal{B}(\mathcal{H}_1\otimes \mathcal{H}_2)\right]=(\mathrm{dim}\mathcal{H}_1)^2 (\mathrm{dim}\mathcal{H}_2)^2.\label{eq:dim}
\end{equation}

\subsection{Choi Isomorphism is Indeed an Isomorphism}
\label{sec:Choi}
For a linear mapping $L:\mathcal{B}(\mathcal{H}_1)\rightarrow\mathcal{B}(\mathcal{H}_2)$, it can be found \cite{Jiang13}:\footnote{Even though transposition is to be introduced in the next section, as the reader should be familiar enough with transposition I think it shouldn't be a problem to put it here.}
\begin{align}
\trace_1 \left[(O\tr\otimes I)\mathscr{T}(L)\right]&=\sum_{i,j} \trace_1 \left[(O\tr\otimes I)(E_{ij}\otimes L(E_{ij}))\right]\nonumber\\
&=\sum_{i,j} \trace(O\tr E_{ij}) L(E_{ij})\nonumber\\
&=\sum_{i,j} \trace(E_{ji}O)L(E_{ij}) \because\text{transposition is TP}\nonumber\\
&=\sum_{i,j}(E_{ij}|O)L(E_{ij})\nonumber\\
&=L\Big(\sum_{i,j}E_{ij}(E_{ij}|O)\Big)\nonumber\\
&=L(O)\; \forall O\in\mathcal{B}(\mathcal{H}_1).\label{eq:Choi}
\end{align}
Put it another way, $L$ is equal to the linear mapping $O\mapsto \trace_1(O\tr\otimes I \mathscr{T}(L))$. 

What we showed through \eqref{eq:Choi} is that there exists a mapping $\mathscr{T}^{-1}$ from $\mathcal{B}(\mathcal{H}_1)\otimes \mathcal{B}(\mathcal{H}_2)$ to  $\mathcal{B}(\mathcal{B}(\mathcal{H}_1),\mathcal{B}(\mathcal{H}_2))$ such that $\mathscr{T}^{-1}\circ \mathscr{T}=\mathcal{I}_1$, the identity mapping on $\mathcal{B}(\mathcal{H}_1)$. It suggests that $\mathscr{T}$ is injective, but it doesn't directly show $\mathscr{T}$ is surjective, c.f. Theorem~\ref{thm:inj}. It's through Theorem~\ref{thm:rpn} and \eqref{eq:dim} we know $\mathscr{T}$ is bijective and therefore an isomorphism. In other words, \eqref{eq:Choi} shows 
\begin{thm}\label{thm:ChoiI}
Choi isomorphism $\mathscr{T}$ is a bijective linear mapping, or isomorphism, from $\mathcal{B}(\mathcal{B}(\mathcal{H}_1),\mathcal{B}(\mathcal{H}_2))$ to $\mathcal{B}(\mathcal{H}_1\otimes \mathcal{H}_2)$. Its inverse is 
\begin{equation}
\mathscr{T}^{-1}: \mathcal{B}(\mathcal{H}_1\otimes \mathcal{H}_2)\ni O_{12}\mapsto 
\left[\mathcal{B}(\mathcal{H}_1)\ni O\mapsto \trace_1 \left((O\tr\otimes I) O_{12}\right)\in\mathcal{B}(\mathcal{H}_2)\right].\label{eq:Choii}
\end{equation}
\end{thm}
\noindent What's inside the square brackets of \eqref{eq:Choii} is a linear mapping from $\mathcal{B}(\mathcal{H}_1)$ to $\mathcal{B}(\mathcal{H}_2)$. As a bijective mapping, $\mathscr{T}\circ \mathscr{T}^{-1}$ and $\mathscr{T}^{-1}\circ \mathscr{T}$ are identity mappings on their respective domains. With (the inverse of) Choi isomorphism we can define a linear mapping from $\mathcal{B}(\mathcal{H}_1)$ to $\mathcal{B}(\mathcal{H}_2)$ by an operator on $\mathcal{H}_1\otimes \mathcal{H}_2$. 

In addition, if the space $\mathcal{B}(\mathcal{H}_1)\rightarrow\mathcal{B}(\mathcal{H}_2)$ is equipped with the Hilbert-Schmidt inner product as $\mathcal{B}(\mathcal{H}_1\otimes \mathcal{H}_2)$ does, then $\mathscr{T}$ is an isometry, so it's also an isomorphism in the Hilbert-space sense (Section~\ref{sec:isometry}); please refer to Appendix~\ref{app:isHS}.
\section{Tensor Product Space, Transposition and Partial Transposition}
\label{sec:bi}

\subsection{Choi Isomorphism on a Bipartite System}
When there are two parties Alice and Bob (A and B), the state is a density operator on a tensor product space $\mathcal{H}_{A_1}\otimes \mathcal{H}_{B_1}$. Let's consider a linear mapping $L$ from 
$\mathcal{B}(\mathcal{H}_{A_1}\otimes \mathcal{H}_{B_1})$ to $\mathcal{B}(\mathcal{H}_{A_2}\otimes \mathcal{H}_{B_2})$. We don't assume $\mathcal{H}_{A_1}$ and $\mathcal{H}_{A_2}$ are isomorphic, similarly for $\mathcal{H}_{B_1}$ and $\mathcal{H}_{B_2}$, so the dimensions before and after $L$ can be different. Assign bases for $\mathcal{B}(\mathcal{H}_{A_1})$ and $\mathcal{B}(\mathcal{H}_{B_1})$ by \eqref{eq:eij}
\begin{equation}
E_{ij}:=\ket{a_i}\bra{a_j} ,\; F_{ij}:=\ket{b_i}\bra{b_j},
\end{equation}
where $\{\ket{a_i}\}$ and $\{\ket{bi}\}$ are orthonormal bases of $\mathcal{H}_{A_1}$ and $\mathcal{H}_{B_1}$. Choi isomorphism of $L$ then is
\begin{equation}
\mathscr{T}(L)=\sum_{i,j,k,l} E_{ij}\otimes F_{kl}\otimes L(E_{ij}\otimes F_{kl}).
\end{equation}
$\mathscr{T}(L)$ is in $\mathcal{B}(\mathcal{H}_{A_1}\otimes \mathcal{H}_{B_1}\otimes \mathcal{H}_{A_2}\otimes \mathcal{H}_{B_2})$.\footnote{Or any space isomorphic to it, like $\mathcal{B}(\mathcal{H}_{A_1}\otimes \mathcal{H}_{B_1})\otimes \mathcal{B}(\mathcal{H}_{A_2}\otimes \mathcal{H}_{B_2})$.} From now on $\{E_{ij}\}$ and $\{F_{ij}\}$ will denote bases of operators.

\subsection{Transposition and Partial Transposition}\label{sec:TPT}
\subsubsection{On operators}
\begin{defi}[\textbf{Transposition and partial transposition}]
Given an orthonormal basis $\{\ket{a_i}\}$ of $\mathcal{H}$ and $E_{ij}=\ket{a_i}\bra{a_j}$, transposition $\mathrm{T}\in\mathcal{B}(\mathcal{B}(\mathcal{H}))$ with respect to $\{a_i\}$ or $\{E_{ij}\}$ is a linear mapping on operators: For any $O\in\mathcal{B}(\mathcal{H})$ whose Fourier expansion in terms of $\{E_{ij}\}$ is (Section~\ref{sec:dual} and \ref{sec:bas})
\begin{equation}
O=\sum_{i,j}(E_{ij}|O)E_{ij},
\end{equation}
its partial transpose is
\begin{equation}
O\tr:=\mathrm{T}(O)=\sum_{i,j}(E_{ij}|O)E_{ji}.
\end{equation}
Partial transposition $\Gamma:\mathcal{B}(\mathcal{H}_A\otimes \mathcal{H}_B)\rightarrow \mathcal{B}(\mathcal{H}_A\otimes \mathcal{H}_B)$, defined as
\begin{equation}
\Gamma:=\mathrm{T}_A\otimes\mathcal{I}_B,
\end{equation}
where $\mathrm{T}_A$ is transposition on $\mathcal{B}(\mathcal{H}_A)$ and $\mathcal{I}_B$ is the identity mapping on $\mathcal{B}(\mathcal{H}_B)$, is also a linear mapping. Like transpose, $O^\Gamma$ will denote the partial transpose of an operator $O\in\mathcal{B}(\mathcal{H}_A\otimes \mathcal{H}_B)$.
\end{defi}

Transposition is an example of positive but non-CP mappings \cite{Horodecki96}, i.e. $T$ is positive but $\Gamma$ is not. They're both TP, i.e. 
\begin{equation}
\trace O\tr=\trace O \text{ and }\trace O^\Gamma=\trace O,
\end{equation}
and the inverses to themselves: $\mathrm{T}\circ \mathrm{T}=\mathcal{I}$ and $\Gamma\circ \Gamma=\mathcal{I_A}\otimes\mathcal{I_B}=\mathcal{I_{AB}}$, so both are bijective and are isomorphisms.\footnote{They're also isomorphisms in the Hilbert-space sense; see Appendix~\ref{app:isHS}.} As both transposition and partial transposition commute with the adjoint mapping, they are both HP.\footnote{The adjoint here is with respect to the inner products on $\mathcal{H}_1$ and $\mathcal{H}_2$.} 

To get familiarized with the basis $\{E_{ij}\}$, let's use it to show $(O_1 O_2)\tr=O_2\tr O_1\tr$:\footnote{It may be neater to prove it through adjoint---the proof below is more like a practice.} First,
\begin{align}
O_1 O_2&=\sum_{i,j,k,l}(E_{ij}|O_1)(E_{kl}|O_2)E_{ij}E_{kl}\nonumber\\
&=\sum_{i,j,k,l}\delta_{jk} (E_{ij}|O_1)(E_{kl}|O_2)E_{il}\nonumber\\
&=\sum_{i,j,l}(E_{ij}|O_1)(E_{jl}|O_2)E_{il}\label{eq:o1o2}
\end{align}
Because
\begin{align}
(E_{ij}|O_1 O_2)=&\sum_{k,l} \Big(E_{ij}\Big|O_1 E_{kl}(E_{kl}|O_2) \Big)\nonumber\\
=&\sum_{k,l} (E_{ij}E_{lk}|O_1  ) (E_{kl}|O_2)\nonumber\\
=&\sum_{k} (E_{ik}|O_1  ) (E_{kj}|O_2)\nonumber\\
=&\sum_{k} (E_{jk}|O_2\tr)(E_{ki}|O_1\tr),
\end{align}
we obtain
\begin{equation}
(O_1 O_2)\tr=\sum_{ij} (E_{ij}|O_1 O_2) E_{ji}=\sum_{i,j,k} (E_{jk}|O_2\tr)(E_{ki}|O_1\tr)E_{ji}=O_2\tr O_1\tr,
\end{equation}
by \eqref{eq:o1o2}. 

\subsubsection{On linear mappings from operators to operators}
The partial transpose of a linear mapping $L:\mathcal{B}(\mathcal{H}_{A_1}\otimes \mathcal{H}_{B_1})\rightarrow \mathcal{B}(\mathcal{H}_{A_2}\otimes \mathcal{H}_{B_2})$ is defined as\footnote{We can similarly define a transpose of a linear mapping from operators to operators, but it's not needed in this thesis.} \cite{Rains99}
\begin{equation}
\Gamma: L\mapsto L^\Gamma:=\Gamma\circ L\circ \Gamma,\label{eq:par}
\end{equation} 
where the first (rightmost) $\Gamma$ in \eqref{eq:par} maps $\mathcal{B}(\mathcal{H}_{A_1}\otimes \mathcal{H}_{B_1})$ to $\mathcal{B}(\mathcal{H}_{A_1}\otimes \mathcal{H}_{B_1})$, and the second one maps $\mathcal{B}(\mathcal{H}_{A_2}\otimes \mathcal{H}_{B_2})$ to $\mathcal{B}(\mathcal{H}_{A_2}\otimes H_{B_2})$. \eqref{eq:par} itself is a linear mapping on linear mappings (from operators to operators):
\begin{equation*}
\Gamma \in \mathcal{B}\left\{\mathcal{B}\left[\mathcal{B}(\mathcal{H}_{A_1}\otimes \mathcal{H}_{B_1}), \mathcal{B}(\mathcal{H}_{A_2}\otimes \mathcal{H}_{B_2})\right]\right\}.
\end{equation*}
Like partial transposition on operators, this partial transposition on linear mappings is invertible and it's its own inverse: $(L^\Gamma)^\Gamma=L$. Hence, it's also an isomorphism.\footnote{See Appendix~\ref{app:isHS}.} Since partial transposition on operators is HP/TP, for any HP/TP mapping $L$ its partial transpose $L^\Gamma$ is also HP/TP. This property of partial transposition will play an important role in Chapter~\ref{ch:EC}.

\subsection{PPT, Peres Criterion and Negativity}\label{sec:neg}
A state $\rho$ whose partial transpose is positive is called a PPT (positive-partial-transpose) state \cite{Horodecki97,Rains99}. PPT states are important in entanglement theory, because of:
\begin{thm}[\textbf{Peres-Horodecki criterion}]\label{thm:peres}
For a density operator $\rho$ on $\mathcal{H}_A\otimes \mathcal{H}_B$ to be separable, it's necessary that $\rho$ is PPT. This is a necessary and sufficient condition if $\mathrm{dim}\mathcal{H}_A=2$ and $\mathrm{dim}\mathcal{H}_B=3$, or if $\mathrm{dim}\mathcal{H}_A=3$ and $\mathrm{dim}\mathcal{H}_B=2$ \cite{Peres96,Horodecki96,Horodecki97}.
\end{thm}

In other words, (the set of) separable states are a subset of (the set of) PPT states, and in certain dimensions the two sets are identical \cite{Chitambar14}. 

For a quantum operation $S$, if its partial transpose is CP, it's called a PPT-preserving operation, or simply \textbf{PPT operation} \cite{Rains99,Cirac01}, because if $S^\Gamma$ is CP and the input state $\rho$ is PPT, then
\begin{equation}
S(\rho)^\Gamma:=(S(\rho))^\Gamma=S^\Gamma(\rho^\Gamma)\geq 0,
\end{equation}
that is, a PPT state is still PPT after a PPT operation; as a comparison, a separable operation preserves separability of a state \cite{Rains99,Cirac01}. 

\textbf{Negativity} is basically a quantitative version of Peres criterion \cite{Vidal02}: For a density operator $\rho$ on $\mathcal{H}_A\otimes \mathcal{H}_B$, its negativity is defined as
\begin{equation}
E_N(\rho):=\left(||\rho^\Gamma||_1-1\right)/2=\trace(\rho^\Gamma)^-,
\end{equation}
where the second equality comes from the discussion in Section~\ref{sec:eig} and partial transposition being TP. A similar quantity/function, the \textbf{logarithmic negativity} is defined as
\begin{equation}
E_L(\rho):=\log||\rho^\Gamma||.
\end{equation}
Again, since partial transposition is TP, $E_L(\rho)$ is always non-negative. Both negativities are motononic under LOCC or a PPT operation, so they are genuine entanglement monontones (or measures) \cite{Vidal02,Plenio05}.

\section{Adjoint, Choi Isomorphism and Complex Conjugation}
\label{sec:l2}
\subsection{Adjoint and Choi Isomorphism}
Having discussed transposition, we can now introduce a lemma relating Choi isomorphism and the adjoint of a linear mapping. With $L:\mathcal{B}(\mathcal{H}_1)\rightarrow\mathcal{B}(\mathcal{H}_2),$ $\mathscr{T}(L)\in \mathcal{B}(\mathcal{H}_1)\otimes \mathcal{B}(\mathcal{H}_2)$, and $\trace_2$ in the following lemma traces out $\mathcal{B}(\mathcal{H}_2)$. The proof of this lemma also serves to demonstrate complex conjugation on operators, which we will treat more rigorously in the next subsection. 
\begin{lem}\label{lem:tl}
For a linear mapping $L:\mathcal{B}(\mathcal{H}_1)\rightarrow\mathcal{B}(\mathcal{H}_2)$,
\begin{equation*}
L^\dagger (I)=\trace_2\mathscr{T}(L)^*.
\end{equation*}
\end{lem}

\begin{proof}
First of all,
\begin{equation}
\trace_2\mathscr{T}(L)^*=\sum_{i,j}E_{ij}\trace L(E_{ij})^*,\label{eq:p2}
\end{equation}
where $ L(E_{ij})^*$ denotes $(L(E_{ij}))^*$. From Section~\ref{sec:Choi}, for any $O\in\mathcal{B}(\mathcal{H}_1)$
\begin{align}
\trace L(O)&=\trace (O\tr\otimes I\mathscr{T}(L))\nonumber\\
&=\sum_{i,j}\trace (O\tr E_{ij})\trace L(E_{ij})\nonumber\\
&=\sum_{i,j}(E_{ij}|O)\trace L(E_{ij})\nonumber\\
&=\Big(\sum_{i,j}E_{ij}\trace L(E_{ij})^*\Big|O\Big)\nonumber\\
&=(\trace_2\mathscr{T}(L)^*|O)\label{eq:p1}
\end{align}
By comparing \eqref{eq:p1} and
\begin{equation*}
\trace L(O)=(I|L(O))=(L^\dagger(I)|O)\label{eq:p4}
\end{equation*}
we can conclude $L^\dagger(I)=\trace_2 \mathscr{T}(L)^*.$
\end{proof}

\subsection{Complex Conjugation and Basis} 
\label{sec:l2b}
In the proof for Lemma~\ref{lem:tl} it's important in which basis we performed complex conjugation, because the transpose and complex conjugate of an operator depend on the basis: Given an orthonormal basis $\{E_{ij}\}$ of $\mathcal{B}(\mathcal{H})$, complex conjugation on $\mathcal{B}(\mathcal{H})$ is
\begin{equation}
*: \mathcal{B}(\mathcal{H})\ni O\mapsto O^*:= \sum_{i,j}(E_{ij}|O)^* E_{ij}=\sum_{i,j}(O|E_{ij}) E_{ij}.
\end{equation}
Obviously complex conjugation is a conjugate linear mapping from $\mathcal{B}(\mathcal{H})$ to $\mathcal{B}(\mathcal{H})$. It also satisfies
\begin{equation}
\trace (O^*)=(\trace O)^* \text{ for any }O\in\mathcal{B}(\mathcal{H}).\label{eq:o*}
\end{equation}

Let's return back to the proof of Lemma~\ref{lem:tl}. If the Choi isomorphism of $L$ is carried out with respect to $\{E_{ij}\}$, then the complex conjugation on $\mathscr{T}(L)\in \mathcal{B}(\mathcal{H}_1)\otimes \mathcal{B}(\mathcal{H}_2)$ should be executed with respect to $\{E_{ij}\otimes F_{kl}\}$, where $\{F_{kl}\}$ is \emph{any orthonormal basis} of $\mathcal{B}(\mathcal{H}_2)$: Because of \eqref{eq:o*} and because trace is basis-independent,  $\trace \left[L(E_{ij})^*\right]$ will always be equal to $\left[\trace L(E_{ij})\right]^*$ in whichever orthonormal basis of $\mathcal{B}(\mathcal{H}_2)$, and $\trace_2\mathscr{T}(L)^*$ will turn out to be the same. More generally, 
\begin{lem}
Let $\mathcal{H}_1$ and $\mathcal{H}_2$ be finite-dimensional Hilbert spaces, $\{E_{ij}\}$ be an orthonormal basis of $\mathcal{B}(\mathcal{H}_1)$ and $\{F_{ij}\}$ and $\{F_{ij}'\}$ be two orthonormal bases of $\mathcal{B}(\mathcal{H}_2)$. Given any operator $O\in \mathcal{B}(\mathcal{H}_1\otimes \mathcal{H}_2)$, $\trace_2 O^*$ will be the same whether the complex conjugate is executed with respect to $\{E_{ij}\otimes F_{kl}\}$ or $\{E_{ij}\otimes F_{kl}'\}$.
\end{lem}
\begin{proof}
Suppose $O=\sum_{i,j} E_{ij}\otimes  B_{ij}.$\footnote{$B_{ij}$ will be uniquely determined; see Theorem 14.5 of Ref.~\cite{Roman}.} Then $O^*=\sum_{i,j} E_{ij}\otimes  B_{ij}^*$, where $B_{ij}$ depend on the basis chosen for $\mathcal{B}(\mathcal{H}_2)$. Hence $\trace_2 O^*=\sum_{i,j} E_{ij}\trace  B_{ij}^*$, which from \eqref{eq:o*} will be the identical regardless of the basis for $\mathcal{B}(\mathcal{H}_2)$.
\end{proof}

Now let's assume the mapping $L$ is HP, so we can perform a spectral decomposition on $\mathscr{T}(L)$ as it is Hermitian (Theorem~\ref{thm:Choi}). Suppose its spectral decomposition is 
\begin{equation}
\mathscr{T}(L)=\sum_i c_i \ket{v_i}\bra{v_i}
\end{equation}
with $\{v_i\}$ being orthonormal. Let
\begin{equation*}
\ket{v_i}=\sum_{j,k}d_{jk}^i\ket{a_j}\ket{b_k},
\end{equation*}
where $\{b_i\}$ is any orthonormal basis of $\mathcal{H}_2,$ and $\{E_{ij}=\ket{a_i}\bra{a_j}\}$ is the basis chosen for Choi isomorphism. Then,
\begin{align}
(\ket{v_i}\bra{v_i})^*&=\sum_{j,k,l,m}{d_{jk}^i}^* d_{lm}^i\ket{a_j}\bra{a_l}\otimes \ket{b_k}\bra{b_m}\nonumber\\
&=\sum_{j,k}{d_{jk}^i}^*\ket{a_j}\ket{b_k}\sum_{l,m}d_{lm}^i\bra{a_l} \bra{b_m}\nonumber\\
&=\ket{v_i^*}\bra{v_i^*},
\end{align}
where
\begin{equation*}
\ket{v_i^*}:=\sum_{j,k}{d_{jk}^i}^*\ket{a_j}\ket{b_k}.
\end{equation*}
Because $\inner{u}{v}^*=\inner{u^*}{v^*}=\inner{v}{u}$ for any vector $u$ and $v$,
\begin{equation}
\inner{v_i^*}{v_j^*}=\inner{v_j}{v_i}=\delta_{ij}.\label{eq:or}
\end{equation}
We thus obtain the complex conjugate of $\mathscr{T}(L)$ with respect to the basis $\{\ket{a_i}\bra{a_j}\otimes \ket{b_k}\bra{b_l}\}$:
\begin{equation}
\mathscr{T}(L)^*=\sum_i c_i\ket{v_i^*}\bra{v_i^*}\label{eq:TL}
\end{equation}
Since $\{v_i\}$ is orthonormal, \eqref{eq:TL} is the spectral decomposition of $\mathscr{T}(L)^*$. From \eqref{eq:or} and \eqref{eq:TL}, $\mathscr{T}(L)^*$ and $\mathscr{T}(L)$ have the same eigenvalues, so
\begin{cor}[of Lemma \ref{lem:tl}]\label{cor:1}
For any HP mapping $L,$
\begin{equation}
\trace L^\dagger(I)=\trace \mathscr{T}(L).
\end{equation}
\end{cor}

\subsection{Adding an Ancilla}
Here let's prove another corollary of Lemma~\ref{lem:tl}: 
\begin{cor}[of Lemma \ref{lem:tl}]\label{cor:2lm1}
For any linear mapping $L:\mathcal{B}(\mathcal{H}_1)\rightarrow \mathcal{B}(\mathcal{H}_2)$,
\begin{equation*}
(\mathcal{I}_a\otimes L)^\dagger(I_{\mathcal{H}_a\otimes \mathcal{H}_2})=I_a\otimes L^\dagger(I_{2}),
\end{equation*}
where $\mathcal{I}_a$ is the identity mapping on operators on the ``ancilla'' space $\mathcal{H}_a$, and the subscript beside each identity operator refers to the domain. 
\end{cor}
From this corollary $||(\mathcal{I}_a\otimes L)^\dagger(I_{\mathcal{H}_a\otimes \mathcal{H}_2})||=||I_{a}||\,||L^\dagger(I_2)||=||L^\dagger(I_2)||.$ 

\begin{proof}
As the domain of each mapping is already made clear from the statement of the corollary, we will ignore all the subscripts. Let $\{G_{ij}\}$ be an orthonormal basis of $\mathcal{B}(\mathcal{H}_a)$. We can find\footnote{As in Section~\ref{sec:bi}, we implicitly assume the basis with respect to which Choi isomorphism is performed is composed of decomposable (Section~\ref{sec:tensor}) vectors/operators, i.e. of the form $G_{ij}\otimes E_{kl}$.}
\begin{equation*}
\mathscr{T}(\mathcal{I}\otimes L)=\Big(\sum_{i,j}G_{ij}\otimes G_{ij}\Big)\otimes \mathscr{T}(L).
\end{equation*}
By Lemma~\ref{lem:tl},
\begin{align*}
(\mathcal{I}\otimes L)^\dagger(I)&=\trace_2 \mathscr{T}(\mathcal{I}\otimes L)^*\\
&=\trace_2 \Big(\sum_{i,j}G_{ij}\otimes G_{ij}\Big)\otimes\trace_2 \mathscr{T}(L)^*\\
&=\sum_{i,j}G_{ij}\delta_{ij}\otimes L^\dagger(I)\\
&=I\otimes L^\dagger(I).
\end{align*}
This proof can be compared with that in Ref.~\cite{Campbell10}.
\end{proof}

\section{HP and TP Mappings}

\subsection{Decomposition of an HP Mapping}\label{sec:dehp}
As an HP mapping $L:\mathcal{B}(\mathcal{H}_1)\rightarrow \mathcal{B}(\mathcal{H}_2)$ has a Hermitian Choi isomorphism (Theorem~\ref{thm:Choi}), so $\mathscr{T}(L)=\widetilde{\mathscr{T}(L)}^+-\widetilde{\mathscr{T}(L)}^-$, with $\widetilde{\mathscr{T}(L)}^\pm\geq 0$ (Section~\ref{sec:eig}). By Theorem~\ref{thm:ChoiI}, we can define such linear mappings:
\begin{equation}
\widetilde{L}_\pm(O):=\trace_1 \big[(O\tr \otimes I)\widetilde{\mathscr{T}(L)}^\pm\big] \text{ for any }O\in\mathcal{B}(\mathcal{H}_1),
\end{equation}
namely 
\begin{equation}
\widetilde{L}_\pm:=\mathscr{T}^{-1}\left(\widetilde{\mathscr{T}(L)}^\pm\right);
\end{equation}
in particular, by eigendecomposing $\mathscr{T}(L)$, 
\begin{equation}
L_\pm:=\mathscr{T}^{-1}\left(\mathscr{T}(L)^\pm\right).
\end{equation}
Since $\widetilde{\mathscr{T}(L)}^\pm\geq 0$,  
\begin{equation}
L=\widetilde{L}_+-\widetilde{L}_-,\; \widetilde{L}_\pm \text{ are CP},
\end{equation}
so we can always decompose an HP mapping as the difference between two CP mappings.

From the discussion in Section~\ref{sec:innerHP} and \ref{sec:adl}, $L^\dagger(I)\geq 0$ for a CP mapping $L.$ Therefore, for an HP mapping $L$, $\widetilde{L}_\pm^\dagger(I)\geq 0,$ and 
\begin{equation}
||\mathscr{T}(L)||_1=\trace L_+^\dagger(I)+\trace L_-^\dagger(I),
\end{equation} 
by Corollary~\ref{cor:1} and $||H||_1=\trace H^++\trace H^-$.

\subsubsection{HPTP mapping}
The equalities shown here will be utilized later. Suppose $L: \mathcal{B}(\mathcal{H}_1)\rightarrow \mathcal{B}(\mathcal{H}_2)$ is HPTP, so $L^\dagger(I_2)=\widetilde{L}_+^\dagger(I_2)-\widetilde{L}_-^\dagger(I_2)=I_1$, and
\begin{equation}
\widetilde{L}_+^\dagger(I_2)=I_1+\widetilde{L}_-^\dagger(I_2),
\end{equation}
leading to
\begin{equation}
\text{dim}\mathcal{H}_1+2\trace \widetilde{L}_-^\dagger(I_2)=\trace \left(I_1+2\widetilde{L}_-^\dagger(I_2)\right)=\trace\left(\widetilde{L}_+^\dagger(I_2)+\widetilde{L}_-^\dagger(I_2)\right).\label{eq:l1}
\end{equation}
Because $||I+P||=1+||P||$ for $P\geq 0$,
\begin{equation}
1+2 ||\widetilde{L}_-^\dagger(I_2)||=||I_1+2\widetilde{L}_-^\dagger(I_2)||=||\widetilde{L}_+^\dagger(I_2)+\widetilde{L}_-^\dagger(I_2)||.\label{eq:n1}
\end{equation}

\subsection{Operator-sum Representation}
\label{sec:opsum}
As mentioned in Section~\ref{sec:adl}, every HP mapping $L:\mathcal{B}(\mathcal{H}_1)\rightarrow \mathcal{B}(\mathcal{H}_2)$ has an operator-sum representation: For any $O\in\mathcal{B}(\mathcal{H}_1)$
\begin{equation}
L(O)=\sum_i c_i V_i O V_i^\dagger,\,V_i\in\mathcal{B}(\mathcal{H}_1,\mathcal{H}_2) \text{ and }c_i\in\mathbb{R},
\end{equation}
and in the case of quantum operations this is often called the Kraus form \cite{Kraus71,Choi75,Bengtsson}. The operator-sum representation isn't unique. 

While quite often a linear mapping is given in an operator-sum representation, sometimes it's not. For example, transposition and partial transposition aren't defined in such expressions (Section~\ref{sec:TPT}), and in Chapter~\ref{ch:EC} we will need the operator-sum representation of the partial transpose of an operation. Here let's demonstrate how to find an operator-sum representation through Choi isomorphism.

Suppose the spectral decomposition of $\mathscr{T}(L)$ is $\mathscr{T}(L)=\sum_i c_i \ket{v_i}\bra{v_i}$, with orthonormal $\{\ket{v_i}=\sum_{j,k} d_{jk}^i \ket{a_j}\ket{b_k}\},$ where $\{a_i\}$ and $\{b_i\}$ are orthonormal bases of $\mathcal{H}_1$ and $\mathcal{H}_2$ respectively. Then by \eqref{eq:Choi} for any $O\in \mathcal{B}(\mathcal{H}_1)$,
\begin{align}
L(O)&=\trace_1\left((O\tr\otimes I) \sum_i c_i \ket{v_i}\bra{v_i}\right)\nonumber\\
&=\sum_{i,j,k,l,m} c_i d_{jk}^i {d_{lm}^i}^*\trace_1\left(O\tr\ket{a_j}\bra{a_l}\otimes I  \ket{b_k}  \bra{b_m}\right)\nonumber\\
&=\sum_{i,j,k,l,m} c_i d_{jk}^i {d_{lm}^i}^* \bracket{a_j}{O}{a_l}\ket{b_k}\bra{b_m}.\label{eq:osum}
\end{align}
Hence by defining
\begin{equation}
V_i:= \sum_{j,k} d^i_{jk}\ket{b_k}\bra{a_j},\label{eq:V}
\end{equation}
\eqref{eq:osum} becomes
\begin{equation}
L(O)=\sum_{i} c_i V_iOV_i^\dagger,
\end{equation}
and we obtain an operator-sum representation of $L$. Note
\begin{equation}
(V_i|V_j)=(v_i|v_j)=\delta_{ij},
\end{equation}
so $\{V_i\}$ is an orthonormal set of linear mappings from $\mathcal{H}_1$ to $\mathcal{H}_2$.

\subsection{Trace Norm of a Hermitian Operator after an HP or HPTP Mapping}
Given an HP mapping $L: \mathcal{B}(\mathcal{H}_1)\rightarrow\mathcal{B}(\mathcal{H}_2)$ and a Hermitian operator $H\in \mathcal{B}(\mathcal{H}_1)$, $L(H)$ is a Hermitian operator in $\mathcal{B}(\mathcal{H}_2)$, and we'd like to find an upper bound for $||L(H)||_1$.

Since $\widetilde{L}_\pm$ are CP,
\begin{equation}
L(H)=(\widetilde{L}_+-\widetilde{L}_-)(H^+-H^-)=\left[\widetilde{L}_+(H^+)+\widetilde{L}_-(H^-)\right]-\left[\widetilde{L}_+(H^-)+\widetilde{L}_-(H^+)\right],
\end{equation}
with
\begin{equation}
\widetilde{L}_+(H^+)+\widetilde{L}_-(H^-)\geq 0 \text{ and }\widetilde{L}_+(H^-)+\widetilde{L}_-(H^+)\geq 0,
\end{equation}
which is therefore a decomposition of $L(H)$ as two positive operators, i.e. 
\begin{equation}
\widetilde{L(H)}^\pm=\widetilde{L}_+(H^\pm)+\widetilde{L}_-(H^\mp),
\end{equation}
c.f. Section~\ref{sec:eig}. By Lemma~\ref{lem:or} \cite{Campbell10},
\begin{equation}
\trace L(H)^\pm\leq \trace\widetilde{L(H)}^\pm=  \trace\left[\widetilde{L}_+(H^\pm)+\widetilde{L}_-(H^\mp)\right],\label{eq:LH}
\end{equation}
and
\begin{equation}
||L(H)||_1=\trace L(H)^++\trace L(H)^- \leq \trace\widetilde{L(H)}^++\trace\widetilde{L(H)}^-.\label{eq:LH1}
\end{equation}
By \eqref{eq:LH} and \eqref{eq:P1P2} we can find
\begin{align}
||L(H)||_1&\leq\trace\left[\widetilde{L}_+(H^+)+\widetilde{L}_-(H^-)\right]+\trace\left[\widetilde{L}_+(H^-)+\widetilde{L}_-(H^+)\right]\nonumber\\
&=(\widetilde{L}_+^\dagger(I)+\widetilde{L}_-^\dagger(I)|H^++H^-)\nonumber\\
&\leq ||\widetilde{L}_+^\dagger(I)+\widetilde{L}_-^\dagger(I)||\,||H||_1.\label{eq:lh1}
\end{align}

If, in addition to being HP, $L$ is also TP, then
\begin{align*}
\widetilde{L(H)}^\pm&= \trace\left[\widetilde{L}_+(H^\pm)+\widetilde{L}_-(H^\mp)\right]\\
&=\trace \left[(\widetilde{L}_-+L)(H^\pm)+\widetilde{L}_-(H^\mp)\right]\\
&=\trace H^\pm+\trace \left[L_-(H^+ +H^-)\right]\\
&\leq \trace H^\pm+||\widetilde{L}_-^\dagger(I)||\;||H||_1,
\end{align*}
by \eqref{eq:P1P2} and $L$ being TP. Therefore from \eqref{eq:LH1} we obtain
\begin{align}
||L(H)||_1&\leq \trace H^++\trace H^-+2||(\widetilde{L}_-)^\dagger(I)||\;||H||_1\nonumber\\
&=||H||_1+2 ||\widetilde{L}_-^\dagger(I)||\;||H||_1\label{eq:lh2}
\end{align}
Note by \eqref{eq:n1}, \eqref{eq:lh1} and \eqref{eq:lh2} are actually the same for HPTP mappings.

More generally, by matrix H\"{o}lder inequality \eqref{eq:Holder} we have the following lemma:
\begin{lem}\label{lem:lh}
Suppose $1 \leq q,p\leq\infty$ and $1/p+1/q=1$. If $L$ is an HP mapping and $H$ is a Hermitian operator,
\begin{equation}
\trace L(H)^\pm\leq \trace\left[\widetilde{L}_+(H^\pm)+\widetilde{L}_-(H^\mp)\right],\label{eq:lh}
\end{equation}
and
\begin{equation}
||L(H)||_1\leq ||\widetilde{L}_+^\dagger(I)+\widetilde{L}_-^\dagger(I)||_p||H||_q.\label{eq:lhq1}
\end{equation}
If $L$ is also TP, then
\begin{align}
\trace L(H)^\pm&\leq \trace H^\pm+\trace \left[L_-(H^+ +H^-)\right]\label{eq:lht}\\
&=\trace H^\pm+||\widetilde{L}_-^\dagger(I)||_p\;||H||_q\label{eq:lht2}
\end{align}
and
\begin{equation}
||L(H)||_1\leq||H||_1+2 ||\widetilde{L}_-^\dagger(I)||_p||H||_q.\label{eq:lhq2}
\end{equation}

\eqref{eq:lh} and \eqref{eq:lht} become equalities if and only if
\begin{equation}
\mathrm{ran}\left(\widetilde{L}_+(H^+)+\widetilde{L}_-(H^-)\right)\perp\mathrm{ran}\left(\widetilde{L}_+(H^-)+\widetilde{L}_-(H^+)\right);\label{eq:conqi}
\end{equation}
In the case $p=\infty$ and $q=1$, for \eqref{eq:lhq1}, \eqref{eq:lht2} and \eqref{eq:lhq2} to become equalities the necessary and sufficient condition is that \eqref{eq:conqi} be obeyed and $\mathrm{ran}H$ be a subspace of the eigenspace of $\widetilde{L}_+^\dagger(I)+\widetilde{L}_-^\dagger(I)$ with the largest eigenvalue, which eigenspace for \eqref{eq:lht2} and \eqref{eq:lhq2} is equivalent to the eigenspace of $\widetilde{L}_\pm^\dagger(I)$ with the largest eigenvalue. Same can be said of the case $p=1$ and $q=\infty$.
\end{lem}

\begin{proof}
What remains to be shown is the last part of the lemma. The condition \eqref{eq:conqi} is just a direct application of Lemma~\ref{lem:or}. From the discussion above, \eqref{eq:conqi} is a premise for \eqref{eq:lhq1}, \eqref{eq:lht2} and \eqref{eq:lhq2} to become equal. 

When $p=\infty$ and $q=1$, from \eqref{eq:lh1} and the discussion below \eqref{eq:P1P2} we know they become equalities if and only if $\mathrm{ran}(H^++H^-)$ is a subspace of the eigenspace of $\widetilde{L}_+^\dagger(I)+\widetilde{L}_-^\dagger(I)$ with the largest eigenvalue. Since $H^+$ and $H^-$ are Hermitian and $\mathrm{ran}H^+$ and $\mathrm{ran}H^-$ are orthogonal, (by the spectral theorem)
\begin{equation}
\mathrm{ran}(H^++H^-)=\mathrm{ran}(H^+-H^-)=\mathrm{ran}H.
\end{equation}
Given any HPTP mapping $L$, $\widetilde{L}_+^\dagger(I)=\widetilde{L}_-^\dagger(I)+I$, their spectral decompositions are exactly the same except for a difference of 1 in all the eigenvalues, so their eigenspaces with the largest eigenvalues are identical, c.f. \eqref{eq:n1}, which also demonstrates the equivalence of \eqref{eq:lh1} and \eqref{eq:lh2} (or \eqref{eq:lhq1} and \eqref{eq:lhq2} with $p=\infty$ and $q=1$) for HPTP mappings.
\end{proof}
We could apply matrix H\"{o}lder inequality to \eqref{eq:lh} too, but the result isn't particularly fruitful so we will do it on the fly when it's really needed.

\chapter{Entangling Capacity of a Quantum Operation}
\label{ch:EC}
In this chapter it's always assumed the Hilbert spaces under consideration are finite-dimensional. As a result the spaces of operators are finite-dimensional too, as are the linear mappings thereon.

\section{Entangling Capacity and Perfect Entangler}
\subsection{Entangling Capacity and the Ancilla}\label{sec:EC}
Consider two Hilbert spaces $\mathcal{H}_1$ and $\mathcal{H}_2$, for which
\begin{equation}
\mathcal{H}_i=\mathcal{H}_i^A\otimes \mathcal{H}_i^B,
\end{equation}
and any deterministic quantum operation (Section~\ref{sec:qop}) $S:\mathcal{B}(\mathcal{H}_1)\rightarrow \mathcal{B}(\mathcal{H}_2)$. The \textbf{entangling capacity} of $S$ with respect to an entanglement measure $E$ between A and B is defined as\footnote{Since there are multiple ways of describing the capability or capacity of an operation to produce entanglement, there doesn't seem to be consensus on the terminology for this quantity. For example it's called the entangling power in Ref.~\cite{Linden09}, and the entangling capacity along with descriptive words may refer to quantities, like in Ref.~\cite{Bennett03,Lari09}. We will address the issue of whether to choose the maximum or supremum near the end of this subsection.} \cite{Kraus01,Leifer03,Bennett03,Chefles05,Linden09,Campbell10,Cohen11}
\begin{equation}
	\text{EC}_E(S):=\max_{\rho}\{E\left( S(\rho)\right)-E(\rho)\},\label{eq:EC}
\end{equation}
where the maximization is over all density operators $\rho$ on $\mathcal{H}_1$. $\mathrm{EC}_E(S)\geq 0$ for any operations, because for any separable state $\rho$ we have $E(\rho)=0$. By the same token, for a probabilistic operation $S$ with sub-operations $S_i$, define the (average) entangling capacity as
\begin{equation}
\text{EC}_E(S):=\max_\rho\Big\{\sum_i p_i E\left(\overline{S_i(\rho)}\right)-E(\rho) \Big\},
\end{equation}
where the overhead bar normalizes the state, i.e. $\overline{P}=P/\trace P$ for any $P\geq 0$, and $p_i$ is the probability that $S_i$ occurs. We can also define the entangling capacity for a sub-operation $S_i$ in a similar manner.

As mentioned in Chapter~\ref{ch:int}, we're interested in the entangling capacity of an operation when aided by an ancilla. Now consider an additional ``ancilla'' space 
\begin{equation}
\mathcal{H}_a=\mathcal{H}_a^A\otimes \mathcal{H}_a^B,
\end{equation}
and let $S:\mathcal{B}(\mathcal{H}_1)\rightarrow \mathcal{B}(\mathcal{H}_2)$ be a deterministic operation, and $\mathcal{I}_a:\mathcal{B}(\mathcal{H}_a)\rightarrow \mathcal{B}(\mathcal{H}_a)$ the identity mapping. The entangling capacity of $S$ with respect to an entanglement measure $E$, when aided by an ancilla, is then defined as \cite{Kraus01,Leifer03,Chefles05,Campbell10}
\begin{equation}
	\text{EC}_E(S):=\max_\rho\{E\left(\mathcal{I}_a\otimes S(\rho)\right)-E(\rho)\},
\end{equation}
where the maximization is over all density operators on $\mathcal{H}_a\otimes \mathcal{H}_1$ and $E$ measures the entanglement between $\mathcal{H}_a^A\otimes \mathcal{H}_i^A$ and $\mathcal{H}_a^B\otimes \mathcal{H}_i^B$. We can define the (average) entangling capacity aided by an ancilla for probabilistic operations and sub-operations in the same way.

In Ref.~\cite{Bennett03,Linden09,Campbell10,Cohen11}, entangling capacity (or its analogy) was actually defined to be the supremum, rather than maximum over all states. However, if the measure $E$ is continuous, then since the space is finite-dimensional, density operators in the space of operators is a compact set $\mathcal{K}$, and $E(\mathcal{K})$ is also compact and is therefore bounded and closed \cite{Loomis,Tu}, so there's no difference between the supremum and maximum. On the other hand, if the system in question is infinite-dimensional, it may be more suitable to choose the supremum over maximum.

From now on the entangling capacity (EC), unless stated otherwise, is always assumed to be aided by an ancilla, and $\mathrm{EC}_N(S)$ and $\mathrm{EC}_L(S)$ denote the entangling capacity with respect to negativity and logarithmic negativity (Section~\ref{sec:neg}) respectively.

\subsection{Perfect Entangler}\label{sec:pee}
An operation is called a \textbf{perfect entangler} (with respect to measure $E$) if its entangling capacity is at the maximum under the constraint of the system. When the measure is negativity (Section~\ref{sec:neg}), with the dimensions of $\mathcal{H}_A$ and $\mathcal{H}_B$ being $d_A$ and $d_B$, because
\begin{equation}
\max ||\rho^\Gamma||_1=\min (d_A,d_B),\label{eq:enlmax}
\end{equation}
negativities are at most
\begin{align}
\max_\rho E_N(\rho)&=\frac{\min(d_A,d_B)-1}{2},\nonumber\\
\max_\rho E_L(\rho)&=\log\min(d_A,d_B);
\end{align}
please read Appendix~\ref{app:en} for details. Hence, if an operation $S$ with codomain $\mathcal{B}(\mathcal{H}_A\otimes \mathcal{H}_B)$ has $\text{EC}_N(S)=(\min(d_A,d_B)-1)/2$ or $\text{EC}_L(S)=\log\min(d_A,d_B),$ then $S$ is a perfect entangler with respect to negativities.

\section{Bounds on Entangling Capacity}
Here again we consider operations and sub-operations $S$ and $S_i$ from $\mathcal{B}(\mathcal{H}_1^A\otimes \mathcal{H}_1^B)$ to $\mathcal{B}(\mathcal{H}_2^A\otimes \mathcal{H}_2^B)$. $d_A$ and $d_B$ refer to the dimensions of $\mathcal{H}_1^A$ and $\mathcal{H}_1^B$, respectively. In this section and the followings the domains (and codomains) of identity operators, unless necessary, will usually be ignored, c.f. Section~\ref{sec:adl}.
\subsection{The Bounds}
 $S^\Gamma_-$ and $\widetilde{S^\Gamma}_-$ refers to the ``negative'' part (Section~\ref{sec:dehp}) of $S^\Gamma=\Gamma\circ S\circ \Gamma$, the partial transpose of $S$ (Section~\ref{sec:TPT}); similarly for ${S^\Gamma_i}_-$ and $\widetilde{S^\Gamma_i}_-$. Not having an overhead tilde indicated that the decomposition is obtained through eigendecomposing the Choi isomorphism (Section~\ref{sec:dehp}).
\begin{pro}\label{pro}
1. There exist upper and lower bounds for entangling capacities of deterministic operations:
\begin{align*}
\frac{||{S^\Gamma_-}^\dagger(I)||_1}{d_A d_B}\leq&\mathrm{EC}_N(S)\leq ||\widetilde{S^\Gamma}_-^\dagger(I)||\,||\rho^\Gamma||_1\\
\log\Big(1+2\frac{||{S^\Gamma_-}^\dagger(I)||_1}{d_A d_B}\Big)\leq&\mathrm{EC}_L(S)\leq\log(1+2||\widetilde{S^\Gamma}_-^\dagger(I)||).
\end{align*}

2. For a probabilistic operation composed of sub-operations $S_i$,
\begin{align*}
\sum_i\frac{||{S_i^\Gamma}_-^\dagger(I)||_1}{d_A d_B} \leq&\mathrm{EC}_N(S)\leq ||\sum_i\widetilde{S^\Gamma_i}^\dagger_-(I)||\,||\rho^\Gamma||_1,\\
\sum_{i}\frac{\trace \mathscr{T}(S_i)}{d_A d_B}\log\frac{||\mathscr{T}(S^\Gamma_i)||_1}{\trace \mathscr{T}(S_i)}\leq&\mathrm{EC}_L(S) 
\leq \log(1+2||\sum_i\widetilde{S^\Gamma_i}^\dagger_-(I)||).
\end{align*}
The upper bounds of part 2 can be applied to a deterministic operation $S=\sum_i S_i$.

3. With an initial negativity $E_N$, the expected negativity, i.e. probability times $p_i$ the actual negativity ${E_N}_i$, after a sub-operation $S_i$ is bounded by:
\begin{equation*}
p_i {E_N}_i\leq 
E_N\Big(||\widetilde{S^\Gamma_i}_-^\dagger(I)||+||\widetilde{S^\Gamma_i}_+^\dagger(I)||\Big)+||\widetilde{S^\Gamma_i}_-^\dagger(I)||.
\end{equation*}
The entangling capacity of a sub-operation is bounded from below by:
\begin{align*}
EC_N(S_i)&\geq \frac{\trace {S_i^\Gamma}_-^\dagger(I)}{\trace S_i^\dagger(I)}=\left(\frac{\trace {S_i^\Gamma}_+^\dagger(I)}{{\trace S_i^\Gamma}_-^\dagger(I)}-1\right)^{-1},\\
EC_L(S_i)&\geq\log\left(1+2\frac{\trace {S^\Gamma_i}_-^\dagger(I)}{\trace  S_i^\dagger(I)}\right)=\log\left[1+2\left(\frac{\trace{S_i^\Gamma}_+^\dagger(I)}{\trace {S_i^\Gamma}_-^\dagger(I)}-1\right)^{-1}\right].
\end{align*}
By Corollary~\ref{cor:2lm1}, all the upper bounds remain the same after the addition of an ancilla: $||(\mathcal{I}\otimes \widetilde{S_i^\Gamma}_-)^\dagger(I)||=||\widetilde{S_i^\Gamma}^\dagger_-(I)||$.
\end{pro}

\subsection{Discussion}\label{sec:dispro1}
Acquired through spectral decomposition, ${S_i^\Gamma}_\pm$ is a natural choice of $\widetilde{S_i^\Gamma}_\pm$, and ${S^\Gamma_i}_-^\dagger(I)$ has the smallest trace (norm) by Lemma~\ref{lem:or} and Corollary~\ref{cor:1}, but it may not have the smallest operator norm: $||\widetilde{S^\Gamma_i}_-^\dagger(I)||<||{S^\Gamma_i}_-^\dagger(I)||$ can occur, as shown in Appendix~\ref{app:smanorm}. In the following discussion I will assume $\widetilde{S_i^\Gamma}_-={S_i^\Gamma}_-$. Note the argument in this paragraph can and should be applied to deterministic operations too.

The upper bounds for $\text{EC}_N(S)$, which aren't state-independent, are actually maximized over all states $\rho$ with a given $||\rho^\Gamma||_1$, c.f. the definition in Section~\ref{sec:EC}.\footnote{This dependence on $||\rho^\Gamma||_1$ can be removed by putting in the maximum value of it from \eqref{eq:enlmax}, which in general isn't ideal.} It can be found\footnote{We will defer the justification of this equality until Section~\ref{sec:nge}.}
\begin{equation}
1+2\frac{||{S^\Gamma_-}^\dagger(I)||_1}{d_A d_B}=\frac{||\mathscr{T}(S^\Gamma)||_1}{d_A d_B},\label{eq:loe}
\end{equation}
so the lower bounds for $\mathrm{EC}_L(S)$ in part 1 and 2 of Proposition~\ref{pro} aren't that different as they seem; besides, in part 2 the summation of the upper bounds can be performed outside the norm by the triangle inequality. 
Since $||{S_i^\Gamma}_-^\dagger(I)||_1/(d_A d_B)$ is the average of the eigenvalues of ${S^\Gamma_i}_-^\dagger(I),$ all the upper bounds are never smaller than the lower ones. 

How entangling an operation can be in terms of negativities is associated with the norms of ${S^\Gamma_i}_-^\dagger(I)$ or ${S^\Gamma_-}^\dagger(I)$. They vanish if and only if the (sub-)operations are PPT, corresponding to the fact that an PPT operation on average can not increase the negativities of any state \cite{Vidal02,Plenio05}. With a small or zero entangling capacity with respect to negativities, even if the operation is entangling, it produces mostly or entirely bound entanglement \cite{Horodecki98,Cirac01}.

For a deterministic operation $S=\sum_i S_i$, the upper bounds from part 2 is a special case of part 1, by choosing $\widetilde{S^\Gamma}_\pm=\sum_i \widetilde{S_i^\Gamma}_\pm.$ With $S=p S_1+(1-p)S_2$, where both $S_i$ are TP (deterministic) and $0\leq p\leq 1$, if $S_2$ is PPT, Proposition~\ref{pro} suggests that $S$ is at most about $p$ times as entangling as $S_1$ is---Mixing an operation with a PPT one in general makes it less entangling.

From part 3 of Proposition~\ref{pro}, if a sub-operation is PPT, e.g. LOCC \cite{Chitambar14}, whether its negativity can increase depends on $||{S_i^\Gamma}^\dagger(I)||/p_i.$ If the state is initially PPT, $S_i$ must be non-PPT for any negativity to be produced, and the amount is bounded by $||{S^\Gamma_i}^\dagger_-(I)||/p_i$; in other words,  no entanglement can be distilled out of a PPT state after PPT (sub-)operations \cite{Horodecki98,Cirac01}.

\subsection{Proof of the Proposition}
\subsubsection{Upper Bounds}
In the proof below we will not explicitly assume the existence of an ancilla: The reader can think of $S$ or $S_i$ as $\mathcal{I}_a\otimes S$ or $\mathcal{I}_a\otimes S_i$. This doesn't influence the result because as stated in Proposition~\ref{pro} the upper bounds remain the same whether there's an ancilla or not by Corollary~\ref{cor:2lm1}.
\subsubsection*{Deterministic operation}
An operation $S$ is CP, so it's HP \cite{Kraus71,Choi75,Stormer}. Because $\Gamma$ is HPTP, $S^\Gamma$ is HPTP. Hence we can apply Lemma~\ref{lem:lh} to $S^\Gamma(\rho^\Gamma)=S(\rho)^\Gamma$:
\begin{align}
E_N(S(\rho))-E_N(\rho)&=\trace S^\Gamma(\rho^\Gamma)^--\trace{\rho^\Gamma}^-\nonumber\\
&\leq ||\widetilde{S^\Gamma}^\dagger_-(I)||\,||\rho^\Gamma||_1,
\end{align}
and
\begin{align}
E_L(S(\rho))-E_L(\rho)&=\log||S^\Gamma(\rho^\Gamma)||_1-\log||\rho^\Gamma||_1\nonumber\\
&=\log(1+\frac{||S^\Gamma(\rho^\Gamma)||_1-||\rho^\Gamma||_1}{||\rho^\Gamma||_1})\nonumber\\
&\leq \log(1+2||\widetilde{S^\Gamma}^\dagger_-(I)||).
\end{align}
The upper bounds from part 1 of Proposition~\ref{pro} are proved. 

For $L=\sum_i L_i$, where $L_i$ are HP, we can choose $\widetilde{L}_\pm=\sum_i \widetilde{L_i}_\pm$, and\footnote{This is conjugate linearity of of the adjoint mapping (Section~\ref{sec:adl}).}
\begin{equation*}
(\sum_i \widetilde{L_i}_-)^\dagger=\sum_i \widetilde{L_i}_-^\dagger.
\end{equation*}
Similarly, for a deterministic $S=\sum_i S_i$, we have $S^\Gamma=\sum_i S^\Gamma_i$, and if we choose $\widetilde{S^\Gamma}_\pm=\sum_i (\widetilde{S^\Gamma_i}_\pm)$, we then obtain
\begin{equation}
E_N(S(\rho))-E_N(\rho)\leq ||\sum_i\widetilde{S^\Gamma_i}^\dagger_-(I)||\,||\rho^\Gamma||_1,\label{eq:N1}
\end{equation}
and
\begin{equation}
E_L(S(\rho))-E_L(\rho)\leq \log(1+2||\sum_i\widetilde{S^\Gamma_i}^\dagger_-(I)||).
\end{equation}
Hence we've demonstrated part 2 of Proposition~\ref{pro} can be applied to deterministic $S=\sum_i S_i$, and it's an application of part 1.

Let's emphasize again $\widetilde{S^\Gamma}_\pm=\sum_i (\widetilde{S^\Gamma_i}_\pm)$ are some of the possible decompositions of $S^\Gamma,$ and they may or may not be the ideal choices that yield the smallest bounds. 
\subsubsection*{Probabilistic operation}
Let's prove the upper bounds from part 2 of Proposition~\ref{pro} for probabilistic operations. The average negativity after $S$ is
\begin{equation}
\sum_i p_i {E_N}_i=\sum_i p_i\frac{{\trace S_i(\rho)^\Gamma}^-}{p_i}=\sum_i {\trace S_i(\rho)^\Gamma}^-.
\end{equation}
Furthermore, by Lemma~\ref{lem:lh}
\begin{align}
\sum_i {\trace S_i(\rho)^\Gamma}^\pm&\leq\sum_i[\trace\widetilde{{S^\Gamma_i}}_+({\rho^\Gamma}^\pm)+\trace \widetilde{{S^\Gamma_i}}_-({\rho^\Gamma}^\mp)]\nonumber\\
&=\trace\left(\sum_i \widetilde{{S^\Gamma_i}}_+\right)({\rho^\Gamma}^\pm)+\trace \left(\sum_i \widetilde{{S^\Gamma_i}}_-\right)({\rho^\Gamma}^\mp).\label{eq:Sga}
\end{align}
Essentially, we now have $S^\Gamma=\sum_i S_i^\Gamma=\sum_i \widetilde{S_i^\Gamma}_+-\sum_i\widetilde{S_i^\Gamma}_-,$ with $\widetilde{S^\Gamma}_\pm=\sum_i \widetilde{S_i^\Gamma}_\pm,$ and $S$ is TP as in the deterministic case. By Lemma~\ref{lem:lh}, we essentially recover the right sides of \eqref{eq:N1} for the bounds of $\sum_i p_i {E_N}_i-E_N,$ where $E_N$ is the initial negativity.

As to the bounds for $\text{EC}_L$ of a probabilistic operation, with an initial logarithmic negativity $E_L$:
\begin{align}
\sum p_i {E_L}_i-E_L&=\sum_i p_i \log ||\rho_i^\Gamma||_1-\log||\rho^\Gamma||_1\nonumber\\
&=\sum_i p_i \log \frac{||S_i^\Gamma(\rho^\Gamma)||_1}{p_i}-\log||\rho^\Gamma||_1 \nonumber\\
&\leq \log \sum_i||S_i^\Gamma(\rho^\Gamma)||_1-\log||\rho^\Gamma||_1\nonumber\\
&=\log \left(1+\frac{\sum_i||S_i^\Gamma(\rho^\Gamma)||_1-||\rho^\Gamma||_1}{||\rho^\Gamma||_1}\right),\label{eq:0}
\end{align}
by the concavity of logarithm and Jensen's inequality \cite{Plenio05,Cover,PapaRudin}. By \eqref{eq:H+H-},
\begin{align}
\sum_i||S_i^\Gamma(\rho^\Gamma)||_1&=\sum_i \left[\trace S_i^\Gamma(\rho^\Gamma)+2\trace S_i^\Gamma(\rho^\Gamma)^-\right]\nonumber\\
&=\trace\sum_i S^\Gamma_i(\rho^\Gamma)+2\sum_i\trace S_i^\Gamma(\rho^\Gamma)^-\nonumber\\
&=\trace \rho^\Gamma+2\sum_i\trace S_i^\Gamma(\rho^\Gamma)^-,\label{eq:1}
\end{align}
because $S^\Gamma=\sum_i S^\Gamma_i$ is TP. Now we're back to \eqref{eq:Sga}. From \eqref{eq:H+H-}, $||\rho^\Gamma||_1=\trace \rho^\Gamma +2\trace {\rho^\Gamma}^-=1+2\trace {\rho^\Gamma}^-$; following the same steps as before, we can obtain the desired result. 

\subsubsection*{Sub-operation}
Let's move on to upper bounds from part 3 of Proposition~\ref{pro}. By \eqref{eq:P1P2} and Lemma~\ref{lem:lh},
\begin{align*}
p_i {E_N}_i=\trace S^\Gamma_i(\rho^\Gamma)^-&\leq ||\widetilde{S^\Gamma_i}_+^\dagger(I)||\;||{\rho^\Gamma}^-||_1+||\widetilde{S^\Gamma_i}_-^\dagger(I)||\;||{\rho^\Gamma}^+||_1\\
&=||\widetilde{S^\Gamma_i}_+^\dagger(I)||\;||{\rho^\Gamma}^-||_1+||\widetilde{S^\Gamma_i}_-^\dagger(I)||\left(1+||{\rho^\Gamma}^-||_1\right)\\
&=E_N\left(||\widetilde{S^\Gamma_i}_+^\dagger(I)||+||\widetilde{S^\Gamma_i}_-^\dagger(I)||\right)+||\widetilde{S^\Gamma_i}_-^\dagger(I)||.
\end{align*}

\subsubsection{Lower bounds}
\label{sec:lo}
\subsubsection*{Deterministic operation}

Employing the method proposed by Campbell \cite{Campbell10}, to obtain the lower bounds for entangling capacities, an ancilla is used. As a reminder, the operation $S$ is a mapping from $\mathcal{B}(\mathcal{H}_1^A\otimes \mathcal{H}_1^B)$ to $\mathcal{B}(\mathcal{H}_2^A\otimes \mathcal{H}_2^B)$, and the ancilla is the space $\mathcal{H}_a=\mathcal{H}_a^A\otimes \mathcal{H}_a^B$. 

Here we further require $\mathcal{H}_a^A=\mathcal{H}_1^A$ and $\mathcal{H}_a^B=\mathcal{H}_1^B$; in other words we demand them to be isomorphic spaces. Let $\{a_i\}$ be an orthonormal basis of $\mathcal{H}_a^A$ and $\mathcal{H}_1^A$, and $\{b_i\}$ be an orthonormal basis of $\mathcal{H}_a^B$ and $\mathcal{H}_1^B$, and $d_i:=\mathrm{dim}\mathcal{H}_1^i$. Define a state $\ket{\Psi}$ as:
\begin{equation}
\ket{\Psi}:=\ket{\Psi_{A}}\otimes \ket{\Psi_{B}},\label{eq:Psi}
\end{equation}
where
\begin{subequations}
\begin{align}
\ket{\Psi_{A}}&:=\frac{1}{\sqrt{d_A}}\ket{a_i}\ket{a_i}\\
\ket{\Psi_{B}}&:=\frac{1}{\sqrt{d_A}}\ket{b_i}\ket{b_i}
\end{align}
\end{subequations}
$\ket{\Psi_{A}}$ and $\ket{\Psi_{B}}$ are therefore maximally entangled states. We can also find
\begin{align}
\ket{\Psi}\bra{\Psi}=&\frac{1}{d_A d_B}\sum_{i,j,k,l} \ket{a_i}\bra{a_j}\otimes\ket{b_k}\bra{b_l}\otimes \ket{a_i}\bra{a_j}\otimes\ket{b_k}\bra{b_l}\nonumber\\
=&\frac{1}{d_A d_B}\sum_{i,j,k,l}E_{ij}\otimes F_{kl}\otimes E_{ij}\otimes F_{kl}.\label{eq:psiAB}
\end{align}
Therefore
\begin{equation}
(\mathcal{I}_a\otimes S)(\ket{\Psi}\bra{\Psi})=\frac{1}{d_A d_B}\sum_{i,j,k,l}E_{ij}\otimes F_{kl}\otimes S\left(E_{ij}\otimes F_{kl}\right)=\frac{1}{d_A d_B}\mathscr{T}(S)
\end{equation}
is essentially a normalized Choi isomorphism. From now on we will ignore the subscript of $\mathcal{I}_a$. As $\ket{\Psi}\bra{\Psi}$ isn't entangled, the entanglement measure of $\mathscr{T}(S)/(d_A d_B)$ can give us a lower bound for entangling capacity. 

With $\Gamma=\mathrm{T}_{A}\otimes\mathcal{I}_B$, where $\mathrm{T}_{A}$ operates on $\mathcal{B}(\mathcal{H}_a^A\otimes \mathcal{H}_1^A)$ and $\mathcal{I}_B$ is the identity mapping on $\mathcal{B}(\mathcal{H}_a^B\otimes \mathcal{H}_1^B)$, we have
\begin{align}
\mathscr{T}(S)^\Gamma&=\Big(\sum_{i,j,k,l}E_{ij}\otimes F_{kl}\otimes S(E_{ij}\otimes F_{kl})\Big)^\Gamma\nonumber\\
&=\Big(\sum_{i,j,k,l}E_{ji}\otimes F_{kl}\otimes S^\Gamma(E_{ji}\otimes F_{kl})\Big)\nonumber\\
&=\mathscr{T}(S^\Gamma),\label{eq:Sgamma}
\end{align}
because $S^\Gamma(O)=S(O^\Gamma)^\Gamma$. Therefore
\begin{equation}
\left[\mathscr{T}(S)^\Gamma\right]^\pm=\mathscr{T}(S^\Gamma)^\pm=\mathscr{T}(S^\Gamma_\pm).
\end{equation}
By Corollary~\ref{cor:1}, $\trace\mathscr{T}(S^\Gamma_\pm)=\trace{S^\Gamma_\pm}^\dagger(I)$, so
\begin{equation}
\trace\left\{\left[\mathcal{I}\otimes S(\ket{\Psi}\bra{\Psi})\right]^\Gamma\right\}^\pm=\frac{1}{d_A d_B}\trace\mathscr{T}(S^\Gamma_\pm)=\frac{1}{d_A d_B}\trace{S^\Gamma_\pm}^\dagger(I).\label{eq:spm}
\end{equation}
Alternatively, we can directly calculate $\trace\mathscr{T}(S^\Gamma_\pm)$ without using Corollary~\ref{cor:1}:
\begin{align}
\trace\mathscr{T}(S^\Gamma_\pm)&=\trace\sum_{i,j,k,l}E_{ji}\otimes F_{kl}\otimes S^\Gamma_\pm(E_{ij}\otimes F_{kl})\nonumber\\
&=\sum_{i,j,k,l}\delta_{ij}\delta_{kl}\trace S^\Gamma_\pm(E_{ij}\otimes F_{kl})\nonumber\\
&=\trace S^\Gamma_\pm\left(\sum_{i,k}E_{ii}\otimes F_{kk}\right)\nonumber\\
&=\trace {S^\Gamma_\pm}^\dagger(I),\label{eq:ts}
\end{align}
because $\trace L(I)=(L^\dagger(I)|I)=\trace (L^\dagger(I))^\dagger$, and $L^\dagger(I)$ is Hermitian as $L$ is HP, \eqref{eq:Kr}. Hence
\begin{equation*}
E_N(\mathcal{I}\otimes S(\ket{\Psi}\bra{\Psi}))=\frac{1}{d_A d_B}\trace\mathscr{T}(S^\Gamma_-)=\frac{1}{d_A d_B}\trace {S^\Gamma_-}^\dagger(I).
\end{equation*}

We can use \eqref{eq:H+H-} to acquire the bound for logarithmic negativity. Or, because $\mathscr{T}(S^\Gamma_\pm)$ are orthogonal, 
\begin{equation}
||\left[\mathcal{I}\otimes S(\ket{\Psi}\bra{\Psi})\right]^\Gamma||_1=\frac{1}{d_A d_B}||\mathscr{T}(S^\Gamma)||_1=\frac{1}{d_A d_B}\left[\trace\mathscr{T}(S^\Gamma_+)+\trace\mathscr{T}(S^\Gamma_-)\right].\label{eq:la}
\end{equation}
By Corollary~\ref{cor:1},
\begin{equation}
||\mathscr{T}(S^\Gamma)||_1= \trace {S^\Gamma_+}^\dagger(I)+\trace {S^\Gamma_-}^\dagger(I).\label{eq:s1}
\end{equation}
Since ${S^\Gamma_+}^\dagger(I)-{S^\Gamma_-}^\dagger(I)=I,$
\begin{equation}
\trace {S^\Gamma_+}^\dagger(I)+\trace {S^\Gamma_-}^\dagger(I)=\trace I+2\trace {S^\Gamma_-}^\dagger(I)=d_A d_B+2\trace {S^\Gamma_-}^\dagger(I),
\end{equation}
and the lower bound for logarithmic negativity can be obtained. We can also use \eqref{eq:l1} to show this relation.

\subsubsection*{Sub-operation}

To obtain the lower bounds for a sub-operation the procedure is pretty much the same, but now we should take the probability into account:
\begin{equation}
p_i=\trace \mathcal{I}\otimes S_i(\ket{\Psi}\bra{\Psi})=\frac{\trace \mathscr{T}(S_i)}{d_A d_B}=\frac{\trace S_i^\dagger(I)}{d_A d_B},\label{eq:pi}
\end{equation}
c.f. \eqref{eq:spm} or \eqref{eq:ts}. Use the same method to obtain $\trace \mathcal{I}\otimes{S^\Gamma_i}_-(\ket{\Psi}\bra{\Psi}).$ Thus the negativity is
\begin{equation}
{E_N}_i(\ket{\Psi}\bra{\Psi})=\frac{\trace {S_i^\Gamma}_-^\dagger(I)}{\trace S_i^\dagger(I)}.\label{eq:ns}
\end{equation}

Similarly,
\begin{equation}
||\mathcal{I}\otimes S_i(\ket{\Psi})\bra{\Psi}||_1=\frac{||\mathscr{T}(S^\Gamma_i)||_1}{d_A d_B}=\frac{\trace {S_i^\Gamma}_+^\dagger(I)+\trace {S_i^\Gamma}_-^\dagger(I)}{d_A d_B},
\end{equation}
c.f. \eqref{eq:s1}, so
\begin{equation}
{E_L}_i(\ket{\Psi}\bra{\Psi})=\log\frac{||\mathscr{T}(S^\Gamma_i)||_1}{\trace S_i^\dagger(I)}=\log\frac{\trace {S_i^\Gamma}_+^\dagger(I)+\trace {S_i^\Gamma}_-^\dagger(I)}{\trace S_i^\dagger(I)}.\label{eq:ls}
\end{equation}
The lower bounds for deterministic operations can be regarded as a special case of sub-operations---this can be easily shown with Corollary~\ref{cor:1} and \eqref{eq:d}.

Because transposition and partial transposition are TP, $\trace \mathscr{T}(L)=\trace \mathscr{T}(L^\Gamma)$. For an HP $L$, it means
\begin{equation*}
\trace \mathscr{T}(L)=\trace \mathscr{T}(L^\Gamma)=\trace \mathscr{T}(L^\Gamma_+)-\trace \mathscr{T}(L^\Gamma_-).
\end{equation*}
Therefore by Corollary~\ref{cor:1}
\begin{equation}
\trace {S_i}^\dagger(I)=\trace {S_i^\Gamma}^\dagger(I)=\trace {S_i^\Gamma}_+^\dagger(I)-\trace {S_i^\Gamma}_-^\dagger(I).\label{eq:s}
\end{equation}
With \eqref{eq:s} we can adjust \eqref{eq:ns} and \eqref{eq:ls} to our liking, such as in Proposition~\ref{pro}. Note in general we can't expect ${S_i^\Gamma}^\dagger(I)\geq 0,$ so its trace may not equal its trace norm.
\subsubsection*{Probabilistic operation}

For the average negativity of a probabilistic operation, by \eqref{eq:pi} and \eqref{eq:ns}
\begin{equation}
\sum_i p_i E_N(\mathcal{I}\otimes S_i(\ket{\Psi}\bra{\Psi}))=\sum_i\frac{ \trace {S_i^\Gamma}_-^\dagger(I)}{d_A d_B}.
\end{equation}
For the average logarithmic negativity, by \eqref{eq:pi} and \eqref{eq:ls}
\begin{equation}
\sum_i p_i E_L(\mathcal{I}\otimes S_i(\ket{\Psi}\bra{\Psi}))=\sum_{i}\frac{\trace \mathscr{T}(S_i)}{d_A d_B}\log\frac{||\mathscr{T}(S^\Gamma_i)||_1}{\trace \mathscr{T}(S_i)},
\end{equation}
Note for the upper bounds the denominator $p_i$ could be removed due to concavity (see \eqref{eq:0}), which can't be applied here. Also note by Corollary~\ref{cor:1} and \eqref{eq:s} there are several expressions for these bounds, just like sub-operations.

\section{Geometrical Interpretation of Bounds on Entangling Capacity and Entanglement}
\label{sec:geo}
In this section, we only discuss deterministic operations.
\subsection{Norms and Metrics Induced by Partial Transposition}\label{sec:d1g}
Let $\mathcal{V}$ and $\mathcal{W}$ be two vector spaces. If there's an injective linear mapping $T:V\rightarrow W$ and a norm (Definition~\ref{def:norm}) $f$ on $W$, then $f\circ T$ is a norm on $V$ \cite{Loomis}. Because partial transposition is bijective both on operators and on linear mappings from operators to operators (Section~\ref{sec:TPT}),
\begin{defi}\label{def:normga}
With the Schatten $p$-norm $||\cdot||_p$ (Section~\ref{sec:normop}), we define such norms for any operator $O\in\mathcal{B}(\mathcal{H}_A\otimes \mathcal{H}_B)$:
\begin{equation*}
||O||_{p,\Gamma}:=||O^\Gamma||_p,
\end{equation*}
and any linear mapping $L$ from $\mathcal{B}(\mathcal{H}^A_1\otimes \mathcal{H}^B_1)$ to $\mathcal{B}(\mathcal{H}^A_2\otimes \mathcal{H}^B_2)$:
\begin{equation*}
||L||_{p,\Gamma}:=||L^\Gamma||_n,
\end{equation*}
where\footnote{This is a norm because Choi isomorphism is one-to-one; moreover being isomorphic the two spaces can be regarded as identical.}
\begin{equation*}
||L||_p:=||\mathscr{T}(L)||_n.
\end{equation*}
These lead to the following metric/distance for operators or linear mappings 
\begin{equation*}
D_{1,\Gamma}(X,Y):=||X-Y||_{1,\Gamma},
\end{equation*}
where $X$ and $Y$ are either operators or linear mappings.
\end{defi}

\subsubsection{Norm, negativity and entangling capacity}\label{sec:nge}

Because $L^\Gamma$ is TP for any (bipartite) TP mapping $L$ (Section~\ref{sec:TPT}), for any operation $S:\mathcal{B}(\mathcal{H}_1^A\otimes \mathcal{H}_1^B)\rightarrow \mathcal{B}(\mathcal{H}_2^A\otimes \mathcal{H}_2^B)$ by \eqref{eq:d} and \eqref{eq:H+H-}, 
\begin{equation}
||S||_{1,\Gamma}=||\mathscr{T}(S^\Gamma)||_1=\trace \mathscr{T}(S^\Gamma)+2\trace \mathscr{T}(S^\Gamma_-)=d_A d_B+2\trace \mathscr{T}(S^\Gamma_-),
\end{equation}
where again $d_A:=\mathrm{dim}\mathcal{H}_1^A$ and $d_B:=\mathrm{dim}\mathcal{H}_1^B$. By Corollary~\ref{cor:1}, $\trace{S^\Gamma_-}^\dagger(I)=\trace \mathscr{T}(S^\Gamma_-)$, so
\begin{equation}
||S||_{1,\Gamma}=d_A d_B+2||{S^\Gamma_-}^\dagger(I)||_1.\label{eq:g1g}
\end{equation}
If $S$ is a PPT operation
\begin{equation}
||S||_{1,\Gamma}=d_A d_B.\label{eq:g2g}
\end{equation}

By the inequality $||O||\leq||O||_1$ \cite{Lidar08,Rastegin12}, $||{S^\Gamma_-}^\dagger(I)||\leq ||{S^\Gamma_-}^\dagger(I)||_1=\trace{S^\Gamma_-}^\dagger(I);$ . From \eqref{eq:g1g} $2 ||{S^\Gamma_-}^\dagger(I)||_1=||S||_{1,\Gamma}-d_A d_B$, so we obtain:
\begin{equation}
||{S^\Gamma_-}^\dagger(I)||\leq ||{S^\Gamma_-}^\dagger(I)||_1=(||S||_{1,\Gamma}-d_A d_B)/2.\label{eq:upeb}
\end{equation}
The difference between the two sides of the inequality, $||{S^\Gamma_-}^\dagger(I)||_1-||{S^\Gamma_-}^\dagger(I)||=\trace \mathscr{T}(S^\Gamma_-)-||{S^\Gamma_-}^\dagger(I)||$, is the sum of the eigenvalues of ${S^\Gamma_-}^\dagger(I)$ minus the largest one. When the operation is PPT, both sides of the inequality become zero.

$||\rho||_{1,\Gamma}$ decides the negativities of a state (Section~\ref{sec:neg}), and $||S||_{1,\Gamma}$ can be considered the negativity of the operation (if normalized), because $||\mathscr{T}(S^\Gamma)||_1=||(\mathscr{T}(S))^\Gamma||_1$ \cite{Stormer,Zanardi01}. By \eqref{eq:upeb} and \eqref{eq:g1g}, $||S||_{1,\Gamma}$ not only corresponds to the lower bounds of entangling capacity from  Proposition~\ref{pro}, but also bounds the upper bound of entangling capacity. Therefore $||S||_{1,\Gamma}$ and $||\rho||_{1,\Gamma}$ \emph{quantify/estimate the entangling capacity and entanglement of an operation, and entanglement of a state}.

Even though the norm $||\cdots||_{1,\Gamma}$ is more suitable for deterministic operations, we can normalize a sub-operation by dividing it by $||\mathscr{T}(S_i)||_1/(d_A d_B)$, and thus reconcile them.

\subsubsection{Distance, negativity and entangling capacity}
From \eqref{eq:g2g}, geometrically PPT operations are a subset of a $d_A d_B$-sphere with center at the origin with respect to $||\cdots||_{1,\Gamma}$.  By the triangle inequality (Section~\ref{sec:normmet}), $|\; ||a||-||b|| \;|\leq ||a-b||$, 
so the length of an operation is bounded by the distance to another operation, and any operations within an open ball with center being a non-PPT operation $S$ and radius $||S||_{1,\Gamma}-d_A d_B$ are non-PPT.  The distance from a non-PPT operation $S$ to a PPT one is at least $||S||_{1,\Gamma}-d_A d_B=2||\mathscr{T}(S^\Gamma)^-||_1$\footnote{As $\trace\mathscr{T}({S^\Gamma})=d_A d_B$, by \eqref{eq:H+H-}, we can obtain this relation.}---It may not be the exact distance, as linear mappings $||S||_{1,\Gamma}-d_A d_B$ away from $S$ are not necessarily quantum operations.  All these can be applied to density operators; see Figure~\ref{fig:norm}. This is somewhat like distance-based entanglement measures \cite{Vedral97}.

\begin{figure}[hbtp!]
\begin{subfigure}[b]{0.5\linewidth}
\includegraphics[width=\linewidth]{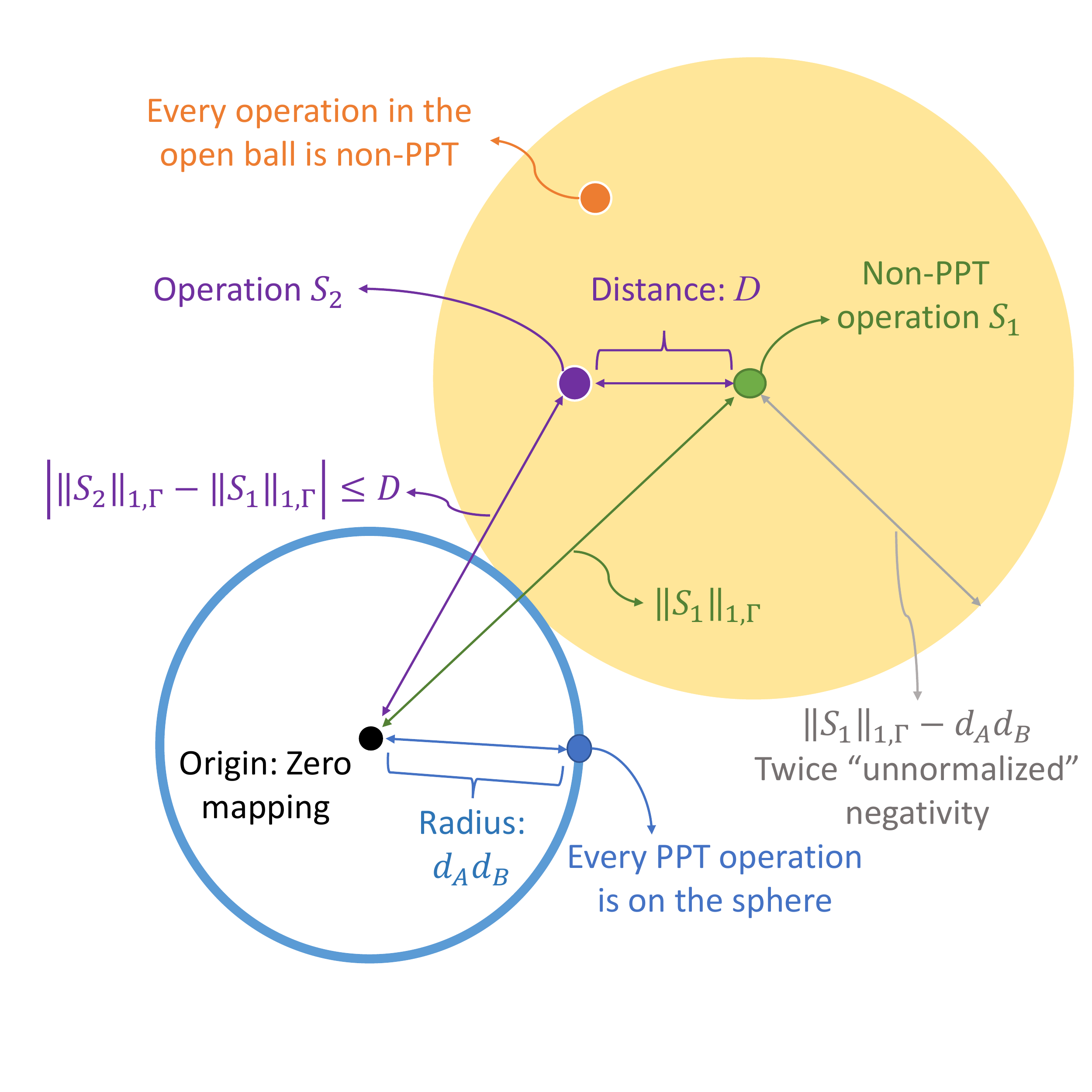}
\caption{PPT and non-PPT operations in the space of  linear mappings.}
\end{subfigure}
\begin{subfigure}[b]{0.5\linewidth}
\includegraphics[width=\linewidth]{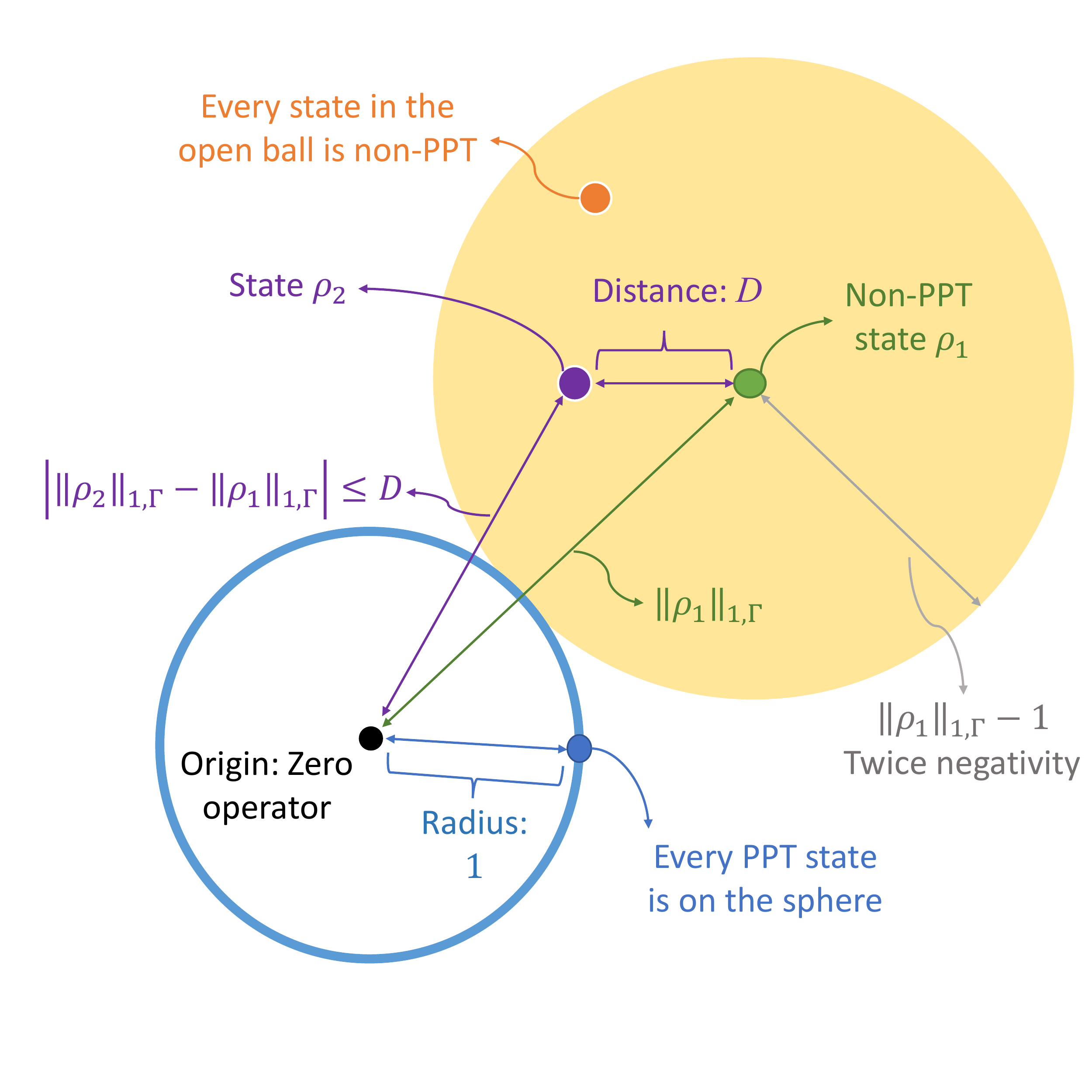}
\caption{PPT and non-PPT states in the space of  operators.}
\end{subfigure}
\caption[Geometry of operations and states]{Geometry of operations and states in the spaces of linear mappings and  operators, with respect to $||\cdots||_{1,\Gamma}$. The entanglement/entangling capacity in terms of negativity is related to the length, and nearby operations/states are similarly entangling/entangled.}
\label{fig:norm}
\end{figure}

\subsubsection{Equivalence of norms}\label{sec:eq}
For any Hermitian operator $H$ on $\mathcal{H}_A\otimes \mathcal{H}_B$ and any HP mapping $L:\mathcal{B}(\mathcal{H}_1^A\otimes \mathcal{H}_1^B)\rightarrow \mathcal{B}(\mathcal{H}_2^A\otimes \mathcal{H}_2^B)$:
\begin{equation*}
\frac{||H||_1}{\min(\mathrm{dim}\mathcal{H}_A,\mathrm{dim}\mathcal{H}_B)}\leq ||H||_{1,\Gamma}\leq \min(\mathrm{dim}\mathcal{H}_A,\mathrm{dim}\mathcal{H}_B)||H||_1.
\end{equation*}
and
\begin{align}
\frac{||L||_1}{\min(d_1^A d_2^A,d_1^B d_2^B)}\leq ||L||_{1,\Gamma}
\leq \min(d_1^A d_2^A,d_1^B d_2^B)||L||_1,
\end{align}
where $d_i^j:=\mathrm{dim}\mathcal{H}_i^j$. Therefore two states or operations that are close in one norm should not be far apart in another, and vice versa. Using whichever norm does not change the topology, by continuity (or linearity) of partial transposition or equivalence of norms \cite{Loomis}. Please refer to Appendix~\ref{app:en} for more details.

\subsection{Distance and Difference in Negativity after Operations}
There is more to the distance from Definition~\ref{def:normga}:\\
\begin{pro}\label{pro:2}
For any density operator $\rho$ and $\rho_i$ on $\mathcal{H}_1^A\otimes \mathcal{H}_1^B$ and deterministic operation $S$, $S_1$ and $S_2$ from $\mathcal{B}(\mathcal{H}_1^A\otimes \mathcal{H}_1^B)$ to $\mathcal{B}(\mathcal{H}_2^A\otimes \mathcal{H}_2^B)$:
\begin{align}
D_{1,\Gamma}\left(S_1(\rho),S_2(\rho)\right)&\leq 2||(S_2^\Gamma-S_1^\Gamma)_\pm^\dagger(I)||\, ||\rho||_{1,\Gamma}\label{eq:pro211}\\
&\leq D_{1,\Gamma}(S_1,S_2) ||\rho||_{1,\Gamma}.\label{eq:pro212}
\end{align}
and
\begin{equation}
D_{1,\Gamma}\left(S(\rho_1),S(\rho_2)\right)\leq \left(1+2||{S^\Gamma_-}^\dagger(I)||\right) D_{1,\Gamma}(\rho_1,\rho_2).\label{eq:pro22}
\end{equation}
Also $(S_2^\Gamma-S_1^\Gamma)_+^\dagger(I)=(S_2^\Gamma-S_1^\Gamma)_-^\dagger(I)$.

\end{pro}
Note $||\rho||_{1,\Gamma}=||\rho^\Gamma||_1$, and $D_{1,\Gamma}\left(X,Y\right)=||X^\Gamma-Y^\Gamma||_1$, so the inequalities from Proposition~\ref{pro:2} bound the trace distances between $S_1(\rho)^\Gamma$ and $S_2(\rho)^\Gamma$ and between $S(\rho_1)^\Gamma$ and $S(\rho_2)^\Gamma$. Because Proposition~\ref{pro} and \ref{pro:2} are proved in similar ways, they look alike. By the triangle inequality, this proposition provides bounds for differences in negativity between different operations on the same state, and between the same operation on different states, in terms of the metric $D_{1,\Gamma}$; see Figure~\ref{fig:pro2}. 

The inequalities from Proposition~\ref{pro:2} also exhibit the continuity of linear mappings, and $||\cdots||_{1,\Gamma}$ as a legitimate norm or $D_{1,\Gamma}$ as a legitimate metric---Since a linear mapping is continuous when only finite-dimensional spaces are considered (c.f. Theorem~\ref{thm:boco}), there should exist relations as in Proposition~\ref{pro:2} for $||\cdots||_{1,\Gamma}$ or $D_{1,\Gamma}$.

\begin{figure}[hbtp!]
\includegraphics[width=0.5\linewidth]{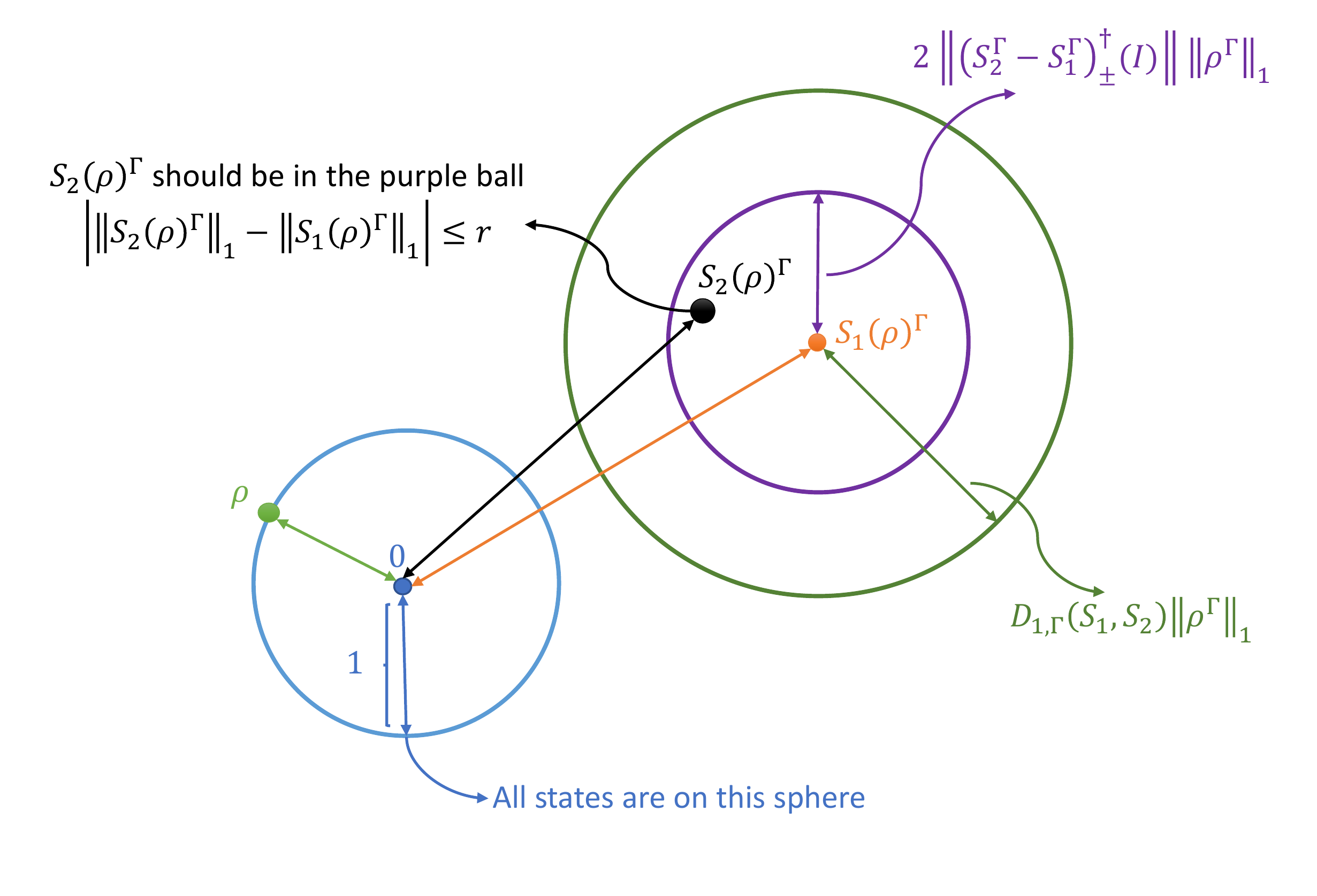}
\includegraphics[width=0.5\linewidth]{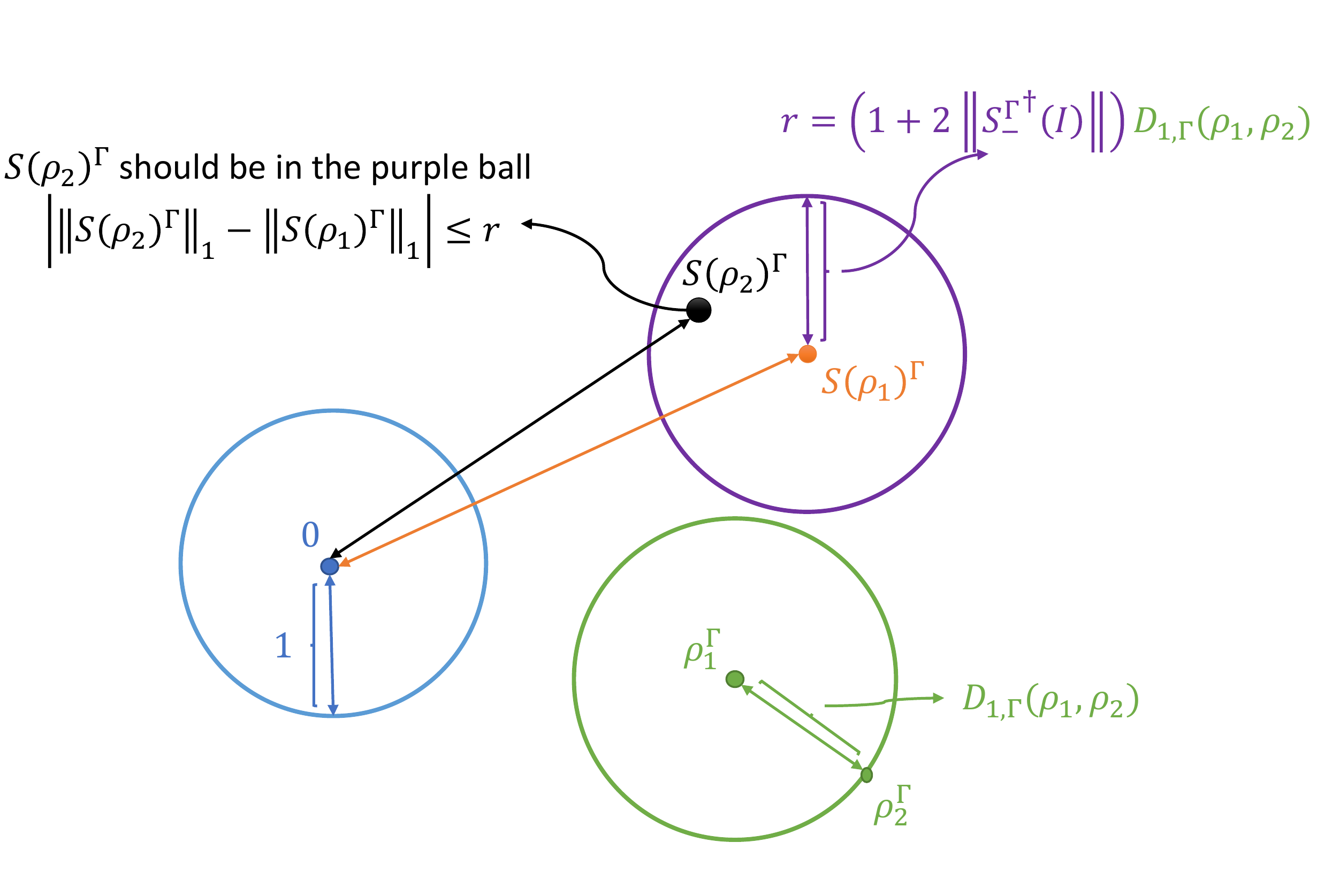}
\caption[Proposition~\ref{pro:2} and its implication]{Proposition~\ref{pro:2} and its implication in the space of operators with respect to trace norm $||\cdots||_1$; similar diagrams can be drawn in terms of $||\cdots||_{1,\Gamma}$. By the triangle inequality, $||S_2(\rho)^\Gamma||_1$ and $||S(\rho_2)^\Gamma||_1$ are bounded, so are the negativities of $S_2(\rho)$ and $S(\rho_2)$. This demonstrates the physical importance of the norm $||\cdots||_{1,\Gamma}$.}
\label{fig:pro2}
\end{figure}

\begin{proof}
Because $S_1$ and $S_2$ are deterministic, ${S_1^\Gamma}^\dagger(I)={S_2^\Gamma}^\dagger(I)=I$, and
\begin{equation*}
({S_2^\Gamma-S_1^\Gamma})^\dagger(I)=0.
\end{equation*}
Hence we can phrase our problem this way: If an HP mapping $L$ has $L^\dagger(I)=0$, given a Hermitian operator $H$ with a fixed trace norm, how large can $||L(H)||_1$ be? Let's begin with
\begin{equation}
L^\dagger(I)=L_+^\dagger(I)-L_-^\dagger(I)=0\Rightarrow L_+^\dagger(I)=L_-^\dagger(I),\label{eq:eqq}
\end{equation}
confirming $(S_2^\Gamma-S_1^\Gamma)_+^\dagger(I)=(S_2^\Gamma-S_1^\Gamma)_-^\dagger(I)$. Then by Lemma~\ref{lem:lh}, 
\begin{equation*}
||L(H)||_1\leq 2||L_\pm^\dagger(I)|| \,||H||_1.\label{eq:srho2}
\end{equation*}
\eqref{eq:pro211} is proved. To show \eqref{eq:pro212}, we need the inequality $||O||\leq||O||_1$ for any (trace-class) operator $O\in\mathcal{B}(\mathcal{H})$ \cite{Lidar08,Rastegin12,Blackadar}, so $||(S_2^\Gamma-S_1^\Gamma)_\pm^\dagger(I)||\leq ||(S_2^\Gamma-S_1^\Gamma)_\pm^\dagger(I)||_1$. By Corollary~\ref{cor:1}, $\trace (S_2^\Gamma-S_1^\Gamma)_\pm^\dagger(I)=\trace \mathscr{T}(S_2^\Gamma-S_1^\Gamma)^\pm$, so by \eqref{eq:eqq}
\begin{equation*}
2\trace \mathscr{T}(S_2^\Gamma-S_1^\Gamma)^\pm=\trace \mathscr{T}(S_2^\Gamma-S_1^\Gamma)^++\trace \mathscr{T}(S_2^\Gamma-S_1^\Gamma)^-=||\mathscr{T}(S_2^\Gamma-S_1^\Gamma)||_1=||S_2-S_1||_{1,\Gamma}.
\end{equation*}

To show \eqref{eq:pro22} is pretty similar, as
\begin{equation*}
S(\rho_1)^\Gamma-S(\rho_2)^\Gamma=S^\Gamma[(\rho_1-\rho_2)^\Gamma].
\end{equation*}
By Lemma~\ref{lem:lh} and $S^\Gamma$ being TP we can easily find
\begin{equation*}
||S^\Gamma((\rho_1-\rho_2)^\Gamma)||_1\leq\left(1+2||{S^\Gamma_-}^\dagger(I)||\right)||\rho_1-\rho_2||_{1,\Gamma}.
\end{equation*}
\end{proof}

\section{Upper Bounds Using Schatten $p$-norm}
\label{sec:pnorm}
Up until now when I utilized Lemma~\ref{lem:lh} to find the upper bounds, I chose $p=\infty$ and $q=1$. If using the most general form of this lemma, the upper bound for negativity from part 1 of Proposition~\ref{pro} becomes
\begin{equation*}
EC_N(S)\leq ||\widetilde{S^\Gamma}^\dagger_-(I)||_p||\rho^\Gamma||_q,
\end{equation*}
and the upper bound for logarithmic negativity will be state-dependent when $q\neq 1.$ The same can be done for the rest of Proposition~\ref{pro} and \ref{pro:2} by following the proofs. However since density operators have unit trace, choosing the trace norm for states and therefore the operator norm for operations seems most natural.

\section{Comparison between Approaches}

\label{sec:up}
From part 2 of Proposition~\ref{pro}, for a deterministic operation $S=\sum_i S_i$, $\text{EC}_L(S)\leq \log(1+2||\sum_i{S^\Gamma_i}^\dagger_-(I)||)$;\footnote{This can be obtained from part 1 of Proposition~\ref{pro} as well; please refer to Section~\ref{sec:dispro1}.} be aware that here we choose ${S^\Gamma_i}_-$ instead of $\widetilde{S^\Gamma_i}_-$. In Ref.~\cite{Campbell10}, it was shown for a deterministic operation $S=\sum_i S_i$ where $S_i(\rho)=V_i \rho V_i^\dagger$ and $V_i$ have Schmidt decompositions \cite{Nielsen03}
\begin{equation}
V_i=\sum_{j(i)} \lambda_{ij} A_{ij}\otimes B_{ij},
\end{equation}
the entangling capacity has an upper bound: 
\begin{equation*}
\text{EC}_L(S)\leq \log\big(\sum_i ||\sum_{j(i)} \lambda_{ij} A_{ij}^\dagger A_{ij}||\,||\sum_{k(i)} \lambda_{ik} B_{ik}^\dagger B_{ik}||\big).
\end{equation*}
Furthermore, it was shown by Campbell for a unitary operator $U=\sum_i \lambda_i A_i\otimes B_i$ the lower bound is \cite{Campbell10}
\begin{equation}
\log \frac{(\sum_i\lambda_i)^2}{d_A d_B},\label{eq:loC}
\end{equation}
while our lower bound is 
\begin{equation}
\log\left(1+2\frac{||{S^\Gamma_-}^\dagger(I)||_1}{d_A d_B}\right).\label{eq:lop}
\end{equation}
Are the bounds of his and ours equivalent in these cases? To answer this question, we can prove:
\begin{lem}\label{lem:com}
For deterministic $S=\sum_i S_i$, where $S_i=V_i \rho V_i^\dagger$ and the Schmidt decompositions of $V_i$ are $V_i=\sum_{j(i)} \lambda_{ij} A_{ij}\otimes B_{ij},$
\begin{align}
1+2||\sum_i {S^\Gamma_i}_-^\dagger(I)||&=||\sum_i\sum_{j(i)} \lambda_{ij}{A_{ij}^*}^\dagger A_{ij}^*\otimes \sum_{k(i)} \lambda_{ik} B_{ik}^\dagger B_{ik}||\label{eq:lem61}\\
&\leq\sum_i ||\sum_{j(i)} \lambda_{ij} A_{ij}^\dagger A_{ij}||\,||\sum_{k(i)} \lambda_{ik} B_{ik}^\dagger B_{ik}||.\label{eq:lem62}
\end{align} 

For a unitary operation $S(\rho)=U\rho U^\dagger$, where $U$ has domain $\mathcal{H}_A\otimes \mathcal{H}_B$ and Schmidt decomposition $U=\sum_i \lambda_i A_i\otimes B_i$,
\begin{equation*}
1+2\frac{||{S^\Gamma_-}^\dagger(I)||_1}{d_A d_B}=\frac{(\sum_i\lambda_i)^2}{d_A d_B},
\end{equation*}
with $d_A:=\mathrm{dim}\mathcal{H}_A$ and $d_B:=\mathrm{dim}\mathcal{H}_B$; in addition, \eqref{eq:lem61} and \eqref{eq:lem62} become identical.
\end{lem}

By this lemma, $\text{EC}_L(S)\leq \log(1+2||\sum_i {S^\Gamma_i}_-^\dagger(I)||)\leq \big(\sum_i ||\sum_{j(i)} \lambda_{ij} A_{ij}^\dagger A_{ij}||\,||\sum_{k(i)} \lambda_{ik} B_{ik}^\dagger B_{ik}||\big)$. and the lower bounds for unitary operations are indeed identical.

Please compare the following proof with Appendix E of Ref.~\cite{Campbell10}, as it is an expansion thereof.
\subsection{Proof for the Upper Bounds}
\subsubsection{Single sub-operation}\label{sec:sisuo}
Let's first consider just one sub-operation $S_i:\mathcal{B}(\mathcal{H}_1^A\otimes \mathcal{H}_1^B)\rightarrow \mathcal{B}(\mathcal{H}_2^A\otimes \mathcal{H}_2^B)$, and drop the subscript $i$ for this moment. Suppose $S(O)=V O V^\dagger$ for any $O\in\mathcal{B}(\mathcal{H}_1^A\otimes \mathcal{H}_1^B)$, with $V=\sum_i \lambda_i A_i\otimes B_i$ being the Schmidt decomposition, where $A_i\in\mathcal{B}(\mathcal{H}_1^A,\mathcal{H}_2^A)$ and $B_i\in\mathcal{B}(\mathcal{H}_1^B,\mathcal{H}_2^B)$. Then
\begin{equation}
S^\Gamma(O)=(V O^\Gamma V^\dagger)^\Gamma
=\sum_{i,j} \lambda_{i}\lambda_{j}A_{j}^*\otimes B_{i} O (A_{i}^*\otimes B_{j})^\dagger.\label{eq:sgf}
\end{equation}
Define
\begin{align}
V_{ij}^\pm&:=\left(A_i^*\otimes B_j\pm A_j^*\otimes B_i \right)/\sqrt{2}\text{ for }i\neq j,\,\nonumber\\
V_{ii}^+&:=A_i^*\otimes B_i.\label{eq:VV}
\end{align}
We can see for $i\neq j$
\begin{equation}
A_{j}^*\otimes B_{i} O (A_{i}^* \otimes B_{j})^\dagger+A_{i}^*\otimes B_{j} O (A_{j}^* \otimes B_{i})^\dagger
=V_{ij}^+ O {V_{ij}^+}^\dagger-V_{ij}^- O {V_{ij}^-}^\dagger;
\end{equation}
Note $V_{ij}^\pm=\pm V_{ji}^\pm$ and therefore $V_{ij}^- O V_{ij}^-=V_{ji}^- O V_{ji}^-$. We now have
\begin{equation}
S^\Gamma(O)=\sum_i \lambda_i^2 V_{ii}^+ O {V_{ii}^+}^\dagger+\frac{1}{2}\sum_{i\neq j} \lambda_i \lambda_j V_{ij}^+ O {V_{ij}^+}^\dagger-\frac{1}{2}\sum_{i\neq j} \lambda_i \lambda_j V_{ij}^- O {V_{ij}^-}^\dagger.\label{eq:SG}
\end{equation}
As each summation on the right side of \eqref{eq:SG} is an operator-sum representation, each is a CP mapping, so we have
\begin{align}
\widetilde{S^\Gamma}_+(O)&=\sum_i \lambda_i^2 V_{ii}^+ O {V_{ii}^+}^\dagger+\frac{1}{2}\sum_{i\neq j} \lambda_i \lambda_j V_{ij}^+ O {V_{ij}^+}^\dagger\nonumber\\
\widetilde{S^\Gamma}_-(O)&=\frac{1}{2}\sum_{i\neq j} \lambda_i \lambda_j V_{ij}^- O {V_{ij}^-}^\dagger.\label{eq:SGa}
\end{align}

Now the problems is, is $\widetilde{S^\Gamma}_\pm=S^\Gamma_\pm$? This is equivalent to $\mathrm{ran}\mathscr{T}(\widetilde{S^\Gamma}_+)\perp \mathrm{ran}\mathscr{T}(\widetilde{S^\Gamma}_-)$, by Lemma~\ref{lem:or} or the spectral theorem (Theorem~\ref{thm:spe}). We can find $\{V_{ij}^+\}$ and $\{V_{kl}^-\}$ are mutually orthogonal:
\begin{equation}
(V_{ij}^+|V_{ij}^-)=\left(A_i^*\otimes B_j+ A_j^*\otimes B_j|A_i^*\otimes B_j- A_j^*\otimes B_i\right)/2=0\;\forall i,j;\label{eq:o1}
\end{equation}
the orthogonality of the rest can be shown by the orthonormality of $\{A_i\}$ and $\{B_i\}$. Choose an orthonormal basis $\{\ket{a_i}\}$ for the composite system $\mathcal{H}_1^A\otimes \mathcal{H}_1^B$. Like \eqref{eq:Psi}, define
\begin{equation}
\ket{\Psi}:=\sum_i \ket{a_i}\ket{a_i}.\label{eq:Psid}
\end{equation}
Then the Choi isomorphism of an operation by a single $V_{kl}^\pm$ is
\begin{equation}
\sum_{ij}\ket{a_i}\bra{a_j}\otimes V_{kl}^\pm \ket{a_i}\bra{a_j}{V_{kl}^\pm}^\dagger=(I\otimes V_{kl}^\pm) \ket{\Psi}\bra{\Psi}(I\otimes {V_{kl}^\pm})^\dagger;\label{eq:o3}
\end{equation}
Therefore $\mathscr{T}(\widetilde{S^\Gamma}_\pm)$ are ensembles of pure states. It was proved in Appendix F of Ref.~\cite{Campbell10} that if operators $O_1$ and $O_2$ are orthogonal, then $I\otimes O_1\ket{\Psi}$ and $I\otimes O_2\ket{\Psi}$ are orthogonal:
\begin{equation}
\inner{(I\otimes O_2)\Psi}{(I\otimes O_1)\Psi}=\sum_{i,j}\inner{a_i}{a_j}\bracket{a_i}{O_2^\dagger O_1}{a_j}
=(O_2|O_1).\label{eq:o4}
\end{equation}
By Lemma~\ref{lem:oen} and the orthogonality of $\{V_{ij}^+\}$ and $\{V_{kl}^-\}$, $\mathrm{ran}\mathscr{T}(\widetilde{S^\Gamma}_+)\perp\mathrm{ran}\mathscr{T}(\widetilde{S^\Gamma}_-)$, so $\widetilde{S^\Gamma}_\pm$ as defined in \eqref{eq:SGa} are equal to $S^\Gamma_\pm.$ Therefore,
\begin{align}
{S^\Gamma_-}^\dagger(I)&=\frac{1}{2}\sum_{i\neq j}\lambda_i\lambda_j {V_{ij}^-}^\dagger V_{ij}^-,\nonumber\\
{S^\Gamma_+}^\dagger(I)&=\sum_i \lambda_i^2 {V_{ii}^+}^\dagger V_{ii}^+  +\frac{1}{2}\sum_{i\neq j} \lambda_i \lambda_j {V_{ij}^+}^\dagger V_{ij}^+,\label{eq:SGA}
\end{align}
and
\begin{align}
{S^\Gamma_+}^\dagger(I)+{S^\Gamma_-}^\dagger(I)&=\sum_i \lambda_i^2 {A_i^*}^\dagger A_i^*\otimes B_i^\dagger B_i+\frac{1}{2}\sum_{i\neq j}\lambda_i\lambda_j\left( {A_i^*}^\dagger A_i^*\otimes B_j^\dagger B_j+{A_j^*}^\dagger A_j^*\otimes B_i^\dagger B_i\right)\nonumber\\
&=\sum_{i=j} \lambda_i\lambda_j {A_i^*}^\dagger A_i^*\otimes B_j^\dagger B_j+\sum_{i\neq j} \lambda_i\lambda_j{A_i^*}^\dagger A_i^*\otimes B_j^\dagger B_j\nonumber\\
&=\sum_{i}\lambda_i{A_i^*}^\dagger A_i^*\otimes \sum_j \lambda_j B_j^\dagger B_j.\label{eq:s+-}
\end{align}

\subsubsection{In the entirety}
Let's come back to the operation $S=\sum_i S_i,$ where $S_i(A)=V_i A V_i^\dagger$ and $V_i=\sum_j \lambda_j A_{ij}\otimes B_{ij}$ is the Schmidt decomposition. By \eqref{eq:n1}, for a TP $L=\sum_i L_i$ with each $L_i$ being HP, we have
\begin{equation}
1+2\Big|\Big|\sum_i {L_i}_-^\dagger(I)\Big|\Big|=\Big|\Big|\sum_i\left({L_i}_+^\dagger(I)+{L_i}_-^\dagger(I)\right)\Big|\Big|;
\end{equation}
along with \eqref{eq:s+-} we obtain,
\begin{align*}
1+2||\sum_i {S^\Gamma_i}_-^\dagger(I)||&=||\sum_i \left({S^\Gamma_i}_+^\dagger(I)+ {S^\Gamma_i}_-^\dagger(I)\right)||\nonumber\\
&=||\sum_i\sum_{j(i)} \lambda_{ij}{A_{ij}^*}^\dagger A_{ij}^*\otimes \sum_{k(i)} \lambda_{ik} B_{ik}^\dagger B_{ik} ||\\
&\leq\sum_i ||\sum_{j(i)} \lambda_{ij}{A_{ij}^*}^\dagger A_{ij}^*||\,|| \sum_{k(i)} \lambda_{ik} B_{ik}^\dagger B_{ik}||\\
&=\sum_i ||\sum_{j(i)} \lambda_{ij}{A_{ij}}^\dagger A_{ij}||\,|| \sum_{k(i)} \lambda_{ik} B_{ik}^\dagger B_{ik}||,\\
\end{align*}
by $||O_1\otimes O_2||=||O_1||\;||O_2||$ and $||O\tr||=||O||$.\hfill~$\square$

Following the method developed in this proof, in Appendix~\ref{app:Sch} I'll discuss more about Schmidt decomposition and spectral decomposition.

\subsection{Proof for the Lower Bounds}

By \eqref{eq:l1},
\begin{equation}
1+2\frac{||{S^\Gamma_-}^\dagger(I)||_1}{d_A d_B}=\frac{\trace\left[{S^\Gamma_+}^\dagger(I)+{S^\Gamma_-}^\dagger(I)\right]}{d_A d_B},\label{eq:eq}
\end{equation}
and by \eqref{eq:s+-},
\begin{equation}
\trace\left[{S^\Gamma_+}^\dagger(I)+{S^\Gamma_-}^\dagger(I)\right]=\sum_i \lambda_i \trace {A_i^*}^\dagger A_i^* \sum_j \lambda_j \trace \lambda_j B_j^\dagger B_j=(\sum_i \lambda_i)^2, \label{eq:sq}
\end{equation}
because $\{A_i\}$ and $\{B_i\}$ are orthonormal sets of operators. \eqref{eq:eq} and \eqref{eq:sq} together demonstrate the equivalence. \hfill~$\square$

\section{PPT-ness and Separability of Unitary Operations and of Pure States}

Here let's show
\begin{pro}\label{pro:sppt}
For a finite-dimensional $\mathcal{H}=\mathcal{H}_A\otimes \mathcal{H}_B$, the following statements for either a unitary operation $S(\rho)=U\rho U^\dagger$ with $U\in\mathcal{U}(\mathcal{H})$ or a pure state $\ket{\psi}\bra{\psi}$ on $\mathcal{H}$ are equivalent:
\begin{enumerate}
\item The rank\footnote{The rank of a tensor (Section~\ref{sec:rank}) is equal to the Schmidt number or Schmidt rank, the number of Schmidt coefficients of a tensor \cite{Breuer,Stormer,Bengtsson}. Note that to define the rank of a tensor inner product isn't required, whereas the Schmidt decomposition is based on inner product.} of $U$ or $\ket{\psi}$ is 1; in other words, $U$ or $\ket{\psi}$ is decomposable (Definition~\ref{def:tensor}).
\item The operation/state is separable.
\item The operation/state is PPT.	
\end{enumerate}
\end{pro} 
Hence, for unitary operations and pure states, separability and PPT-ness are no different. If a unitary operation is found to be PPT, then after this operation any separable state will remain separable. That PPT pure states are necessarily separable has already been shown in, e.g. \cite{Vidal02,Zyczkowski02}, where the negativity of a pure state $\ket{\psi}$ was explicitly calculated, given its Schmidt decomposition. Below we will provide an alternative proof for pure states.

\subsection{Proof for Pure States}\label{sec:pps}
By definition a pure state is separable if and only if its (Schmidt) rank is 1. Suppose the Schmidt decomposition of a pure state is
\begin{equation}
\ket{\psi}=\sum_i \lambda_i \ket{a_i}\ket{b_i}.
\end{equation}
Since partial transposition with respect to different bases are related unitarily, we can do this with respect to any basis, and let's choose $\{\ket{a_i}\}$. We obtain
\begin{equation}
\rho^\Gamma=(\ket{\psi}\bra{\psi})^\Gamma=\sum_{i,j}\lambda_i\lambda_j \ket{a_j}\bra{a_i}\otimes \ket{b_i}\bra{b_j}.
\end{equation}
There's a decomposition of $\rho^\Gamma$ as $\rho^\Gamma=\widetilde{\rho^\Gamma}^+-\widetilde{\rho^\Gamma}^-$, where
\begin{subequations}
\begin{align}
\widetilde{\rho^\Gamma}^+&=\sum_{i}\lambda_i^2(\ket{a_i}\ket{b_i})(\bra{a_i}\bra{b_i})+\sum_{i<j}\lambda_i \lambda_j\frac{\ket{a_i}\ket{b_j}+\ket{a_j}\ket{b_i}}{\sqrt{2}}\frac{\bra{a_i}\bra{b_j}+\bra{a_j}\bra{b_i}}{\sqrt{2}},\label{eq:rho+}\\
\widetilde{\rho^\Gamma}^-&=\sum_{i<j}\lambda_i \lambda_j\frac{\ket{a_i}\ket{b_j}-\ket{a_j}\ket{b_i}}{\sqrt{2}}\frac{\bra{a_i}\bra{b_j}-\bra{a_j}\bra{b_i}}{\sqrt{2}}.\label{eq:rho-}
\end{align}
\end{subequations}
\eqref{eq:rho+} and \eqref{eq:rho-} are ensembles of $\widetilde{\rho^\Gamma}^\pm$. As the vectors are orthonormal, \eqref{eq:rho+} and \eqref{eq:rho-} together are spectral decomposition of $\rho^\Gamma$ by Lemma~\ref{lem:or} or the spectral theorem, i.e. $\widetilde{\rho^\Gamma}^\pm={\rho^\Gamma}^\pm$, c.f. the proof in Section~\ref{sec:sisuo}.

Therefore, because $\lambda_i>0$, ${\rho^\Gamma}^-=0$, i.e. $\rho$ is PPT if and only if the rank of $\ket{\psi}$ is 1.

\subsection{Proof for Unitary Operations}
Let $U=\sum_i \lambda_i A_i\otimes B_i$, with $A_i\in \mathcal{B}(\mathcal{H}_A)$ and $B_i\in \mathcal{B}(\mathcal{H}_B)$, $\{a_i\}$ and $\{b_i\}$ be orthonormal bases of $\mathcal{H}_A$ and $\mathcal{H}_B$ respectively. 

It's easy to see a unitary operation $S(O)=U O U^\dagger$ is separable if $U=\sum_i \lambda_i A_i\otimes B_i$ has rank 1. To show it's separable only if the Schmidt rank is 1, as in Section~\ref{sec:lo} we define a space $\mathcal{H}_a:=\mathcal{H}_a^A\otimes \mathcal{H}_a^B$ with $\mathcal{H}_a^A=\mathcal{H}_A$ and $\mathcal{H}_a^B=\mathcal{H}_B$, and define
\begin{equation}
\ket{\Psi}=\ket{\Psi_A}\ket{\Psi_B},
\end{equation}
where
\begin{subequations}
\begin{align}
\ket{\Psi_A}&:=\sum_i \ket{a_i}\ket{a_i},\\
\ket{\Psi_B}&:=\sum_i \ket{b_i}\ket{b_i}.
\end{align}
\end{subequations}
Then we apply $I_{a}\otimes U=I_a^A\otimes I_a^B\otimes U$ on $\ket{\Psi}$:
\begin{equation}
I_{a}\otimes U\ket{\Psi_{AB}}=\sum_i\lambda_i (I_A\otimes A_i\ket{\Psi_A})\otimes (I_B\otimes B_i\ket{\Psi_B}).\label{eq:IU}
\end{equation}
If $U$ is separable, then the state above is separable with respect to A and B. Because of \eqref{eq:o4}, $I_A\otimes A_i\ket{\Psi_A}$ are orthogonal with $I_A\otimes A_j\ket{\Psi_A}$ for $i\neq j$, likewise for B. Therefore, \eqref{eq:IU} is a Schmidt decomposition of $I_{AB}\otimes U\ket{\Psi_{AB}}$ (after normalization), and this vector is separable if and only if the Schmidt rank of $U$ is 1. Therefore, no unitary operators with Schmidt rank higher than 1 are separable, so a unitary operator is separable if and only if its Schmidt rank is 1.

From \eqref{eq:VV} and \eqref{eq:SGa}, it's clear that a unitary operation is PPT if and only if there exist no $V_{ij}^-$, which happens if and only if the Schmidt rank of the unitary operator is 1. Thus a unitary operation is PPT if and only if it's separable. 

\eqref{eq:eq} and \eqref{eq:sq} also confirm this: A unitary operation $S$ is PPT if and only if $\trace\mathscr{T}(S^\Gamma)^-=||{S^\Gamma_-}^\dagger(I)||_1=0$, i.e. $(\sum_i \lambda_i)^2=d_A d_B.$ Because each $\lambda_i> 0$, $(\sum_i\lambda_i)^2\geq \sum_i \lambda_i^2$, and they are equal if and only if there's only one $\lambda_i$. Since 
\begin{equation}
\trace U^\dagger U=\trace I=(U|U)=(\sum_i \lambda_i A_i\otimes B_i|\sum_j \lambda_j A_j \otimes B_j)=\sum_i \lambda_i^2=d_A d_B,
\end{equation}
$(\sum_i \lambda_i)^2=d_A d_B$ to if and only if the rank of $U$ is 1, i.e. separable. \hfill~$\square$

\section{Exact Entangling Capacity}\label{sec:bae}
Here we consider deterministic operations from $\mathcal{B}(\mathcal{H}_1^A\otimes \mathcal{H}_1^B)$ to $\mathcal{B}(\mathcal{H}_2^A\otimes \mathcal{H}_2^B)$ and density operators in $\mathcal{B}(\mathcal{H}_1^A\otimes \mathcal{H}_1^B)$. We'd like to answer the question: Under what conditions do the upper and lower bounds from Proposition~\ref{pro} become identical, so that the bounds are the exact entangling capacity? And if so, with which states can we attain the entangling capacity?  
\subsection{Basic Unitary Operators}
Let's start off with what's already known. Among operators on $\mathcal{H}_A\otimes \mathcal{H}_B$, basic unitary operators are those whose Schmidt decompositions $U=\sum_i \lambda_{i} A_{i}\otimes B_{i}$ have all the $A_i$ and $B_i$ proportional to unitary operators \cite{Campbell10}. It was shown all $2\otimes 2$ unitary operators are basic \cite{Campbell10}. The upper and lower bounds for the entangling capacity of a basic unitary operator are identical (with $\widetilde{S^\Gamma}_\pm=S^\Gamma_\pm$), so they are the exact entangling capacity \cite{Campbell10}.

\subsection{Reaching the Entangling Capacity}
To understand the proposition below, it's recommended to read Section~\ref{sec:fispe}, \ref{sec:en}, \ref{sec:dehp} and \ref{sec:opsum} first:
\begin{pro}\label{pro:ecs}
Suppose $\widetilde{S^\Gamma}_\pm=S^\Gamma_\pm,$ that the operator-sum representations of $S^\Gamma_\pm$ are\footnote{The index $i$ depends on $\pm.$}
\begin{equation}
S^\Gamma_\pm(O)=\sum_i c_i^\pm V_i^\pm O {V_i^\pm}^\dagger,\,c_i^\pm>0,\label{eq:sg+-}
\end{equation}
and that for the input state $\rho$
\begin{equation}
{\rho^\Gamma}^\pm=\sum_{i} \ket{\psi^\pm_i}\bra{\psi^\pm_i}\label{eq:rg+-}
\end{equation}
are ensembles\footnote{The index $i$ depends on $\pm$ too. The eigensemble is an obvious choice, but any ensemble will do.} of ${\rho^\Gamma}^\pm$. The upper bound of the entangling capacity given by Proposition~\ref{pro} is reached if and only if
\begin{enumerate}
\item $\mathrm{ran}\left[S^\Gamma_-({\rho^\Gamma}^+)+S^\Gamma_+({\rho^\Gamma}^-)\right]$ and $ \mathrm{ran}\left[S^\Gamma_+({\rho^\Gamma}^+)+S^\Gamma_-({\rho^\Gamma}^-)\right]$ are orthogonal, which is equivalent to the orthogonality of the following vectors:
\begin{equation}
\inner{V_i^\pm \psi_j^\mp}{V_k^+ \psi_l^+}=0 \text{ and } 
\inner{V_i^\pm \psi_j^\mp}{V_k^- \psi_l^-}=0 \;\forall i,j,k,l.\label{eq:oeq}
\end{equation}

\item $\mathrm{ran}\rho^\Gamma$ is a subspace of the eigenspace corresponding to the largest eigenvalue of ${S^\Gamma_\pm}^\dagger(I)$. 
\end{enumerate}

The upper and lower bounds of entangling capacities from Proposition~\ref{pro} are the same if and only if ${S^\Gamma_\pm}^\dagger(I)\propto I$, and the second condition above is satisfied by any state $\rho$ when ${S^\Gamma_\pm}^\dagger(I)\propto I$.
\end{pro}

\begin{proof}

First, for $P\geq 0$, $||P||_1/d=\trace P/d$, where $d$ is the dimension of the space, is the average of its eigenvalues (zero included), whereas $||P||$ is its largest eigenvalue. Hence they're the same if and only if $P\propto I,$ and $||{S^\Gamma_\pm}^\dagger(I)||_1/(d_A d_B)=||{S^\Gamma_\pm}^\dagger(I)||$ if and only if ${S^\Gamma_\pm}^\dagger(I)\propto I.$

The second condition above is a direct application of Lemma~\ref{lem:lh}, and it's satisfied when ${S^\Gamma_-}^\dagger(I)\propto I$, c.f. the discussion below \eqref{eq:P1P2}. The first condition is also due to Lemma~\ref{lem:lh}. From \eqref{eq:sg+-} and \eqref{eq:rg+-} we acquire ensembles of $S^\Gamma_\pm({\rho^\Gamma}^+)+S^\Gamma_\mp({\rho^\Gamma}^-)$:
\begin{equation*}
S^\Gamma_\pm({\rho^\Gamma}^+)+S^\Gamma_\mp({\rho^\Gamma}^-)=\sum_{i,j} c_i^\pm V_i^\pm \ket{ \psi_j^+}\bra{ \psi_j^+}{V_i^\pm}^\dagger+\sum_{k,l} c_k^\mp V_k^\mp \ket{ \psi_l^-}\bra{ \psi_l^-}{V_k^\mp}^\dagger.  
\end{equation*}
By Lemma~\ref{lem:oen}, $\mathrm{ran}\big[S^\Gamma_-({\rho^\Gamma}^+)+S^\Gamma_+({\rho^\Gamma}^-)\big]\perp \mathrm{ran}\big[S^\Gamma_+({\rho^\Gamma}^+)+S^\Gamma_-({\rho^\Gamma}^-)\big]$ if and only if
\begin{align*}
\inner{V_i^\pm \psi_j^\mp}{V_k^+ \psi_l^+}&=0 \,\forall i,j,k,l\text{ and}\\ 
\inner{V_i^\pm \psi_j^\mp}{V_k^- \psi_l^-}&=0 \,\forall i,j,k,l,
\end{align*}
which is \eqref{eq:oeq}.
\end{proof}

The parts related to the upper bounds only are also applicable to any $\widetilde{S^\Gamma}_\pm$, as a result of Lemma~\ref{lem:lh}. The core of this proposition doesn't rely on operator-sum representations and ensembles of positive operators, but in practical calculation they come in handy. The proposition gives the necessary and sufficient condition for the upper and lower bounds to be identical---it does not matter whether it is unitary or even basic. When it's satisfied, the exact entangling capacity is acquired, and $||S||_{1,\Gamma}$ from Definition~\ref{def:normga} reflects its true entangling capacity.

\subsubsection{PPT states}
If the state $\rho$ is PPT, ${\rho^\Gamma}^-=0$, so $\psi_i^-=0$ and \eqref{eq:rg+-} becomes:
\begin{equation}
\inner{V_i^- \psi_j^+}{V_k^+ \psi_l^+}=0 \;\forall i,j,k,l.\label{eq:oeqppt}
\end{equation}
In particular, \eqref{eq:oeqppt} is the equation that needs to be solved if $\rho$ is separable.
\subsubsection{Pure separable states} 
Given an orthonormal basis $\{\ket{a_i}\}$ of $\mathcal{H}_1^A$, with the pure separable state $\ket{\psi}=\ket{\psi_1}\ket{\psi_2}$, where $\ket{\psi_1}=\sum c_i\ket{a_i}\in\mathcal{H}_1^A$ and $\ket{\psi_2}\in\mathcal{H}_1^B$, the partial transpose of it is still pure and separable:
\begin{equation*}
(\ket{\psi}\bra{\psi})^\Gamma=\ket{\psi_1^*}\bra{\psi_1^*}\otimes\ket{\psi_2}\bra{\psi_2},
\end{equation*}
where $\ket{\psi_1^*}=\sum c_i^*\ket{a_i}.$ Hence, ${\rho^\Gamma}^+=\ket{\psi_1^*}\bra{\psi_1^*}\otimes\ket{\psi_2}\bra{\psi_2}$ and ${\rho^\Gamma}^-=0,$ and \eqref{eq:oeq} (or \eqref{eq:oeqppt}) becomes
\begin{equation}
(\bra{\psi_1^*}\bra{\psi_2}){V_i^-}^\dagger V_j^+(\ket{\psi_1^*}\ket{\psi_2})=0\;\forall i,j.\label{eq:vv}
\end{equation}

\subsection{Examples}\label{sec:ex}

Here I'll provide several examples whereby to demonstrate the method outlined above. The initial states are pure and separable for simplicity. For Proposition~\ref{pro:ecs} to be applicable, $\widetilde{S^\Gamma}_\pm=S^\Gamma_\pm$.
\subsubsection{A $2\otimes 2$ unitary operation}\label{sec:22}

Consider a $2 \otimes 2$\footnote{It means the space $\mathcal{H}_A\otimes \mathcal{H}_B$ has $\mathrm{dim}H_A=\mathrm{dim}H_B=2$.} unitary operation $S(\rho)=U\rho U^\dagger$, where
\begin{equation}
U=\begin{pmatrix}
\cos\alpha & \sin\alpha &0 &0\\
-\sin\alpha & \cos\alpha &0 &0\\
0&0&\cos\beta&\sin\beta\\
0&0&-\sin\beta&\cos\beta
\end{pmatrix}.
\end{equation}
Since all $2\otimes 2$ unitary operators are basic \cite{Campbell10}, we can find:
\begin{equation}
{S^\Gamma_-}^\dagger(I)=\frac{\left|\sin(\beta-\alpha)\right|}{2} I\propto I.\label{eq:SI}
\end{equation}

With the input state $\ket{\psi}=\ket{\psi_1}\otimes \ket{\psi_2}$ for which
\begin{equation}
\ket{\psi_i}=\cos\theta_i\ket{\uparrow}+e^{i\phi_i}\sin\theta_i\ket{\downarrow},\label{eq:state}
\end{equation}
\eqref{eq:vv}, depending on the sign of $\sin(\beta-\alpha)$, becomes
\begin{align*}
\{\bra{\psi_1^*}\bra{\psi_2}{V^-}^\dagger V_i^+\ket{\psi_1^*}\ket{\psi_2} \}=\{&\mp 2i\sec(\beta+\alpha)\cos(2\theta_1)\sin(2\theta_2)\sin\phi_2,2\frac{\cos(2\theta_1)\pm i\sin(2\theta_2)\sin\phi_2}{-1+\cos(\alpha+\beta)\mp\sin(\alpha+\beta)},\nonumber\\
&\mp\sec^2\frac{\beta+\alpha}{2}\frac{\cos(2\theta_1)\mp i\sin(2\theta_2)\sin\phi_2}{\pm 1+\tan((\beta+\alpha)/2)}\},
\end{align*}
taking the top ones when $\sin(\beta-\alpha)\geq 0.$ For all of them to vanish,
\begin{equation}
\theta_1=\pi/4+n\pi/2, \text{ and } \theta_2=m\pi/2 \text{ or }\phi_2=p\pi,\,m,n,p\in \mathbb{Z},\label{eq:22sol}
\end{equation}
e.g. $(1,0,1,0)/\sqrt{2}$. Interestingly, the solution has nothing to do with the values of $\alpha$ and $\beta$.

From \eqref{eq:SI}, this unitary operation is a perfect entangler with respect to negativities (Section~\ref{sec:pee}) when $\beta-\alpha=\pi/2+n\pi$, $n\in\mathbb{Z}$, and a pure separable state will have maximal negativity (for a state in/on $2\times 2$) after this operation if and only if it obeys \eqref{eq:22sol}. 

On the other hand, when $\beta-\alpha=n\pi$, $n\in\mathbb{Z}$, this operation is PPT and therefore entirely non-entangling, because a state $\rho$ on $2\times 2$ is separable if and only if it's PPT by Peres-Horodecki criterion (Theorem~\ref{thm:peres}). However, as Proposition~\ref{pro:sppt} suggests, because this operation is unitary, it's separable if and only if it's PPT, and the exact dimensions don't really matter.

\subsubsection{Generalized CNOT gate}
Here's another $2\otimes 2$ unitary operator:
\begin{equation}
U=\begin{pmatrix}
\cos\alpha & \sin\alpha &0 &0\\
-\sin\alpha & \cos\alpha &0 &0\\
0&0&\cos\beta&\sin\beta\\
0&0&\sin\beta&-\cos\beta
\end{pmatrix}.
\end{equation}
When $\alpha=0$ and $\beta=\pi/2$, this is a controlled-NOT (CNOT) gate \cite{Barenco95,Steane98,Deutsch99}:
\begin{equation*}
\begin{pmatrix}
1 & 0 &0 &0\\
0 & 1 &0 &0\\
0&0&0&1\\
0&0&1&0
\end{pmatrix}.
\end{equation*}
Regardless of $\alpha$ and $\beta$,
\begin{equation}
(S^\Gamma_-)^\dagger(I)=I/2.
\end{equation}
Hence it's a perfect entangler with respect to negativities.

To find the optimal separable pure state \eqref{eq:state}, we need to solve \eqref{eq:vv}:
\begin{subequations}
\begin{align}
0=&-4\cos^2\theta_1 \sec \alpha\sin\beta\left[\csc(\alpha+\beta)+\cos\beta\sec\alpha\left(\cos 2\theta_2\cot(\alpha+\beta)+\cos\phi_2\sin 2\theta_2\right) \right]\nonumber\\
&+4\sin^2\theta_1\left[\cos^2\theta_2\left(-1+2\csc(\alpha+\beta)\sec\alpha\sin\beta \right)+\sin^2\theta_2-\cos\phi_2\sin 2\theta_2\tan\alpha\right],\label{eq:23a}\\
0=&-\cos\beta\cos^2 \theta_2\csc(\alpha+\beta)\sec\alpha\left[3\cos 2\theta_1+\cos(\alpha+2\beta)\cos^2\theta_1\sec\alpha\right]\nonumber\\
&+\cos(2\alpha+\beta)\cos^2\theta_2\csc(\alpha+\beta)\sec\alpha\sin^2\theta_1-2\cos^2\beta\cos^2\theta_1\cos\phi_2\sec^2\alpha\sin 2\theta_2\nonumber\\
&+2\cos\phi_2\sin^2\theta_1\sin 2\theta_2+\sin^2\theta_2\left(-\cos^2\theta_1\sec^2\alpha\sin 2\beta+2\sin^2\theta_1\tan\alpha\right),\label{eq:23b}\\
0=&\cos^2\theta_1\sec^2\alpha\left[\cos(\alpha+\beta)\cos2\theta_2+\cos\phi_2\sin(\alpha+\beta)\sin 2\theta_2\right].\label{eq:23c}
\end{align}
\end{subequations}
From \eqref{eq:23c}, $\cos\theta_1=0$ or $\cos(\alpha+\beta)\cos2\theta_2+\cos\phi_2\sin(\alpha+\beta)\sin 2\theta_2=0.$ When $\cos\theta_1=0$, it can happen that there's no solution, and we are not going to address it.\footnote{For example, with $\cos\theta_1=0$, if $\tan\alpha=0$ and $\beta=\pi/2$, the right side of \eqref{eq:23a} becomes $1$ and \eqref{eq:23b} becomes $\cos\phi_2\sin 2\theta_2=0$. By continuity there are more than one $\alpha$ and $\beta$ such that there's no solution.} When $\cos(\alpha+\beta)\cos2\theta_2+\cos\phi_2\sin(\alpha+\beta)\sin 2\theta_2=0$, 
\begin{equation}
\cos\phi_2=-\cot(\alpha+\beta)\cot 2\theta_2.\label{eq:sol1}
\end{equation}
Replacing $\cos\phi_2$ in \eqref{eq:23a} and \eqref{eq:23b} results in
\begin{subequations}
\begin{align}
0&=\cos 2\theta_1 \csc (\alpha+\beta)\sec \alpha\sin\beta,\\
0&=\cos 2\theta_1 \csc(\alpha+\beta)\sec\alpha\cos\beta.
\end{align}\label{eq:sol2}
\end{subequations}
Therefore from \eqref{eq:sol1} and \eqref{eq:sol2} a solution is\footnote{Since the range of $\cos$ is bounded, it's better to choose $\phi_2$ first and solve for $\theta_2$ than the other way around.}
\begin{equation}
\theta_2=\frac{1}{2}\arctan\left(\frac{\tan(\alpha+\beta)}{\cos\phi_2}\right)\text{ and }\theta_1=\pi/4+n\pi/2,\,n\in\mathbb{Z}.
\end{equation}

For example, for the CNOT gate, i.e. $\alpha=0$ and $\beta=\pi/2$, one of the solutions is $\theta_1=\pi/4$, $\phi_1=0$, $\theta_2=0$, $\phi_2=0$, that is $(\ket{\uparrow}+\ket{\downarrow})\ket{\uparrow}/\sqrt{2}$, which after the CNOT gate will be transformed into the maximally entangled state $(\ket{\uparrow}\ket{\uparrow}+\ket{\downarrow}\ket{\downarrow})/\sqrt{2}.$

\subsubsection{A $2\otimes 3$ unitary operation}
Now consider a $2\otimes 3$ unitary operation $S(\rho)=U\rho U^\dagger$, with:
\begin{equation}
U=\left(\begin{array}{c|c}
I_{3}&0_{3}\\\hline
0_{3}&\begin{matrix}
\cos\beta&\sin\beta&0\\
-\sin\beta&\cos\beta&0\\
0&0&1
\end{matrix}
\end{array}\right)\cdot
\left(\begin{array}{c|c}
I_{4}&0_{4}\\\hline
0_{4}&\begin{matrix}
\cos\alpha&\sin\alpha\\
-\sin\alpha&\cos\alpha&
\end{matrix}
\end{array}\right),
\end{equation}
where $I_n$ and $0_n$ denote $n\times n$ identity and zero matrices respectively.

As $2\otimes 3$ unitary operations in general are not basic, the upper and lower bounds normally do not coincide. For them to be the same, we need ${S^\Gamma_-}^\dagger(I)\propto I,$ i.e. the maximum and minimum eigenvalues of ${S^\Gamma_-}^\dagger(I)\propto I$ should be identical. From Fig.~\ref{fig:2x3} they seem to be the same at $1/2$ when $\alpha=2\pi/3$ and $\beta=0$, which can indeed be verified analytically, so
\begin{equation*}
\text{at }\alpha=2\pi/3\text{ and }\beta=0,\,{S^\Gamma_-}^\dagger(I)=\frac{1}{2}I.
\end{equation*}
At this particular point, the operation is a perfect entangler, and obeying \eqref{eq:oeq} alone is necessary and sufficient to reach the entangling capacity.

\begin{figure}[hbtp!]
\centering
\includegraphics[width=0.6\linewidth]{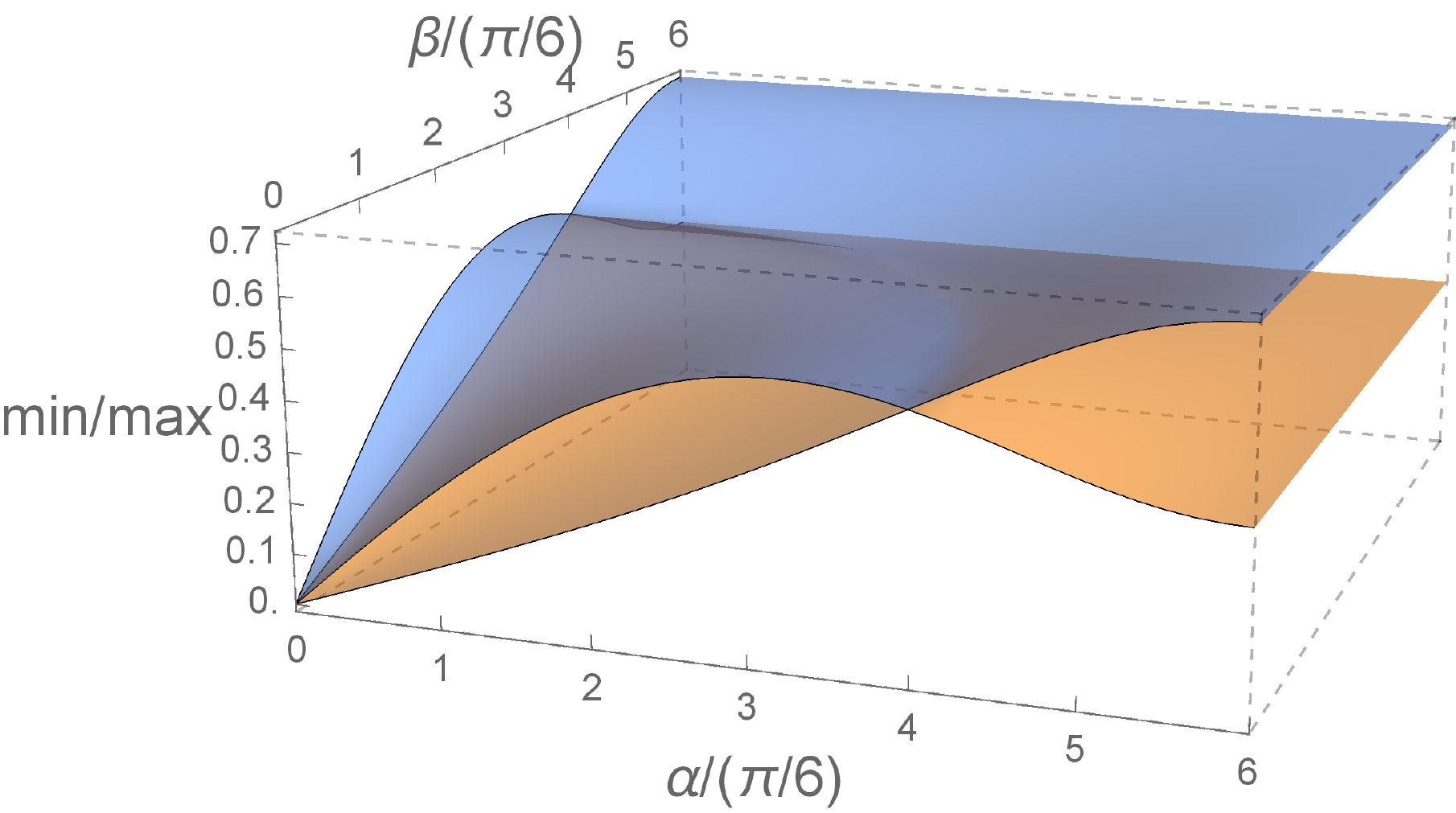}
\caption[The max/min eigenvalues of ${S^\Gamma_-}^\dagger(I)$]{The max/min eigenvalues of ${S^\Gamma_-}^\dagger(I)$, with eigenvalue 0 taken into account.}
\label{fig:2x3}
\end{figure}

Assume again the state is pure and separable
\begin{equation}
\ket{\psi}=(\cos\theta_1\ket{\uparrow}+\sin\theta_1\ket{\downarrow})\otimes(\cos\theta_2\sin\phi\ket{0}+\sin\theta_2\sin\phi\ket{1}+\cos\phi\ket{2}).
\end{equation}
At $\alpha=2\pi/3$ and $\beta=0,$ \eqref{eq:vv} becomes
\begin{subequations}
\begin{align}
0=&4\sin^2\theta_1-3f,\label{eq:e1}\\
0=&4g\sin^2\theta_1 -f,\label{eq:e2}\\
0=&\cos^2\theta_1\left(1-3g\right),\label{eq:e3}
\end{align}
\end{subequations}
where
\begin{equation*}
f=\cos^2\theta_1\left(1+g\right),\,g=\sin^2\phi\cos^2\theta_2.\label{eq:fg}
\end{equation*}
From \eqref{eq:e3}, either $\cos^2\theta_1= 0$ or $1-3 g=0$. However, $\cos^2\theta_1\neq 0$; otherwise $f$ would vanish, and from \eqref{eq:e1}, $\sin^2\theta_1$ would be zero, a contradiction. Hence, $g=1/3,$ so $f=4\cos^2\theta_1/3$, and \eqref{eq:e2} is satisfied. \eqref{eq:e1} now becomes
\begin{equation*}
4\sin^2\theta_1-4\cos^2\theta_1=0=-4\cos 2\theta_1.
\end{equation*}
Thus the solution is 
\begin{equation*}
\sin^2\phi\cos^2\theta_2=1/3 \text{ and }\cos2\theta_1=0,
\end{equation*}
e.g. $1/\sqrt{6}(\ket{\uparrow}+\ket{\downarrow})(\ket{0}+\sqrt{2}\ket{2})$.

\section{Conclusion}

We found upper and lower bounds for entangling capacities, generalizing the work by Campbell \cite{Campbell10}, and explained their physical  and geometric significance, as illustrated by Figure~\ref{fig:norm} and \ref{fig:pro2}. Giving them norms and metrics, Definition~\ref{def:normga}, may help up better comprehend entanglement and PPT. Let's  reemphasize that $||\cdots||_{1,\Gamma}$ is a valid norm, like trace norm, on  linear mappings or operators, and the norms are bounded from above and below by one another (Section~\ref{sec:eq}).

Our result reflects that operations and states are tight-knit: The entanglement or negativity of operations is correlated with that of states. In Ref.~\cite{Zanardi01} a similar phenomenon was  observed for entangling power in terms of linear entropy. We expect this to happen in other facets of states and operations: For example, a more non-separable operation could make states more non-separable, if a suitable measure of non-separability is used \cite{Chitambar12}.

Entanglement of a dynamical system can be studied with this method in a state-independent way. For example, if the $2\otimes 2$ unitary operator in Section~\ref{sec:22} has time-dependent $\alpha(t)$ and $\beta(t)$, we can investigate how entangling the system is as it evolves, and in the case where the optimal states don't depend on the parameters of the operation, like in Section~\ref{sec:22}, we can track the exact entanglement at any time $t$ for those states, and the entanglement rate can be obtained \cite{Dur01,Bennett03}.

We also found for unitary operations and pure states separability and PPT-ness are identical. Hence, the deviation between PPT-ness and separability starts with mixed states \cite{Horodecki97}, and non-unitary operations \cite{Horodecki97,Chitambar14}.

Because PPT operations are defined to be CP after partial transposition, ancillas come in naturally: As we decompose an HP mapping into CP rather than positive parts, the upper bounds are unaltered with the addition of an ancilla. Furthermore, a non-PPT operation may be positive after partial transposition, and not be able to create any negativity due to positivity, but after adding an ancilla, it will fail to become positive after partial transposition. Studying operations through Choi isomorphism makes ancillas a natural fit.

\chapter{Density Operators on $L^2(\mathbb{R}^N)$ and Gaussian States}\label{ch:gaf}

\section{Basics}

\subsection{Covariance Matrix}\label{sec:cov}
\begin{defi}[\textbf{Covariances and variances}]
The covariance between two random variables $X$ and $Y$ is defined as \cite{deGosson,Gut}
\begin{equation}
\sigma(X,Y):=\langle (X-\langle X\rangle)(Y-\langle Y\rangle) \rangle=\langle XY\rangle-\langle X\rangle\langle Y\rangle,
\end{equation}
where $\langle \cdot \rangle $ is the expectation value of a random variable over a probability distribution. The variance of a random variable $X$ is
\begin{equation}
\sigma(X,X):=\langle (X-\langle X\rangle)(Y-\langle Y\rangle) \rangle=\langle X^2\rangle-\langle X\rangle^2.
\end{equation}
$\sqrt{\sigma (X,X)}$ is the \textbf{standard deviation} of $X$. 
\end{defi}

Let $(X_1,\cdots,X_n)$ be a vector of random variables. The matrix formed by $\sigma (X_i,X_j)_{1\leq i\leq n,1\leq j\leq n}$ is called the \textbf{covariance matrix} of this random vector, which is a $2n\times 2n$ real symmetric matrix. The off-diagonal elements are covariances, whereas diagonal ones are variances. By Jensen's inequality \cite{PapaRudin,Cover}, the variance of a random variable is non-negative, since $x\mapsto x^2$ is a convex function; more generally, a covariance matrix is always \emph{positive semi-definite} \cite{Gut}. Here's a theorem regarding vectors of random variables \cite{Gut}:
\begin{thm}\label{thm:rv}
Let $\mathbf{X}$ and $\mathbf{Y}$ be vectors of random variables. If $\mathbf{X}$ has a mean vector $\langle \mathbf{X} \rangle$ and covariance matrix $\sigma$, and $\mathbf{Y}=\mathbf{a}+T\mathbf{X}$, where $T$ is an operator/matrix, then $\mathbf{Y}$ has a mean vector $\mathbf{a}+T\langle \mathbf{X} \rangle$ and a covariance matrix $T\sigma T\tr$.
\end{thm}

For a quantum observable $O$, given a state $\rho$ the variance is then \cite{Simon94}
\begin{equation}
\langle O^2\rangle- \langle O\rangle^2=\trace (O^2\rho)-\trace (O\rho)^2.
\end{equation}
assuming $O\rho$ and $O^2\rho$ are still trace-class.\footnote{If $O$ is a bounded operator, they will be trace-class; if not, they may not be; see Section~\ref{sec:HSinner}.} The covariance of two observables $O_1$ and $O_2$ over a state $\rho$ is
\begin{equation}
\trace ( \frac{O_1O_2+O_2 O_1}{2} \rho)-\langle O_1\rangle\langle O_2\rangle,
\end{equation}
where we make similar assumptions. That we use the operator $(O_1 O_2+O_2 O_1)/2$ can be regarded as quantization of the classical variable $o_1 o_2$, if $O_1$ and $O_2$ have classical counterparts $o_1$ and $o_2$.

Let $\boldsymbol{\xi}$ be a vector of quadrature observables $(x_1,p_1,\cdots,x_n,p_n)$. The covariance matrix of $\boldsymbol{\xi}$, as defined above, is
\begin{equation}
\sigma_{ij}=\sigma(\xi_i,\xi_j).
\end{equation}
Note in previous chapters we grouped the position coordinates and momentum coordinates into two bunches, but here it's more convenient to put the position and momentum coordinates of each mode together. Hence some matrices will look slightly different, but other than that it makes no difference. From now on whenever we mention a covariance matrix it's always that of $\boldsymbol{\xi}$.

\subsection{Symplectic Spectrum}\label{sec:syms}
\begin{thm}[\textbf{Williamson's theorem}]
For a positive-definite symmetric real $2n\times 2n$ matrix $M$, there exists $S\in Sp(n)$ such that
\begin{equation}
S M S\tr=\mathrm{diag}(\nu_1,\nu_1,\nu_2,\nu_2,\cdots,\nu_n,\nu_n),\label{eq:Wi}
\end{equation}
which diagonal matrix is unique up to reordering. $\nu_i$ are called the \textbf{symplectic eigenvalues} of $M$, and $\pm i\nu_j$ are the eigenvalues of $\Omega M^{-1}$, where
\begin{equation}
\Omega=\mathrm{diag}(\begin{pmatrix}
0 & 1\\ -1 &0
\end{pmatrix},\begin{pmatrix}
0 & 1\\ -1 &0
\end{pmatrix},\dotsc).
\end{equation}
\end{thm}

\eqref{eq:Wi} is refereed to as the \textbf{Williamson normal/diagonal form} of $M$. Given the symplectic eigenvalues, $\mathrm{Spec}_\omega(M):=\{\nu_1,\nu_2,\cdots,\nu_n\}$ is called the \textbf{symplectic spectrum} of $M$.\footnote{The subscript $\omega$ refers to the symplectic form $\omega$.} Remark that the theorem doesn't say that the symplectic matrix that diagonalizes $M$ is unique, but the symplectic spectrum is the same; more generally, the symplectic spectrum is a symplectic invariant: $\mathrm{Spec}_\omega(M)=\mathrm{Spec}_\omega(SMS\tr)$ for any $S\in Sp(n)$.

\subsection{Gaussian State}
\begin{defi}[\textbf{Gaussian state}]
A quantum density operator $\rho$ on $L^2(\mathbb{R}^n)$ is said to be a Gaussian state if its Wigner quasi-probability distribution $W_\rho$ is Gaussian.
\end{defi}

The density matrix of a Gaussian state, namely the (integral) kernel of $\rho$, is consequently also of a Gaussian form, but the quadrature polynomial of the exponent may contain non-real coefficients. When the state is pure, studying a Gaussian state via the density matrix isn't a hindrance, but calculation can get complicated when dealing with general (mixed) Gaussian states. In Appendix~\ref{app:entbos} the reader can have a look at the calculation involved in this approach, where useful results could be obtained because of extra assumptions.

Given a Gaussian state, we can remove its first moments by local unitary operations, which are exactly translations of positions and momenta and which don't affect the covariances and variances \cite{Simon94}.\footnote{And a Gaussian state keeps Gaussian after translations.} Then a Gaussian state $\rho$ can be totally characterized by its covariance matrix $\sigma$ \cite{Simon00,Adesso14}:
\begin{equation}
W_\rho(\boldsymbol{\xi})=\mathcal{N}\exp\left(-\frac{1}{2}\boldsymbol{\xi}\tr\sigma^{-1}\boldsymbol{\xi}\right),\label{eq:Gau}
\end{equation} 
where the normalization factor $\mathcal{N}$ depends on $\sigma$ and the number of modes. Because the distribution is Gaussian, the covariance matrix is \emph{positive definite}. The covariance matrix alone can represent a Gaussian state, and sometimes we'll just call the covariance matrix $\sigma$ the ``state.''

On the contrary, given any positive-definite covariance matrix, the Weyl quantization (times $(2\pi \hbar)^n$) of a phase-space Gaussian like \eqref{eq:Gau} isn't necessarily a Gaussian state---being normalized and Schwartz implies it is trace-class and has unit trace; however, the operator may not be positive, so it may fail to be a density operator. This problem will be answered in the next subsection.

\subsection{Valid Quantum States}

For a density operator on $L^2(\mathbb{R}^n)$, because of the commutation relation $[x_i,p_j]=i\hbar\delta_{ij}$, there is a condition on its covariance matrix of $\boldsymbol{\xi}=(x_1,p_1,\cdots,x_n,p_n)$---Violating this condition implies the operator isn't a density operator. Simon et al. proved that \cite{Simon94,deGosson,Arvind95} 
\begin{thm}\label{thm:valid}
A trace-class self-adjoint $\rho$ on $L^2(\mathbb{R}^n)$ with $\trace \rho=1$ is a valid quantum state, i.e. $\rho\geq 0$ only if its covariance matrix $\sigma$ satisfies
\begin{equation}
\sigma+i\hbar\Omega/2\geq 0,\label{eq:goodstate}
\end{equation}
which is equivalent to each symplectic eigenvalue of $\sigma$ being no smaller than $\hbar/2$. This condition becomes necessary and sufficient for a phase-space Gaussian, namely an operator whose Weyl symbol (or Wigner quasi-probability distribution) is Gaussian.
\end{thm}

This theorem is actually the uncertainty relation in another guise. Unlike some papers $\hbar$ is kept in this thesis, so whichever convention is adopted or whichever system is considered the result can be easily adjusted.

\subsection{Entanglement of a Density Operator on $L^2(\mathbb{R}^n)$}

\subsubsection{Transposition, partial transposition and separability}
The transposition of a density operator on $L^2(\mathbb{R}^n)$ corresponds to a reversal of momentum of the phase space: The Weyl symbol of $\rho\tr$ is the phase-space function $(\mathbf{x},\mathbf{p})\mapsto W_\rho(\mathbf{x},-\mathbf{p})$ \cite{Simon00}. Hence for a multipartite state a reversal of momentum of one party is essentially a partial transposition. Caution that since the reversal of momenta isn't a symplectic transform, Theorem~\ref{thm:meta} can't be applied, and $\rho\tr$ and $\rho$ aren't unitarily equivalent.

By Peres criterion, all separable states are PPT, so by Theorem~\ref{thm:valid} all separable states satisfy the following condition:
\begin{equation}
\widetilde{\sigma}+i\hbar\Omega/2\geq 0,\label{eq:Simon}
\end{equation}
where, henceforth, a tilde on top indicates the quantity is of the partially transposed operator, or equivalently of the phase-space function with some of the momentum coordinates reversed. \eqref{eq:Simon} is known as the \textbf{Peres-Horodecki-Simon criterion} \cite{Peres96,Horodecki96,Simon00}. This is a necessary condition for separability, but not sufficient in general.

\subsubsection{Entanglement of a Gaussian state}\label{sec:eG}
It was proved by Simon \cite{Simon00,Werner01} that for a two-mode Gaussian state \eqref{eq:Simon} is a necessary and sufficient condition for it to be separable; namely it is separable if and only if its partial transpose is positive, or if and only if its negativity vanishes. 

For a Gaussian state $\rho$, the logarithmic negativity between party A and B, which may contain more than one mode each, with the covariance matrix of $\rho^\Gamma$ being $\widetilde{\sigma}$ is
\begin{equation}
E_L=\sum_i \max \left(\log \frac{\hbar}{2\widetilde{\nu}_i},0\right),\label{eq:LGau}
\end{equation}
where $\widetilde{\nu}_i$ are the symplectic eigenvalues of $\widetilde{\sigma}$ \cite{Vidal02}. 
\subsubsection{Negativity of a two-mode Gaussian}
For a two-mode Gaussian state, let's suppose its covariance matrix is
\begin{equation}
\sigma=\begin{pmatrix}
A & C \\ C\tr & B 
\end{pmatrix},
\end{equation}
$A$, $B$ and $C$ being $2\times 2$ submatrices. There's an invariant under local symplectic transformation \cite{Serafini04}:
\begin{equation}
\Delta:= \det A+\det B+2\det C.\label{eq:delta}
\end{equation}

The symplectic eigenvalues of $\sigma$ can be determined by \cite{Serafini04}
\begin{equation}
\nu^\pm=\sqrt{\frac{\Delta\pm\sqrt{\Delta^2-4\det\sigma}}{2}}.
\end{equation}
Similarly, that of the partially transposed state $\widetilde{\sigma}$ are
\begin{equation}
\widetilde{\nu}^\pm=\sqrt{\frac{\widetilde{\Delta}\pm\sqrt{\widetilde{\Delta}^2-4\det\sigma}}{2}};\label{eq:num}
\end{equation}
note $\det\sigma=\det\widetilde{\sigma}$. Since \eqref{eq:goodstate} basically states that the symplectic eigenvalues are no smaller than $\hbar/2$, $\widetilde{\nu}_+\geq \hbar/2$ by $\det\sigma=\det\widetilde{\sigma}$, and \eqref{eq:LGau} thus becomes
\begin{equation}
E_{L}=\max\left(\log\frac{\hbar}{2\widetilde{\nu}^-},0\right);\label{eq:LG}
\end{equation}
namely the smaller symplectic value of $\widetilde{\sigma}$ alone decides the negativity.

\section{Multisymmetric Gaussian states}
\subsection{Symmetric Gaussian States, and the Standard Form}
Consider the space $L^2(\mathbb{R}^N)$. Let's assume the Gaussian state has symmetric modes, which can be characterized by its $2N\times 2N$ covariance matrix \cite{Adesso04,Serafini05}: 
\begin{equation}
\sigma_N'=\begin{pmatrix}
\alpha & \beta  & \cdots & \beta \\
\beta  & \alpha& \beta & \vdots \\
\vdots & \beta & \ddots  & \beta \\
\beta  & \cdots & \beta & \alpha
\end{pmatrix},\label{eq:symmat}
\end{equation}
where $\alpha$ and $\beta$ are $2\times 2$ real and symmetric submatrices, and the subscript $N$ is for the total number of modes. By performing $N$ appropriate local symplectic transformations on each of the $N$ modes, corresponding to $N$ local unitary transformations on the density operator,\footnote{Such transforms are actually metaplectic. We will have a proper discussion about them in Section~\ref{sec:or}.} the covariance matrix \eqref{eq:symmat} can always be cast into the ``standard form'' \cite{Duan00,Simon00,Serafini04,Adesso04_2,Serafini05}:
\begin{equation}
\sigma_N'\mapsto \sigma_N:\,\alpha\mapsto\text{diag}(a,a),\, \beta\mapsto\text{diag}(b,c). \label{eq:abc}
\end{equation}
Because local unitary operations don't change entanglement, as far as only entanglement and other properties that are invariant under local unitary operations are concerned, regardless of the original $\sigma_N'$, the standard form $\sigma_N$ or the parameters thereof, $a$, $b$, and $c$ will determine the block entanglement. To put in another way, equivalent classes of symmetric Gaussian states are considered, with states in the same standard form regarded as identical. 

\subsection{Multisymmetric Gaussian state}\label{sec:mul}
A natural extension of symmetric Gaussian states is bisymmetric, and multisymmetric Gaussian states, which contains multiple families of modes, and the state is invariant under local permutation between modes within a family \cite{Adesso08,Serafini05,Kao16}. For an $N$-symmetric Gaussian states, its covariance matrix will take this form:
\begin{equation}
\sigma=\begin{pmatrix}
\sigma_{11}&\sigma_{12}&\cdots&\sigma_{1N}\\
\sigma_{21}&\sigma_{22}&\cdots&\sigma_{2N}\\
\vdots     &\cdots     &\ddots&\vdots     \\
\sigma_{N1}&\sigma_{N2}&\cdots&\sigma_{NN}
\end{pmatrix},
\end{equation}
where each $\sigma_{ii}$ is a submatrix like \eqref{eq:symmat} : 
\begin{equation}
\sigma_{ii}=\begin{pmatrix}
\alpha_i & \beta_i  & \cdots & \beta_i \\
\beta_i  & \alpha_i& \beta_i & \vdots \\
\vdots & \beta_i & \ddots  & \beta_i \\
\beta_i  & \cdots & \beta_i & \alpha_i
\end{pmatrix},\, i=1,2,\cdots,N,
\end{equation}
and 
\begin{equation}
\sigma_{ij}=\begin{pmatrix}
\beta_{ij} & \cdots  & \beta_{ij} \\
\vdots     & \ddots  & \vdots \\
\beta_{ij}  & \cdots & \beta_{ij}
\end{pmatrix},i,j=1,2,\cdots,N,\,i\neq j.
\end{equation}
$\alpha_i$, $\beta_i$ and $\beta_{ij}$ are $2\times 2$ symmetric real matrices. Note $\sigma_{ij}=\sigma_{ji}$. As discussed previously, we can assume each $\sigma_{ii}$ is in the standard form \eqref{eq:abc}:
\begin{equation}
\alpha_{ii}=\text{diag}(a_i,a_i),\,\beta_i=\text{diag}(b_i,c_i).
\end{equation}

\section{Metaplectic Transform on States}\label{sec:or}
\subsection{Motive}
\subsubsection{Symplectic transform on the phase space}
By Theorem~\ref{thm:meta}, a symplectic transform $S$ on the phase space corresponds to a (or two) unitary operator $\hat{S}$, and given a Gaussian density operator $\rho$, we have $W_\rho \circ S^{-1}\weyl \hat{S}\rho \hat{S}^{-1}$, where $\hat{S}\rho \hat{S}^{-1}$ is still a density operator. \eqref{eq:Gau} becomes
\begin{align}
W_{\hat{S}\rho \hat{S}^{-1}}(\xi)=W_\rho(S^{-1}\xi)&=\mathcal{N}\exp\left(-\frac{1}{2}(S^{-1}\xi)\tr \sigma^{-1} S^{-1}\xi \right)\nonumber\\
&=\mathcal{N}\exp\left(-\frac{1}{2}\xi\tr (S\sigma^{-1}S\tr)^{-1} \xi \right).\label{eq:Ws}
\end{align}
Hence, after a metaplectic transform $\hat{S}\in Mp(n)$, the state is still Gaussian, with a covariance matrix $S\sigma S\tr$. One thing to note here is $\det S=1$ for any $S\in Sp(n)$, so it's measure/volume-preserving, and $\int W_\rho \,d\boldsymbol{\xi}=\int W_{\hat{S}\rho \hat{S}^{-1}}\,d\boldsymbol{\xi}=1.$ The ``standard form'' of a symmetric Gaussian state, \eqref{eq:abc}, is an example of this technique. 

According to \eqref{eq:Ws}, we can alter the phase-space Gaussian to our liking, as long as the transform is symplectic. So the problem is, what symplectic transform can ease our job? The Williamson normal form is clearly one of the candidates, but if there are extra structures on the Gaussian states, e.g. symmetry, then there may exist desirable symplectic transforms to simplify the task.

In the argument above we've assumed $\rho$ is a density operator, but this holds true for operators whose Weyl symbols are Gaussian. In particular, after the reversal of momentum coordinates/partial transposition, the operator $\rho^\Gamma$ may become non-positive but the phase space function would still be Gaussian. We can perform symplectic transforms on the phase space, and the corresponding operator will be unitarily equivalent to $\rho^\Gamma$.  

Here's a remark: By Theorem~\ref{thm:rv}, the covariance matrix $\sigma$ always becomes $S\sigma^{-1}S\tr$ under the symplectic transform $S$, whether the phase-space function is Gaussian or not, but being Gaussian ensures that the state will remain Gaussian after such a transform.

\subsubsection{Local transform}
If the metaplectic transform $\hat{S}$ is limited to $m$ modes, collectively called A with the rest of the $n$ modes B, i.e. $\hat{S}=\hat{S}_A\otimes I_B$, then the corresponding symplectic transform is of the form $S=S_A\oplus I_B$ (Section~\ref{sec:dip}). By the same token, if $\hat{S}=\hat{S}_A\otimes \hat{S}_B$, then $\pi_{Mp}(\hat{S})=S_A\oplus S_B$, where $\pi_{Mp}(\hat{S}_A)=S_A$ and  $\pi_{Mp}(\hat{S}_B)=S_B$.

Here's an informal argument: Let $f\in L^2(\mathbb{R}^m)$ and $g\in L^2(\mathbb{R}^n)$, and choose a unitary operator $U=\hat{S}_A\otimes I_B$. The Wigner transform of $f\otimes g$, which is exactly the Weyl symbol of $\ket{f\otimes g}\bra{f\otimes g}$, is $W_f\otimes W_g$. Note $W_f\otimes W_g(z_1,z_2)$ is identified as $W_f(z_1)W_g(z_2)$; see Section~\ref{sec:tf}. On the other hand we can obtain 
\begin{equation}
U \ket{f\otimes g}\bra{f\otimes g}U^{-1}=\ket{\hat{S}_A f\otimes g}\bra{\hat{S}_A f\otimes g}\weyl (W_f\circ S_A^{-1})\otimes W_g.
\end{equation}
This suggests on the operator $\ket{f\otimes g}\bra{f\otimes g}$, $\pi_{Mp}(\hat{S}_A\otimes I_B)$ acts like $\hat{S}_A\otimes I_B$; the same can be applied to operators of the form $\ket{f_1\otimes g_1}\bra{f_2\otimes g_2}$ (Moyal transform). As $L^2(\mathbb{R}^{m+n})=L^2(\mathbb{R}^m)\hat{\otimes}L^2(\mathbb{R}^n)$, we can expect this to be true for  $\ket{f}\bra{g}$, where $f,g\in L^2(\mathbb{R}^{m+n})$. Since a trace-class operator has a discrete spectrum, c.f. \eqref{eq:rhoij}, this should holds true for all trace-class operators on $L^2(\mathbb{R}^{m+n})$. This of course needs a more concrete justification by looking at the structure of the metaplectic group, but we'll leave it as it is.

This permits us to perform a local symplectic transform on the phase space, with the corresponding metaplectic transform(s) also being local. As entanglement (measure) is invariant under local unitary transforms, we're able to adjust the Gaussian state without changing the entanglement.

\subsection{Symmetric State}
The transform $S$ we're going to perform on the phase space is one where we apply orthogonal transforms to positions and momenta separately, i.e. $S=O_{\mathbf{x}}\oplus O_{\mathbf{p}}$, and which is symplectic by picking suitable $O_{\mathbf{x}}$ and $O_{\mathbf{p}}$.

Let's consider an $n$-symmetric state whose covariance matrix is of the form \eqref{eq:symmat}, $\sigma_n$. First, the covariances (and variances) of $\mathbf{x}=(x_1,x_2,\dotsc,x_n)$ constitute an ``$\mathbf{x}$-covariance'' matrix, defined as $\sigma_{ij}^\mathbf{x}:=\sigma(x_i,x_j)$, which has the same diagonal elements and the same off-diagonal ones due to symmetry. If we want to diagonalize it, we can utilize the fact that for an $n\times n$ matrix $M_{ij}=\epsilon+\delta_{ij}(\lambda-\epsilon)$, its eigenvalues are $\lambda+(n-1)\epsilon$ and $\lambda-\epsilon$, the latter having $(n-1)$-fold degeneracy.\footnote{Please read Appendix~\ref{app:spe} for a proof.} The eigenvector corresponding to the non-degenerate eigenvalue is $(1,1,\dotsc,1)$. That is, we can introduce these coordinates:
\begin{equation}
\mathbf{X}:=(X_n=\frac{1}{\sqrt{n}}\sum_{i=1}^{n}x_i,u_2,u_3,\dotsc,u_n)=O_\mathbf{x}\mathbf{x},\label{eq:Xn}
\end{equation}
where $O_\mathbf{x}$ is an orthogonal matrix that diagonalizes the $\mathbf{x}$-covariance matrix.

For $S=O_{\mathbf{x}}\oplus O_{\mathbf{p}}$ to be symplectic, we require the transformation for $\mathbf{p}=(p_1,p_2,\dotsc,p_n)$ to be
\begin{equation}
\mathbf{P}=(P_n=\frac{1}{\sqrt{n}}\sum_{i=1}^{n}p_i,\Pi_2,\Pi_3,\dotsc,\Pi_n)=O_{\mathbf{p}}\mathbf{p},\label{eq:Pn}
\end{equation}
where 
\begin{equation}
O_{\mathbf{x}}=O_{\mathbf{p}}:=O
\end{equation}
This transformation can also diagonalize the $\mathbf{p}$-covariance matrix $\sigma_{ij}^\mathbf{p}:=\sigma(p_i,p_j)$ for the same reason. $O_{\mathbf{x}}\oplus O_{\mathbf{p}}$ is indeed a symplectic transform; see e.g. Chapter 2.1 of Ref.~ \cite{deGosson}.

These transforms imply
\begin{equation}
\sigma(X_n,\Pi_i)=\sigma(u_i,P_n)=0,\,i\neq j,
\end{equation}
and \cite{Kao16}:
\begin{align}
\sigma(X_n,X_n)&=\sigma(x_i,x_i)+(n-1)\sigma(x_i,x_j),\nonumber\\
\sigma(P_n,P_n)&=\sigma(p_i,p_i)+(n-1)\sigma(p_i,p_j),\nonumber\\
\sigma(u_i,u_i)&=\sigma(x_i,x_i)-\sigma(x_i,x_j),\nonumber\\
\sigma(\Pi,\Pi)&=\sigma(p_i,p_i)-\sigma(p_i,p_j),\,i\neq j,\label{eq:sig}
\end{align}
the details of which is shown in Appendix~\ref{app:covaiance}. \emph{$\sigma(u_i,u_i)$ and $\sigma(\Pi_i,\Pi_i)$ doesn't depend on $n$ and the index $i$ at all}, so from now on we will ignore the subscripts of them, labeling them as $u$ and $\Pi$ for a symmetric state.  Note $\sigma_n$ doesn't need to assume the standard form \eqref{eq:abc} for the results above to apply.  We have effectively found a symplectic transformation that diagonalizes the covariance matrix $\sigma_n$.

\subsection{Multisymmetric State}
With a covariance matrix given in Section~\ref{sec:mul}, we can perform symplectic transforms like in the previous subsection on each family of symmetric modes. If the $i$-th family has $N_i$ modes,
\begin{equation}
\mathbf{X}_i=O_i\mathbf{x}_i \text{ and } \mathbf{P}_i=O_i\mathbf{p}_i,
\end{equation}
where $O_i$ are orthogonal matrices with the first rows being $(1,1,\cdots,1)/\sqrt{N_i}$ and
\begin{equation}
\mathbf{X}_i=(X_i, u_{i,2},\cdots,u_{i,N_i}) \text{ and }\mathbf{P}_i=(P_i, \Pi_{i,2},\cdots,\Pi_{i,N_i}).
\end{equation}

After such a transform, the correlations between two families $i$ and $j$ will be ``concentrated'' to $(X_i,P_i,X_j,P_j)$, with the covariances of other modes/coordinates between the families totally disappearing. For the proof please again refer to Appendix~\ref{app:covaiance}. With this transform, the covariance matrix of a multisymmetric state is greatly simplified. This concentration of correlation into two modes by applying local metaplectic transforms is sometimes called ``unitary localization'' \cite{Adesso04_2,Serafini05,Adesso08}.

\subsection{Blocks of Modes from a Symmetric Gaussian State}
Now let's consider two blocks from a symmetric Gaussian state that has $N$ modes in total, and the two blocks contain $n_1$ and $n_2$ modes. As the state of each block is still symmetric, we can employ the same method as in the last subsection, by treating each block as a family of symmetric modes.  Because of symmetry, the only quantities to differentiate the blocks are the numbers of modes, and we'll mark the blocks with $n_1$ and $n_2$.

Because $\sigma(u,u)$ and $\sigma(\Pi,\Pi)$ has no dependency on $n$, as shown in \eqref{eq:sig} or Appendix~\ref{app:covaiance}, whether a particular $\sigma(u,u)$ or $\sigma(\Pi,\Pi)$ ``comes from'' either block doesn't matter. After unitary localization, the correlation between the two blocks is accumulated to the two modes $(X_{n_1},P_{n_1})$ and $(X_{n_2},P_{n_2})$. Because the block entanglement is now equivalent to that between these two modes, with other modes being irrelevant, by Simon criterion (Section~\ref{sec:eG}) having a positive partial transpose implies separability and vice versa, which gives weight to using negativities as the entanglement measure \cite{Adesso04_2,Serafini05}. 

To obtain the covariance matrix of the unitarily localized modes, while taking into account these two blocks are part of $N$ symmetric modes, we can make use of \eqref{eq:sig}:
\begin{equation}
\sigma_{N:n_1|n_2}=\begin{pmatrix}
\alpha_{N:n_i} & \beta_{N:n_1|n_2}\\
\beta_{N:n_1|n_2} & \alpha_{N:n_i}
\end{pmatrix},\label{eq:sigma}
\end{equation}
where $N$ and $n_i$ refer to the numbers of modes, and  
\begin{align}
\alpha_{N:n_i}=&\text{diag}\left[\sigma(X_{n_i},X_{n_i}),\sigma(P_{n_i},P_{n_i})\right]\nonumber\\
=&\frac{1}{N} \text{diag}\big[n_i \sigma(X_N,X_N)+(N-n_i)\sigma(u,u),n_i \sigma(P_N,P_N)+(N-n_i)\sigma(\Pi,\Pi)\big],\label{eq:sigman}\\
\beta_{N:n_1|n_2}=&\text{diag}\left[\sigma(X_{n_1},X_{n_2}),\sigma(P_{n_1},P_{n_2})\right]\nonumber\\
=&\frac{\sqrt{n_1 n_2}}{N}\text{diag}\big[\sigma(X_N,X_N)-\sigma(u,u),\sigma(P_N,P_N)-\sigma(\Pi,\Pi)\big],\label{eq:sigmade}
\end{align}
by replacing $\sigma(x_i,x_i)$ and $\sigma(x_i,x_j)$ with $\sigma(X_N,X_N)$ and $\sigma(u,u)$, and replacing $\sigma(p_i,p_i)$ and $\sigma(p_i,p_j)$ with $\sigma(P_N,P_N)$ and $\sigma(\Pi,\Pi)$. $(X_N,P_N)$ are the modes after we unitarily localize all the $N$ modes of the state, and that's part of the scheme of how we take the whole state, rather than the two blocks alone, into account. \eqref{eq:sigman} and \eqref{eq:sigmade} will diagonal if we choose the standard form \eqref{eq:abc}.

\chapter{Supremum of Block Entanglement for Symmetric Gaussian States}
\label{ch:SGS}

\section{Block Entanglement of a Symmetric Gaussian State}
\subsection{Parameters for the Entanglement}\label{sec:par}
We are interested in the block entanglement of an $N$-mode symmetric Gaussian state. Clearly, the block entanglement of an $N$-mode symmetric Gaussian state can be (along with numbers of modes) determined by $(a,b,c)$ of \eqref{eq:abc}, which three parameters characterize the state (in the standard form). Yet according to \eqref{eq:sigman} and \eqref{eq:sigmade}, it can be determined also by the following four variances: $\sigma(X_N,X_N)$, $\sigma(P_N,P_N)$, $\sigma(u,u)$ and $\sigma(\Pi,\Pi)$, but now with one more variable. To eliminate it, we introduce the ratio:
\begin{equation}
r:= \sigma(X_N,X_N)/\sigma(u,u)\label{eq:r}
\end{equation}
and
\begin{equation}
\nu_D:= \sqrt{\sigma(u,u)\sigma(\Pi,\Pi)},\,\nu_N:=\sqrt{\sigma(X_N,X_N)\sigma(P_N,P_N)}.\label{eq:ci}
\end{equation}
$\nu_D$ and $\nu_N$ are actually the symplectic eigenvalues for the symmetric covariance matrix $\sigma_N$ \eqref{eq:abc}, with $\nu_D$ being degenerate, and the $N$ of $\nu_N$ for its $N$-dependence, c.f. \eqref{eq:sig}. Hence, the formulation from \eqref{eq:Xn} to \eqref{eq:ci} is just the other side of the coin of the work done by Adesso, Serafini and Illuminati \cite{Adesso04_2,Serafini05}, who found the symplectic eigenvalues \eqref{eq:ci}, but here we have one additional parameter \eqref{eq:r} that allows us to describe a symmetric Gaussian state with these ``global'' parameters. Using \eqref{eq:sig}, $(a,b,c)$ of \eqref{eq:abc} are related to $(\nu_D,\nu_N,r)$ by, c.f. \eqref{eq:symmat} and \eqref{eq:abc}:
\begin{align}
\det\alpha&=a^2=\frac{\nu_N^2(N+r-1)+\nu_D^2 r(N-1)(N+r-1)}{N^2 r}=\left(\frac{\hbar}{2\mu_1}\right)^2,\nonumber\\
\det\beta&=bc=\frac{(1-r)(\nu_D^2r-\nu_N^2)}{N^2 r},\nonumber\\
\det\sigma_2&=\det\begin{pmatrix}
\alpha & \beta\\
\beta & \alpha
\end{pmatrix}=(a^2-b^2)(a^2-c^2)=\frac{\nu_D^2(N+2r-2)\left[2\nu_N^2+\nu_D^2 r(N-2)\right]}{N^2 r}=\left(\frac{\hbar}{2}\right)^4\frac{1}{\mu_2^2},\label{eq:local}
\end{align}
where $\mu_1$ and $\mu_2$ are the one- and two-mode purities \cite{Serafini04}.

\emph{The advantage of applying the parameters $(\nu_D,\nu_N,r)$ over $(a,b,c)$ is that $\nu_D\geq \hbar/2$, $\nu_N\geq \hbar/2$ and $r>0$ are necessary conditions for the the state to be valid} by Theorem~\ref{thm:valid}, therefore giving us a simpler constraint. To see that this indeed the case, after we execute the symplectic/orthogonal transform \eqref{eq:abc} becomes\footnote{The covariances $\sigma(\mathbf{X}_j,\mathbf{P}_k)$ completely vanish in the standard form; see Appendix~\ref{app:covaiance}.}
\begin{equation}
\text{diag}\left[
\sigma(X_N,X_N),\sigma(P_N,P_N),\sigma(u,u),\sigma(\Pi,\Pi),\sigma(u,u),\sigma(\Pi,\Pi),\cdots,\sigma(u,u),\sigma(\Pi,\Pi)\right].
\end{equation}
For each mode has to obey the uncertainty relation or Theorem~\ref{thm:valid}, $\sqrt{\sigma(q,q)\sigma(p,p)}\geq \hbar/2$, leading to $\nu_D\geq\hbar/2$ and $\nu_N\geq\hbar/2.$ $r>0$ simply reflects the fact that each variance should be positive, or has the same sign. Were we to define $r$ as $\sigma(P_N,P_N)/\sigma(\Pi,\Pi)$, we could obtain the same final result (in the next chapter). 

Since the ratio $\nu_N/\nu_D$ will prove important later, we replace $\nu_N$ with
\begin{equation}
\gamma:= \nu_N/\nu_D,\label{eq:gamma}
\end{equation}
which satisfies $\gamma\geq\hbar/(2\nu_D)$ because $\nu_N=\nu_D\gamma\geq\hbar/2,$ and we'll use the three parameters $(\nu_D,\gamma,r)$ from now on, replacing $(a,b,c)$ of \eqref{eq:abc}.

\subsection{Block Entanglement for a Symmetric Gaussian State}
Because the (indexes of the) blocks are interchangeable, the entanglement will depend on the sum and difference of $n_1$ and $n_2$:
\begin{equation}
n_s:= n_1+n_2\geq 2, \, n_d:= |n_1-n_2|\geq 0,\label{eq:nsnd}
\end{equation}
satisfying $N\geq n_s>n_d$. It's worth pointing out that only the absolute value of $n_1-n_2$ matters, not its sign, as implied by the interchangeability. By \eqref{eq:num} the smaller symplectic eigenvalue of $\widetilde{\sigma}_{N:n_1|n_2}$, partial transpose of \eqref{eq:sigma}, is a function of $(\nu_D,\gamma,r)$:
\begin{equation}
(\widetilde{\nu}^-)^2= f_{N:n_s,n_d}(\nu_D,\gamma,r):=\left(\widetilde{\Delta}_{N:n_1|n_2}-\sqrt{(\widetilde{\Delta}_{N:n_1|n_2})^2-4\det\sigma_{N:n_1|n_2}}\right)/2,\label{eq:f}
\end{equation}
where $\Delta$ is defined in \eqref{eq:delta} and $n_s,n_d$ after the colon refers to the mode numbers of the unitarily localized blocks \eqref{eq:nsnd}. That $f$ in \eqref{eq:f} can indeed be a function of $(\nu_D,\gamma,r)$ validates $\nu_D,\gamma,r$ as the variables of block entanglement.

The domain of $f_{N:n_s,n_d}$ is
\begin{equation}
\text{dom}(f_{N:n_s,n_d})=\{(\nu_D,\gamma,r):\nu_D\geq\hbar/2,\nu_D\gamma\geq\hbar/2,r>0\},\label{eq:domain}
\end{equation}
which is defined by the set of all valid Gaussian symmetric states,\footnote{To be more accurate, Gaussian symmetric states which are considered equivalent by us and which have the same entanglement.} as discussed in the last subsection. By \eqref{eq:LG} the logarithmic negativity between two blocks is
\begin{equation}
E_{L}^{N:n_1|n_2}=E_{L}^{N:n_s,n_d}=\log \max\left(\frac{\hbar}{2\widetilde{\nu}^-},1\right)=\frac{1}{2}\log\max\left(\frac{\hbar^2}{4f_{N:n_s,n_d}},1\right).\label{eq:L}
\end{equation}

Let's summarize what we have done so far: Any symmetric Gaussian state of the form \eqref{eq:symmat} can be turned into the "standard form" by local unitary operations, with parameters given in \eqref{eq:abc}. After orthogonal/symplectic transforms and with a proper choice of parameters, the logarithmic negativity between two blocks \eqref{eq:L} becomes a function with a clear domain \eqref{eq:domain}. We will show how this parametrization will help us find the suprema of negativities. 

\section{Search for the Supremum/infimum}\label{sec:sear}
By \eqref{eq:L} the supremum (Definition~\ref{def:sup}) of $E_{L}$, if it exists, corresponds to the infimum (Definition~\ref{def:sup}) of $f_{N:n_s,n_d}$. Hence our goal is to find this infimum. Let's start with
\begin{align}
f_{N:n_s,n_d}(\nu_D,&\gamma,r)=\frac{\nu_D^2}{2N^2 r}\Big\{\left(Nn_s-n_d^2\right)r^2+\left[2N^2-2Nn_s+n_d^2(1+\gamma^2)\right]r+\left(Nn_s-n_d^2\right)\gamma^2\nonumber\\
&-\Big\{\left(Nn_s-n_d\right)^2r^4-2n_d^2\left[2N^2+n_d^2(1+\gamma^2)-Nn_s\left(3+\gamma^2\right)\right]r^3+\big[N^2\left(4n_d^2(1+\gamma^2)-2n_s^2\gamma^2\right)\nonumber\\
&+n_d^4\left(\gamma^4+4\gamma^2+1\right)-4Nn_sn_d^2\left(1+2\gamma^2\right)\big]r^2-2n_d^2\left[2N^2+n_d^2\left(1+\gamma^2\right)-Nn_s\left(3+\gamma^2\right)\right]\gamma^2r\nonumber\\
&+\left(Nn_s-n_d^2\right)^2\gamma^4\Big\}^{1/2}\Big\},\label{eq:ff}
\end{align}
continuous everywhere in the domain \eqref{eq:domain}. With three variables, first let's vary $r$ only while keeping $\nu_D$ and $\gamma$ fixed, by defining
\begin{equation}
F_{\nu_D,\gamma}^{N:n_s,n_d}(r):= f_{N:n_s,n_d}(\nu_D,\gamma,r)
\end{equation}
as a function of $r$ only. Henceforth we'll ignore the superscripts and subscripts of $f$ and $F$. $F$ is of this form:
\begin{equation}
F(r)=\frac{\nu_D^2}{N^2}\frac{p(r)-\sqrt{h(r)}}{r},\label{eq:F}
\end{equation}
where $p$ and $h$ are polynomials of second and forth orders respectively, whose coefficients are made up of $(\gamma,N,n_s,n_d)$, \emph{without} $\nu_D.$ Because $F(r)$ is a differentiable function of $r\in(0,\infty)$, the global minimum or infimum lies on the boundary or where the derivative vanishes. Let us first examine the behavior of $F(r)$ at the boundary. 
\subsection{Boundary}
Using L'H\^{o}pital's rule, the limits of $F$ as $r$ approaches the boundary $0$ and $\infty$ are
\begin{equation}
\lim_{r \rightarrow 0}F(r)=\lim_{r\rightarrow \infty}F(r)=\nu_D^2\frac{Nn_s-n_s^2}{N n_s-n_d^2},\label{eq:f0}
\end{equation}
\emph{independent of} $\gamma$, or $\nu_N$. Note it's non-negative and smaller than $\nu_D^2$ because $N\geq n_s>n_d$. 
\subsection{Critical Point}
Next, we shall investigate where $F'(r)=0$. To find it we need
\begin{lem}
For a function 
\begin{equation}
y(x)=\frac{p(x)-\sqrt{h(x)}}{x},\label{eq:y}
\end{equation}
where $p$ and $h$ are second- and fourth-order polynomials:
\begin{equation}
p(x)=\sum_{i=0}^{2}a_i x^i,\,h(x)=\sum_{i=0}^{4}b_i x^i,
\end{equation}
if not only 
$\lim_{x \rightarrow 0} y(x)$ and $\lim_{x \rightarrow \infty}y(x)$ exist but also 
\begin{equation}
\lim_{x \rightarrow 0}y(x)=\lim_{x \rightarrow \infty}y(x),\label{eq:lim}
\end{equation}
then the solutions of $y'(x)=0$ are $\pm\sqrt{a_0/a_2}.$ 
\end{lem}

\begin{proof}
The existence of $\lim_{x \rightarrow 0} y(x)$ and $\lim_{x \rightarrow \infty}y(x)$ implies
\begin{equation}
a_0=\sqrt{b_0},\, a_2=\sqrt{b_4}.\label{eq:ai}
\end{equation}
Applying L'H\^{o}pital's rule again, the equivalence of these limits implies
\begin{equation}
a_1-\frac{b_1}{2\sqrt{b_0}}=a_1-\frac{b_3}{2\sqrt{b_4}}\Rightarrow \frac{b_1}{\sqrt{b_0}}=\frac{b_3}{\sqrt{b_4}}.\label{eq:bi}
\end{equation}
$y'(x)=0$ induces
\begin{equation}
4h(p'x-p)^2-(h'x-2h)^2=0.
\end{equation}
With \eqref{eq:ai} and \eqref{eq:bi}, we obtain
\begin{equation}
x^2(a_2x^2-a_0)^2=0,
\end{equation}
whose roots are $\pm\sqrt{a_0/a_2}$.
\end{proof}

The function $F$, \eqref{eq:F}, conforms to the form \eqref{eq:y}. Because $\lim_{r \rightarrow 0}F(r)=\lim_{r \rightarrow \infty}F(r)$, \eqref{eq:f0}, it satisfies the condition \eqref{eq:lim}. The square root of the ratio is exactly $\gamma$, so the critical point is at $r=\gamma.$ In addition, $F(0)=F(\infty)$ (which will denote $\lim_{r \rightarrow 0}F(r)$ and $\lim_{r \rightarrow \infty}F(r)$ hereafter), differentiability of $F(r)$, and the existence of only one critical point imply that $F(\gamma)$ is either a global maximum or minimum, and that $F(0)=F(\infty)$ is either the supremum or infimum; that is,
\begin{equation}
\inf\mathrm{ran}F=\min\left(F(\gamma),F(0)=F\infty\right),
\end{equation}
where $\mathrm{ran}$ denotes the range (Section~\ref{sec:map}).
\subsection{Value at the Critical Point}
\label{sec:Vcri}
However, instead of directly comparing $F(0)=F(\infty)$ and $F(\gamma)$, let's see how small $F(\gamma)$ can be, because if $F(\gamma)\geq\hbar^2/4$, then the entanglement measure would be zero by \eqref{eq:L}; namely, this point is of no interest to us at all.  From \eqref{eq:ff},
\begin{equation}
F(\gamma)=\frac{\nu_D^2}{2N^2}\Bigg[2N^2+2Nn_s(\gamma-1)+n_d^2(\gamma-1)^2
-n_d|\gamma-1|\sqrt{4N^2+4Nn_s(\gamma-1)+n_d^2(\gamma-1)^2}\Bigg].\label{eq:Fgamma}
\end{equation}
Let's define
\begin{equation}
g(\gamma):= F(r=\gamma)\label{eq:g}
\end{equation}
\emph{as a continuous function of} $\gamma$, and
\begin{equation}
q:= \hbar/(2\nu_D)\in(0,1].
\end{equation}
Because $\nu_N=\nu_D\gamma\geq \hbar/2$ by \eqref{eq:gamma}, the domain of $g$ is
\begin{equation}
\text{dom}(g)=[q,\infty).
\end{equation}
Because there's a term $|\gamma-1|$ in \eqref{eq:Fgamma}, we consider two cases separately for $g(\gamma)$, i.e.$\gamma> 1$ and $1\geq\gamma\geq q$:

\subsubsection{$\gamma> 1$} 
Defining $x:= \gamma-1> 0$, let's solve 
\begin{equation}
g(\gamma=x+1)< \hbar^2/4,
\end{equation}
which becomes
\begin{equation}
\frac{2N^2+2x N n_s+x n_d\left(x n_d-\sqrt{n_d^2x^2+4N(N+n_s x)}\right)}{2N^2}<\frac{\hbar^2}{4\nu_D}=q^2.
\end{equation}
This is equivalent to
\begin{equation}
2N^2(1-q^2)+n_d^2 x^2+2N n_s x<n_d x\sqrt{n_d^2 x^2+4N(N+n_s x)}.
\end{equation}
Since both sides are positive for $q\in (0,1]$, this leads to
\begin{equation*}
\left(2N^2+n_d^2 x^2+2N n_s x-2N^2 q^2\right)^2<(n_d x)^2\left[n_d^2 x^2+4N(N+n_s x)\right],
\end{equation*}
\begin{equation}
\Rightarrow N^2(n_s^2-n_d^2q^2)x^2+2N^3n_s(1-q^2)x+N^4(1-q^2)^2=c_2 x^2+c_1 x+c_0<0,
\end{equation}
a quadratic inequality. By $n_s>n_d$ (see \eqref{eq:nsnd}) and $q\in (0,1]$, we obtain $c_2>0$, $c_1\geq 0$ and $c_0\geq 0.$ 
Hence, $x$ has no positive solutions.
\subsubsection{$1\geq\gamma\geq q$} 
Now consider $\gamma$ in the interval $[q,1]$. Define $x:= 1-\gamma\in [0,1-q]$. Following pretty much the same procedure as above, the solution of $g(\gamma=1-x)< \hbar^2/4$ is the interval
\begin{equation}
x\in(\frac{N(1-q^2)}{n_s+n_d q},\frac{N(1-q^2)}{n_s-n_d q}).
\end{equation}
The greatest lower bound of the interval above is however no smaller than $1-q$:
\begin{equation}
\frac{N(1-q^2)}{n_s+n_d q}=\frac{N(1+q)}{n_s+n_d q}(1-q)\geq 1-q
\end{equation}
by $N\geq n_s> n_d\geq 0$ and $q\in (0,1]$. As we've required $x\in [0,1-q]$, no solution exists.

Therefore, $g(\gamma)$ \eqref{eq:g} is no smaller than $h^2/4$ on its whole domain $\gamma\in[q,\infty)$, so if $\inf \text{ran}F=\min (F(\gamma),F(0)=F(\infty))$ is smaller than $\hbar^2/4,$ $\inf F$ is definitely $F(0)=F(\infty).$

\section{The Suprema}
Here's the main result of this chapter:
\begin{pro}\label{pro:sup}
If there are $N$ modes at a symmetric Gaussian state, for two blocks containing $n_1$ and $n_2$ modes from the $N$ modes the supremum of logarithmic negativity between them is\footnote{Here $\sup$ refers to the supremum of the image, i.e. $\sup\mathrm{ran}$, or the supremum of the entanglement measure over all symmetric Gaussian states given the constraint.}
\begin{equation}
\sup E_{L}^{N:n_s,n_d}=\frac{1}{2}\log\left(1+\frac{n_s^2-n_d^2}{n_s(N-n_s)}\right),\label{eq:sup2}
\end{equation}
and that of negativity is
\begin{equation}
\sup E_{N}^{N:n_s,n_d}=\frac{1}{2}\left[\left(1+\frac{n_s^2-n_d^2}{n_s(N-n_s)}\right)^{1/2}-1\right].\label{eq:sup3}
\end{equation}

If instead the degenerate symplectic eigenvalue $\nu_D$ \eqref{eq:ci} is given, 
\begin{equation}
\sup E_{L:\nu_D}^{N:n_s,n_d}=\max\left[\frac{1}{2}\log\frac{\hbar^2}{4\nu_D^2}\left(1+\frac{n_s^2-n_d^2}{n_s(N-n_s)}\right),0\right],\label{eq:sup1}
\end{equation}
and
\begin{equation}
\sup E_{N:\nu_D}^{N:n_s,n_d}=\max\left\{\frac{1}{2}\left[\frac{\hbar}{2\nu_D}\left(1+\frac{n_s^2-n_d^2}{n_s(N-n_s)}\right)^{1/2}-1\right],0\right\}.\label{eq:sup4}
\end{equation} 
\end{pro}

Remark that the former two are positive, while the latter two are nonnegative. Furthermore, by intermediate value theorem\cite{DuistermaatI,Kantorovitz}, with two blocks at given $(N,n_s,n_d)$ and perhaps $\nu_D$, for all elements in the following interval:
\begin{equation}
[0,\sup E),\label{eq:int}
\end{equation}
where $\sup E$ may be \eqref{eq:sup2}$\sim$\eqref{eq:sup4} depending on whether $\nu_D$ is fixed and which measure is adopted, we can always find a corresponding symmetric Gaussian state, unless $\sup E=0$ (which can only happen with fixed $\nu_D$, \eqref{eq:sup1} and \eqref{eq:sup4}), in which case the entanglement measure can only be zero. In other words, \eqref{eq:int} is the range of the mapping that is measuring the block entanglement. Note these suprema aren't maxima, so the interval \eqref{eq:int} is not closed.

\begin{proof}
Recall the entanglement (measure) is a function of $f$ or $F$, \eqref{eq:L}. In this proof we will assume $\inf \mathrm{ran}F< \hbar^2/4$, so $F(0)=F(\infty)$ is always the infimum of $\mathrm{ran}F$. If not, the entanglement measure will be zero, so either way it doesn't matter much, but it will save a few words in the discussion below. 

First, let's fix the variable $\nu_D$. Since at a given $\gamma$, $f$ \eqref{eq:ff} or $F$ approaches its smallest when $r\rightarrow 0$ and $r\rightarrow \infty$, to find the infimum of $\mathrm{ran}f$ as we vary $\gamma$ we can first look for $\lim_{r\rightarrow\infty}F_{\nu_D,\gamma}(r)=\lim_{r\rightarrow 0}F_{\nu_D,\gamma}(r)$, and see how these two limits vary with $\gamma$. As \eqref{eq:f0} turns out, the limits are independent of $\gamma$ (or $\nu_N$), so $\lim_{r\rightarrow\infty}F_{\nu_D,\gamma}(r)=\lim_{r\rightarrow 0}F_{\nu_D,\gamma}(r)$ is the infimum of $\mathrm{ran}f_{\nu_D}(\gamma,r)$, where we define $f_{\nu_D}(\gamma,r):=f(\nu_D,\gamma,r)$ as a function of $(\gamma,r)$. From \eqref{eq:L} and \eqref{eq:f0}, we can obtain \eqref{eq:sup1}.

Likewise, because the infimum of $f$ over all $(\nu_D,\gamma,r)$ is essentially the infimum of $\lim_{r\rightarrow\infty} F_{\nu_D,\gamma}(r)=\lim_{r\rightarrow 0}F_{\nu_D,\gamma}(r)$ over all $(\nu_D,\gamma)$ and because $\nu_D\geq\hbar/2$, we acquire \eqref{eq:sup2}.

By the relation $ E_{N}=(t^{ E_{L}}-1)/2$, assuming the base is $t$, the suprema for negativity $ E_{N}$ are easily derived.
\end{proof}

\section{Decrease of Entanglement with Increasing $n_d$}
Here's another interesting result:
\begin{pro}\label{pro:de}
Given $N$ modes at a specific symmetric Gaussian state, after fixing $n_s$, $E_{L}^{N:n_s,n_d}$ and $E_{N}^{N:n_s,n_d}$ decrease monotonically with $n_d$ unless $E_{L}^{N:n_s,n_d}$ or $E_{N}^{N:n_s,n_d}$ is zero already. 
\end{pro}
Note this is a general result, not only limited to the suprema. Hence, with everything else being equal, the more even the partition is (less $n_d$), the greater the negativities are, i.e. stronger entanglement. 

\begin{proof}
As shown in Section~\ref{sec:Vcri}, for the function $f$ or $F$, \eqref{eq:ff} or \eqref{eq:F}, because $F(\gamma)\geq\hbar^2/4$ there's a region around $\gamma$ where $F$ is always larger than or equal to $\hbar^2/4$, regardless of the values of other parameters, so in this region the entanglement vanishes, by \eqref{eq:f}. Since $\nu_D\geq \hbar/2,$ if $F(r)\geq \nu_D^2,$ then $F(r)\geq\hbar^2/4.$ Let's first try to solve $F(r)=\nu_D^2.$ It can be easily found out the roots are $r=1$ or $r=\gamma^2.$ Because $F(\gamma)\geq \hbar^2/4$ is the global maximum,\footnote{If it's the global minimum, then $F(r)\geq\hbar^2/4$ at all $r$, which isn't of our concern.} $F(r)\geq\hbar^2/4$ between\footnote{When $\gamma\geq 1$, it's the interval $[1,\gamma^2]$, else it's $[\gamma^2,1]$. Note it's $r$ in this interval implying $F(r)\geq \hbar^2/4$, not the other way around.} $r=1$ and $r=\gamma^2$, that is, we're only interested in this subset of positive $r$:\footnote{$r$ is required to be positive for the state to be valid; see \eqref{eq:domain}.}
\begin{equation}
\mathbb{S}:=\{r:(r-1)(r-\gamma^2)>0\}.\label{eq:S}
\end{equation}
Any other positive $r$ result in zero negativities.

Treating $f$ as a function of $n_d,$ it takes this form
\begin{equation}
\frac{\nu_D^2}{2N^2r}\left(a_2 n_d^2+a_0-\sqrt{b_4n_d^4+b_2n_d^2+b_0}\right),\,\text{with }b_4=a_2^2,\,\text{and }b_2=a_2 a_0.
\end{equation}
Partially differentiate it with respect to $n_d$:
\begin{equation}
\partiald{f}{n_d}=\frac{\nu_D^2n_d}{2N^2r}\left(2a_2-\frac{2b_4n_d^2+b_2}{\sqrt{b_4n_d^4+b_2n_d^2+b_0}}\right).\label{eq:partial}
\end{equation}
We'd like to see whether $\partial f/\partial n_d$ is always larger than zero when $f<\hbar^2/4$, which would imply
\begin{equation}
\partiald{E^{N:n_s,n_d}}{n_d}<0\label{eq:E}
\end{equation}
for negativities. Since 
\begin{equation}
a_2=-(r-1)(r-\gamma^2),
\end{equation}
$a_2$ is always negative for $r\in \mathbb{S}$. On the other hand, we can find
\begin{equation}
2b_4n_d^2+b_2=-2(r-1)(r-\gamma^2)\left[2rN(N-n_s)+(N n_s-n_d^2)(r^2+\gamma^2)+n_d^2r(1+\gamma^2)\right]<0,
\end{equation}
because $N\geq n_s>n_d$ and $r\in \mathbb{S}$. Therefore by \eqref{eq:partial}, if 
\begin{equation}
|2a_2|<\left|\frac{2b_4n_d^2+b_2}{\sqrt{b_4n_d^4+b_2n_d^2+b_0}}\right|,
\end{equation}
then $\partial f/\partial n_d>0.$ We can verify this by squaring and subtraction (also c.f. \eqref{eq:f})
\begin{align}
\left(2b_4n_d^2+b_2\right)^2-(2a_2)^2\left(b_4n_d^4+b_2n_d^2+b_0\right)&=4a_2^2(a_0^2-b_0)\nonumber\\&=16a_2^2N^2r\left[(N-n_s)+n_sr\right]\left[r(N-n_s)+n_s\gamma^2\right]>0,
\end{align} 
because $N\geq n_s.$ Hence with everything else fixed, as $n_d$ increases, the entanglement measure always decreases, unless the entanglement measure is zero already.
\end{proof}

\section{Discussions}
\subsection{Boundedness and Block Sizes}
Suppose we have two blocks, with $n_1$ and $n_2$ modes, from $N$ symmetric modes. As $N$ increases, the supremum of bipartite entanglement between them decreases, unless $n_1+n_2=N.$ What is the reason behind? By increasing the total number of interchangeable modes, we actually impose a stronger constraint on individual modes, i.e. $(a,b,c)$ of \eqref{eq:abc}. This should be clearer with the help of \eqref{eq:sig}: 
In terms of the parameters $(\nu_D,\nu_N,r)$ as defined by \eqref{eq:r} and \eqref{eq:ci}, which  \emph{are interchangeable with} $(a,b,c)$, we can find as $r\rightarrow 0$ or $r\rightarrow \infty$,
\begin{equation}
\sigma(X_{N'},X_{N'})\sigma(P_{N'},P_{N'})\rightarrow\begin{cases}
-\infty & \text{if } N'>N\\
\nu_N & \text{if } N'=N\\
\infty &\text{if } N'<N
\end{cases}.
\end{equation}
This implies that at both limits the symmetric state remains legitimate with $N$ or less total modes, but becomes invalid if there are $N'>N$ symmetric modes in total, violating \eqref{eq:goodstate}; hence we're left with ``less'' choices (still infinite) of $(a,b,c)$ as $N$ grows. Please heed that with the same ``local state,'' or $(a,b,c)$, the entanglement between two blocks of given sizes should be the same, irregardless of how large $N$ is: What $N$ does is restricting the choice of ``local states,'' which in turn restricts entanglement. A valid local state for $N=N_1$ may be invalid for $N=N_2>N_1.$

By Proposition~\ref{pro:de}, for all symmetric Gaussian states, at fixed $n_s$ and $N$ a larger $n_d$ decreases the negativities of block entanglement, which was observed by Serafini et al. \cite{Serafini05}. The suprema from Proposition~\ref{pro:sup} comply with this phenomenon as well: As the blocks become more equal, they become more entangled. This suggests an optimal strategy of ``gathering'' entanglement \cite{Serafini05}.

At the same state, the block entanglement between subsystems should be smaller than or equal to that between the parent systems. This is reflected by the fact that a larger $n_s$ implies larger suprema, which can be checked by partially differentiating them. Furthermore, we can replace $n_d$ and $n_s$ with $n_1$ and $n_2$, and find that these suprema do increase by adding more modes to either block.
\subsection{Monogamy of Entanglement and Multipartite Entanglement}
Monogamy of entanglement \cite{Coffman00,Terhal04}, describing the phenomenon that entanglement between multiple parties can't be shared freely, can be embodied by an inequality first introduced by Coffman, Kundu, and Wooters (CKW) \cite{Coffman00,Adesso06Int}:
\begin{equation}
E^{p_1|(p_2,\dotsc,p_N)}\geq \sum_{i=1}^{N}E^{p_1|p_N},\label{eq:ckw}
\end{equation}
where $p_i$ denotes the $i$-th party, and $E^{p_i|p_j}$ is the bipartite entanglement between party $i$ and party $j$. In the tripartite case, the difference between the left and right hand sides of the inequality above is the residual entanglement, characterizing the tripartite entanglement \cite{Adesso06}. However, not all genuine entanglement measures obey this inequality; instead, a class of entanglement measures has been chosen to satisfy this condition, such as tangle \cite{Coffman00} and cotangle \cite{Adesso06,Hiroshima07}. A stronger form of entanglement monogamy was also proposed \cite{Adesso07PRL,Adesso08}:
\begin{equation}
E^{p_1|(p_2,\dotsc,p_N)}=\sum_{j=2}^{N}E^{p_1|p_j}+\sum_{k>j=2}^{N}E^{p_1|p_j|p_k}+\dots+E^{\underline{p_1}|p_2|\dots|p_N},\label{eq:strmo}
\end{equation}
where all terms except for the first are multipartite entanglements with more than two parties, called the residual entanglements, but the last one isn't ``genuine.'' The genuine residual $N$-partite entanglement $E^{p_1|p_2|\dots|p_N}$ is found by minimizing $E^{\underline{p_{i_1}}|p_{i_2}|\dots|p_{i_N}}$\footnote{$i_2<i_3<\cdots<i_N$, so there's only one symbol for a particular residual entanglement.} over all probing parties $\underline{p_{i_1}}$, $i_1=1,2,\cdots ,N$. In principle with the information of bipartite entanglements any multipartite entanglement can be attained by iteration; please read Ref.~\cite{Adesso07PRL,Adesso08} for details.

Even though negativity and logarithmic negativity, the focus of this work, in general don't obey CKW inequality \cite{Adesso06}, our results show some interesting aspects regarding Gaussian multipartite entanglement: The bipartite block entanglement, with respect to negativities, becomes unbounded when $n_s=N$, bounded otherwise. If we have $m$ blocks from $N$ symmetric modes, with $n_1+\dots+n_m<N$, then since all bipartite entanglement between them is bounded, \eqref{eq:strmo} implies that the $m$-partite entanglement among these $m$ blocks is also bounded. On the other hand, if $n_1+\dots+n_m=N$, because the bipartite entanglement between one and the remaining parties is unbounded, we expect the multipartite entanglement involving all of them to be unbounded as well. In short, if the blocks fill all the symmetric modes, the multipartite entanglement is expected to be unbounded; bounded if not.

A few more words can be also said on the boundedness: Because when $n_1+n_2=N$ the entanglement is unbounded, that the maximal entanglement between blocks of given sizes decreases with increasing $N$ isn't due to a limited amount of entanglement distributed to more parties (by \eqref{eq:ckw}, the monogamy of entanglement), leading to every party gaining less: There's (potentially) infinite entanglement to begin with.

Rigorously speaking, this problem should be treated with a suitable measure that obeys \eqref{eq:ckw}, like cotangle. However, this is beyond our current scope, and we hope this primitive discussion can inspire future works.
\subsection{Available Entanglement}\label{sec:ava}
Because logarithm is monotonic, from Proposition~\ref{pro:sup} how $\sup E$ behaves depends on
\begin{equation}
K(N,n_s,n_d):=\frac{n_s^2-n_d^2}{n_s(N-n_s)}.\label{eq:K}
\end{equation}

For $n_s<N$, to maximize $K$ clearly we need $n_s=N-1$, and\footnote{Here $\max K$ refers to the maximum of the range of $K$ with a fixed $N$.}
\begin{equation}
\max K=\begin{cases}
N-1 & \text{for odd } N \text{ and by choosing }n_d=0\\
N-1-(N-1)^{-1} & \text{for even } N \text{ and by choosing }n_d=1
\end{cases},\label{eq:n-1}
\end{equation}
which clearly becomes larger as $N$ increases. Therefore, excluding the case $n_s=N$, as the total number of symmetric modes $N$ increases, we can actually have more entanglement at out disposal.

By \eqref{eq:sup1} and \eqref{eq:sup4} at a given $\nu_D$ no entanglement can exist between any two blocks if
\begin{equation}
K(N,n_s,n_d)+1\leq \frac{4\nu_D^2}{\hbar^2},
\end{equation}
so by \eqref{eq:n-1} it means when
\begin{equation}
\frac{4\nu_D^2}{\hbar^2}\geq\begin{cases}
N & \text{for odd } N\\
N-(N-1)^{-1} & \text{for even } N
\end{cases},
\end{equation}
two blocks with $n_s<N$ can only be separable; in other words for symmetric modes with $\nu_D$ satisfying the condition above any two blocks are separable unless $n_s=N.$
\subsection{Purities}
If we can gain information of $\nu_D$, then it's possible to lower the supremum, i.e. using \eqref{eq:sup1} instead of \eqref{eq:sup2}. A symmetric Gaussian \emph{pure} state has $\nu_D=\nu_N=\hbar/2$, and for bisymmetric \cite{Serafini05} or multisymmetric Gaussian \emph{pure} states \cite{Adesso08} composed of several families of symmetric modes $\nu_D=\hbar/2$, so for those states the upper bound cannot be lowered. Because the global purity $\mu$ of symmetric Gaussian modes is \cite{Serafini04,Adesso04_2}
\begin{equation}
\mu= \frac{1}{\nu_N \nu_D^{N-1}}\left(\frac{\hbar}{2}\right)^N=\frac{1}{\gamma}(\frac{\hbar}{2\nu_D})^N,\label{eq:mu}
\end{equation}
which fixes the parameter $\gamma.$ However, since the the suprema lie on the boundary $r\rightarrow 0$ and $r\rightarrow \infty$, $\gamma$ doesn't really play a role in determining the values of the suprema. Hence if the global purity is known, which even though can ``shrink'' the domain, the suprema stay the same.

As a comparison, Adesso et al. \cite{Adesso04} found that the entanglement between two modes of a Gaussian state is bounded from below and above if the global (two-mode) and marginal (one-mode) purities are known. However, were we to ascertain the bounds imposed by purities of multiple modes (e.g. global purity) and of one mode for symmetric Gaussian states, it either would constitute a complicated constraint of the domain for the current parameters in use, i.e. $(\nu_D,\gamma,r)$ or $(a,b,c)$ of \eqref{eq:abc}, or while reducing the variables by one, would require a suitable parameter and technique to find the bounds, and hence not the subject of our work at this stage.
\subsection{Approximation and Two Modes}\label{sec:app}
Assuming the base of logarithm is $e$, when $N\gg n_s$ (and thus $N\gg n_d$) \eqref{eq:sup2} approximates
\begin{equation}
\sup E_{L}^{N:n_s,n_d}\simeq \frac{n_s^2-n_d^2}{2n_s N}\simeq \frac{n_s}{2N},\label{eq:approx}
\end{equation}
where the second approximation can be made if $n_d\ll n_s$ or $n_d=0$. By the same token, we acquire $\sup E_{N}^{N:n_1|n_2}\simeq \left(n_s^2-n_d^2\right)/\left(4n_s N\right)$, and $\sup E_{N}^{N:n_1|n_2}\simeq n_s/(4N)$ when $n_d\ll n_s.$ Here scaling with $N$ and $n_s$ is demonstrated in an explicit manner: proportional to $n_s$ and inversely proportional to $N$ when $N\gg n_s$.

Now suppose $n_1=n_2=1$. \eqref{eq:sup2} becomes
\begin{equation}
\sup E_{L}^{N:1|1}=\frac{1}{2}\log\left(1+\frac{2}{N-2}\right),
\end{equation}
as obtained Ref.~\cite{Kao16}, which was proved for a bisymmetric \cite{Serafini05} Gaussian pure state under the assumption that the coefficients of the wave function are real.\footnote{The proof is also provided in Appendix~\ref{app:entbos}, with notation updated.} Taking the base to be $e$,
\begin{equation}
\sup E_{L}^{N:1|1}\simeq \frac{1}{N-2}\simeq\frac{1}{N}.
\end{equation}
$1/N$ usually is a good approximation for large $N$. But if we want an approximation that also serves as an upper bound, then $1/(N-2)$ should be chosen. The reason is
\begin{equation}
\frac{1}{N}<\sup E_{L}^{N:1|1}=\frac{1}{2}\log\left(1+\frac{2}{N-2}\right)<\frac{1}{N-2}.
\end{equation}
The proof of the above inequality is simple. First,
\begin{equation}
\frac{d(1/N)}{dN}=-\frac{1}{N^2},\,\frac{d(\sup E_{L}^{N:1|1})}{dN}=-\frac{1}{N(N-2)},\, \frac{d}{dN}(\frac{1}{N-2})=-\frac{1}{(N-2)^2},
\end{equation}
and thus
\begin{equation}
\frac{d}{dN}(\frac{1}{N}-\sup E_{L}^{N:1|1})>0 \text{ and } \frac{d}{dN}(\sup E_{L}^{N:1|1}-\frac{1}{N-2})>0 \text{ for }N>2,
\end{equation}
meaning that both $\frac{1}{N}-\sup E_{L}^{N:1|1}$ and $\sup E_{L}^{N:1|1}-\frac{1}{N-2}$ monotonically increase for $N>2$. Because
\begin{equation}
\lim_{N\rightarrow\infty}\frac{1}{N}=\lim_{N\rightarrow\infty}\sup E_{L}^{N:1|1}=\lim_{N\rightarrow\infty}\frac{1}{N-2}=0,
\end{equation}
$1/N-\sup E_{L}^{N:1|1}$ and $\sup E_{L}^{N:1|1}-1/(N-2)$ should approaches 0 as $N\rightarrow\infty.$ By their monotonicity, we have
\begin{equation}
\frac{1}{N}-\sup E_{L}^{N:1|1}<0 \text{ and } \sup E_{L}^{N:1|1}-\frac{1}{N-2}<0.
\end{equation}
This completes the proof. Therefore, $E_{L}^{N:1|1}$ may be larger than $1/N$ but will always be smaller than $1/(N-2).$

\section{Conclusion}
Proposition~\ref{pro:sup} shows the least upper bounds of block entanglement for Gaussian state with $N$-symmetric modes in terms of logarithmic negativity and negativity, and if $\nu_D$ of \eqref{eq:ci}, the degenerate symplectic eigenvalue can be known, we have tighter bounds. Such upper bounds originates from the symmetry of the state, and the basic requirement that the state obey the uncertainty relation, satisfying \eqref{eq:goodstate}. These two conditions together impose a constraint on the modes, which becomes more stringent for each individual mode as the total number of modes $N$ increases; that is, with ``local'' parameters like $(a,b,c)$ of \eqref{eq:abc} to describe a symmetric Gaussian state, the domain shrinks as a result of increasing $N$. 

Despite the dwindling domain of local parameters, the possibility for the blocks to swell even more, i.e. larger $n_s=n_1+n_2$ can not only compensate for that, but the achievable maximal block entanglement, or unitarily localizable entanglement can increases as a result of increasing total modes \eqref{eq:n-1}, at the expense of entanglement between blocks of fixed sizes, say $n_1=n_2=1$, c.f. some previous works \cite{Adesso04_2,Serafini05}. Yet at whichever $N$, for two blocks with $n_1+n_2=N$, there exists no upper bound for their block entanglement. In other words, the unitarily localizable entanglement will be wasted if it's not gathered across all the symmetric modes, but if the localization can be done in full, then the total number of symmetric modes makes little difference, in terms of the amount of available resource that is entanglement. Moreover, by the ``strong'' monogamy of entanglement, the same can be anticipated for $m$-partite entanglement, i.e. whether it's bounded depends on if $n_1+n_2+\dots+n_m=N.$

Gaussian states, as a subset of states of an infinite-dimensional space, have an unbounded entanglement (with respect to measures like negativities) if there's no constraint put on them. In this regard, symmetry under mode swapping is a very strong condition, which renders the entanglement bounded. Ascertaining the exact bound was then made possible by choosing global parameters with a clear boundary, \eqref{eq:domain}. 

\chapter{Summary}\label{ch:con}

\section{Entangling Capacity}
In Section~\ref{sec:HSinner}, \ref{sec:HPF} and Chapter~\ref{ch:li}, I introduced the necessary concepts for deriving the results in Chapter~\ref{ch:EC}, in particular, 
\begin{itemize}
\item Hilbert-Schmidt inner product as an inner product on operators,
\item and the associated adjoint of a linear mapping on operators;
\item Choi isomorphism as an isomorphism between linear mappings on operators and tensor product of operators.
\item How to decompose a Hermitian operator as the difference between two positive operators,
\item and similarly, how to decompose an HP mapping as the difference between two CP mappings.
\item Several kinds of linear mappings on operators: CP, TP, HP mappings and quantum operations (CPTP).
\end{itemize}

In Chapter~\ref{ch:EC}, I demonstrated there exist upper bounds of entangling capacity with respect to negativities (Proposition~\ref{pro}), where I showed bounds for entangling capacity of a deterministic operation, average entangling capacity of a probabilistic operation, and lower bounds for a sub-operation, and since the bounds for deterministic operations are less complex, emphasis is put on them. The bounds are related to PPT-ness of an operation, which to a degree this is hardly surprising, because negativity is by definition how a state violates PPT-ness, so we essentially found that non-PPT-ness of an operation has to do with non-PPT-ness of the resultant state.  These bounds, in a different and more limited form, were discovered by Campbell in Ref.~\cite{Campbell10}, and in Section~\ref{sec:up} I explained how these two seemingly distinct bounds are related. 

A norm (or two norms) was defined in Section~\ref{sec:geo}, which quantifies the PPT-ness of an operation or state. It's essentially the ``negativity,'' but by showing it's indeed a norm or metric we're able to give it a geometrical significance. Since the bounds for entangling capacity in Proposition~\ref{pro} have to do with PPT-ness, it implies that the length of an operation bounds how much longer the state can get at most.  Proposition~\ref{pro:2} furthers this concept, by showing the operational importance of the distance between two operations or between two states.

Proposition~\ref{pro:sppt} was actually a by-product of this project: During the proof for Lemma~\ref{lem:com}, I found a unitary operation is PPT if and only if it's separable, and it's very natural to extend this result to pure states.

In Proposition~\ref{pro:ecs} we showed that the upper and lower bounds (taking $\widetilde{S^\Gamma}_\pm=S^\Gamma_\pm$) in Proposition~\ref{pro} become identical if and only if ${S^\Gamma_\pm}^\dagger(I)\propto I$, and under what condition the upper bounds shown in Proposition~\ref{pro} can be saturated. This is particularly useful when ${S^\Gamma_\pm}^\dagger(I)\propto I$, which is obeyed by all unitary operations on $2\otimes 2$ space, so theoretically we can obtain the entangling capacities with respect to negativity of all unitary operations on $2\otimes 2$ space---of course practically this may be very hard to solve. Some examples were prepared in Section~\ref{sec:ex} for showcasing.

\section{Supremum of Block Entanglement for Symmetric Gaussian States}

In Chapter~\ref{ch:gaf} we discussed the techniques used for treating so-called continuous-variable (CV) systems, or density operators on $L^2(\mathbb{R}^n)$ in a more mathematical language . By Weyl quantization (Section~\ref{sec:Weyl}), we transform a density operator on $L^2(\mathbb{R}^n)$ to a real function on the phase space $\mathbb{R}^{2n}$ (Section~\ref{sec:psv}), and the metaplectic transform, a unitary operator on $L^2(\mathbb{R}^n)$, becomes a symplectic transform on the phase space. Therefore we can change the coordinate of a phase-space function symplectically, knowing that the density operators will be unitarily equivalent. A symplectic/orthogonal transform for (multi)symmetric states was then introduced that can greatly simplify the covariance matrix, and it was employed on blocks of modes from a symmetric Gaussian state.

In Chapter~\ref{ch:SGS}, first I brought in three ``global'' parameters for quantifying the entanglement of a symmetric Gaussian state: As the total number of modes $N$ plays an important role, this parameterization are the reasonable choice and would facilitate the derivation. 

We then showed how to derive Proposition~\ref{pro:sup} in Section~\ref{sec:sear}. The reader could observe that it's basically repeated applications of basic concepts from fundamental calculus---Not that it was a particularly easy task, as it wasn't always obvious where to use what, and the equations, if not knowing how to tide them up, could be very messy, and to get anything useful out of them would be difficult, if not nigh impossible, without the aid of computers. What is laid out in Section~\ref{sec:sear} is an optimized procedure.

Proposition~\ref{pro:sup} is one the primary results of this thesis. It shows there exists an upper bound of block entanglement for symmetric Gaussian states, and the lowest upper bound is determined by the number of modes in the two blocks, and how many modes the symmetric Gaussian state has in total. Specifically, it shows being Gaussian and symmetric limits the choice of ``local'' states, and this limit gets stronger as the total number of modes $N$ increases.

Proposition~\ref{pro:sup} also shows that there's no bound if $n_1+n_2=N$; namely, the ``global'' state as a whole doesn't have its block entanglement bounded. In addition, in Section~\ref{sec:ava} it was proved that a higher $N$, i.e. total number of modes, actually heightens the supremum of block entanglement if the block sizes are allowed to increase as well. With fixed block sizes, a higher $N$ results in a lower supremum. 

Finally, in Section~\ref{sec:app} some simple approximations are made, so that how the supremum varies as $N$ and $n_s=n_1+n_2$ can be easily observed. The case $n_1=n_2=1$ is one I observed in my master's thesis \cite{Kao} and proved under extra assumptions in Ref.~\cite{Kao16}. 
\section{Outlook}
Whether the entangling capacity in terms of another measure is attainable depends on how complex the measure is, which other techniques shall be employed or developed to examine. The method developed for finding the upper bounds of entangling capacities in this thesis has potential for studying other related topics of quantum operations. Since the technique is very linear in nature, it has some limitations, but there should be some subjects where it can be applied, for example where trace norm is used for quantifying. As entanglement isn't the only area of interest in quantum information, other important quantities or topics can also be scrutinized. 

The current method for attaining the bound of entangling capacity doesn't always yield an optimal result, in that the upper bound obtained may be far higher than the operation is really capable of. Is it possible to tighten the bounds by using other methods?

Whether a similar approach can be adopted to obtain the entangling capacity/capability in terms of negativity for Gaussian states can be investigated. In Ref.~\cite{Wolf03} metaplectic transforms were already studied, so the next step should be the more general Gaussian operations \cite{Gideke02}.

Given any symmetric Gaussian state, Proposition~\ref{pro:de} shows that evenly-sized blocks have higher negativity. Is this a general phenomenon regardless of entanglement measures, and does it hold true for other types of symmetric systems? This may be an interesting question to study.

The method employed in this work may be helpful in other related issues, where (multi)symmetric Gaussian states are the subject. Some lessons learned, or properties found may be applicable or should be checked when handling other symmetric systems: For example, to characterize the system sometimes it may be better to choose ``global'' parameters. 

When carrying out the study for Ref.~\cite{Kao18}, i.e. Chapter~\ref{ch:gaf} and \ref{ch:SGS}, it was found that there seems to exist an upper bound for block entanglement from different families of symmetric modes too. This may deserve further research.

\appendix

\chapter{Equivalence of Norms}
\label{app:en}
Here we'd like to show for any Hermitian operator $H$ on $\mathcal{H}_A\otimes\mathcal{H}_B$,
\begin{equation}
\frac{||H||_1}{\min(d_A,d_B)}\leq ||H||_{1,\Gamma}\leq \min(d_A,d_B)||H||_1,
\end{equation}
where $d_A:=\mathrm{dim}\mathcal{H}_A$ and $d_B:=\mathrm{dim}\mathcal{H}_B$. By Ref.~\cite{Vidal02}, for a state with Schmidt decomposition $\ket{\psi}=\sum_i \lambda_i \ket{a_i}\ket{b_i}$
\begin{equation}
||(\ket{\psi}\bra{\psi})^\Gamma||_{1}=||(\ket{\psi}\bra{\psi})||_{1,\Gamma}=\Big(\sum_i \lambda_i\Big)^2.
\end{equation}
This can be proved with the aid of a unitary operator as in Ref.~\cite{Vidal02}, or
\begin{equation}
[(\ket{\psi}\bra{\psi})^\Gamma]^2=\sum_i \lambda_i^2 \ket{a_i}\bra{a_i}\otimes\sum_j\lambda_j^2 \ket{b_j}\bra{b_j},
\end{equation}
so $|(\ket{\psi}\bra{\psi})^\Gamma|=\sum_{i,j} \lambda_i\lambda_j \ket{a_i}\bra{a_i}\otimes \ket{b_j}\bra{b_j}$ and $||(\ket{\psi}\bra{\psi})^\Gamma||_1=(\sum_i\lambda_i)^2.$

Because the Schmidt rank of a vector in $\mathcal{H}_A\otimes\mathcal{H}_B$ is at most $\min(d_A,d_B),$ by Schwarz inequality
\begin{equation}
\Big(\sum_i \lambda_i\Big)^2\leq \sum_{i}\lambda_i^2\sum_{j}1^2=\min(d_A,d_B).
\end{equation} 
Hence, for any Hermitian operator $H$ with spectral decomposition $H=\sum_i h_i \ket{\psi_i}\bra{\psi_i}$,
\begin{equation}
||H^\Gamma||_1=||H||_{1,\Gamma}\leq \sum_i |h_i|\,||(\ket{\psi_i}\bra{\psi_i})||_{1,\Gamma}\leq \min(d_A,d_B)\sum_i |h_i|=\min(d_A,d_B)||H||_1.
\end{equation}
Since $(H^\Gamma)^\Gamma=H$ (and because $\Gamma$ is HP), for the same reason $||H||_1\leq \min(d_A,d_B)||H^\Gamma||_{1}=\min(d_A,d_B)||H||_{1,\Gamma}$. 

For any HP mapping $L:\mathcal{B}(\mathcal{H}_1^A\otimes\mathcal{H}_1^B)\rightarrow\mathcal{B}(\mathcal{H}_2^A\otimes\mathcal{H}_2^B)$, $\mathscr{T}(L)\in \mathcal{B}(\mathcal{H}_1^A\otimes\mathcal{H}_1^B\otimes\mathcal{H}_2^A\otimes\mathcal{H}_2^B)$, $\mathscr{T}(L^\Gamma)=\mathscr{T}(L)^\Gamma$, c.f. \eqref{eq:Sgamma}. The Schmidt rank of a vector in $(\mathcal{H}_1^A\otimes \mathcal{H}_2^A)\otimes (\mathcal{H}_1^B\otimes \mathcal{H}_2^B)$ is at most $\min(d_1^A d_2^A,d_1^B d_2^B),$ where $d_i^j:=\mathrm{dim}\mathcal{H}_i^j$, so from the discussion above we have
\begin{equation}
\frac{||L||_1}{\min(d_1^A d_2^A,d_1^B d_2^B)}\leq ||L||_{1,\Gamma}\leq \min(d_1^A d_2^A,d_1^B d_2^B)||L||_1.
\end{equation}
For non-Hermitian operators and non-HP mappings, similar relations should exist.
\chapter{Isometry with Respect to Hilbert-Schmidt Inner Product}
\label{app:isHS}
Here we'll show isometry (see Section~\ref{sec:isometry}) of Choi isomorphism, transposition and partial transposition with respect to Hilbert-Schmidt inner product. I will consider linear mappings $\mathcal{B}(\mathcal{H}_1)\rightarrow \mathcal{B}(\mathcal{H}_2),$ and $\{E_{ij}\}$ is an orthonormal basis of $\mathcal{H}_1$, where $\mathcal{H}_i$ are finite-dimensional. Because of isometry and bijectivity, these linear mappings are actually unitary mappings, so they are Hilbert space isomorphisms (Section~\ref{sec:isometry}). 
\section{Choi Isomorphism}
By \eqref{eq:in2}, for two linear mappings $L_1,\,L_2:\mathcal{B}(\mathcal{H}_1)\rightarrow \mathcal{B}(\mathcal{H}_2)$ their Hilbert-Schmidt inner product is:\footnote{Assume the inner product adopted on $\mathcal{B}(\mathcal{H}_2)$ is also Hilbert-Schmidt inner product.}
\begin{align}
(\mathscr{T} (L_1)|\mathscr{T} (L_2))&=\Big(\sum_{i,j}E_{ij}\otimes L_1(E_{ij})\Big|\sum_{k,l}E_{kl}\otimes L_1(E_{kl})\Big)\nonumber\\
&=\sum_{i,j,k,l}(E_{ij}|E_{kl})(L_1(E_{ij})|L_2(E_{kl}))\nonumber\\
&=\sum_{i,j}(L_1(E_{ij})|L_2(E_{ij}))\nonumber\\
&=(L_1|L_2).
\end{align}
\section{Transposition and Partial Transposition}
\subsection{Operators}
First, it's easy to see
\begin{equation}
(E_{ij}\tr|E_{kl}\tr)=(E_{ij}|E_{kl}).
\end{equation}
Therefore, transposition on operators is isometric, by linearity (and conjugate linearity) of inner product and transposition. To show isometry of partial transposition on operators, in addition to (conjugate) linearity, the other property we need is \eqref{eq:OQ}:
\begin{equation}
((O_1\otimes Q_1)^\Gamma|(O_2\otimes Q_2)^\Gamma)=(O_1\tr|O_2\tr)(Q_1|Q_2)=(O_1\otimes Q_1|O_2\otimes Q_2).
\end{equation}
\subsection{Linear Mappings from Operators to Operators}
For any linear mapping $L:\mathcal{B}(\mathcal{H}_1^A\otimes\mathcal{H}_1^A)\rightarrow \mathcal{B}(\mathcal{H}_2^A\otimes\mathcal{H}_2^A)$, with $\{E_{ij}\}$ and  $\{F_{ij}\}$ being orthonormal bases of $\mathcal{B}(\mathcal{H}_1^A)$ and $\mathcal{B}(\mathcal{H}_1^B)$, we can find
\begin{align}
(L^\Gamma|L^\Gamma)&=\sum_{i,j,k,l} (L^\Gamma(E_{ij}\otimes F_{kl})|L^\Gamma(E_{ij}\otimes F_{kl}))\nonumber\\
&=\sum_{i,j,k,l} (L(E_{ji}\otimes F_{kl})^\Gamma|L(E_{ji}\otimes F_{kl})^\Gamma)\nonumber\\
&=\sum_{i,j,k,l} (L(E_{ji}\otimes F_{kl})|L(E_{ji}\otimes F_{kl}))\nonumber\\
&=||L||^2_{HS}.
\end{align}
Isometry of the transposition on such mappings can be shown in the same way. Or we can instead prove by $\mathscr{T}(L^\Gamma)=\mathscr{T}(L)^\Gamma$, the isometry of Choi isomorphism and that of (partial) transposition on operators.

Because of the isometry of partial transposition, Hilbert-Schmidt norm is insensitive to PPT-ness of operations and states: The norm doesn't change after partial transposition whether the operations/states are PPT or not, unlike trace norm. Additionally, as operations and states are normalized according to trace norm,\footnote{By \eqref{eq:d}, $||\mathscr{T}(S)||_1$ is a fixed number.} it's difficult to compare different operations or states with Hilbert-Schmidt norm. These are the reasons why Hilbert-Schmidt norm doesn't see much use in this thesis.

\chapter{${S^\Gamma_i}_-$ does not Necessarily Have the Smallest Operator Norm}
\label{app:smanorm}
For a Hermitian operator $H$, because there exist $\widetilde{H}^\pm$ such that $\widetilde{H}^\pm-H^\pm\ngeq 0$ (see Section~\ref{sec:eig}), for an operation $S$, $\widetilde{S^\Gamma_i}_--S{^\Gamma_i}_-$ is not always CP and $\widetilde{S^\Gamma_i}_-^\dagger(I)-{S^\Gamma_i}_-^\dagger(I)$ can be non-positive. Therefore it can happen that $||{S^\Gamma_i}_-^\dagger(I)||> ||\widetilde{S^\Gamma_i}_-^\dagger(I)||$; nevertheless, ${S^\Gamma_i}_-^\dagger(I)$ has the smallest trace (norm), because of Lemma~\ref{lem:or} and Corollary~\ref{cor:1}. How to minimize the operator norm under such decompositions may be an interesting mathematical problem. 

To show $||{S^\Gamma_i}_-^\dagger(I)||\nleq ||\widetilde{S^\Gamma_i}_-^\dagger(I)||$ in general, we consider a deterministic operation $S=pS_1+(1-p)S_2$, where $0\leq p\leq 1$ and $S_i$ are unitary operations. We choose $2\otimes 3$ and $3\otimes 3$ unitary operations, the former being 
\begin{equation}
\left(\begin{array}{c|c}
I_{3}&0_{3}\\\hline
0_{3}&\begin{matrix}
\cos\beta&\sin\beta&0\\
-\sin\beta&\cos\beta&0\\
0&0&1
\end{matrix}
\end{array}\right)\cdot
\left(\begin{array}{c|c}
I_{4}&0_{4}\\\hline
0_{4}&\begin{matrix}
\cos\alpha&\sin\alpha\\
-\sin\alpha&\cos\alpha&
\end{matrix}
\end{array}\right),
\end{equation}
and the latter being
\begin{equation}\left(\begin{array}{c|c}
I_{6}&0_{6\times 3}\\\hline
0_{3\times 6}&\begin{matrix}
\cos\beta&\sin\beta&0\\
-\sin\beta&\cos\beta&0\\
0&0&1
\end{matrix}
\end{array}\right)\cdot
\left(\begin{array}{c|c}
I_{7}&0_{7\times 2}\\\hline
0_{2\times 7}&\begin{matrix}
\cos\alpha&\sin\alpha\\
-\sin\alpha&\cos\alpha&
\end{matrix}
\end{array}\right),\label{eq:3x3}
\end{equation}
where $I_n$ are $n\times n$ identity matrices and $0_{m\times n}$ are $m\times n$ zero matrices. The result is illustrated by Fig.~\ref{fig:U}, where we can see $||{S^\Gamma}_-^\dagger(I)||> ||\widetilde{S^\Gamma}_-^\dagger(I)||$ with  $\widetilde{S^\Gamma}_-=p{S_1^\Gamma}_-+(1-p){S_2^\Gamma}_-$ in some cases.

\begin{figure}[hbtp!]
	\begin{subfigure}[b]{0.5\linewidth}
		\centering
		\includegraphics[width=\linewidth]{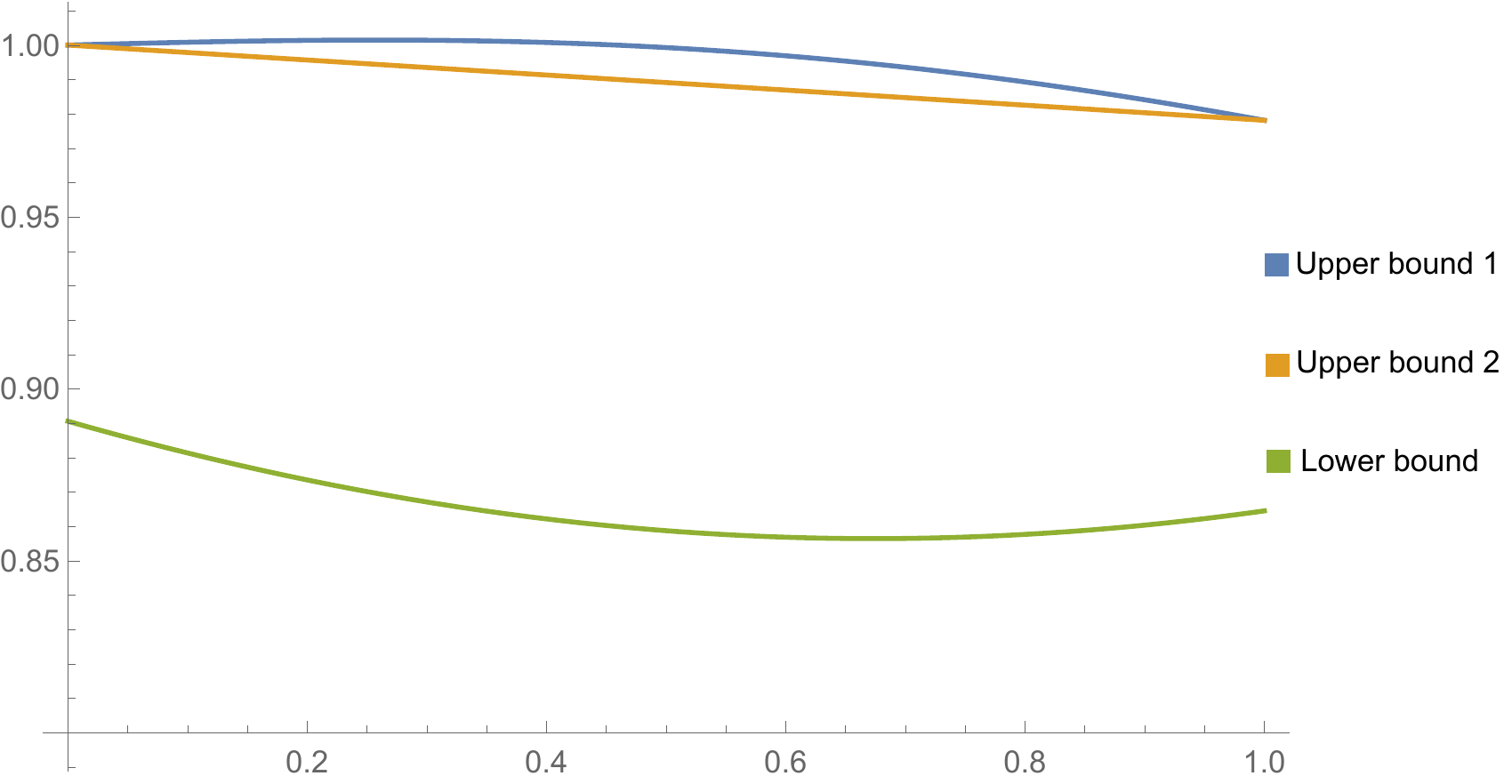}
		\caption{$2\otimes 3$ operations. Here upper bound 1 is larger than 2 when $0<p<1.$}\label{fig:23}
	\end{subfigure}
	\begin{subfigure}[b]{0.5\linewidth}
		\centering
		\includegraphics[width=\linewidth]{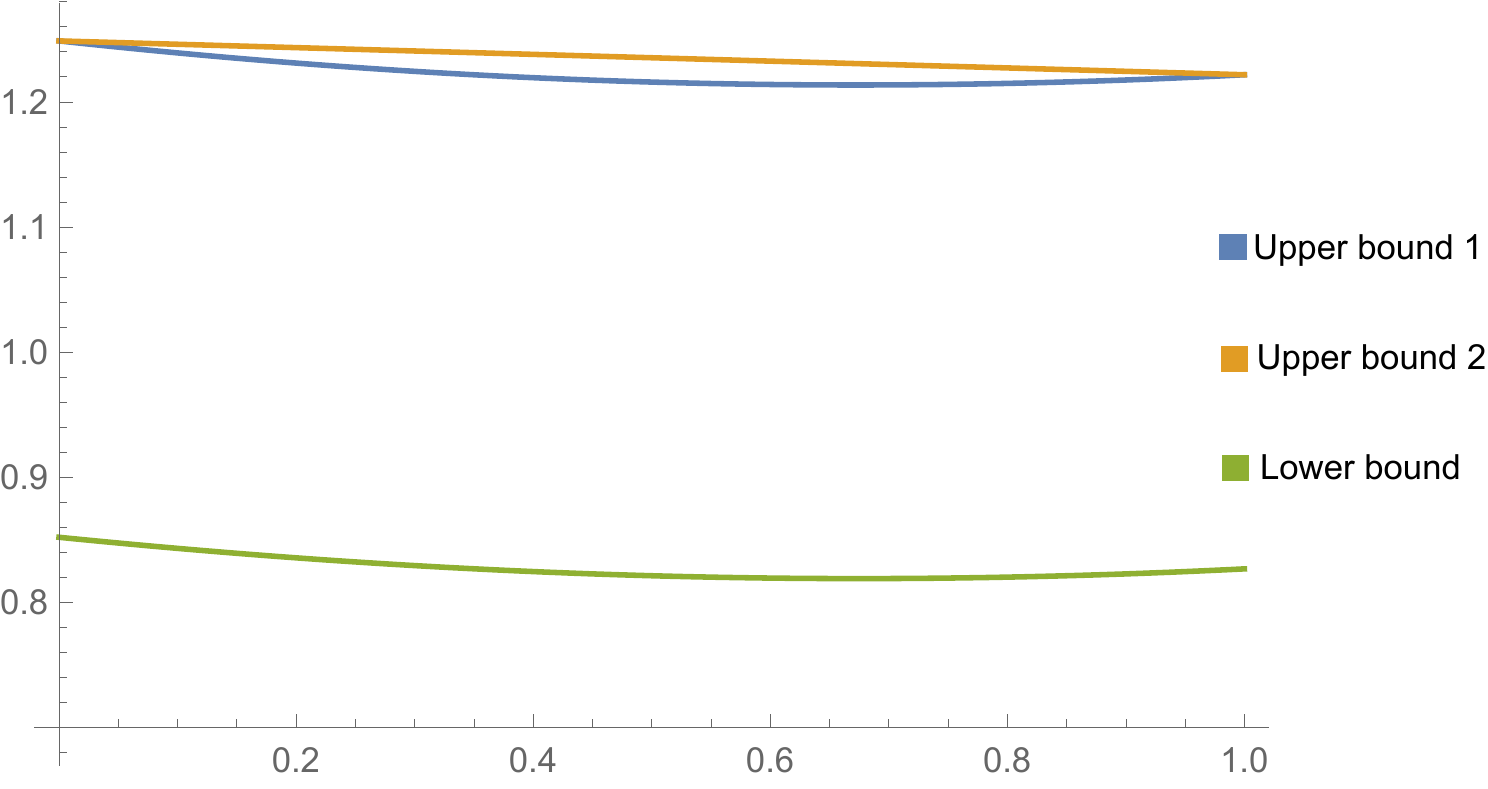}
		\caption{$3\otimes 3$ operations. Here upper bound 2 is larger than 1 when $0<p<1.$}\label{fig:33}
	\end{subfigure}
	\caption[Upper and lower bounds for $\text{EC}_L$]{Upper and lower bounds for $\text{EC}_L$ (base 2) of operations of the form $S(\rho)=p S_1+(1-p)S_2$ with unitary $S_i$, the lower bound accompanied as a reference. For both the $2\otimes 3$ and $3\otimes 3$ operations, $S_1$ has $\alpha=\pi/3$ and $\beta=\pi/5,$ and $S_2$ has for $S_2$, $\alpha=\pi/4$ and $\beta=\pi/3.$ Upper bound 1 refers to that obtained with $S^\Gamma_-$, and upper bound 2 with $p{S_1^\Gamma}_-+(1-p){S_2^\Gamma}_-$. Fig.~\ref{fig:23} clearly demonstrate that ${S^\Gamma_-}^\dagger(I)$ in general doesn't have the smallest operator norm among all $\widetilde{S^\Gamma}^\dagger_-(I)$.}
	\label{fig:U}
\end{figure}

\chapter{Schmidt Decomposition and Spectral decomposition}
\label{app:Sch}
The method from Ref.~\cite{Campbell10} used to show Lemma~\ref{lem:com} and Proposition~\ref{pro:sppt} can be further generalized as follows. Please read the proof in Section~\ref{sec:up} first, and refer to it as necessary.

For a (sub-)operation $S(A)=VAV^\dagger$ where the Schmidt decomposition of $V$ is $\sum_i \lambda_i A_i\otimes B_i,$ $S^\Gamma_\pm$ can be determined from the Schmidt decomposition of $V$, in particular \eqref{eq:VV}. Let's express \eqref{eq:SGa} as
\begin{align}
S^\Gamma_+(O)&:=\sum_i \lambda_i^2 V_{ii}^+ O {V_{ii}^+}^\dagger+\sum_{i< j} \lambda_i \lambda_j V_{ij}^+ O {V_{ij}^+}^\dagger\nonumber\\
S^\Gamma_-(O)&:=\sum_{i< j} \lambda_i \lambda_j V_{ij}^- O {V_{ij}^-}^\dagger.\label{eq:s+-o}
\end{align}
Because both $\{A_i\}$ and $\{B_i\}$ are composed of orthonormal operators, $(V_{ij}^\pm|V_{kl}^\pm)=\delta_{ik}\delta_{jl}$, and thus $\{V_{ij}^+: i\leq j\}$ and $\{V_{ij}^-:i<j\}$ are sets of orthonormal operators, and $V_{ij}^+\perp V_{kl}^-$. Therefore after Choi isomorphism, $\mathscr{T}(S^\Gamma)$ becomes an ensemble of orthogonal pure states, c.f. \eqref{eq:o3} and \eqref{eq:o4}, giving us a spectral decomposition of it. 

To find the eigenvalues and eigenvectors of $\mathscr{T}(S^\Gamma)$, because $\{V_{ij}^+:i\leq j\}$ and $\{V_{ij}^-:i<j\}$ are normalized and because of \eqref{eq:o4}, we can easily see $\mathscr{T}(S^\Gamma)$ has:
\begin{equation*}
\text{eigenvectors } I\otimes V_{ij}^\pm\ket{\Psi} \text{ with corresponding eigenvalues } \pm\lambda_i \lambda_j,\; i\leq j\text{ for }V_{ij}^+ \text{ and }i< j\text{ for }V_{ij}^-.\label{eq:Schei}
\end{equation*}
For example, for a unitary operation $S(A)=UAU^\dagger$ where $U=\sum_{i=1}^2 \lambda_i A_i\otimes B_i$ with Schmidt rank 2, $\mathscr{T}(S^\Gamma)$ has three positive eigenvalues and a negative one: $\lambda_1^2$, $\lambda_2^2$, and $\pm\lambda_1\lambda_2.$

In general, given a linear mapping $L$ of this form (c.f. \eqref{eq:sgf}):
\begin{equation}
L(O)=\sum_{i,j} A_i\otimes B_j O (A_j\otimes B_i)^\dagger,\; A_i\in\mathcal{B}(\mathcal{H}_1^A,\mathcal{H}_2^A),\;B_i\in\mathcal{B}(\mathcal{H}_1^B,\mathcal{H}_2^B),
\end{equation}
we can find a decomposition of $L$: $L=\widetilde{L}_+-\widetilde{L}_-$. If $\{A_i\}$ and $\{B_i\}$ are sets of orthogonal operators, then $\widetilde{L}_\pm=L_\pm,$ and we can determine the spectrum of $\mathscr{T}(L)$ by $\{A_i\}$ and $\{B_i\}$ alone.

Similarly, with an operator of the form
\begin{equation}
H=\sum_{i,j}\ket{v_j}\bra{v_i}\otimes\ket{w_i}\bra{w_j},
\end{equation}
where $\ket{v_i}$ and $\ket{w_i}$ aren't necessarily orthonormal, we can obtain
\begin{align}
\widetilde{H}^+&=\sum_{i} (\ket{v_i}\ket{w_i})(\bra{v_i}\bra{w_i})+\frac{1}{2}\sum_{i<j} (\ket{v_i}\ket{w_j}+\ket{v_j}\ket{w_i})(\bra{v_i}\bra{w_j}+\bra{v_j}\bra{w_i})\nonumber\\
\widetilde{H}^-&=\frac{1}{2}\sum_{i<j} (\ket{v_i}\ket{w_j}-\ket{v_j}\ket{w_i})(\bra{v_i}\bra{w_j}-\bra{v_j}\bra{w_i}).
\end{align}
If $\{\ket{v_i}\}$ and $\{\ket{v_i}\}$ are orthogonal sets of vectors, this gives us the spectral decomposition of $H$. This was utilized in Section~\ref{sec:pps}.
\chapter{Spectrum of a Symmetric Matrix with the Same Diagonal Entries and the Same Off-diagonal Ones}
\label{app:spe}
Consider a $n\times n$ matrix $M(n)$:
\begin{equation}
M(n)_{ij}=b+\delta_{ij}(a-b).
\end{equation}
It can be easily verified
\begin{equation}
\det M(2)=(a-b)(a+b)=(a-b)^{2-1}[a+(2-1)b].
\end{equation}
Therefore by mathematical induction, let's suppose 
\begin{equation}
\det M(n)=(a-b)^{n-1}[a+(n-1)b]\text{ for }n=2,3,\cdots, N-1.
\end{equation}
And
\begin{align}
\det M(N)&=\det\begin{pmatrix}
a     & b     & b    &\cdots& \cdots &\cdots & b\\
b     & a     & b    & b    & \cdots &\cdots & b\\
0     & b-a   & a-b  & 0    & 0      &\cdots & 0\\
0     & b-a   & 0    & a-b  & 0      &\cdots & 0\\
\vdots& \vdots& 0    & 0    & \ddots &\ddots &\vdots\\
\vdots&\vdots &\vdots&\vdots&\ddots  &\ddots &\vdots\\
0     & b-a   & 0    &0     &\cdots  &\cdots & a-b 
\end{pmatrix}\label{eq:det1}\\
&=\det\begin{pmatrix}
a-b   & b-a   & 0    & 0    & \cdots &\cdots & 0\\
b     & a     & b    & b    & \cdots &\cdots & b\\
0     & b-a   & a-b  & 0    & 0      &\cdots & 0\\
0     & b-a   & 0    & a-b  & 0      &\cdots & 0\\
\vdots& \vdots& 0    & 0    & \ddots &\ddots &\vdots\\
\vdots&\vdots &\vdots&\vdots&\ddots  &\ddots &\vdots\\
0     & b-a   & 0    &0     &\cdots  &\cdots & a-b 
\end{pmatrix}\label{eq:det2}
\end{align}
Note the determinants of the $(N-1)\times (N-1)$ submatrices in the bottom right corner of both \eqref{eq:det1} and \eqref{eq:det2} are $(a-b)^{N-2}[a+(N-2)b]$, because I didn't subtract the lower $N-1$ rows by the first row.\footnote{Same can be said of the $(N-1)\times (N-1)$ submatrices in the bottom left corner.} Hence
\begin{align}
\det M(N)&=(a-b)(a-b)^{N-2}[a+(N-2)b]-b \det \begin{pmatrix}
b-a   & 0    & 0    &\cdots &0\\
b-a   &a-b   & 0    &\cdots &0\\
\vdots& 0    &a-b   &\ddots &0\\
\vdots&\vdots&\ddots&\ddots &\vdots\\
b-a   &0     &0     &\cdots &a-b
\end{pmatrix}\nonumber\\
&=(a-b)(a-b)^{N-2}[a+(N-2)b]-b(b-a)(a-b)^{N-2}\nonumber\\
&=(a-b)^{N-1}[a+(N-1)b]
\end{align}
Therefore $\det M(n)=(a-b)^{n-1}[a+(n-1)b]$ for all integral $n\geq 2.$

To find the spectrum of $M(n)$ is to solve
\begin{equation}
\det(M(n)-\lambda I_n)=0
\end{equation}
$M(n)-\lambda I_n$ also has identical diagonal elements and identical off-diagonal ones, so 
\begin{equation}
\det(M(n)-\lambda I_n)=(a-\lambda-b)^{n-1}[a-\lambda+(n-1)b],
\end{equation}
and the eigenvalues are $a-b$ and $a+(n-1)b$, the degeneracy of the former one is $(n-1)$-fold. It's apparent that $(1,1,\cdots,1)$ is an eigenvector of $M(n)$, corresponding to the eigenvalue $a+(n-1)b$. 

\chapter{Covariances of a Multisymmetric System}
\label{app:covaiance}

Here we consider a multisymmetric system, consisting of families of symmetric modes. As a reminder, 
\begin{equation}
\mathbf{X}_i=O_i\mathbf{x}_i \text{ and } \mathbf{P}_i=O_i\mathbf{p}_i,
\end{equation}
where $O_i$ are orthogonal matrices with the first rows being $(1,1,\cdots,1)/\sqrt{N_i}$ and
\begin{equation}
\mathbf{X}_i=(X_i, u_{i,2},\cdots,u_{i,N_i}) \text{ and }\mathbf{P}_i=(P_i, \Pi_{i,2},\cdots,\Pi_{i,N_i}).
\end{equation}
In Ref.~\cite{Kao,Kao16} I used an operator approach to obtain a similar result, but here I instead will deal with the covariances directly, and treat $x$ and $p$ more as random variables than operators.

\section{First Moment}
The first moments, or means, don't show up often in our discussion in Chapter~\ref{ch:gaf} and \ref{ch:SGS}, as they can be removed by local unitary transforms, but they are included here for completeness. 

Since modes from the same family are identical, we immediately recognize
\begin{align}
\langle x_{i,j}\rangle&=\langle x_{i,1}\rangle=\dotsb=\langle x_{i,N_i}\rangle=\frac{\sqrt{N_i}}{N_i}\langle X_i\rangle=\frac{1}{\sqrt{N_i}}\langle X_i\rangle,\\
\langle p_{i,j}\rangle&=\langle p_{i,1}\rangle=\dotsb=\langle p_{i,N_i}\rangle=\frac{1}{\sqrt{N_i}}\langle P_{i,1}\rangle.
\end{align}
Since $O_i$ is an orthogonal matrix whose first row is $(1,1,\cdots,1)/\sqrt{N_i}$, by orthogonality
\begin{align}
\mean{u_{i,j}}&=\sum_{k=1}^{N_i} (O_i)_{jk}\mean{x_{i,k}}=\mean{x_{i,1}}\sum_{k=1}^{N_i} (O_i)_{jk}=0,\\
\mean{\Pi_{i,j}}&=\sum_{k=1}^{N_i} (O_i)_{jk}\mean{p_{i,k}}=0.
\end{align}

\section{Variances and Covariances}
\subsection{The Same Family}

From Appendix~\ref{app:spe}, we can see
\begin{align}
\sigma(X_i,X_i)&=\sigma(x_{i,j},x_{i,j})+(N_i-1)\sigma(x_{i,j},x_{i,k}),\,j\neq k,\\
\sigma(P_i,P_i)&=\sigma(p_{i,j},p_{i,j})+(N_i-1)\sigma(p_{i,j},p_{j,j}),\,j\neq k,
\end{align}
and
\begin{align}
\sigma(u_{i,j},u_{i,j})&=\sigma(x_{i,j},x_{i,j})-\sigma(x_{i,j},x_{i,k}),\,j\neq k,\\
\sigma(P_i,P_i)&=\sigma(p_{i,j},p_{i,j})-\sigma(p_{i,j},p_{i,k}),\,j\neq k,
\end{align}
from which we can obtain, for example
\begin{align}
N_i\sigma(x_{i,j},x_{i,j})&=\sigma(X_i,X_i)+(N_i-1)\sigma(u_{i,j},u_{i,j}),\\
N_i\sigma(x_{i,j},x_{i,k})&=\sigma(X_i,X_i)-\sigma(u_{i,j},u_{i,j}).
\end{align}

Next, let's investigate the covariances between position and momentum variables. Let $y$ denote $\sigma(x_{i,m},p_{i,m})$ and $z$ denote $\sigma(x_{i,m},p_{i,n})$ for $m\neq n$:
\begin{align}
\sigma\left((\mathbf{X}_i)_j,(\mathbf{P}_i)_k\right)&=\sum_{m,n}(O_i)_{jm}(O_i)_{kn}\sigma(x_{i,m},p_{i,n})\nonumber\\
&=y\sum_{m}(O_i)_{jm}(O_i)_{km}+z\sum_{\underset{m\neq n}{m,n}}(O_i)_{jm}(O_i)_{kn}\nonumber\\
&=y\delta_{jk}+z\sum_{\underset{m\neq n}{m,n}}(O_i)_{jm}(O_i)_{kn}.\label{eq:yz}
\end{align}
Because the rows of $O_i$ except for the first are orthogonal to $(1,1,\dotsc,1)$ when $j\neq 1$ or $k\neq 1$, the second summation of \eqref{eq:yz} vanishes. When $j=k=1$, $(O_i)_{1,m}=1/\sqrt{N_i}$ for all $m$, and because there are $N_i(N_i-1)$ terms, we have
\begin{equation}
\sigma\left((\mathbf{X}_i)_j,(\mathbf{P_i})_k\right)=
\begin{cases}
\sigma(x_{i,m},p_{i,m})+N_i(N_i-1) \sigma(x_{i,m},p_{i,n}),\,m\neq n & \text{when }j=k=1,\\
\delta_{jk} \sigma(x_{i,m},p_{i,m}) & \text{otherwise}
\end{cases}.
\end{equation} 
If using the standard form \eqref{eq:abc}, $\sigma\left((\mathbf{X}_i)_j,(\mathbf{P_i})_k\right)=0$, since $\sigma(x_{i,m},p_{i,n})=0$ for all $m,n$.

\subsection{Different Families}
In this subsection we will always assume $i\neq j$. We will just calculate the example below as a showcase:
\begin{align}
\sigma ((\mathbf{X}_i)_k,(\mathbf{X}_j)_l)&=\sum_{m,n} (O_i)_{km}(O_j)_{ln}\sigma(x_{i,m},x_{j,n})\nonumber\\
&=\sigma(x_{i,m},x_{j,n})\sum_{m} (O_i)_{km}\sum_n (O_j)_{ln}.
\end{align}
Following the discussion below \eqref{eq:yz}, we obtain
\begin{equation}
\sigma ((\mathbf{X}_i)_k,(\mathbf{X}_j)_l)=\begin{cases}
\sqrt{N_i N_j}\sigma(x_{i,m},x_{j,n}) & \text{when }i=j=1\\
0 & \text{otherwise}
\end{cases}.
\end{equation}
Similarly,
\begin{align}
\sigma ((\mathbf{X}_i)_k,(\mathbf{P}_j)_l)&=\begin{cases}
\sqrt{N_i N_j}\sigma(x_{i,m},p_{j,n}) & \text{when }i=j=1\\
0 & \text{otherwise}
\end{cases},\\
\sigma ((\mathbf{P}_i)_k,(\mathbf{P}_j)_l)&=\begin{cases}
\sqrt{N_i N_j}\sigma(p_{i,m},p_{j,n}) & \text{when }i=j=1\\
0 & \text{otherwise}
\end{cases}.
\end{align}

\chapter{Supremum for Symmetric Pure Gaussian States}
\label{app:entbos}
Here's an early effort of mine to find the supremum of negativity for symmetric Gaussian pure states, with some additional assumptions \cite{Kao16}. This was rendered redundant due to the more general result of Chapter~\ref{ch:SGS}, but the reader may find the approach here interesting, and maybe easier to comprehend.
\section{One Family of Symmetric Modes}
\label{subsec:entn1}
The symmetric Gaussian state is assumed to be pure:
\begin{equation}
\psi=\mathcal{N}\exp[-a\sum_{i=1}^{N}x_{i}^2+2b\sum_{i,j>i}^{N}x_{i}x_{j}].\label{eq:statei}
\end{equation}
If we assume that $a$ and $b$ in \eqref{eq:statei} are real, then with basic but lengthy algebras we can reduce it to
\begin{equation}
E_{L}^{N:1|1}=\begin{cases}
\frac{1}{2}\ln d & \text{if } d>1\\
0 & \text{otherswise}
\end{cases}, \text{ where }
d=\begin{cases}
\frac{a+b}{a-b} & \text{if }b\geq 0\\
\frac{a+b-b N}{a+3b-b N} & \text{otherwise}
\end{cases}.\label{eq:Ldab}
\end{equation}
The above equation clearly shows that when $b=0$ the entanglement measure vanishes, which corresponds to the fact that at $b=0$ the state is separable. Hence, we're more interested in circumstances where $b\neq 0$, so defining $r:= a/b$ we have
\begin{equation}
d=\begin{cases}
1+\frac{2}{r-1} & \text{if }b\geq 0\\
1+\frac{2}{N-r-3} & \text{otherwise}
\end{cases}.\label{eq:dabr}
\end{equation}
There are also constraints for the wave function \eqref{eq:statei} such that it's square-integrable:
\begin{equation}
a-(N-1)b>0, \;a+b>0, \;a>0;\text{ i.e.}\begin{cases}
r>N-1 & \text{if }b\geq 0\\
r<-1 & \text{otherwise}
\end{cases},\label{eq:con}
\end{equation}
taking into account $N\geq 2.$ Now let's discuss the value of $d$ and thus $E_{L}^{N:1|1}$ by the signs of $b$.
\subsection{$b<0$}
$d$ in this case monotonically increases with $r$, except for the discontinuity at $N-r-3=0,$ namely at $r=N-3\geq -1$. Under the constraint \eqref{eq:con}, $r$ never crosses this discontinuity, and thus the suprema of $d$ and $E_{L}^{N:1|1}$ are given by
\begin{equation}
\lim_{r\rightarrow -1^-} d=1+\frac{2}{N-2} \text{ so } \lim_{r\rightarrow -1^-}E_{L}^{N:1|1}=\frac{1}{2}\ln(1+\frac{2}{N-2}),\label{eq:supbl}
\end{equation}
because of the monotonicity of the logarithm. Now it's obvious that the maximum value of $E_{L}^{N:1|1}$ decreases as $N$ increases. In addition, at $N=2$ both $d$ and $E_{L}^{N:1|1}$ aren't bounded from above.

At large $N$, because $r<-1$, $1/(N-r-3)\ll 1$ for all possible $r$ and therefore we have
\begin{equation}
E_{L}^{N:1|1}\simeq \frac{1}{2}\frac{2}{N-r-3}\simeq \frac{1}{N-r}.\label{eq:nmr}
\end{equation}
Whenever $|r|\ll N$, $E_{L}^{N:1|1}\simeq 1/N.$ Hence, the larger $N$ is, the larger the state space where $E_{L}^{N:1|1}\simeq 1/N$ is.

\subsection{$b>0$}
$d$ now monotonically decreases as $r$ increases. Since the discontinuity happens at $r=1$, while the constraint \eqref{eq:con} requires $r>N-1\geq 1$, the suprema of $d$ and $E_{L}^{N:1|1}$ are
\begin{equation}
\lim_{r\rightarrow (N-1)^+} d=1+\frac{2}{N-2} \text{ and } \lim_{r\rightarrow (N-1)^+}E_{L}^{N:1|1}=\frac{1}{2}\ln(1+\frac{2}{N-2}),
\end{equation}
same as the suprema for $b<0,$ \eqref{eq:supbl}.

As \eqref{eq:Ldab} or \eqref{eq:dabr} shows, $d$ adopts different forms according to the sign of $b$. More interestingly while $d$ for $b>0$ doesn't depend on $N$, $d$ for $b<0$ does. At first glance there seems to be discontinuity, but as \eqref{eq:Ldab} shows $d$ under both circumstances approaches to $0$ as $b$ approaching 0, so the entanglement measure changes continuously as $b$ switches the sign.
\section{Two Families of Symmetric Modes}
\label{subsec:entn2}
Now the state is
\begin{equation}
\psi(\mathbf{x}_1,\mathbf{x}_2)=\mathcal{N}\exp[-\sum_{i=1}^{N_1}b_{11}x_{1,i}^2+2b_1\sum_{i>j}^{N_1}x_{1,i} x_{1,j}-\sum_{i=1}^{N_2}b_{22}x_{2,i}^2+2b_2\sum_{i>j}^{N_2}x_{2,i} x_{2,j}+c_{12}\sum_{i=1}^{N_1}\sum_{j=2}^{N_2}x_{1,i}x_{2,j}].
\end{equation}
We'll consider the negativity between two modes in the first family. After the transformation on the coordinate:
\begin{equation}
\psi(\mathbf{X}_1,\mathbf{X}_2)=\mathcal{N}\exp[-\epsilon_1 X_1^2-\epsilon_2 X_2^2+c_{12}X_1 X_2]\exp[-\zeta_1\sum_{i=1}^{N_1-1} u_{1,i}^2-\zeta_2\sum_{j=1}^{N_2-1}u_{2,j}^2],
\end{equation}
where $X_1$, $X_2$, $u_{1,i}$ and $u_{2,j}$ are as defined in Section~\ref{sec:or} and $\epsilon_i:= b_{ii}-(N_i-1)b_i$ and $u_{i,j}$ $\zeta_i:= b_{ii}+b_i.$ Again assuming $b_{ii}$, $b_i$, $c_{12}$ are all real (so are $\epsilon_i$ and $\zeta_i$), the logarithmic negativity $E_{L}^{N_1:1|1}$ of this state can be reduced to
\begin{equation}
E_{L}^{N:1|1}=\begin{cases}
\frac{1}{2}\ln d & \text{if } d>1\\
0 & \text{otherswise}
\end{cases}, \text{ where }
\frac{1}{d}=1+\frac{1}{N_1}(-2+\frac{\epsilon_1}{\zeta_1}-\frac{4\epsilon_2\zeta_1}{c_{12}^2-4\epsilon_1\epsilon_2})+\frac{\sqrt{(\epsilon_1 c_{12}^2-4\epsilon_1^2\epsilon_2+4\epsilon_2\zeta_1^2)^2}}{N_1\zeta_1(c_{12}^2-4\epsilon_1\epsilon_2)}.
\end{equation}
Therefore, the exact form of $d$ depends on the sign of
\begin{equation}
\kappa:=\epsilon_1 c_{12}^2-4\epsilon_1^2\epsilon_2+4\epsilon_2\zeta_1^2.
\end{equation}
\subsection{$\kappa>0$}
If it's larger than 0, then
\begin{equation}
d=\frac{\zeta_1 N_1}{2\epsilon_1+(N_1-2)\zeta_1}=1+\frac{2b_1}{b_{11}-b_1}.
\end{equation}
As before, there are constraint on these state parameters to make the wave function square integrable. Here, the constraints on $b_1$ and $b_{11}$ are
\begin{equation}
b_{11}>0, \,b_{11}+b_1>0 \text{ and }  b_{11}-(N_1-1)b_1>0.
\end{equation}
If $b_1\leq0,$ then $d\leq1$ because $b_{11}>0$ and there's no entanglement. Since here we want to find the upper bound of $E_{L}^{N:1|1},$ we don't have to take this into account. When $b_1>0$, we only have to consider $b_{11}-b_1>0$, because otherwise $d<1$.

Now assuming $b_{11}-b_1>0$ and $b_1>0$ for the reasons stated above, the last constraint $b_{11}-(N_1-1)b_1>0$ leads to
\begin{equation}
b_{11}-b_1>(N_1-2)b_1\Rightarrow \frac{b_1}{(N_1-2)b_1}=\frac{1}{N_1-2}>\frac{b_1}{b_{11}-b_1},
\end{equation}
and thus
\begin{equation}
d<1+\frac{2}{N_1-2}
\end{equation}
\subsection{$\kappa<0$}
If it's smaller than 0, then
\begin{equation}
d=(1-\frac{2}{N_1}+\frac{8\epsilon_2\zeta_1}{N_1(4\epsilon_1\epsilon_2-c_{12}^2)})^{-1}.
\end{equation}
To ensure the wave function is square integrable, we need the constraints
\begin{equation}
4\epsilon_1\epsilon_2-c_{12}^2>0,\,\epsilon_2>0\text{ and } \zeta_1>0,
\end{equation}
which implies
\begin{equation}
d<(1-\frac{2}{N_1})^{-1}=1+\frac{2}{N_1-2}.
\end{equation}

Because $d$ has the same upper bound whether $\kappa>0$ or $\kappa<0$, we get
\begin{equation}
d<1+\frac{2}{N_1-2} \text{ or } \sup E_{L}^{N:1|1}=\frac{1}{2}\ln(1+\frac{2}{N_1-2}).
\end{equation}

\end{document}